\documentclass[a4paper, 10pt, DIV = 13]{scrartcl}
\usepackage[utf8]{inputenc}
\usepackage[english]{babel}
\usepackage{amsmath}
\usepackage{amsthm}
\usepackage{amsfonts}
\usepackage{amssymb}
\usepackage{wasysym}
\usepackage{lmodern}
\usepackage{hyperref}
\usepackage{xcolor}
\usepackage{hyperref}
\usepackage[normalem]{ulem}
\hypersetup{breaklinks=true}
\usepackage{tikz}
\usepackage{minitoc}

\usepackage{mathtools}
\mathtoolsset{showonlyrefs}

\definecolor{darkred}{rgb}{0.5,0,0}
\definecolor{darkgreen}{rgb}{0,0.5,0}
\definecolor{darkblue}{rgb}{0,0,0.5}
\hypersetup{colorlinks,
linkcolor=darkblue,
filecolor=darkgreen,
urlcolor=darkred,
citecolor=darkblue}

\numberwithin{equation}{section}

\usepackage{enumitem}
\setlist{nosep}
\setlist{noitemsep}
\setlist{leftmargin=*}

\usepackage{mathtools}

\DeclareMathOperator{\Div}{\mathrm{div}}

\newtheorem{proposition}{Proposition}[section]
\newtheorem{lemma}[proposition]{Lemma}
\newtheorem{corollary}[proposition]{Corollary}

\newtheorem{claim}[proposition]{Claim}

\theoremstyle{definition}
\newtheorem{theorem}{Theorem}
\newtheorem*{theorem*}{Theorem}
\newtheorem{definition}[proposition]{Definition}
\newtheorem{remark}[proposition]{Remark}

\newcommand{\Z}{\mathbb{Z}}
\renewcommand{\epsilon}{\varepsilon}
\newcommand{\hal}{\frac{1}{2}}

\newcommand{\R}{\mathbb{R}}
\newcommand{\dd}{\mathtt{d}}
\newcommand{\sq}{\square}
\newcommand{\CC}{\mathrm{C}}
\newcommand{\Nn}{\mathcal{N}}
\newcommand{\nn}{\mathsf{n}}
\newcommand{\Dd}{\mathsf{D}}

\newcommand{\Vv}{\mathrm{V}}
\newcommand{\La}{\Lambda}

\newcommand{\ind}{\mathtt{1}}
\newcommand{\kk}{\mathsf{k}}
\newcommand{\ff}{\mathsf{f}}
\newcommand{\tphi}{\widetilde{\varphi}}

\DeclareMathOperator{\supp}{supp}
\DeclareMathOperator{\Id}{Id}
\DeclareMathOperator{\dist}{dist}
\DeclareMathOperator{\dive}{div}

\newcommand{\0}{\mathsf{0}}
\newcommand{\1}{\mathsf{1}}
\newcommand{\2}{\mathsf{2}}
\newcommand{\3}{\mathsf{3}}

\newcommand{\Conf}{\mathrm{Conf}}

\newcommand{\Fluct}{\mathrm{Fluct}}
\newcommand{\muc}{\mm_0}
\newcommand{\mus}{\mm_s}

\newcommand{\LN}{\Sigma_N}
\newcommand{\XN}{\mathrm{X}_N}
\newcommand{\bXN}{\mathbf{X}_N}
\newcommand{\X}{\mathrm{X}}
\newcommand{\bX}{\mathbf{X}}

\newcommand{\Di}{\mathsf{Dis}}
\newcommand{\Points}{\mathsf{Pts}}
\newcommand{\Elec}{\mathrm{E}}

\newcommand{\Elece}{\Elec_{\veta}}
\newcommand{\Lploc}{\mathrm{L}_{\mathrm{loc}}^p}

\newcommand{\feta}{\mathsf{f}_{\eta}}
\newcommand{\HH}{\mathfrak{h}}
\newcommand{\nHH}{\nabla \HH}

\newcommand{\mm}{\mathbf{m}}

\newcommand{\f}{\mathbf{f}}
\newcommand{\fN}{\f_N}

\newcommand{\KNbeta}{\mathrm{K}^\beta_{N}}
\newcommand{\PNbeta}{\mathbb{P}^\beta_{N}}

\newcommand{\E}{\mathbb{E}}
\newcommand{\EN}{\mathbb{E}^{\beta}_N}
\renewcommand{\P}{\mathbb{P}}

\newcommand{\F}{\mathsf{F}}

\newcommand{\FN}{\F_N}

\newcommand{\Fs}{F_s}
\newcommand{\Phis}{\Phi_s}
\newcommand{\vPhis}{\vec{\Phi}_s}
\newcommand{\psis}{\psi_s}
\newcommand{\rmax}{r_\mathrm{max}}

\newcommand{\Hone}{\mathrm{H}^1}

\newcommand{\Remain}{\mathrm{Rem}}

\newcommand{\G}{\mathsf{G}}

\newcommand{\rr}{\mathsf{r}}
\newcommand{\veta}{\vec{\eta}}
\newcommand{\vr}{\vec{\rr}}

\newcommand{\bCc}{\bar{\Cc}}

\newcommand{\Cond}{\Big |}

\newcommand{\loc}{\mathrm{loc}}

\newcommand{\EE}{\mathcal{E}}
\renewcommand{\AA}{\mathcal{A}}
\newcommand{\Lai}{\La_{i}}
\newcommand{\Laj}{\La_{j}}

\newcommand{\ext}{\mathrm{Ext}}
\newcommand{\ErrorCI}{\mathrm{ErrorCI}}
\newcommand{\Int}{\mathrm{Int}}
\newcommand{\tInt}{\widetilde{\Int}}

\newcommand{\Vext}{\Vv^\mathrm{ext}}
\newcommand{\Ext}{\mathrm{Ext}}
\newcommand{\Eext}{\mathcal{E}_\ext}
\newcommand{\Xext}{\mathrm{X}_\ext}
\newcommand{\nnext}{\nn_{\ext}}
\newcommand{\tn}{\nn}
\newcommand{\tni}{\nn_i}

\newcommand{\bXext}{\bX_{\ext}}
\newcommand{\bXi}{\bX_i}
\let\cXi\Xi
\renewcommand{\Xi}{\X_i}
\newcommand{\fLi}{\f_{\Lai}}
\newcommand{\fExt}{\f_{\Ext}}
\newcommand{\KNbetai}{\mathrm{K}^\beta_{\tni, \Lai, \Vext}}

\newcommand{\vu}{\vec{u}}
\newcommand{\dij}{\mathrm{d}_{ij}}

\newcommand{\Lab}{\Lambda^{\mathrm{bulk}}}

\newcommand{\Cc}{\mathtt{C}}
\newcommand{\CLL}{\hyperref[prop:local_laws]{\mathtt{C}_{\mathtt{LL}}}} 
\newcommand{\Cins}{\mathtt{C}_{\ref{prop:chase_the_goat}}}

\newcommand{\bchi}{\overline{\chi}}

\newcommand{\Vvi}{\Vv_{i}}
\newcommand{\tVi}{\widetilde{\Vv}_i}
\newcommand{\Dii}{\mathsf{D}_{i}}
\newcommand{\Djj}{\mathsf{D}_{j}}

\newcommand{\chir}{\chi_{\mathrm{uv}}}

\newcommand{\h}{\mathfrak{h}}

\renewcommand{\ss}{\mathsf{s}}
\newcommand{\Vback}{\Vv_{\mathrm{back}}}
\newcommand{\tFN}{\widetilde{\mathsf{F}}_N}

\newcommand{\EnerPoints}{\mathsf{EnerPts}}

\newcommand{\Fta}{\mathcal{F}_{\tau}}
\newcommand{\Gsi}{\mathcal{G}_{\sigma}}

\newcommand{\ErrorQI}{\mathsf{ErrorQI}}

\newcommand{\hT}{\hat{T}}

\newcommand{\eR}{\epsilon_R}
\newcommand{\eRR}{\eR \cdot R}

\newcommand{\Ener}{\mathrm{Ener}}

\newcommand{\ErrAv}{\mathrm{ErrAve}}

\newcommand{\Psimt}{\psi_{-t}}

\newcommand{\Ani}{\mathsf{Ani}}
\newcommand{\WS}{\mathsf{WellSpread}}
\newcommand{\PhiX}{\Phi \cdot}
\newcommand{\PhiS}{\Phi \# }

\newcommand{\ErrorScr}{\mathrm{ErrorScr}}
\newcommand{\bY}{\mathbf{Y}}
\newcommand{\Y}{\mathrm{Y}}
\newcommand{\heta}{\h_{\veta}}
\newcommand{\hetaX}{\heta^\X}
\newcommand{\hetaY}{\heta^\Y}
\newcommand{\nhetaX}{\nabla \hetaX}
\newcommand{\nhetaY}{\nabla \hetaY}
\newcommand{\nhh}{\nabla \hat{h}}

\newcommand{\Main}{\mathrm{Main}}
\newcommand{\Rem}{\mathrm{Rem}}
\newcommand{\dyi}{\delta_{y_i}^{(\eta_i)}}
\newcommand{\dhi}{\hat{\delta}_{y_i}}
\newcommand{\dyj}{\delta_{y_j}^{(\eta_j)}}
\newcommand{\dhj}{\hat{\delta}_{y_j}}
\newcommand{\tell}{\tilde{\ell}}

\newcommand{\hci}{\tilde{\h}_i^\Y}
\newcommand{\hciX}{\tilde{\h}_i^\X}

\renewcommand{\O}{\mathcal{O}}

\def\Xinttt#1{\mathchoice
{\XXinttt\displaystyle\textstyle{#1}}%
{\XXinttt\textstyle\scriptstyle{#1}}%
{\XXinttt\scriptstyle\scriptscriptstyle{#1}}%
{\XXinttt\scriptscriptstyle\scriptscriptstyle{#1}}%
\!\int}
\def\XXinttt#1#2#3{{\setbox0=\hbox{$#1{#2#3}{\int}$}
\vcenter{\hbox{$#2#3$}}\kern-.5\wd0}}

\def\dashint{\Xinttt-}

\renewcommand{\Int}{\mathrm{int}}
\renewcommand{\ext}{\mathrm{ext}}

\newcommand{\hrho}{\hat{\rho}}

\renewcommand{\Xext}{\X^{\ext}}

\newcommand{\VV}{\mathrm{V}}
\newcommand{\PP}{\mathcal{P}(\R^2)}
\newcommand{\muV}{\mu_{\VV}}
\newcommand{\SV}{S_\VV}
\newcommand{\SSV}{S^*_\VV}
\newcommand{\cV}{c_\VV}
\newcommand{\hnu}{\h^\nu}
\newcommand{\WW}{\mathrm{W}}
\newcommand{\nuloc}{|\nu|_{\mathrm{loc}}}
\newcommand{\muW}{\mu_\WW}
\newcommand{\IV}{\mathcal{I}_\VV}
\newcommand{\IW}{\mathcal{I}_\WW}
\newcommand{\hmuV}{\mathfrak{h}^{\muV}}
\newcommand{\uV}{u_{\VV}}

\newcommand{\Xint}{\X^{\mathrm{int}}}

\newcommand{\bXint}{\bX^{\mathrm{int}}}
\renewcommand{\bXext}{\bX^{\mathrm{ext}}}

\renewcommand{\Int}{\mathsf{Int}}
\renewcommand{\Ext}{\mathsf{Ext}}

\newcommand{\Var}{\mathrm{Var}}

\usepackage{csquotes}

\usepackage[backend=bibtex, style=alphabetic,giveninits,doi=false,isbn=false,url=false,maxalphanames=6, maxbibnames= 6]{biblatex}
\renewbibmacro{in:}{}
\usepackage{pgfplots}
\pgfplotsset{compat=1.16}

\newcommand{\DD}{\mathfrak{D}}

\newcommand{\boxs}{\mathfrak{B}}

\newcommand{\cor}{}
\newcommand{\corT}{}

\usepackage{authblk}

\begin{document} 
\renewcommand{\stctitle}{}
\author{\large{Thomas Leblé}}
\affil{\small{Université de Paris-Cité, CNRS, MAP5 UMR 8145, F-75006 Paris, France}}
\title{\LARGE{The two-dimensional one-component plasma is hyperuniform}}
\date{\small{\today}}
\maketitle

\vspace{-1cm}

\begin{abstract}
We prove that at all positive temperatures in the bulk of a classical two-dimensional one-component plasma (also called Coulomb or log-gas, or jellium) the variance of the number of particles in large disks grows (strictly) more slowly than the area. In other words the system is hyperuniform. 
\end{abstract}

\section{Introduction}
\label{sec:intro}
Let $N \geq 1$ be an integer, let $\LN$ be the disk of center $0$ and radius $\sqrt{\frac{N}{\pi}}$, let $\XN := (x_1, \dots, x_N)$ denote a $N$-tuple of points in $\LN$ and let $\bXN := \sum_{i=1}^N \delta_{x_i}$ be the associated atomic measure.  We let $\fN$ be the signed “fluctuation” measure on $\LN$ defined by $\fN := \bXN - \mm_0$, where $\mm_0$ is the Lebesgue measure on $\LN$. We think of $\bXN$ as a collection of \emph{point particles} in $\LN$ and of $\mm_0$ as the \emph{background measure}. We define the \textit{logarithmic interaction energy} $\FN(\bXN)$ as:
\begin{equation}
\label{def:FN}
\FN(\bXN) := \hal \iint_{(x,y) \in \LN \times \LN, x \neq y} - \log |x-y| \dd \fN(x) \dd \fN(y),
\end{equation}
\cor{if the points are distinct, and $+\infty$ otherwise.}

Let $\beta$ be a positive real number that will be fixed throughout the paper. We define a probability density $\PNbeta$ on the space of $N$-tuples of points in $\LN$ by setting:
\begin{equation}
\label{def:PNbeta}
\dd \PNbeta(\XN) := \frac{1}{\KNbeta} \exp\left( - \beta \FN(\bXN) \right) \dd \XN,
\end{equation}
with a normalizing constant $\KNbeta := \int_{\LN \times \dots \times \LN}  \exp\left( - \beta \FN(\bXN) \right) \dd \XN$ called the \emph{partition function}.
Here and below, $\dd \XN$ denotes the Lebesgue measure on the Cartesian product $\LN \times \dots \times \LN \subset (\R^2)^N$. The probability measure $\PNbeta$ is the \textit{canonical Gibbs measure of the two-dimensional one-component plasma} (2DOCP) at \textit{inverse temperature} $\beta$. We denote the expectation under $\PNbeta$ by $\EN$.

Let $\Omega$ some (Borel measurable) subset of $\LN$. We denote by $\Points(\bXN, \Omega)$ the number of points of $\bXN$ in $\Omega$ and by $\Di(\bXN, \Omega)$ the \textit{discrepancy} (of $\bXN$) in~$\Omega$, defined by:
\begin{equation}
\label{def:Discrepancy_Omega}
\Di(\bXN, \Omega) := \Points(\bXN, \Omega) - \mm_0(\Omega).
\end{equation}
By construction we have $\Di(\bXN, \LN) = 0$, which corresponds to the fact that the system is globally neutral. However it cannot be perfectly locally neutral, and $\Di(\bXN, \Omega)$ is meant to measure \textit{charge fluctuations} in $\Omega$. The \textit{number variance} in $\Omega$ is defined as the variance under $\PNbeta$ of $\Points(\bXN, \Omega)$, or equivalently of $\Di(\bXN, \Omega)$. 

\smallskip

For all $x$ in $\R^2$ and $R > 0$, we let $\DD(x, R)$ be the disk of center~$x$ and radius $R$.

\subsection{Main result}
Hyperuniformity of a system is defined by the fact that: 
\begin{quote}{\cite[Section 1]{torquato2018hyperuniform}}
\textit{(...) the number variance of particles within a spherical observation window of radius $R$ grows more slowly than the window volume in the large-$R$ limit.}
\end{quote}
Our main conclusion is the following:
\begin{theorem*}[The 2DOCP is hyperuniform]
\label{theo:main} For each $N \geq 1$, let $x = x(N)$ be a point in the bulk of $\LN$ and let $R = R(N)$ be such that $R(N) \to \infty$. Then the number variance in $\DD(x, R)$ is $o(R^2)$ as $N \to \infty$.
\end{theorem*}
We give a more precise statement in Section \ref{sec:strategy}, explaining what we mean by “the bulk of $\LN$” and specifying a quantitative upper bound on the number variance of the form $\O\left(\frac{R^2}{\log^{c}(R)}  \right)$ for some $c > 0$.

\subsubsection*{\cor{Context.}}
Hyperuniformity of the 2DOCP \textit{at all positive temperatures} is a forty year old prediction\footnote{Those results are rigorous to the extent that authors make use of so-called “clustering assumptions”, i.e. they assume properties of the two-point correlation function at large distances in order to derive certain identities (called “sum rules”) which, among other things, imply hyperuniformity. As explained in \cite{martin1988sum}: “\textit{The results obtained in this way are exact (i.e. do not follow from approximations), but not all of them are rigorously proven, in so far as some reasonable properties (e.g. the type of decay of the correlations) are assumed to hold a priori.}”. Unfortunately, obtaining mathematically rigorous statements about the large-distance properties of the 2DOCP's two-point correlation function (for $\beta \neq 2$) is extremely challenging.}
from statistical physics, see \cite{martin1980charge,lebowitz1983charge,martin1988sum,jancovici1993large,levesque2000charge} (which use a different terminology, \cor{as the term “hyperuniformity” had not been coined yet}), or \cite{torquato2018hyperuniform} again: \textit{“OCP fluid phases at [all] temperatures must always be hyperuniform”}. 

The full physical prediction says not only that the number variance is negligible with respect to the area of the disk, but that it should even be \emph{comparable to the perimeter} ($\O(R)$ and not only $o(R^2)$). 

We give here the first mathematical proof of hyperuniformity at all temperatures, however our upper bound on the number variance remains far from the conjectured sharp estimate.

\subsection{The 2DOCP}
The two-dimensional \emph{one-component plasma} is a well-studied model of statistical physics, \cor{also known as  Coulomb or log-gas, or jellium - we discuss terminology in Appendix \ref{sec:discussion_model}}. Besides the papers cited above, we refer to \cite[Chapter 15]{forrester2010log} for a presentation of exact and approximate results.

When defining the energy $\FN$ in \eqref{def:FN} and the canonical Gibbs measure $\PNbeta$ in \eqref{def:PNbeta}, we think of $(x_1, \dots, x_N)$ as the positions of point particles in $\LN$ which all carry the same electric charge $+1$, and which are immersed in a uniform neutralizing background of constant density on $\LN$. The logarithmic potential \cor{through which the “charges” interact} can then \cor{be thought of as} the electrostatic interaction potential in dimension $2$. 

Instead of placing a background measure, one sometimes \cor{defines a similar model by imposing} a confining potential \cor{felt by each particle, the most common choice being a quadratic potential}. \cor{We discuss this further in Appendix \ref{sec:discussion_model} and explain why our analysis would also carry through in that case}. \cor{Finally, note that we have chosen the radius of $\LN$ in such a way that the neutralizing background is exactly the Lebesgue measure - with a different choice, it should be normalized.}

\subsubsection*{Mathematical aspects of Coulomb gases}
As mathematical objects, Coulomb gases (under various forms: one-, two- or three-dimensional, with one or two components...) have attracted much interest. Concerning the 2DOCP itself, topics that have been investigated in the last years alone are very diverse and include: lower bounds on the minimal distance between points \cite{ameur2018repulsion}, concentration inequalities for the empirical measure of the particles (\cite{chafai2018concentration}, \cor{see also \cite{padilla2023concentration} for a sharper bound}), upper bounds for the local density of points \cite{lieb2019local}, generalizations to Riemannian manifolds (\cite{berman2019sharp}, \cite{garcia2019concentration}), Wegner's and clustering estimates \cite{Thoma:2022aa}... to quote only a few among many others. We refer to \cor{\cite{serfaty2024lectures} for a thorough presentation including} an overview of motivations, ranging from constructive approximation to the study of the Quantum Hall Effect, via random matrix theory. The long-range and singular nature of the pairwise potential raises considerable analytic challenges.

In \cite{MR3788208,MR4063572,serfaty2020gaussian} it was shown that the 2DOCP exhibits strong forms of rigidity at all scales and for all values of the temperature as far as fluctuations of \emph{smooth} linear statistics are concerned. Regarding fluctuations of charges, i.e. the statistics of an indicator function, \cite{armstrong2019local} (see also \cite{MR3719060,bauerschmidt2017local} for weaker controls) imply that the discrepancy within a disk of radius $R$ is $\O(R)$ with high probability. Our goal here is \cor{to} improve such bounds to $o(R)$ and to thus prove hyperuniformity.

\subsubsection*{The Ginibre ensemble}
When the inverse temperature parameter $\beta$ is equal to $2$, the probability measure $\PNbeta$ admits an interpretation as the joint law of the complex eigenvalues of a $N \times N$ non-Hermitian random matrix, known as the (finite) Ginibre ensemble, after \cite{ginibre1965statistical}. The model belongs to the class of \textit{determinantal processes} (see e.g. \cite[Section 4.3.7]{hough2006determinantal}) and is amenable to exact computations. 

In that specific case, the standard deviation of the number of points in a disk of radius $R$ is known: it scales like $R^{\hal}$ (in accordance with predictions from physics), see \cite{shirai2006large} and \cite{osada2008variance}. Thus for $\beta =2$ our result is not new, and our bound is far from the optimal one.

\subsubsection*{The hierarchical model}
In \cite{chatterjee2019rigidity}, S. Chatterjee studied Coulomb gases in dimensions $\{1 ,2, 3\}$ (see \cite{ganguly2020ground} for an extension to higher dimensions), giving sharp estimates on the number variance up to logarithmic corrections. His model is a \textit{hierarchical Coulomb gas}, where the physical space is decomposed into a tree-like structure following an old recipe of Dyson \cor{(see \cite[Sec. 3]{dyson1969existence})}. The hierarchical model has two built-in properties: conditional independence of sub-systems and a self-similar nature. A large part of the present analysis is devoted to finding approximate analogues of those features for the “true” (non-hierarchical) model. \cor{We also mention the recent paper \cite{nishry2024large} which studies charge fluctuations for the hierarchical model, as discussed in Section \ref{sec:JLM} below.}

\subsection{Hyperuniformity}
The term “hyperuniform(ity)” has been coined in the theoretical chemistry literature by S. Torquato (see \cite{torquato2003local,torquato2018hyperuniform} for surveys), an alternative terminology due to J. Lebowitz is “superhomogeneous/ity”. \cor{A closely related notion is \emph{incompressibility}, which consists in finding \emph{upper bounds} on the local number of points, see e.g. \cite{lieb2019local} or \cite[Thm. 5]{Thoma:2022aa}.}

\cor{Hyperuniformity can (often) be phrased as the identity:
\begin{equation*}
\int_{\R^d} \left(\rho_2(x) - 1 \right) \dd x = -1,
\end{equation*}
where $\rho_2$ is the two-point correlation function of the process (provided it exists and the integral makes sense). This is an instance of what is called a “sum rule” in the physics literature, see e.g. \cite{gruber1980equilibrium,martin1988sum}. Another possible way to phrase this is $S(0) = 0$, where $S$ is the so-called “structure factor”.}

As we already mentioned, a system is hyperuniform when the number variance in a large ball is asymptotically negligible with respect to the volume of said ball. Of course the definition of hyperuniformity needs an adaptation for finite systems, as one should take both the size of the system and of the “spherical window” to infinity. We refer to \cite{Coste} for a mathematical presentation of hyperuniformity and a survey of related results.

\subsubsection*{Examples of hyperuniform systems}
For non-interacting random particles (forming a Bernoulli point process or a Poisson point process with constant intensity) the number variance grows exactly like the volume, hence hyperuniformity can be thought of as the property of systems that are “much more rigid” than independent particles regarding discrepancy at large scales. Examples of two-dimensional objects that are proven to be hyperuniform include: the Ginibre ensemble (as mentioned above) and its generalizations (see \cite{charles2019entanglement} for a recent result in connection with the Quantum Hall Effect), averaged and perturbed lattices (see \cite{gacs1975problem}), the zeroes of random polynomials with i.i.d Gaussian coefficients (see \cite{forrester1999exact,shiffman2008number})...

In some cases (see \cite[Section 5.3.2]{torquato2018hyperuniform}) the number variance grows as the perimeter (in fact that was initially chosen as the definition of hyperuniformity in \cite{torquato2003local} but is now called “class I hyperuniformity”), which is the slowest possible growth (see \cite{beck1987irregularities}). The 2DOCP is predicted to be class I hyperuniform \cor{in the bulk} at all temperatures, with good tails on the probability of large charge fluctuations, as we explain next.

\subsubsection*{The Jancovici-Lebowitz-Manificat law}
\label{sec:JLM}
In \cite{jancovici1993large}, Jancovici, Lebowitz and Manificat made precise predictions concerning the probability of observing large charge fluctuations within the 2DOCP (and its three-dimensional version). Their statement is significantly more precise\footnote{Tail estimates as in \eqref{JLM} readily imply the strongest type of hyperuniformity (Type I).} than hyperuniformity, as they argue that for all $\alpha > \hal$:
\begin{equation}
\label{JLM}
\P \left[ \text{The discrepancy \cor{in a disk of radius $R$ is larger than $R^\alpha$}} \right] \sim \exp(-R^{\varphi(\alpha)}) \quad \textbf{(“JLM law”)}
\end{equation}
where the rate $\varphi(\alpha) > 0$ is an explicit piecewise affine function of $\alpha$. This was later checked for $\beta =2$ through explicit computations, see \cite{shirai2006large,fenzl2020precise}, while the general case is open. Remarkably, although the original prediction of \cite{jancovici1993large} deals with Coulomb gases, it was first verified in \cite{NSV} for a different model, namely the zeros of the Gaussian Entire Function (see \cite{ghosh2018point} for a survey of related results).

\cor{Concerning the two-dimensional one-component Coulomb gas, partial results have recently been obtained: the full conjecture has been proven in \cite{nishry2024large} for the \emph{hierarchical} model, while \cite[Thm. 1]{Thoma:2022aa} confirms the JLM law (in the non-hierarchical model) but only for the “higher” values of $\alpha$.}

\subsubsection*{Fekete points and the low temperature regime}
A situation of interest that we do not consider here is the case of energy minimizers, which formally corresponds to taking $\beta = + \infty$ i.e. a zero temperature. A long-standing conjecture posits that minimizers of $\FN$ form a lattice as $N \to \infty$ and in that regard Fekete points should exhibit excellent rigidity properties at finite $N$. For studies of such questions we refer to \cite{ameur2012beurling,ameur2017density} and, in a different line of work, to \cite[Thm.~3]{nodari2015renormalized}. See also \cite{ameur2020planar,Marceca:2022aa} or the $\beta$-dependent statements of \cite{armstrong2019local} for a more general “low-temperature” regime where $\beta$ is large but not necessarily infinite.

\subsection{Open questions}
Besides the obvious natural goal of obtaining the sharp number variance estimate as well as proving the aforementioned “JLM laws”, let us mention three other open directions of research.

\subsubsection*{The 3DOCP.}
The claims of \cite{jancovici1993large} are made in dimension $2$ \emph{and} $3$, moreover sharp hyperuniformity estimates were obtained for the \emph{hierarchical} model in dimension $3$ in \cite{chatterjee2019rigidity}, it is thus natural to ask whether one can prove hyperuniformity for the 3DOCP.

Unfortunately, in dimension $3$ there is no (known) value of $\beta$ for which the model would be integrable and predictions could be tested. Some understanding of fluctuations of smooth linear statistics was recently obtained in \cite{serfaty2020gaussian} but with considerable new difficulties compared to the two-dimensional case.

\subsubsection*{Riesz interactions.} 
The 2DOCP can be seen as a member of one-parameter family of two-dimensional systems called \emph{Riesz gases} for which the interaction potential is taken as a certain power $|x-y|^{-\ss}$ of the distance (the case $\ss = 0$ corresponds by convention to the logarithmic interaction). Rigidity properties of those systems depend on the value of $\ss$ as it governs both the level of repulsiveness at $0$ and, most importantly, the long- or short-range nature of the interaction. 

\cor{Discrepancy bounds were given in \cite[Lemma 2.2]{petrache2017next} (finite-volume) and \cite[Lemma 3.2]{MR3735628} (infinite-volume), expressed purely in terms of the (average) Riesz energy of a (random) point configuration. In particular, if a stationary random point configuration $\bX$ has “finite Riesz energy”, then it is known that: 
\begin{equation*}
\E[ \Di(\bX, \DD(0,R))^2 ] = \O(R^{2 + \ss}) \text{ as } R \to \infty.
\end{equation*}
We will not define Riesz energy of an random infinite point configuration here, we refer instead to the articles cited above or to \cite{leble2016logarithmic} for an alternative take which does not use the “electric” formalism.}

\cor{A fairly wild guess, based on the empirical observation that the true size of fluctuations is often one order of magnitude smaller than energy-based bounds, would be that the number variance under the canonical Gibbs measure of a $2d$ Riesz gas is $\O\left(R^{1 + \ss}\right)$, but this is very speculative.}

\cor{The rigidity of one-dimensional Riesz systems was recently studied in \cite{boursier2021optimal}, which shows, in particular, that the number variance in an interval of length $R$ is $\O(R^{\ss})$. In other words, $1d$ Riesz gases are hyperuniform for $\ss \in [0, 1)$. This is consistent with the idea that when $\ss > 1$, one leaves the realm of “long-range” systems and that “short-range” systems cannot be hyperuniform, see \cite{dereudre2024non} for a recent result in that direction.}

\subsubsection*{High temperature regime.}
There has been interest in studying the “high-temperature” regime for the 2DOCP, i.e. when $\beta \to 0$ as $N \to \infty$.

\cor{The “local laws” paper \cite{armstrong2019local} allows to take the inverse temperature $\beta$ depending on $N$. Their results puts forward the fact that if $\beta = \beta(N) \to 0$, then there is a characteristic lengthscale $\rho_\beta \sim C \log (\beta^{-\hal}) \beta^{-\hal}$ tending to $+ \infty$ such that the “good behavior” of the system (expressed in terms of the local laws) can only be shown \emph{at scales larger than $\rho_\beta$}. In other words, one could expect the system to be “disordered” at scales $\ell \ll \rho_\beta$, and the good rigidity properties to only arise at scales $\ell \gg \rho_\beta$.}

\cor{Other results (see e.g. \cite{akemann2019high,padilla2020large} for the Coulomb gas, \cite{Lambert_2021} for general interacting systems at high temperature, \cite{hardy2021clt} for the $1d$ log-gas) indicate that the inverse temperature regime where $\beta$ is of order $N^{-1}$ is the threshold that should distinguish between “Ginibre-like” and “Poisson-like” properties.}

In view of those results, it is tempting to ask whether the 2DOCP stays hyperuniform as long as $\beta N \to \infty$,  and whether one can observe an interesting transition along $\beta = \frac{c}{N}$, as $c$ moves down from $+ \infty$ to $0$. \cor{Moreover, it seems plausible than any good discrepancy estimate should require the lengthscale $R$ to be larger than $\rho_\beta$ as introduced in \cite{armstrong2019local}, in other words there would be a lower bound on the speed at which we take $R = R(N)$ depending on $N$.}

\medskip

\subsection{Precise statement of the result, proof strategy and tools}
\label{sec:strategy}
\cor{Our goal is to prove hyperuniformity of the 2DOCP, which (by definition) is an asymptotic result on the number variance in large disks. The main result of the paper is in fact a \emph{non-asymptotic} tail bound on the probability of having a “large” discrepancy (here, specifically, we can handle discrepancies larger than $R (\log R)^{-0.3}$) in a disk $\DD(x,R)$ of radius $R$. The two important features of this bound are:
\begin{enumerate}
	\item It is valid for all $N$, for all $R$ and for all $x$ (in particular one can take $x$ and $R$ depending on $N$) \emph{provided} certain conditions are verified: $x$ must be at a certain distance from the edge (see \eqref{eq:assumption_distance} below) and $N, R$ must be large enough depending on $\beta$ and on that distance.
	\item The tail estimate is \emph{sub-algebraic} i.e. it decays faster than any polynomial in $R$. \\
	In particular, it can thus be used to control the second moment of $\Di(\bXN, \DD(x, R))$ via the usual “layer-cake” representation, \emph{yielding a $o(R^2)$ bound, which proves hyperuniformity}.
\end{enumerate}
}

\begin{theorem}[Non-asymptotic tail probability bound on discrepancies in large disks] 
\label{theo:main2}
Let $\delta > 0$ be fixed. For all $N$ and $R$ large enough (both depending on $\beta$ and $\delta$), for all $x$ in $\LN$ such that $x$ is “in the bulk” in the following sense:
\begin{equation}
\label{eq:assumption_distance}
\dist(\DD(x, R), \partial \LN) \geq \delta \sqrt{N}, 
\end{equation}
we have:
\begin{equation}
\label{eq:resultat}
\PNbeta\left(\left\lbrace \left|\Di(\bXN, \DD(x, R))\right| \geq R (\log R)^{-0.3} \right\rbrace   \right) \leq \exp\left( - \log^{1.5} R \right).
\end{equation}
\end{theorem}
Since good exponential tails for $\Di$ at values higher than $R$ are already known (see e.g. \cite[Thm~1]{armstrong2019local}), this new sub-algebraic tail valid for values of the discrepancy between $R \log^{-0.3} R$ and $R$ \cor{provides the crucial improvement to go from $\O(R^2)$ to $o(R^2)$ in our controls on the number variance}, and it implies the desired asymptotic result:
\begin{corollary}[Hyperuniformity of the 2DOCP]
\cor{For each $N \geq 2$, let $x = x(N)$ be a point in $\LN$ and let $R = R(N)$ be a positive radius. Assume that we are taking “large disks” in the sense:
\begin{equation*}
\lim_{N \to \infty} R(N) = + \infty,
\end{equation*}
and that the disks that we are considering are located “in the bulk” in the sense:
\begin{equation*}
\liminf_{N \to \infty} \frac{1}{\sqrt{N}} \dist(\DD(x(N), R(N)), \partial \LN) > 0
\end{equation*}
Then we have, as $N \to \infty$:}
\begin{equation}
\label{eq:HU_main}
\Var[\Di(\bXN, \DD(x,R)) ] \leq \EN \left[ \Di^2(\bXN, \DD(x, R))\right] = \O\left( \frac{R^2}{\log^{0.6} R} \right) = o(R^2).
\end{equation}
\end{corollary}
The proof will be given in Section \ref{sec:conclusion}. \corT{As discussed there, with the current strategy, there is a serious technical obstruction to obtaining better than $\frac{R^2}{\log R}$ in the right-hand side, with this in mind the exponent $0.6$ is somewhat arbitrary and has not been optimized - we preferred to keep a lighter presentation of the proof of the main theorem by choosing some “nice” constants by hand.}

\subsubsection*{General strategy}
Our main source of inspiration is \cite[Section 4]{NSV}, in which Nazarov-Sodin-Volberg prove JLM-like estimates on the probability of having large “charge” fluctuations in the disk $\DD(0,R)$ \emph{for zeros of the Gaussian Entire Function}. Their argument can be roughly summarized as follows: fix $\alpha \in (\hal, 1)$ and assume there is a discrepancy of size $R^\alpha$ in the disk $\DD(0,R)$.
\begin{enumerate}[leftmargin=*]
	\item Show that the discrepancy can be captured along the boundary $\partial \DD(0,R)$.
	\item Split that boundary into $\approx R$ pieces of size $\approx 1$. Take $M$ large and apply a basic pigeonhole argument: there exists a family of $\approx \frac{R}{M}$ pieces that capture a discrepancy of size at least $\frac{R^{\alpha}}{M}$ and which are “well-separated” (distances between pieces are multiples of $M$).
	\item Then comes the main probabilistic work:
	\begin{enumerate}
		\item Show that these well-separated pieces are \emph{approximately independent}.
		\item Show that the discrepancy on each piece (of size $\approx 1$) is typically $\O(1)$.
		\item Show that the discrepancy on each piece is centered.
	\end{enumerate}
	\item Apply Bernstein's concentration inequality: if $\{D_i\}_i$ is a family of $\approx \frac{R}{M}$ independent centered random variables of size $\O(1)$ then:
	\begin{equation*}
	\P \left( \left\lbrace \sum_i D_i \geq R^{\alpha}  \right\rbrace \right) \lesssim \exp\left(- \frac{R^{2\alpha -1}}{M}\right).
	\end{equation*} 
\end{enumerate}
This gives them \emph{sharp} bounds on the number variance \emph{and} the correct tail probabilities. Whilst not aiming so high, we follow here a similar strategy, with much effort to translate it to the context of the 2DOCP (besides point (2) and (4), which are easy).

\subsubsection*{Electric approach, sub-systems}
The technical core is the “electric approach” to Coulomb gases as developed by S. Serfaty and co-authors. In that regard, our main imports are:
\begin{itemize}[leftmargin=*]
\item The general spirit of controlling fluctuations through the electric energy (see Section \ref{sec:logEnergy}).
\item Local laws (up to microscopic scale) as in \cite{armstrong2019local} (see Proposition \ref{prop:local_laws} and Proposition \ref{prop:local_law_SUBSYS}).
\item Optimal bounds for fluctuations of smooth linear statistics as in \cite{MR3788208,serfaty2020gaussian,MR4063572}. Point (1), for example follows fairly easily from such bounds.
\item “Smallness of the anisotropy”, for which we refer to \cite{MR3788208,serfaty2020gaussian} (see also \cite{MR4063572}).
\end{itemize}
In addition, we put forward the role of so-called “sub-systems”, which arise as restrictions of the full system to some region and which we view as two-dimensional Coulomb systems in their own right, possibly with a small global non-neutrality, and (most importantly) feeling the effect of some harmonic exterior potential. As such, this object is not new - it has been sometimes called a “conditional” or “local” measure in the literature (see e.g. \cite{MR4063572,MR3192527}). 

We provide here a thorough study of its behavior, through a global law \cor{(Proposition \ref{prop:global_law_SUBSYS})} and local laws \cor{(Proposition \ref{prop:local_law_SUBSYS})}, which show that with high probability the sub-systems, although under the influence of an external potential, retain most of the rigidity properties of the full system - as is the case “for free” in the hierarchical model studied in \cite{chatterjee2019rigidity}. This requires a precise study of the external potential “felt” by a sub-system \cor{(Proposition \ref{prop:VextIsOftenGood})} and \cor{(as a technical input)} of the perturbation of the equilibrium measure that it induces within the sub-system, for which we rely on the analysis of \cite[Section 2.]{bauerschmidt2017local}, with some modifications. \cor{To summarize informally our analysis of those “sub-systems”: fix some disk $\La$ in the bulk of $\LN$, and for a given configuration $\bXext$ outside $\La$, consider the law of $\bXN \cap \La$ knowing that $\bXN \cap \bar{\La} = \bXext$. Then:
\begin{itemize}
	\item This law can be seen as the Gibbs measure of a “generalized” 2ODCP, either with some external potential (Section \ref{sec:with_external}) or some perturbed background measure (Section \ref{sec:PnLV}).
	\item With high $\PNbeta$-probability (Proposition \ref{prop:VextIsOftenGood}), the exterior configuration $\bXext$ is such that the law of $\bXN \cap \La$ knowing that $\bXN \cap \bar{\La} = \bXext$ enjoys the same good properties (Sections \ref{sec:good_properties} and \ref{conseqSubSys}) as the original Gibbs measure $\PNbeta$. In particular, local laws hold.
\end{itemize}
}

The fact that the local laws \cite{armstrong2019local,MR3719060} extend to sub-systems allows one to treat Point (3) (b), which seem unsurprising, but requires to consider objects of size $\sim 1$, hence to control the system \emph{down to the smallest scales}. Such local laws are also crucial for the approximate translation-invariance argument which we present now.

\subsubsection*{Approximate translation-invariance}
Step (3) (c) is void in the context of \cite{NSV} because the underlying point process is infinite and translation-invariant (which implies that linear statistics are centered). However, it turns out to be a major roadblock when adapting the proof to a Coulomb gas context. Indeed, in sub-systems (or even in the full system if one does not impose artificial boundary conditions) there is absolutely no obvious translation-invariance. We introduce an “approximate translation-invariance” result, valid for both the full system and sub-systems, which is crucial for Step (3) (c). The point being that shifting a function (or more precisely, averaging over shifts) acts as a mollification and enables us to compare the expectation of a discrepancy, or of any non-smooth linear statistics, to that of a smoother one.

There is a series of results in mathematical statistical mechanics \emph{à la} H.-O. Georgii devoted to prove translation invariance of infinite volume Gibbs measures in contexts where stationarity is not built-in, see e.g. \cite{MR1886241,richthammer2007translation}, following earlier works by Fröhlich-Pfister \cite{frohlich1981absence,frohlich1986absence}. The basic idea is to construct suitable “localized translations” in the form of diffeomorphisms acting as a given translation in a large box while leaving the majority of the system unchanged, and to control the effect on the energy of such changes of variables. 

Following the wisdom found in a remark of \cite[Chap. 3, Sec.~III.7]{simon2014statistical}, we construct a localized translation that varies very slowly (in terms of its $H^1$ norm), and a careful revisit of a computation done in \cite{serfaty2020gaussian}, together with Serfaty's “smallness of anisotropy” trick, allows us to conclude that this localized translation can be chosen to have an arbitrarily small effect on the energy. 

\cor{Our most general result on approximate translation-invariance is given Proposition \ref{prop:Quantitative_invariance}. An easy-to-state consequence is the following:
\begin{corollary}
Let $\G$ be a bounded function on the space $\Conf$ of point configurations, which is local in the sense that $\G(\bX)$ depends only on the restriction of $\bX$ to some bounded set of $\R^2$, and let $\vec{v} \in \R^2$.
\begin{equation*}
\E \left[ \G(\bXN) \right] = \hal \left( \E\left[ \G(\bXN + \vec{v}) \right]  + \E\left[ \G(\bXN - \vec{v}) \right]  \right) + o_N(1).
\end{equation*}
\end{corollary}
In particular, if $\bXN$ converges in distribution (up to extraction) to some random infinite point configuration $\bX$, we have:
\begin{equation*}
\E \left[ \G(\bX) \right] = \hal \left( \E\left[ \G(\bX + \vec{v}) \right]  + \E\left[ \G(\bX - \vec{v}) \right]  \right),
\end{equation*}
which implies that $\bX$ and $\bX + \vec{v}$ have the same distribution, and since $\vec{v}$ is arbitrary we deduce that \emph{every limit point of the law of $\bXN$ is translation-invariant} (see \cite[Thm 3.]{leble2024dlr} for a slightly different proof).
}

\subsubsection*{Approximate independence}
Finally, one needs to find some kind of independence between the sub-systems. We introduce a new “approximate independence” argument (which we believe to be of interest on its own) for sub-systems that are well-separated, conditionally on the \emph{number of points} in each of them. The simple idea is that two domains $\Lai, \Laj$ with $\nn_i, \nn_j$ points contribute an interaction energy given to first order by: 
  \begin{equation*}
  -\left(\nn_i - |\Lai|\right)\left(\nn_j - |\Laj|\right) \log \dist(\Lai, \Laj),
  \end{equation*}
while the precise arrangements of the points inside each domain should matter only to a lower order. \cor{Our main result in that regard is given in Proposition \ref{prop:CI}.}

\cor{Despite our phrasing, one should not think of distant subsystems as being truly independent. What our analysis shows, however is that \emph{conditionally on the number of points in each subsystem}, events that depend on several distant subsystems and would be very unlikely for independent subsystems are also very unlikely for the true Gibbs measure.}

\medskip

\begin{remark}[On scaling conventions]
In the rest of the paper we often use results found in \cite{armstrong2019local} and \cite{serfaty2020gaussian}. It is worth noting that although both papers deal (among other things) with the same 2DOCP model, they do not use the same scaling convention. In \cite{armstrong2019local} the authors work with the so-called “blown-up scaling” (as in the present paper), which is more common in the physics literature and for which lengthscales range from $\sim 1$ (the nearest-neighbor scale) to $\sim N^{1/2}$ (the diameter of the system), whereas \cite{serfaty2020gaussian} uses the random matrix theory convention where the local scale is $\sim N^{-1/2}$ and the global scale is $\sim 1$. In particular, a length scale $\ell$ in the present paper translates into $N^{1/2} \ell$ in \cite{serfaty2020gaussian}, and the corresponding modifications must be applied when quoting results.
\end{remark}

\subsection*{Acknowledgements}
\textit{We thank Alon Nishry for discussions about the JLM law, out of which this project grew. We thank New York University, Tel Aviv University and the Institute for Advanced Study for their hospitality.}
\textit{Part of this work was done while the author was a member of the School of Mathematics at the IAS with the support of the Florence Gould foundation. We thank Prof. S. Torquato and Prof. T. Spencer for stimulating discussions.}

\textit{A previous version contained a serious error which invalidated the proof. We are grateful to Eric Thoma for quickly spotting the dubious passage and mentioning it to us.}

\textit{We thank Sylvia Serfaty for helpful discussions about various aspects of the proof, and Fréderic Marbach for his advice concerning the “spin wave” construction.}

\textit{We thank the anonymous referees for their careful reading and constructive criticism.}

\section{Preliminary results}
\label{sec:preliminary}
\cor{In this section, we gather notation, definitions, and important recent results on finite volume 2DOCP's, most of them drawn from \cite{armstrong2019local} (local laws) and \cite{serfaty2020gaussian} (control on smooth linear statistics). We will occasionally need small extensions of those results, whose proofs are given in the appendix.}

\subsection{Some notations.}
\label{sec:notation}
 We denote indicator functions by $\1$. We denote by $|\Omega|$ the Lebesgue measure of a measurable subset $\Omega$. For $x \in \R^2$ and $r > 0$ we let $\DD(x,r)$ be the (closed) disk of center $x$ and radius $r$. We let $\sq(x,r)$ be the square of center $x$ and sidelength $r$ (with sides parallel to the axes of $\R^2$). We recall the notation $\Di$ for discrepancies: if $\bX$ is a point configuration and $\Omega$ is a measurable subset, we let\footnote{\cor{We use $\mm_0(\Omega)$ when we want to emphasize that it is the mass given to $\Omega$ \emph{by the background measure}.}} $\Di(\bX, \Omega) := \Points(\bX, \Omega) - \mm_0(\Omega)$.

\subsubsection*{Size of derivatives.} If $M$ is a matrix we let $\|M\|$ be its Euclidean norm. If $\varphi$ is a bounded map we write $|\varphi|_\0$ for its sup-norm. If $\varphi$ is differentiable, we let $\Dd \varphi$ be its differential and introduce the following notation:
\begin{itemize}
     \item $|\varphi|_{\1, \star}(x)$ is the size of $\Dd \varphi$ at a given point $x$, $|\varphi|_{\1, \star}(x) := \|\Dd \varphi(x)\|$.
     \item $|\varphi|_{\1, \loc}(x)$ is the size of $\Dd \varphi$ around a given point $x$, $|\varphi|_{\1, \loc}(x) := \sup_{|x'-x| \leq 1} |\varphi|_{\1, \star}(x')$
     \item $|\varphi|_{\1, \Omega}$ is the size of $\Dd \varphi$ in a given domain $\Omega$, $|\varphi|_{\1, \Omega} := \sup_{x \in \Omega} |\varphi|_{\1, \star}(x)$.
     \item Finally, $|\varphi|_{\1}$ is the sup-norm of $\Dd \varphi$.
 \end{itemize} 
 We define similarly $|\varphi|_{\kk, \star}$, $|\varphi|_{\kk, \loc}$, $|\varphi|_{\kk, \Omega}$ and $|\varphi|_{\kk}$ for $\kk \geq 2$.

\subsubsection*{Point configurations.}
For all Borel subsets $\Lambda$ of $\R^2$, we let $\Conf(\Lambda)$ be the space of locally finite\footnote{In fact all of the point configurations considered in this paper are finite.} point configurations on $\Lambda$, endowed with the vague topology of Radon measures and the associated Borel $\sigma$-algebra. When $\Lambda$ is not specified, we use the notation $\Conf$ for point configurations on $\R^2$. We will denote by $\bX \mapsto \bX \cap \Lambda$ the natural projection $\Conf \twoheadrightarrow \Conf(\Lambda)$. We may write “$x \in \bX$” to express the fact that $\bX$ has an atom at a given point $x \in \R^2$. We say that a measurable function $\G$ on $\Conf$ is $\La$-local when for all $\bX$ in $\Conf$ we have $\G(\bX) = \G(\bX \cap \La)$. We say that a measurable subset (an “event”) $\EE$ of $\Conf$ is $\La$-local when its indicator function $\ind_{\EE}$ is $\La$-local.

\subsubsection*{Constants.} Unless specified otherwise, $\Cc$ denotes a universal constant, which may change from line to line, and $\Cc_\beta$ a constant that depends only on $\beta$. We will use $\bCc$ for constants that may depend on $\beta$ and the parameter $\delta$ (as in \eqref{eq:assumption_distance}). We write $A \preceq B$ or $A = \O(B)$ if $|A| \leq \Cc |B|$.

\cor{By “universal”, we mean that the constant does not depend on any of the other parameters ($N, \beta, R, \delta$...), but it might depend on the convention we chose when defining the model, for instance on the fact that we choose the support of the system to be a disk of radius $\sqrt{\frac{N}{\pi}}$ and not a square of sidelength $\sqrt{N}$ etc.}

\subsubsection*{Fluctuations.} If $\varphi$ is a piecewise continuous function on $\R^2$ we define the \emph{fluctuation of (the linear statistics associated to) $\varphi$} as the following random variable:
\begin{equation}
\label{def:Fluct}
\Fluct[\varphi] := \int_{\LN} \varphi(x) \dd \fN(x) = \cor{\int_{\LN} \varphi(x) \dd \left(\bXN - \mm_0\right)(x) = \sum_{i =1}^N \varphi(x_i) - \int_{\LN} \varphi(x) \dd x}.
\end{equation}

\subsection{Electric fields} 
\label{sec:ElectricFields}
We recall that $-\log$ satisfies $- \Delta (-\log) = 2\pi \delta_0$ on $\R^2$ in the sense of distributions. In particular for all smooth enough test functions $f$, the following identity holds: 
\begin{equation}
\label{eq:identiteLaplaceLog}
f(x) = \frac{-1}{2\pi} \int_{\R^2} -\log |x-y| \Delta f(y) \cor{\dd y}.
\end{equation}

Let $\La \subset \R^2$ and let $\bX$ be a point configuration in $\La$.
\begin{definition}[True electric potential and electric field]
\label{def:true}
We let\footnote{\cor{Technically speaking, the objects $\HH, \nHH$ depend on the choice of $\La$. In our case, $\La$ will always be either $\LN$ (i.e. the support of the finite-$N$ 2DOCP) or the domain of “definition” of our generalized 2DOCP's (see Section \ref{sec:CI}) and will never be ambiguous, so we refrain from mentioning it in the notation.}} $\HH^{\bX}$ (resp. $\nHH^{\bX}$) be the \textit{true electric potential} (resp. \textit{true electric field}) generated by $\bX$ (in $\La$), namely the map (resp. vector field) defined on $\R^2$ by:
\begin{equation*}
\HH^{\bX}(x) := \int_{\La} - \log |x-y| \dd \left(\bX - \mm_0 \right)(y),\quad  \text{resp. } \nHH^{\bX}(x) = \int_{\La} - \nabla \log |x-y| \dd \left(\bX - \mm_0 \right)(y).
\end{equation*}
It is easy to check that $\nHH^{\bX}$ is in $\cap_{p \in [1,2)} \Lploc(\R^2, \R^2)$ but fails to be in $\mathrm{L}^2$ around each point charge, and that the following identity is satisfied on $\La$ in the sense of distributions:
\begin{equation}
\label{DeltaHH}
- \Delta \HH^{\bX} = - \dive \nHH^\bX  = 2\pi \left(\bX - \mm_0 \right).
\end{equation}
\end{definition}

\begin{definition}[Compatible electric fields]
\label{def:compatible}
Let $\Elec$ be a vector field in $\cap_{p \in [1,2)} \Lploc(\R^2, \R^2)$. We say that $\Elec$ is an \textit{electric field compatible with} $\bX$ on $\Lambda$ whenever we have: $- \dive \Elec = 2\pi \left( \bX - \mm_0 \right)$ on $\La$ in the sense of distributions.
\end{definition}
Obviously the true electric field is a compatible electric field, however it is not the only one as can be seen by adding any divergence-free smooth vector field on $\La$ to $\nHH^\bX$. 

In order to “take care” of the singularities, one often proceeds to a truncation of the fields near each point charge.
\begin{definition}[Truncation and spreading out Dirac masses]
\label{def:trunc}
For $\eta > 0$ we let $\feta$ be the function:
\begin{equation*}
	\ff_\eta(x) := \max\left( -\log \frac{|x|}{\eta}, 0\right) = \begin{cases} -\log|x| + \log|\eta| & \text{if } x \leq \eta \\ 0 & \text{if } x \geq \eta \end{cases}.
\end{equation*}
For each point $x$ of $\bX$, let $\eta(x)$ be a non-negative real number. The data of $\vec{\eta} = \{\eta(x), x \in \bX\}$ is called a \textit{truncation vector}. If $\HH^\bX$ is the true electric potential, we let $\HH^\bX_{\veta}$ (resp. $\nHH^\bX_{\veta}$) be the (true) \textit{truncated} electric potential (resp. field) given by:
\begin{equation*}
\HH^\bX_{\veta} := \HH^\bX - \sum_{x \in \bX} \ff_{\eta(x)}(\cdot - x), \quad \text{resp. } \nHH^\bX_{\veta} = \nHH^\bX - \sum_{x \in \bX} \nabla \ff_{\eta(x)}(\cdot - x).
\end{equation*}
We are thus effectively replacing $-\log|x - \cdot|$ by $-\log \eta$ near each point charge.

Another way to think about $\ff_\eta$ is that we are truncating the singularities by \textit{smearing out} the point charge $\delta_x$ \textit{à la} Onsager. Indeed when computing the divergence of $\nHH^\bX_{\veta}$ one finds that the atom at each point $x \in \bX$ has been replaced by a measure $\delta_{x}^{(\eta(x))}$ of mass $1$ uniformly spread on the circle of center $x$ and radius $\eta(x)$. \cor{We now have (cf. \eqref{DeltaHH}):
\begin{equation}
\label{DeltaHHeta}
- \Delta \HH^{\bX_{\veta}} = - \dive \nHH^{\bX_{\veta}}  = 2\pi \left(\sum_{x \in \bX} \delta_{x}^{(\eta(x))} - \mm_0 \right).
\end{equation}
}We refer to \cite[Section 2.2 \& Appendix B.1]{armstrong2019local} or to \cite[Sec 3.1]{serfaty2020gaussian} for more details. 

The truncation procedure can be extended to define $\Elec_{\veta}$, where $\Elec$ is any electric field compatible with $\bX$, by setting:
\begin{equation*}
\Elec_{\veta} := \Elec - \sum_{x \in \bX} \nabla \ff_{\eta(x)}(\cdot - x).
\end{equation*}
\end{definition}

For every point $x$ of $\bX$ we define the “nearest-neighbor” distance $\rr(x)$ as:
\begin{equation}
\label{def:nn_distance}
\rr(x) := \frac{1}{4} \min \left( \min_{y \in \bX, y \neq x} |x-y|, 1 \right).
\end{equation}
In particular $\rr(x)$ is always smaller than $1/4$. We let $\vec{\rr} = (\rr(x), x \in \bX)$ be the associated truncation vector. We will sometimes use instead the vector $s \vec{\rr}$ with $s < 1$. 

\subsection{Logarithmic energy and electric fields}
\label{sec:logEnergy}
The logarithmic interaction energy $\FN$ defined in \eqref{def:FN} can be expressed in terms of the true electric field generated by $\bXN$ (Definition \ref{def:true}) as follows.
\begin{itemize}
	\item Taking a uniform truncation vector $\eta(x) = \eta > 0$ for all $x$ in $\bXN$, we have the following equality in the limit as $\eta \to 0$:
\begin{equation*}
\FN(\bXN) = \hal \lim_{\eta \to 0} \left( \frac{1}{2\pi} \int_{\R^2} |\nHH^{\bXN}_{\veta}|^2 + N \log \eta \right).
\end{equation*}
(This quantity is “almost” non-decreasing as $\eta \to 0$, see e.g. \cite[Lemma B.1.]{armstrong2019local}.)

\item On the other hand, taking for each $x$ a truncation $\eta(x) \leq \rr(x)$, we get a non-asymptotic identity:
\begin{equation}
\label{secondformulation}
\FN(\bXN) = \hal \left( \frac{1}{2\pi} \int_{\R^2} |\nHH^{\bXN}_{\veta}|^2 + \sum_{x \in \bXN} \log \eta(x)\right)  - \sum_{x \in \bXN} \int_{\DD(x,\eta(x))} \ff_{\eta(x)}(t-x) \dd t.
\end{equation} 
\end{itemize}
The first formulation can be found in \cite{sandier2012ginzburg,rougerie2016higher} and the second one in \cite{MR3788208} or \cite[Lemma 2.2.]{armstrong2019local}. Of course, expressing electrostatic interactions in terms of the electric field and “smearing out” point charges are both old ideas. One sees that in the small truncation limit there is a compensation between the explosion of $\int_{\R^2} |\nHH^{\bXN}_{\veta}|^2$ and the very negative terms $\sum_{x \in \bXN} \log \eta(x)$, hence the name “renormalized energy” given in \cite{MR3353821} (this renormalization procedure appeared in \cite{bethuel1994ginzburg} and the idea of using nearest-neighbor distances was borrowed from \cite{gunson1977two}). Despite their apparent mutual cancellation, in general the “positive” and “negative” parts can both be compared to $\FN(\bXN)$ as explained e.g. in \cite[Lemma B.2]{armstrong2019local}.

In the sequel we use the expression “electric energy” when referring to quantities of the type $\int_{\Omega} |\Elec_{\vec{\eta}}|^2$ where $\Elec$ is an electric field, $\vec{\eta}$ a truncation vector and $\Omega$ is some subset of $\R^2$, and we write:
\begin{equation}
\label{def_Ener}
\Ener(\bXN, \Omega) :=  \int_{\Omega} |\nHH^{\bXN}_{\vr}|^2, \quad \Ener_s(\bXN, \Omega) :=  \int_{\Omega} |\nHH^{\bXN}_{s \vr}|^2 \ (\text{for } s \in (0,1)).
\end{equation}
From the identity \eqref{secondformulation} we deduce that for $s < 1$:
\begin{equation}
\label{eq:EnersEner}
\Ener_s(\bXN, \Omega) = \Ener(\bXN, \Omega) - \Points(\bXN, \Omega) \log s,
\end{equation}
For convenience, we will sometimes write $\EnerPoints$ for the sum of $\Ener$ and $\Points$ in a given domain.

\subsection{Global and local laws}
The following properties of $\FN$ are well-known (see \cite[Lemma 3.7]{armstrong2019local}):
\begin{itemize}
	\item It is bounded below: there exists a constant $\Cc$ \cor{not depending on $N$} such that for all $N\geq 2$ and for all $\bXN$, we have $\FN(\bXN) \geq - \Cc N$.
	\item It is typically of order $N$ in the sense that for some constant $\Cc_\beta$ depending only on $\beta$:
	\begin{equation}
	\label{eq:global_law_FN}
\EN\left[ \exp\left(\frac{\beta}{2} \FN(\bXN)\right) \right] \leq \exp\left(\Cc_\beta N \right).
	\end{equation}
\end{itemize}
We refer to \eqref{eq:global_law_FN} as a “global law”, controlling the system at macroscopic scale. The next proposition expresses the fact that the system is also well-behaved down to some large microscopic scale. Let us first introduce two important distances that will play a role in Proposition \ref{prop:local_laws} and in the rest of the paper.

\smallskip
\textbf{The smallest microscopic scale $\rho_\beta$.} We refer to a certain \emph{minimal lengthscale} introduced in \cite{armstrong2019local} and denoted by $\rho_\beta$. \cor{The precise expression of $\rho_\beta$ will not be important for us, but we include it here for completeness:
\begin{equation*}
\rho_\beta = C \max\left(1, \beta^{-\hal} \left(1 + \max(- \log \beta, 0)\right)^\hal \right),
\end{equation*}
where $C$ is some constant not depending on $\beta$.
}The quantity $\rho_\beta$ corresponds to the length-scale above which good rigidity properties are proven.
\smallskip

\textbf{Distance to the boundary.} For technical (and possibly physical) reasons, properties of the full system are easier to understand when one looks “away from the edge”, namely at some non-trivial distance of the boundary $\partial \LN$. For $x$ a point in $\LN$ and $\ell > 0$, we will say that “$(x, \ell)$ satisfies \eqref{condi:LL}” when: 
\begin{equation}
\label{condi:LL}
\dist(\sq(x, \ell), \partial \LN) \geq \Cc_\beta N^{1/4},
\end{equation}
where $\Cc_\beta$ is some large enough constant introduced in \cite[(1.16)]{armstrong2019local}.

Let us note that in the statement of Theorem \ref{theo:main2} we assume that $\dist(\DD(x, R), \partial \LN) \geq \delta \sqrt{N}$ which is clearly a stronger assumption, at least for $N$ large enough depending only on $\delta, \beta$.

\begin{proposition}[Local laws]
\label{prop:local_laws}
There exists some universal constant $\Cc$, and a “local laws” constant $\CLL$ depending only on $\beta$ such that the following holds. Let $x$ be a point in $\LN$ and $\ell$ a lengthscale such that:
\begin{enumerate}
    \item  $\ell \geq \rho_\beta$ (the length-scale $\ell$ is larger than the “minimal” one)
    \item $(x, \ell)$ satisfy \eqref{condi:LL} (we are sufficiently “far from the edge”).
\end{enumerate}
Then we control the electric energy in $\sq(x, \ell)$ in exponential moments:
	\begin{equation}
    \label{eq:LocalLaws}
\log \EN \left(\exp\left(\frac{\beta}{2} \Ener\left(\bXN, \sq(x, \ell) \right) \right) \right) \leq \CLL \beta \ell^2.
	\end{equation}
Moreover, we have the following control on the number of points:
\begin{equation}
\label{eq:number_of_points_LL}
 \log \EN \left(\exp\left(\frac{\beta}{\Cc} \Points\left(\bXN, \sq(x, \ell)\right)\right) \right) \leq \CLL \beta \ell^2,
\end{equation}
together with a discrepancy bound:
\begin{equation}
    \label{discr_Di3}
\log \EN \left(\exp\left(\frac{\beta}{\Cc} \frac{\Di^2\left(\bXN, \sq(x, \ell) \right)}{\ell^2}  \right) \right)  \leq \CLL \beta.
	\end{equation}
\end{proposition}
\begin{proof}[Proof of Proposition \ref{prop:local_laws}]
This is a subset of the statements in \cite[Theorem 1]{armstrong2019local}, see Section \ref{sec:proof_locAS} for a technical discussion. 
\end{proof}
As can be seen from Proposition \ref{prop:local_laws}, $\Ener(\bXN, \sq(x, \ell)), \Points(\bXN, \sq(x, \ell))$ and $\EnerPoints(\bXN, \sq(x, \ell))$ are all expected to be of the same order as the area $\ell^2$.

\begin{remark}
\label{rem:LocalLaws}
For our purposes, we will repeatedly use the local laws under the following form: for $\Cc_\beta$ large enough, the probability of having more than $\Cc_\beta \ell^2$ points (or an energy higher than $\Cc_\beta \ell^2$) in a given square $\sq(x, \ell)$ is smaller than $\exp\left( - \ell^2 / \Cc_\beta \right)$.

One can replace squares by disks (or by any “reasonable” shape) in the previous statement, however it is worth observing that local laws do not directly yield interesting controls of the energy (or number of points) on very thin strips, rectangles with diverging aspect ratio, boundaries of squares, thin annuli etc. In those situations, one must resort to splitting the region into squares, applying the local laws to each square and then using a union bound. 
\end{remark}

\subsection{The electric energy controls fluctuations}
\label{sec:apriori}
The next lemma expresses how the electric energy controls linear statistics of Lipschitz functions.
\begin{lemma}[Bounds on fluctuations - Lipschitz case]
\label{lem:apriori} Let $\bX$ be a point configuration in $\R^2$, and let $\varphi$ be a function in $\CC^1(\R^2)$ with compact support. Let $\Elec$ be any electric field compatible with $\bX$ on $\supp \varphi$ in the sense of Definition \ref{def:compatible}. Let $\Omega$ be a domain containing a $1$-neighborhood of $\supp \nabla \varphi$. We have:
	\begin{equation}
	\label{eq:apriori_one_var}
	\left| \int \varphi(x) \dd \left(\bX - \mm_0 \right)(x) \right| \leq \left(\int_{\R^2} |\nabla \varphi|^2 \right)^\hal \left( \int_{\Omega} |\Elec_{\vr}|^2\right)^\hal + |\varphi|_{\1, \Omega} \Points(\bX, \Omega)
	\end{equation}

\subsubsection*{\textbf{Localized case}}	
Assume that $\tilde{\Omega}_1, \dots, \tilde{\Omega}_m$ cover $\supp \nabla \varphi$, and that for each $i$ the domain $\Omega_i$ contains a $1$-neighborhood of $\tilde{\Omega}_i$, then we can replace the right-hand side of \eqref{eq:apriori_one_var} by:
\begin{equation}
	\label{eq:apriori_one_var_localized}
\sum_{k=1}^m |\varphi|_{\1, \Omega_i} \times \left( \left(\int_{\Omega_i} |\Elec_{\vr}|^2\right)^\hal \times \left|\Omega_i \right|^\hal  + \Points(\bX, \Omega_i) \right) 
\end{equation}
\end{lemma}
\begin{remark}
\label{rem:commentsApriori}
Controls of the type \eqref{eq:apriori_one_var} have appeared under various forms in previous works see e.g. \cite[Lemma 5.1]{MR3353821}, \cite[Proposition 2.5]{MR3788208} or \cite[Lemma B.5]{armstrong2019local}, they are usually phrased as: “the electric energy controls the fluctuations”. The electric energy in a given domain $\Omega$ is typically of order $|\Omega|$ and so is the number of points in $\Omega$, cf. Proposition \ref{prop:local_laws}, thus \eqref{eq:apriori_one_var} bounds the typical fluctuations of $\varphi$ by $|\varphi|_{\1} |\supp \nabla \varphi|$ whereas a naive $\mathrm{L}^{\infty}$ bound would rather give $|\varphi|_{\0} \times |\supp \varphi|$. Since our test functions often live on some large lengthscale $\ell$ with $|\varphi|_{\1}$ comparable to $\ell^{-1} |\varphi|_{\0}$, there is indeed an improvement. 
\end{remark}

We give the proof of Lemma \ref{lem:apriori} in Section \ref{sec:proof_apriori}. Compared to existing results, here we simply emphasize the role played by the support of the gradient (instead of the whole support of the test function), which yields more accurate estimates when $\varphi$ is a sharp cut-off function. We will repeatedly use its “localized version” \eqref{eq:apriori_one_var_localized}, whose proof is a straightforward \cor{extension of } \eqref{eq:apriori_one_var}.

\subsection{Finer bound on fluctuations for smooth test functions}
\label{sec:FinerBound}
For test functions with higher regularity the bound of Lemma \ref{lem:apriori} on fluctuations of linear statistics can be improved (see \cite{MR3788208,MR4063572,serfaty2020gaussian} as well as \cite{rider2007noise,ameur2011fluctuations} for the $\beta =2$ case). In particular if $\varphi_\ell(x) := \bar{\varphi}(x/\ell)$ for a fixed reference smooth test function $\bar{\varphi}$, then $\varphi_\ell$ lives at scale $\ell \gg 1$ yet its fluctuations remain \textit{bounded} as $\ell \to \infty$ (with high probability). One can refer e.g. to \cite[Thm.~1]{serfaty2020gaussian} and \cite[Cor.~2.1]{serfaty2020gaussian} for such results. 

In this paper we will occasionally need a more specific statement valid in the radially symmetric case.

\begin{proposition}[Finer bound, the $C^2$ radially symmetric case]
\label{prop:bound_fluct_radial} 
There exists a constant $\Cc_\beta$ depending only on $\beta$ \cor{(and on our model)} such that the following holds. Let $x$ be a point in $\LN$ and let $\varphi$ be a test function which is radially symmetric around $x$, with compact support. Assume that $\varphi$ is in $\CC^2(\R^2)$ and let $\AA$ be an annulus containing (a $1$-neighborhood of) the support of $\Delta \varphi$. Let $s$ be a real number satisfying:
\begin{equation}
\label{eq:condition_on_s}
|s| \leq \frac{\pi \beta}{4 |\varphi|_{\2}}.
\end{equation}
Then the exponential moments of the fluctuations of $\varphi$ satisfy:
\begin{equation}
\label{eq:bound_fluct_radial}
\log \EN \left[  \exp\left( s \Fluct[\varphi]  \right) \right] = \frac{s^2}{4 \pi \beta} \int_{\R^2} |\nabla \varphi|^2 + \log \EN \left[\exp \left(s |\varphi|_{\2} \Cc_\beta \left(\EnerPoints(\bX_N, \AA) \right) \right) \right].
\end{equation}
\end{proposition}
Although it does not appear as such in the literature, Proposition \ref{prop:bound_fluct_radial} can be easily deduced from the tools of \cite{serfaty2020gaussian}, we give a proof in Section \ref{sec:proofCLT}.	

\subsection{Wegner's estimates}
The recent paper \cite{Thoma:2022aa} provides upper bounds on the $k$-point correlation functions of Coulomb gases (in dimension $2$ and higher). Among many other things, it states so-called “Wegner's estimates” i.e. uniform controls on the correlation functions, \cor{and it provides “overcrowding estimates”, namely upper bounds on the probability of having way too many points in a disk, which is what we will need here.}

We will use \cite[Theorem \cor{3}]{Thoma:2022aa}, which is valid even for $r \ll 1$ (sub-microscopic scale):
\begin{lemma}
\label{lem:cluster_bounds}
There exists $\Cc_\beta$ such that for all $x$ in $\LN$ and \emph{for all $r > 0$}, if $(x,r)$ satisfies \eqref{condi:LL} we have:
\begin{equation}
\label{eq:cluster_bound}
\PNbeta \left( \left\lbrace \Points(\DD(x, r)) \geq 1  \right\rbrace \right) \leq \Cc_\beta r^2. 
\end{equation}
\end{lemma}
\begin{remark}
\label{rem:Thoma}
\cor{The results of \cite{Thoma:2022aa} are stated for a 2DOCP with soft confinement, which is not our choice here, and one could wonder if the presence of a “hard wall” might affect the conclusions. An inspection of the proof shows that it is not the case - at least as long as one remains, say, at positive macroscopic distance from the edge (which is our definition of “bulk” in \eqref{eq:assumption_distance}) and that indeed, as written in \cite[Rem. 1.9]{Thoma:2022aa}: \textit{“one can apply [the] arguments to other Coulomb-type systems like finite volume jelliums provided our arguments do not “run into” domain boundaries”}.}
\end{remark}

\section{Locating discrepancies near the boundary}
\label{sec:locating}
\newcommand{\DisIns}{\mathsf{DisIns}}
\newcommand{\DisInsP}{\mathsf{DisIns}^+}
\newcommand{\DisInsM}{\mathsf{DisIns}^-}
Let $0 \leq \eR \leq 1$ (to be chosen later). The goal of this section \cor{is} to show that any “large” discrepancy $\epsilon_R R$ within a disk of radius $R$ can be found (with high probability) \emph{near the boundary of the disk}.

\subsection{Cornering the discrepancy in an annulus close to the boundary}
Let $L \geq 100$ be a lengthscale to be chosen later.  Let $z$ be a point and $R \geq \rho_\beta$ such that $\DD(z, R)$ satisfies \eqref{condi:LL}. In this section we sometimes simply write $\DD_r$ (for $r > 0$) instead of $\DD(z, r)$. 

Assume that the disk $\DD(z,R)$ contains $\eRR$ too many\footnote{We mean $\Points(\bXN, \DD_R) \geq \pi R^2 + \eRR$. points} points (excesses and deficits can be handled similarly). Compare two idealized situations: in the first one, the excess of charges is spread uniformly over $\DD_R$, in the sense that $\Di(\bXN, \DD_r)$ behaves like $\frac{r^2}{R^2} \times \eRR$ for $r \leq R$. In the second situation, the excess is concentrated in a thin strip near $\partial \DD_R$ and immediately compensated by a default on the other side of the boundary. 

According to \cite{jancovici1993large}, the physically realistic picture is the second one: \textit{“Macroscopic electrostatics of conductors implies that, for a given value of [the discrepancy] $Q$, the dominant configurations are such that $Q$ is concentrated in a layer on the inner side of the boundary of the disk, while a charge $-Q$ accumulates in a layer on the outer side.”}. 

For a mathematical statement, \cor{we introduce the event $\DisIns := \DisInsP \cup \DisInsM$, where:
\begin{equation*}
\DisInsP := \bigcap_{R - 2L \leq r \leq R-L} \left\lbrace \Di(\bXN, \DD(z, r)) \geq \hal \eRR \right\rbrace,
\end{equation*}
\begin{equation*}
\DisInsM := \bigcap_{R - 2L \leq r \leq R-L} \left\lbrace \Di(\bXN, \DD(z, r)) \leq - \hal \eRR \right\rbrace, \end{equation*}
The event $\DisInsP$ (resp. $\DisInsM$) represents the situation in which  an excess (resp. deficit) of charges is found in all the disks $\DD_r$ for a fairly wide region $r \in [R-2L, R-L]$ - in which case the discrepancy is thus \emph{not} “concentrated in a layer on the inner side of the boundary of the disk”}. The goal of the next proposition is to rule out this possibility.

\begin{proposition}
\label{prop:chase_the_goat}
There exists a constant $\Cins(\beta)$ depending only on $\beta$ such that the following holds. Let $s$ be a real parameter. Assume that $R, L, \eR, s$ satisfy:
\begin{equation}
\label{eq:s_and_L}
\begin{cases}
& \Cins \frac{1}{\eR} \leq L \leq \frac{R}{10} \\
& 0 \leq s \leq \frac{1}{\Cins} \min\left( \frac{L^3}{R}, L \eR \right).
\end{cases}
\end{equation}
Then we have:
\begin{equation}
\label{eq:chase_the_goat}
\PNbeta \left[ \DisIns \right] \leq \exp\left(- \frac{s \eRR}{4}\right).
\end{equation}
\end{proposition}
We postpone the proof of Proposition \ref{prop:chase_the_goat} to Section \ref{sec:proof_chase_the_goat}. It relies on Proposition \ref{prop:bound_fluct_radial}.

Next, we argue that if the complementary event occurs, then one can find an annulus of width $\approx L$ near the boundary of $\DD_R$ that carries a large discrepancy.
\begin{lemma}[Cornering the discrepancy]
\label{lem:finding_discr}
If $\Di(\bXN, \DD_R) \geq \eRR$ \cor{(resp. $\Di(\bXN, \DD_R) \leq -\eRR$)} and if, in contrast to the event considered in \eqref{eq:chase_the_goat}, there exists a radius $r$ with $R-2L \leq r \leq R-L$ such that:
\begin{equation}
\label{bonr}
\Di(\bXN, \DD_r) < \hal \eRR \quad \cor{\text{(resp. $\Di(\bXN, \DD_r) > - \hal \eRR$)}}
\end{equation}
then there exists an integer $k$ with $0 \leq k \leq R^{2}$ such that the discrepancy in the annulus $\DD_R \setminus \DD_{R - 2L + \frac{k L}{R^2}}$ is larger than $\frac{1}{4} \eRR$ \cor{(resp. smaller than $-\frac{1}{4} \eRR$)}.
\end{lemma}
\begin{proof}[Proof of Lemma \ref{lem:finding_discr}]
\cor{Let us first treat the case of a positive discrepancy.} Let $r$ be such that \eqref{bonr} holds and let $0 \leq k \leq R^2$ be the integer such that:
\begin{equation*}
r_k := R - 2L + \frac{k L}{R^2} \leq r < R - 2L + \frac{(k+1) L}{R^2}.
\end{equation*}
It is easy to see that $\Di(\bXN, \DD_{r_k}) \leq \frac{3}{4} \eRR$ (a rough bound), indeed we have:
\begin{equation*}
\Di(\bXN, \DD_{r_k}) \leq \Di(\bXN, \DD_r) + |\DD_r \setminus \DD_{r_k}|,
\end{equation*}
with equality if and only if there is no point in the annulus $\DD_r \setminus \DD_{r_k}$, whose area is of order $ \frac{RL}{R^{2}} \leq 1$.

Since the total discrepancy in $\DD_R$ is at least $\eRR$ and the one inside $\DD_{r_k}$ is at most $\frac{3}{4} \eRR$ then the annulus $\DD_R \setminus \DD_{r_k}$ must carry a discrepancy of size at least $\frac{1}{4} \eRR$.

\cor{In the case of a negative discrepancy, instead of taking $r_k$ a bit smaller than $r$ we take it a bit larger and observe that the discrepancy cannot go down too much between $\DD_r$ and $\DD_{r_k}$. Hence we must have $\Di(\bXN, \DD_{r_k}) \geq \frac{- 3\eRR}{4}$, so the annulus $\DD_{R} \setminus \DD_{r_k}$ carries a default of points at least $\frac{\eRR}{4}$ (in absolute value).}
\end{proof}

\subsection{A well-separated family of boxes carrying the discrepancy} 
\label{sec:well-separated}
Let $r$ be in $[R-2L, R-L]$.

\subsubsection*{Decomposition in boxes.}
We now split the annulus $\DD_R \setminus \DD_{r}$ into smaller parts that we call \emph{boxes}. 

\begin{definition}[Decomposition in boxes]
\label{def:boxes}
For $i \in \{0, \dots, \frac{R}{L} -1 \}$, we let the $i$-th “box” $\boxs_i$ be the intersection of the annulus $\DD_R \setminus \DD_{r}$ with a certain angular sector of center $z$.

\cor{Let us denote by $z' \mapsto (\rho, \theta) \in [0, + \infty) \times [0, 2\pi)$ the polar coordinates of a point $z' \in \R^2$, with the origin placed at $z$. The box $\boxs_i$ is then defined as:
\begin{equation*}
\boxs_i := \left\lbrace z' \in \R^2, \quad r \leq \rho \leq R, \quad \frac{\theta}{2\pi} \in \left[\frac{L}{R}i, \frac{L}{R}(i+1)\right]  \right\rbrace.
\end{equation*}
Note that in polar coordinates, such a box is simply a rectangle.} 

The boundary of each such box is made of two line segments of equal length in $(L, 2L)$ and two concentric circular arcs which subtend the same angle at the center and whose arclengths are different but both in $\left(\frac{L}{2}, L\right)$. The shape is symmetric with respect to the straight line joining the midpoints of both arcs. \cor{Here is a quick representation of those “boxes”, paving the annulus $\DD_R \setminus \DD_{r}$:}
\begin{center}
\begin{tikzpicture}[scale=0.1]
    \def\R{4}  
    \def\r{3.2}  
    \def\n{8}  

    \draw[thick] (0,0) circle (\R); 
    \draw[thick] (0,0) circle (\r); 

    \foreach \i in {0,...,\n} {
        \draw[thick] ({\r*cos(360/\n*\i)},{\r*sin(360/\n*\i)}) -- ({\R*cos(360/\n*\i)},{\R*sin(360/\n*\i)});
    }
\end{tikzpicture}
 \end{center}
We sometimes call “a box of size $L$” any domain $\boxs$ that fits the previous description. We let $\omega_i$ be the center of mass of $\boxs_i$, which only serves as a convenient reference point.
\end{definition}

\subsubsection*{Introducing the parameters $M$ and $T$}
Let $T$ be a lengthscale, and $M$ be an integer, both to be chosen at the end, such that:
\begin{equation}
\label{eq:condi_T2logTML}
T \geq 10 L, \quad 100 \leq M \leq \frac{R}{L}, \quad  T \leq \frac{ML}{100}, \quad T^2 \log T \leq ML.
\end{equation}
(The fourth condition (which could be weakened), implies the third one but for clarity we write them all down.)

\begin{lemma}[Some well-separated boxes carry the discrepancy]
\label{lem:MV1}
Assume that the discrepancy in $\DD_R \setminus \DD_{r}$ is larger than $\frac{1}{4} \eRR$ \cor{(resp. smaller than $-\frac{1}{4} \eRR$)}. Then there exists $l \in \{0, \dots, M-1\}$ such that:
\begin{equation}
\label{eq:box_carry}
\sum_{i \equiv l \text{ mod } M} \Di(\bXN, \boxs_i) \geq \frac{\eRR}{4 M} \ \cor{\text{ \Big(resp. } \sum_{i \equiv l \text{ mod } M} \Di(\bXN, \boxs_i) \leq - \frac{\eRR}{4 M} \text{\Big)}.}
\end{equation}
\end{lemma}
\begin{proof}[Proof of Lemma \ref{lem:MV1}]
It follows from a straightforward pigeonhole argument.
\end{proof}

\medskip
From now on, we assume to have chosen such a $l \in \{0, \dots, M-1\}$ and we write the corresponding boxes as $\boxs_1, \dots, \boxs_\Nn$, where $\Nn = \O\left( \frac{R}{ML} \right)$ is the number of box in each “well-separated” family. Moreover, we see each box $\boxs_i$ as being contained in a large disk $\La_i := \DD(\omega_i, T)$, where $T$ is as above.
\medskip

Let $\dij$ be the distance between $\La_i$ and $\La_j$. 
\begin{lemma}
We have for each fixed $i$:
\begin{equation}
\label{sumdij}
\sum_{j, j \neq i} \frac{1}{\dij} = \O\left(\frac{\log R}{ML}\right), \quad \sum_{j, j \neq i} \frac{1}{\dij^2} = \O\left(\frac{1}{M^2L^2}\right)
\end{equation}
\end{lemma}
\begin{proof}
Let us observe that between two “consecutive” boxes in the family $\{i \equiv l \text{ mod } M\}$ considered in \eqref{eq:box_carry}, there is a distance of order $ML$ (since $ML \geq 100 T$ by assumption, this is also comparable to the distance between two consecutive \emph{disks}). We can compare the sum to an harmonic sum (in the first case) or a converging Riemann series (in the second case).
\end{proof}

\subsubsection*{Plan for the next two sections}
We now want to treat each box $\boxs_i$ as living in its own smaller version of a 2DOCP contained in $\La_i$, which leads us to the next two sections. 

Think of a “sub-system” as the random collection of particles contained in a given sub-domain $\La \subset \LN$ with a reasonable shape (e.g. a square or a disk). These particles feel the influence of each other, but also of the full system in $\LN$ because the logarithmic interaction is long-range. Hence sub-systems are typically not isolated and not independent from each other. 
\begin{enumerate}
    \item In Section \ref{sec:approx_CI} we observe that if we condition on the values of the discrepancies (or equivalently of the number of particles) in domains that are well-separated, then the corresponding sub-systems acquire a form of independence.
    \item In Section \ref{sec:CI} we show that “typical” sub-systems, seen as slight generalizations of the 2DOCP model introduced earlier (in \eqref{def:FN}, \eqref{def:PNbeta}), retain most of the properties of the \emph{full} system mentioned in Section~\ref{sec:preliminary}. 
\end{enumerate}

\section{Approximate conditional independence for sub-systems}
\label{sec:approx_CI}
In this section, we consider a family  $\{\Lai, 1 \leq i \leq \Nn\}$ of $\Nn \geq 2$ disjoint disks in $\LN$ (forming our “sub-systems”). We will eventually apply the  results below to the $\Lai$'s chosen in Section \ref{sec:well-separated}, but the statements in the present section are  general.

Let $\Ext$ be the complement $\Ext := \LN \setminus \cup_{i =1}^\Nn \Lai$. Let $\Vext$ be the logarithmic potential generated by the system in $\Ext$, namely:
\begin{equation}
\label{def:Vext} \Vext(x) := \int_{\Ext} - \log |x-y| \dd \fN(y),
\end{equation}
\cor{where we recall that $\fN$ is the fluctuation measure $\fN := \bXN - \mm_0$.}

The potential $\Ext$ is harmonic on all the $\Lai$'s and depends only on the configuration in $\Ext$.

\subsection{Decomposition of the interaction, conditional independence error}
\label{sec:decompo}
We let $\Int[\La_1, \dots, \La_\Nn]$ be the true logarithmic interaction between the sub-systems, namely:
\begin{equation}
\label{def:IntSS}
\Int[\La_1, \dots, \La_\Nn] := \hal \sum_{1 \leq i \neq j \leq \Nn} \iint_{\Lai \times \Laj} - \log |x-y| \dd \fN(x) \dd \fN(y).
\end{equation}
By expanding the double integral defining $\FN$, we may write:
\begin{equation}
\label{decompo_inter}
\FN(\bXN) = \F_{\Ext} (\bXN \cap \Ext) + \Int[\La_1, \dots, \La_\Nn] + \sum_{i=1}^\Nn\left( \F_{\Lai}(\bXN \cap \Lai) + \int_{\Lai}  \Vext(x) \dd \fN(x) \right)
\end{equation}
with $\F_{\Ext}, \F_{\Lai}$ \cor{defined by:
\begin{multline*}
\F_{\Ext}(\bXN) := \hal \iint_{(x,y) \in \Ext \times \Ext, x \neq y} - \log |x-y| \dd \fN(x) \dd \fN(y), \\ \F_{\Lai}(\bXN) := \hal \iint_{(x,y) \in \Lai \times \Lai, x \neq y} - \log |x-y| \dd \fN(x) \dd \fN(y).
\end{multline*}
}

On the other hand, for each $1 \leq i \leq \Nn$, let $\Dii$ be the discrepancy of $\bXN$ in $\Lai$, and let $\tInt[\Dd_1, \dots, \Dd_\Nn]$ be the approximation of $\Int[\La_1, \dots, \La_\Nn]$ given by:
\begin{equation}
\label{def:tIntSS}
\tInt[\Dd_1, \dots, \Dd_\Nn] := \hal \sum_{1 \leq i \neq j \leq \Nn} - \Dii \Djj \log |\omega_i - \omega_j|,
\end{equation}
where $\omega_i$ denotes the center of $\La_i$. We define the quantity $\ErrorCI$ as follows:
\begin{equation}
\label{def:ErrorCI}
\ErrorCI\left[\bXN | \left(\La_1, \dots, \La_\Nn\right) \right] := \left| \Int[\La_1, \dots, \La_\Nn] - \tInt[\Dd_1, \dots, \Dd_\Nn] \right|,
\end{equation}
and we use it below in order to measure a “conditional independence error”. \cor{The domains $\La_1, \dots, \La_\Nn$ being always unambiguous, in the sequel we will often shorten the notation and simply write $\ErrorCI\left[\bXN \right]$.}

\subsection{Bounds on the conditional independence error}
For $1 \leq i \neq j \leq \Nn$, define the distance $\dij$ as $\dij := \dist(\Lai, \Laj)$ \cor{and  $\nn_i := \Points(\bXN, \Lai)$}. Assume that for all $1 \leq i \leq \Nn$ the disk $\Lai$ has radius $T$, and that:
\begin{equation}
\label{eq:boundnni}
\max_{1 \leq i \leq \Nn} \nn_i \leq 10 T^2, \quad \min_{1 \leq i \neq j \leq \Nn} \dij \geq 10 T.
\end{equation}
\begin{lemma}[The size of $\ErrorCI$]
\label{lem:sizeCI}
We have, if \eqref{eq:boundnni} holds:
\begin{equation}
\label{sizeCI}
\left|\ErrorCI\left[\bXN | \left(\La_1, \dots, \La_\Nn\right) \right] \right| = \O(T^5) \sum_{1 \leq i \neq j \leq \Nn} \frac{1}{\dij}.
\end{equation} 
\end{lemma}
\begin{proof}[Proof of Lemma \ref{lem:sizeCI}]
For $x \in \Lai, y \in \Laj$ (with $i \neq j$) since the diameter of the \cor{disks} is $\O(T)$ and the mutual distances satisfy \eqref{eq:boundnni}, a Taylor's expansion yields:
\begin{equation}
\label{expansion_logxy}
 \log |x - y | =  \log |\omega_i - \omega_j| + \O\left( \frac{T}{|\omega_i-\omega_j|}\right),
\end{equation}
with a universal implicit constant. Integrating \eqref{expansion_logxy} against the fluctuation measures in $\Lai$ and $\Laj$ yields:
\begin{equation*}
\left| \iint_{\Lai \times \Laj} \log |x-y| \dd\fN(x) \dd\fN(y) - \Dii \Djj \log |\omega_i - \omega_{j}|  \right| \leq (|\Lai| + \nn_i) \cdot (|\Laj| + \nn_j) \cdot \O \left( \frac{T}{\dij}\right).
\end{equation*}
Using the bound on $\nn_i, \nn_j$ given by \eqref{eq:boundnni} and summing over $i \neq j$, we get \eqref{sizeCI}.
\end{proof}

\begin{remark}
\label{rem_sharper_ERRORCI}
It is possible to reduce the order of magnitude of $\ErrorCI$ by expanding the interaction in a more precise way. \cor{While this might be useful for other purposes, here it would not help improving our variance bounds.}
\end{remark}

\newcommand{\VextL}{\mathrm{V}^{\mathrm{ext}}}
\renewcommand{\Vv}{\VextL}
\newcommand{\FL}{\F_{\Lambda}}
\newcommand{\fL}{\f_{\Lambda}}
\newcommand{\XL}{\X_{(\Lambda)}}
\newcommand{\xL}{x^{(\Lambda)}}
\newcommand{\Xn}{\X_\nn}
\newcommand{\bXn}{\bX_\nn}
\newcommand{\mL}{\mm_{\Lambda}}
\newcommand{\PnLV}{\P^\beta_{\nn, \Lambda, \Vv}}
\newcommand{\KnLV}{\mathrm{K}^\beta_{\nn, \Lambda, \Vv}}
\newcommand{\EnLV}{\E^\beta_{\nn, \Lambda, \Vv}}
\newcommand{\DT}{\mathsf{D}}
\newcommand{\wVextL}{\tilde{\mathrm{V}}^{\mathrm{ext}}}
\newcommand{\AUN}{\mathsf{A}_\mathrm{1}}
\newcommand{\mut}{\mu^{(t)}}
\newcommand{\muz}{\mu^{(0)}}
\newcommand{\ErrorEnergy}{\mathrm{ErrEnerTrans}}
\newcommand{\zetaW}{\zeta_\WW}
\newcommand{\PnL}{\mathbb{P}_{\nn, \La}^\beta}
\newcommand{\KnL}{\mathrm{K}_{\nn, \La}^\beta}
\newcommand{\Cgood}{\hyperref[prop:VextIsOftenGood]{\mathtt{C}_{\mathtt{Good}}}} 
\newcommand{\VextLi}{\mathrm{V}^{\mathrm{ext}}_{\La_i}}
\newcommand{\PnLL}{\mathbb{P}_{\nn, \La}^\beta}
\newcommand{\EnLL}{\mathbb{E}_{\nn, \La}^\beta}
\newcommand{\KnLL}{\mathrm{K}_{\nn, \La}^\beta}

\subsection{2DOCP's with harmonic external field}
\label{sec:with_external}
In the following we consider two-dimensional one-component plasmas whose energy takes into account the effect of an external field $\VextL$ on each particle.

\begin{itemize}
	\item Let $\VextL$ (the \textit{external field} on $\La$) be a lower semi-continuous function on $\R^2$ that is harmonic on (the interior of) $\Lambda$. Let us emphasize that $\VextL$ is harmonic hence very regular in the interior of $\La$ but we do not \emph{a priori} control $\VextL$ or its derivatives near $\partial \La$. The situation is thus different from the choice of an external weight/potential as frequently found in the literature on log-gases.  

	\item Let $\Xn$ denote a $\nn$-tuple of points $\Xn = (x_1, \dots, x_\nn)$ in $\Lambda$, let $\bXn := \sum_{i=1}^\nn \delta_{x_{i}}$ be the associated atomic measure, and let $\fL := \bXn - \mm_0 \ind_{\La}$ be the signed fluctuation measure on $\La$. Let $\FL(\bXn)$ be the logarithmic interaction energy:
	\begin{equation}
	\label{def:FL}
	\FL(\bXn) := \hal \iint_{(x,y) \in \La \times \La, x \neq y} - \log |x-y| \dd \fL(x) \dd \fL(y).
	\end{equation}
\end{itemize}
We define a probability density $\PnLV$ on the space of $\nn$-tuples of points in $\Lambda$ by setting:
\begin{equation}
\label{def:PnLV}
\dd \PnLV(\Xn) := \frac{\exp\left( - \beta \left( \FL(\bXn) + \int_{\La} \VextL(x) \dd \fL(x) \right) \right)}{\KnLV} \dd \Xn,
\end{equation}
where $\KnLV$ is the partition function, namely the normalizing constant:
\begin{equation}
\label{def:KnLV}
\KnLV := \int_{\Lambda^\nn} \exp\left( - \beta \left( \FL(\bXn) + \int_{\La} \VextL(x) \dd \fL(x)  \right) \right) \dd \Xn,
\end{equation}
and $\dd \Xn$ is the Lebesgue measure on $\Lambda^\nn$. We may now state the main result of this section.

\renewcommand{\Vv}{\mathrm{V}}

\subsection{Approximate conditional independence}
\label{sec:PropCI}
\begin{proposition}[Approximate conditional independence]
\label{prop:CI} \ 
		\begin{itemize}
		 	\item For $1 \leq i \leq \Nn$, let $\G_i$ be a measurable function on $\Conf$ with \emph{non-negative} real values and let $\EE_i$ be a measurable subset of $\Conf$. Assume that $\G_i$, $\EE_i$ are $\Lai$-local in the sense of Section \ref{sec:notation}, \cor{i.e. they depend only on the configuration within $\Lai$.}
		\item Let $\Eext$ be a measurable subset of $\Conf$, assume that $\Eext$ is $\Ext$-local.
		\item Denote by $\EE_N$ the following event: $\EE_N := \Eext \cap \bigcap_{i = 1}^\Nn \EE_i.$	
		\item Finally, we say that a family $\{\tn_i\}_{1 \leq i \leq \Nn}$ of integers is “admissible” (we write “$\{\tn_i\}\ \mathrm{adm.}$” below) when there exists $\bXN \in \EE_N$ such that $\Points(\bXN, \Lai) = \nn_i$ for all $1 \leq i \leq \Nn$.
	\end{itemize}
	We have:
\begin{multline}
\label{eq:CI}
\EN \left[ \prod_{i=1}^\Nn \G_i(\bXN) \ind_{\EE_N} \right] \leq \exp\left(2 \beta \sup_{\bXN \in \EE_N} \ErrorCI(\bXN) \right) \\ \times \sup_{\bXext \in \Eext, \{\tn_i\} \ \mathrm{adm.}} \prod_{i=1}^\Nn \E^\beta_{\tn_i, \Lai, \Vext} \left[ \G_i(\bX) \Cond \EE_i \right].
\end{multline}
Moreover, with the same assumptions, the following lower bound holds:
\begin{multline}
\label{eq:CI_LB}
\EN \left[ \prod_{i=1}^\Nn \G_i(\bXN) \ind_{\EE_N} \right] \geq \exp\left(- 2 \beta \sup_{\bXN \in \EE_N} \ErrorCI(\bXN) \right) \times \PNbeta(\EE_N) \\
\times  \inf_{\bXext \in \Eext, \{\tn_i\} \ \mathrm{adm.}} \prod_{i=1}^\Nn \E^\beta_{\tn_i, \Lai, \Vext} \left[ \G_i(\bX) \Cond \EE_i \right].
\end{multline}
Notice that one $\sup$ is now an $\inf$, the error term $\sup \ErrorCI$ now appears in the exponent with a minus sign, and there is an extra factor $\PNbeta(\EE_N)$.
\end{proposition}
\begin{proof}[Proof of Proposition \ref{prop:CI}] 
Let us start by using the definition \eqref{def:PNbeta} of the Gibbs measure and write:
\begin{equation*}
\EN \left[ \prod_{i=1}^\Nn \G_i(\bXN) \ind_{\EE_N} \right] = \frac{1}{\KNbeta} \int_{(\LN)^N} \exp\left( - \beta \FN(\bXN) \right) \times \prod_{i=1}^\Nn \G_i(\bXN) \ind_{\EE_N}(\bXN) \dd \XN.
\end{equation*}

\subsubsection*{\cor{Step 1. Using the decomposition \eqref{decompo_inter}.}}
Using the decomposition \eqref{decompo_inter} of the logarithmic interaction $\FN(\XN)$, we obtain:
\begin{multline*}
\EN \left[ \prod_{i=1}^\Nn \G_i(\bXN) \ind_{\EE_N} \right] = \frac{1}{\KNbeta} \int_{(\LN)^N} \exp\left( - \beta \F_{\Ext} (\bXN \cap \Ext)\right)  \exp\left(- \beta \Int[\La_1, \dots, \La_\Nn] \right) \\
\times \left[ \prod_{i=1}^\Nn \exp \left( - \beta \left( \F_{\Lai}(\bXN \cap \Lai) + \int_{\Lai} \Vext(x) \dd \fN(x) \right) \right) \times \G_i(\bXN) \right] \ind_{\EE_N}(\bXN) \dd \XN.
\end{multline*}

\subsubsection*{\cor{Step 2. Introducing the number of points and the conditional independence error.}}
Next:
\begin{itemize}
 	\item We use a complete system of events by fixing the number $\tni$ of points in each $\Lai$ (and thus also the number $\nnext$ of points in $\Ext$). The knowledge of $\tni$ is equivalent to fixing the discrepancy $\Dii$ in $\Lai$. Since we are working under the event $\EE_N$, the $\{\tn_i\}$ must be admissible as defined in the statement of Proposition \ref{prop:CI}.
 	\item Up to a combinatorial factor, we may then decompose the $N$-tuple $\XN$ into a $\nnext$-tuple $\Xext$ of points in $\Ext$, and $\X_i$ ($1 \leq i \leq \Nn$), where each $\X_i$ is a $\tni$-tuple of points in $\Lai$.
 	\item We decompose the Lebesgue measure $\dd \XN$ accordingly.
 	\item We write $\bXext$ and $\bXi$ ($1 \leq i \leq \Nn$) for the associated atomic measures. We have the identities $\bXext = \bXN \cap \Ext$ and $\bXi = \bXN \cap \Lai$ ($1 \leq i \leq \Nn$).
	\item By our locality assumptions \cor{and by definition of $\EE_N$}: $\ind_{\EE_N}(\bXN) = \ind_{\Eext}(\bXext) \times \prod_{i=1}^\Nn \1_{\EE_i}(\bXi)$.
 	\item By our locality assumption, $\G_i(\bXN) = \G_i(\bXi)$.
 	\item We introduce the \cor{“fluctuation”} measures $\fExt := \bXext - \mm_0 \ind_{\Ext}$ and $\fLi := \bXi - \mm_0 \ind_{\Lai}$ (for $i = 1, \dots, \Nn$) \cor{relative to each part of the decomposition}.
 	\item Finally, using the definition \eqref{def:ErrorCI} we replace $\Int[\La_1, \dots, \La_\Nn]$ by $\tInt[\Dd_1,\dots, \Dd_\Nn]$, up to an error quantified by $\ErrorCI$. The quantity $\exp\left(-\beta \tInt[\Dd_1, \dots, \Dd_\Nn]\right)$ can be taken outside the integrals because it only depends on the data of $\{\tni\}_{1 \leq i \leq \Nn}$.
 \end{itemize} 
We obtain the following upper bound:
\begin{multline}
\label{CI_UB_1}
\EN \left[ \prod_{i=1}^\Nn \G_i(\bXN) \ind_{\EE_N} \right] \leq \exp\left(\beta \sup_{\bXN \in \EE_N}  \ErrorCI(\bXN)\right)  \\
\times \frac{1}{\KNbeta} \sum_{  \{\tn_i\} \ \mathrm{adm.} } \binom{N}{\tn_1 \dots \tn_\Nn} \exp\left(-\beta \tInt[\Dd_1, \dots, \Dd_\Nn] \right) \times   \int_{(\Ext)^{\nnext}} \exp\left(-\beta \F_{\Ext}(\bXext) \right) \ind_{\Eext}(\bXext) \\
 \times \left[ \prod_{i=1}^\Nn \int_{(\Lai)^{\tni}} \exp \left( - \beta \left( \F_{\Lai}(\bXi) + \int_{\Lai} \Vext(x) \dd \fLi(x) \right) \right) \times  \G_i(\bXi) \ind_{\EE_i}(\bXi) \dd \Xi \right] \dd \Xext,
\end{multline}
\cor{and the corresponding lower bound (recall that, by assumption, the $\G_i$'s have non-negative values):
\begin{multline}
\label{CI_LB_1}
\EN \left[ \prod_{i=1}^\Nn \G_i(\bXN) \ind_{\EE_N} \right] \geq \exp\left(-\beta \sup_{\bXN \in \EE_N}  \ErrorCI(\bXN)\right)  \\
\times \frac{1}{\KNbeta} \sum_{  \{\tn_i\} \ \mathrm{adm.} } \binom{N}{\tn_1 \dots \tn_\Nn} \exp\left(-\beta \tInt[\Dd_1, \dots, \Dd_\Nn] \right) \times   \int_{(\Ext)^{\nnext}} \exp\left(-\beta \F_{\Ext}(\bXext) \right) \ind_{\Eext}(\bXext) \\
 \times \left[ \prod_{i=1}^\Nn \int_{(\Lai)^{\tni}} \exp \left( - \beta \left( \F_{\Lai}(\bXi) + \int_{\Lai} \Vext(x) \dd \fLi(x) \right) \right) \times  \G_i(\bXi) \ind_{\EE_i}(\bXi) \dd \Xi \right] \dd \Xext.
\end{multline}}

\cor{\subsubsection*{Step 3. Using our new notation for 2DOCP's with external field.}}
We may conveniently condense \eqref{CI_UB_1}, \cor{\eqref{CI_LB_1}} using the notation \cor{introduced in Section \ref{sec:with_external}}. In view of the definitions \eqref{def:PnLV}, \eqref{def:KnLV}, for each $1 \leq i \leq \Nn$,  we write:
\begin{multline*}
\int_{(\Lai)^{\tni}} \exp \left( - \beta \left( \F_{\Lai}(\bXi) + \int_{\Lai} \Vext(x) \dd \fLi(x) \right) \right) \G_i(\bXi) \ind_{\EE_i}(\bXi) \dd \Xi \\
 = \E^\beta_{\tni, \Lai, \Vext} \left[ \G_i(\bX) \Cond \EE_i \right] \times \P^\beta_{\tni, \Lai, \Vext} \left[ \EE_i \right] \times \KNbetai,
\end{multline*}
and we may thus re-write \eqref{CI_UB_1} as:
\begin{multline}
\label{CI_UB_2}
\EN \left[ \prod_{i=1}^\Nn \G_i(\bXN) \ind_{\EE_N} \right] \leq \exp\left(\beta \sup_{\bXN \in \EE_N} \ErrorCI(\bXN)\right)  \\
\times \frac{1}{\KNbeta} \sum_{  \{\tn_i\} \ \mathrm{adm.} } \binom{N}{\tn_1 \dots \tn_\Nn} \exp\left(-\beta \tInt[\Dd_1, \dots, \Dd_\Nn] \right)  \times \int_{(\Ext)^{\nnext}} \exp\left(-\beta \F_{\Ext}(\bXext) \right) \ind_{\Eext}(\bXext)  \\
 \times \left[ \prod_{i=1}^\Nn \E^\beta_{\tni, \Lai, \Vext}\left[ \G_i(\bX) \Cond \EE_i \right] \times \P^\beta_{\tni, \Lai, \Vext} \left[ \EE_i \right] \times \KNbetai \right] \dd \Xext,
\end{multline}
\cor{and \eqref{CI_LB_1} as:
\begin{multline}
\label{CI_LB_2}
\EN \left[ \prod_{i=1}^\Nn \G_i(\bXN) \ind_{\EE_N} \right] \geq \exp\left(- \beta \sup_{\bXN \in \EE_N} \ErrorCI(\bXN)\right)  \\
\times \frac{1}{\KNbeta} \sum_{  \{\tn_i\} \ \mathrm{adm.} } \binom{N}{\tn_1 \dots \tn_\Nn} \exp\left(-\beta \tInt[\Dd_1, \dots, \Dd_\Nn] \right)  \times \int_{(\Ext)^{\nnext}} \exp\left(-\beta \F_{\Ext}(\bXext) \right) \ind_{\Eext}(\bXext)  \\
 \times \left[ \prod_{i=1}^\Nn \E^\beta_{\tni, \Lai, \Vext}\left[ \G_i(\bX) \Cond \EE_i \right] \times \P^\beta_{\tni, \Lai, \Vext} \left[ \EE_i \right] \times \KNbetai \right] \dd \Xext.
\end{multline}}

\subsubsection*{\cor{Step 4. Comparing the local expectations.}}
\cor{Note that the expectations $\E^\beta_{\tni, \Lai, \Vext}$ in the integrands of \eqref{CI_UB_2}, \eqref{CI_LB_2} are quantities that, for each $i$, still depend on the configuration $\bXext$ through the expression \eqref{def:Vext} of $\Vext$. We now replace them (through a lower and an upper bound) by quantities which do not depend on the precise arrangement of points in $\Ext$, but simply on the fact that $\bXext$ is in $\Eext$.} To do so, for all $\bXext \in \Eext$ and all admissible $\{\tn_i\}$, we use the simple upper bound:
\begin{equation*}
\prod_{i=1}^\Nn \E^\beta_{\tni, \Lai, \Vext}\left[ \G_i(\bX) \Cond \EE_i \right] \leq \sup_{\bXext \in \Eext \{\tn_i\} \ \mathrm{adm.}} \prod_{i=1}^\Nn \E^\beta_{\tn_i, \Lai, \Vext} \left[ \G_i(\bX) \Cond \EE_i \right],
\end{equation*}
and the lower bound
\begin{equation*}
\prod_{i=1}^\Nn \E^\beta_{\tni, \Lai, \Vext}\left[ \G_i(\bX) \Cond \EE_i \right] \geq \inf_{\bXext \in \Eext \{\tn_i\} \ \mathrm{adm.}} \prod_{i=1}^\Nn \E^\beta_{\tn_i, \Lai, \Vext} \left[ \G_i(\bX) \Cond \EE_i \right].
\end{equation*}
We get:
\begin{multline}
\label{after_sup}
\EN \left[ \prod_{i=1}^\Nn \G_i(\bXN) \ind_{\EE_N} \right] \leq \exp\left(\beta \sup_{\bXN \in \EE_N} \ErrorCI(\bXN) \right) \times \sup_{\bXext \in \Eext, \{\tn_i\} \mathrm{adm.}} \prod_{i=1}^\Nn \E^\beta_{\tn_i, \Lai, \Vext} \left[ \G_i(\bX) \Cond \EE_i \right]  \\
\times \frac{1}{\KNbeta}  \sum_{  \{\tn_i\} \ \mathrm{adm.} } \binom{N}{\tn_1 \dots \tn_\Nn} \exp\left(-\beta \tInt[\Dd_1, \dots, \Dd_\Nn] \right)  \int_{(\Ext)^{\nnext}} \exp\left(-\beta \F_{\Ext}(\bXext) \right) \ind_{\Eext}(\bXext)
\\ \times \left[ \prod_{i=1}^\Nn  \P^\beta_{\tni, \Lai, \Vext} \left[ \EE_i \right] \times \KNbetai \right] \dd \Xext,
\end{multline}
\cor{as well as the corresponding lower bound (notice the minus sign in front of $\ErrorCI$ in the exponent and the $\inf$ instead of $\sup$):}
\begin{multline}
\label{after_sup_2}
\EN \left[ \prod_{i=1}^\Nn \G_i(\bXN) \ind_{\EE_N} \right] \geq \exp\left(- \beta \sup_{\bXN \in \EE_N} \ErrorCI(\bXN) \right) \times \inf_{\bXext \in \Eext, \{\tn_i\} \mathrm{adm.}} \prod_{i=1}^\Nn \E^\beta_{\tn_i, \Lai, \Vext} \left[ \G_i(\bX) \Cond \EE_i \right]  \\
\times \frac{1}{\KNbeta}  \sum_{  \{\tn_i\} \ \mathrm{adm.} } \binom{N}{\tn_1 \dots \tn_\Nn} \exp\left(-\beta \tInt[\Dd_1, \dots, \Dd_\Nn] \right)  \int_{(\Ext)^{\nnext}} \exp\left(-\beta \F_{\Ext}(\bXext) \right) \ind_{\Eext}(\bXext)
\\ \times \left[ \prod_{i=1}^\Nn \P^\beta_{\tni, \Lai, \Vext} \left[ \EE_i \right] \times \KNbetai \right] \dd \Xext.
\end{multline}

\subsubsection*{Step 5. Treating the partition function.}
\cor{Let us start with a simple observation: by definition, we have:
\begin{equation*}
\KNbeta = \int_{(\LN)^N} \exp\left( - \beta \FN(\bXN) \right) \dd \XN \geq \int_{(\LN)^N} \ind_{\EE_N} \exp\left( - \beta \FN(\bXN) \right) \dd \XN,
\end{equation*}
and on the other hand we have $\PNbeta(\EE_N) = \frac{1}{\KNbeta} \int_{(\LN)^N} \ind_{\EE_N} \exp\left( - \beta \FN(\bXN) \right) \dd \XN$, which yields:}
\begin{multline}
\label{KNbetaAN}
\cor{\left( \int_{(\LN)^N} \ind_{\EE_N} \exp\left( - \beta \FN(\bXN) \right) \dd \XN\right)^{-1} \times \PNbeta(\EE_N) = \frac{1}{\KNbeta}} \\
\cor{\leq \left( \int_{(\LN)^N} \ind_{\EE_N} \exp\left( - \beta \FN(\bXN) \right) \dd \XN\right)^{-1}}.
\end{multline}

\cor{We now re-express the integral appearing in \eqref{KNbetaAN} using our setup. In order do that,} we take all the $\G_i$'s equal to the constant function $1$ in \eqref{after_sup} and \eqref{after_sup_2}, and we simplify by $\frac{1}{\KNbeta}$, which gives:
\begin{multline}
\label{again}
\int_{(\LN)^N} \ind_{\EE_N} \exp\left( - \beta \FN(\bXN) \right) \dd \XN  \\
\leq \exp\left(\beta \sup_{\bXN \in \EE_N} \ErrorCI(\bXN) \right) \times \sum_{ \{\tn_i\} \ \mathrm{adm.} } \binom{N}{\tn_1 \dots \tn_\Nn} \exp\left(-\beta \tInt[\Dd_1, \dots, \Dd_\Nn] \right) \\
 \times  \int_{\Ext^{\nnext}} \exp\left(-\beta \F_{\Ext}(\bXext) \right) \ind_{\Eext}(\bXext)  \left[ \prod_{i=1}^\Nn \P^\beta_{\tni, \Lai, \Vext} \left[ \EE_i \right] \times \KNbetai \right] \dd \Xext,
\end{multline}
\cor{as well as a lower bound:
\begin{multline}
\label{again2}
\int_{(\LN)^N} \ind_{\EE_N} \exp\left( - \beta \FN(\bXN) \right) \dd \XN  \\
\geq \exp\left(- \beta \sup_{\bXN \in \EE_N} \ErrorCI(\bXN) \right) \times \sum_{ \{\tn_i\} \ \mathrm{adm.} } \binom{N}{\tn_1 \dots \tn_\Nn} \exp\left(-\beta \tInt[\Dd_1, \dots, \Dd_\Nn] \right) \\
 \times  \int_{\Ext^{\nnext}} \exp\left(-\beta \F_{\Ext}(\bXext) \right) \ind_{\Eext}(\bXext)  \left[ \prod_{i=1}^\Nn  \P^\beta_{\tni, \Lai, \Vext} \left[ \EE_i \right] \times \KNbetai \right] \dd \Xext,
\end{multline}}

\subsubsection*{Step 6. Conclusion}
\cor{Combining \eqref{after_sup}, the second inequality in \eqref{KNbetaAN} and \eqref{again} we obtain \eqref{eq:CI} (note that many terms cancel out in the ratio). Combining \eqref{after_sup_2}, the first equality in \eqref{KNbetaAN} and \eqref{again2} we obtain the “converse” inequality \eqref{eq:CI_LB}.}
\end{proof}

\section{Generalized 2DOCP's arising as sub-systems}
\label{sec:CI}
Let $\Lambda$ be a disk of center $\omega \in \R^2$ and radius $T$, and let $\nn \geq 1$ be an integer, corresponding to the number of points in $\La$. In general we may have $\nn \neq |\La|$.

\renewcommand{\Vv}{\VextL}
\newcommand{\RR}{\mathfrak{R}}
\subsection{Good external potentials, good sub-systems}
\label{sec:good_good}
Let $\VextL$ be an external field as in Section \ref{sec:with_external} (we will eventually use this for the particular choice \eqref{def:Vext}). We introduce two definitions to pinpoint “good situations” for the generalized 2DOCP measure $\PnLV$ introduced in \eqref{def:PnLV}.

\cor{We start by describing conditions on the external potential $\VextL$. Ideally, we would like it to be smooth, with a derivative of order $1$, and to be generated by charges (see \eqref{def:Vext}) that are at distance $\sim 1$ from the boundary $\partial \La$ and nicely spread out around it. However, later on in “real” situations i.e. when we consider $\VextL$ as in \eqref{def:Vext} where $\bXN$ is sampled from the Gibbs measure,  we will not quite be able to guarantee that this is often enough the case. So we need to relax our constraints a bit and we end up with the following definition, which might seem somewhat artificial, but is designed to ensure that a generalized 2DOCP with a “good” external potential retains some key good properties, see Section \ref{sec:good_properties}.}
\begin{definition}[Good potential]
\label{def:GoofPotential}
We say that $\VextL$ is a “\emph{good external potential} on $\La$ with constant $\bCc$” when the following holds:
\begin{description}
	\item[1. Control up to the edge.] There exists a function $\wVextL$ on $\La$ satisfying:
\begin{equation}
\label{wVextLVextL}
\wVextL(x) = \VextL(x) \text{ if } \dist(x, \partial \La) \geq 1, \quad  \wVextL(x) \leq \VextL(x) + 100 \text{ for all } x \in \La,
\end{equation}
such that:
\begin{equation}
\label{wV_up_to_the_edge}
\left|\wVextL(x) - \wVextL(\omega)\right| \leq \bCc \times T \times \log^{3} T.
\end{equation}
\cor{This can be thought of as a substitute for the naive wish of $\VextL$ having a derivative of order $1$ on $\La$.}
\item[2. A technical decomposition.] $\VextL$ can be decomposed as the sum 
\begin{equation}
\label{Vextassum}
\VextL = \h^\nu + \RR,
\end{equation}
 where $\h^\nu$ is the logarithmic potential generated by some positive measure supported on \cor{the annulus of width $\hT := \log T$ defined by:
\begin{equation*}
\DD(\omega, T + \hT) \setminus \DD(\omega, T),
\end{equation*}}
and $\RR$ is harmonic in $\La$. We ask for the following control on the derivative of $\RR$ up to the edge:
\begin{equation}
|\RR|_{1, \La} \leq \bCc \times \log^2 T,
\end{equation}
and that moreover we control the mass of $\nu$ locally at scale $\hT = \log T$:
\begin{equation}
\label{assum_control_nuloc}
\sup_{x \in \partial \La} \nu(\DD(x,\log T)) \leq \bCc \log^2 T.
\end{equation}
\cor{These conditions can be thought of as expressing the fact that $\VextL$ should come from an underlying point configuration which is nicely spread near $\partial \La$ - but by “near” we mean here “in a $\log T$ neighbourhood” and by “nicely spread” we only mean that points cannot concentrate too much according to \eqref{assum_control_nuloc}.}
\end{description}
\end{definition}

\cor{We now describe some conditions on the configuration \emph{inside} $\La$. Ideally, we would the system within $\La$ to be globally neutral, and that all points be at distance $\sim 1$ from the boundary $\partial \La$. Again, it is not quite possible to guarantee that, so we place weaker convenient assumptions.}
\begin{definition}[Good sub-system]
\label{def:goodSS}
\label{def:EELA}
We consider the event “$\La$ is a good-system” defined as the subset of all point configurations $\bX \in \Conf$ such that:
\begin{enumerate}
\item The discrepancy $\Di(\bX, \La)$ satisfies: $|\Di(\bX, \La)| \leq T \log^2 T$.
\item $\bX$ belongs to the event $\EE_\La$ defined by:
\begin{enumerate}
    \item  There is absolutely no point in $\La$ at distance $\leq e^{- \log^2 T}$ from $\partial \La$.
    \item  There is no more than $T \log T$ points in $\La$  at distance $\leq 1$ from $\partial \La$.
\end{enumerate}
	\end{enumerate}
This event is of course $\La$-local and even $\{x \in \La, \dist(x, \partial \La) \leq 1\}$-local \cor{i.e it only depends on the configuration in a $1$-neighborhood of $\partial \La$.}
\end{definition}

\subsubsection*{The effective external potential is often good.}
Let $\La_i := \DD(\omega_i, T)$ ($1 \leq i \leq \Nn$) be the disks introduced at the end of Section \ref{sec:well-separated}. We define a common external potential for all $\La_i$'s by setting (cf. \eqref{def:Vext}):
\begin{equation}
\label{def:VextLi}
\VextL(x) := \int_{\LN \setminus \cup_{i=1}^\Nn \Lai} - \log |x-y| \dd \fN(y).
\end{equation}

\begin{proposition}[The effective external potential is often good]
\label{prop:VextIsOftenGood}
There exists a constant $\bCc$ depending only on $\beta$ (and on the choice of $\delta$ as in Assumption \ref{eq:assumption_distance}) such that the following holds.
With $\PNbeta$-probability greater than $1 - \Nn \exp\left(- \frac{\log^2 T}{\bCc} \right)$, for all $i = 1, \dots, \Nn$  the external potential $\VextL$ is a good external potential on $\La_i$ with constant $\bCc$ (in the sense of Definition \ref{def:GoofPotential}).
\end{proposition}
The proof of Proposition \ref{prop:VextIsOftenGood} is elementary but cumbersome. We postpone it to Section \ref{sec:proofVext}. Since the constant $\bCc$ given by Proposition \ref{prop:VextIsOftenGood} depends only on $\beta, \delta$, let us keep in mind that if we say that some constant depends on the “good external potential constant” $\bCc$ then in fact it itself only depends on $\beta, \delta$.

\subsubsection*{The sub-systems are often good}
\begin{lemma}
\label{lemma:oftenGood}
For $\Cc_\beta$ large enough, if $T \geq \Cc_\beta$, then with $\PNbeta$-probability $\geq 1 - \Nn \exp\left(- \frac{\log^2 T}{\Cc_\beta} \right)$, for all $i = 1, \dots, \Nn$, the conditions of Definition \ref{def:goodSS} (expressing the fact that “$\La_i$ is a good sub-system”) are satisfied.
\end{lemma}
\begin{proof}
Using the “discrepancy” part \eqref{discr_Di3} of the local laws (Proposition \ref{prop:local_laws}) we see that if $T$ is large enough (depending on $\beta$) then for any fixed $i$ we have:
\begin{equation*}
\PNbeta\left( |\Di(\bXN, \Lai)| \geq T \log^{\cor{2}} T \right) \leq \exp\left( - \frac{\log^{\cor{4}} T}{\Cc_\beta} \right),
\end{equation*}
\cor{which yields the first item in Definition \ref{def:goodSS}.} Checking the conditions \cor{defining the event} $\EE_\La$ requires more care.

\begin{claim}[The conditions of $\EE_\La$ are often met]
\label{claim:EELA_often}
For any fixed $i$ we have:
\begin{equation*}
\PNbeta \left( \EE_{\La_i} \right) \geq 1 - \exp\left(- \frac{\log^2 T}{\Cc_\beta} \right).
\end{equation*}
\end{claim}
\begin{proof}[Proof of Claim \ref{claim:EELA_often}]
\cor{Recall that for $\bX$ to be in $\EE_{\La_i}$, we require two things: absolutely no point in a $\leq e^{- \log^2 T}$-neighborhood of $\partial \La_i$, and no more than $T \log T$ points in a $1$-neighborhood of $\partial \La_i$.}

From the local laws, it is easy to see that the second condition is often satisfied. Indeed we can cover the $1$-neighborhood of $\partial \La_i$ by $\O(T)$ squares of sidelength $1$, each of which contains at most $\frac{\log T}{100 \Cc}$ points with probability $\geq 1 - \exp\left(- \frac{\log^2 T}{\Cc_\beta} \right)$ (for $T$ large enough). We conclude with a union bound \cor{on those $\O(T) = e^{\log T + \O(1)}$ squares}.

The first condition is more subtle. We rely on the “one-particle cluster” bound from \cite{Thoma:2022aa} recalled in Lemma \ref{lem:cluster_bounds}. 
\begin{itemize}
	\item First, cover the region $\Gamma_i := \{z \in \La, \dist(z, \partial \La_i) \leq e^{- \log^2 T}\}$ by $\O(T e^{\log^2 T})$ disks of radius $r = 10 e^{- \log^2 T}$.
	\item For each disk, we know from \eqref{eq:cluster_bound} that the probability of it being occupied by at least one particle is smaller than $\Cc_\beta r^2$.
	\item Then an union bound shows that the probability of at least one point falling anywhere in $\Gamma_i$ is bounded by $\Cc_\beta T e^{- \log^2 T}$, which concludes the proof of the claim.
\end{itemize}
\end{proof}
\cor{We have thus proven that the probability of $\bX$ to be in $\EE_{\La_i}$ for a given $i$ is larger than $1 - \exp\left(- \frac{\log^2 T}{\Cc_\beta} \right)$.} Finally, we use a union bound over $\Nn$ such events to handle all the $\La_i$'s at once, with probability larger than $1 - \Nn \exp\left(- \frac{\log^2 T}{\Cc_\beta} \right)$.
\end{proof}

Until the end of Section \ref{sec:CI}, we consider a “good external potential” $\VextL"$ with constant $\bCc$ as in Definition \ref{def:GoofPotential} and we assume that $|\nn - |\La|| \leq T \log T$. 

\subsection{2DOCP's with non-uniform neutralizing background}
\label{sec:PnLV}
Instead of adding $\VextL$, let us consider 2DOCP's in which the “neutralizing background” is no longer the uniform one, but a perturbation thereof.

\begin{itemize}
		\item Let $\mm$ be a (non-negative) measure on $\La$ and let $\zeta$ be some non-negative function on $\R^2$. Assume that $\mm$ is supported in $\La$ and that $\zeta$ vanishes on the support of $\mm$. 
		
		\item Assume that the measure $\mm$ can be written as the sum of a measure which has a bounded density with respect to the Lebesgue measure $\mm_0$ on $\La$ and of a singular measure which has a bounded density with respect to the arc-length measure $\dd s$ on $\partial \La$. 

		This assumption will be justified later in Section \ref{sec:applications_to_generalized}. In particular, it implies that the logarithmic self-interaction energy of $\mm$ given by $\iint - \log|x-y| \dd \mm(x) \dd \mm(y)$ is finite.

		\item Let $\Xn$ be a $\nn$-tuple of points in $\La$ and $\bXn$ be the associated point configuration, we let $\FL(\bXn, \mm)$ be the logarithmic interaction energy \emph{computed with respect to} $\mm$, namely:
		\begin{equation}
		\label{def:FLnuW}
			\FL(\bXn, \mm) := \hal \iint_{(x,y) \in \LN \times \LN, x \neq y} - \log |x-y| \dd (\bX_\nn - \mm)(x) \dd (\bX_\nn - \mm)(y).
		\end{equation}
	\end{itemize}

We then define a probability density $\PnL(\cdot, \mm, \zeta)$ on the space of $\nn$-tuples of points in $\Lambda$ by setting:
\begin{equation}
\label{def:PnLV2}
\dd \PnL(\Xn, \mm, \zeta) := \frac{\exp\left( - \beta \left( \FL(\bXn, \mm) + \nn \sum_{i=1}^\nn \zeta(x_i) \right) \right)}{\KnL(\mm, \zeta)} \dd \Xn,
\end{equation}
where $\KnL(\mm, \zeta)$ is the partition function, namely the normalizing constant:
\begin{equation}
\label{def:KnLV2}
\KnLV(\mm, \zeta) := \int_{\Lambda^\nn} \exp\left( - \beta \left( \FL(\bXn, \mm) + \nn \sum_{i=1}^\nn \zeta(x_i) \right) \right) \dd \Xn.
\end{equation}

\subsubsection*{\cor{Connection between the two notions of generalized 2DOCP's.}}
We have just introduced two ways to extend the definition of a 2DOCP:
\begin{enumerate}
	\item By adding the influence of an external potential $\VextL$ to $\FL(\bXn)$, as in Section \ref{sec:with_external}. 
	\item By changing the background measure $\mm$ and considering $\FL(\bXn, \mm)$ (and adding $\zeta$) as in Section \ref{sec:PnLV}.
\end{enumerate}
In fact, these two points of view are, to a large extent, equivalent. Indeed, there is a way to pass from an external potential $\VextL$ to the corresponding background measure, which is often called the associated “equilibrium measure”, and vice-versa. In other words, there is a correspondence: 
\begin{equation}
\label{informalcorrespondance}
\text{ $\VextL$ } \leftrightarrow \text{  $\mm$ and $\zeta$, such that $\PnLV = \PnL(\cdot, \mm, \zeta)$}.
\end{equation}
  Moreover, for external potentials that are “nice” enough, the determination of the corresponding equilibrium/background measure is \corT{well understood. We give here a brief explanation, and refer e.g. to \cite[Chap. 2]{serfaty2024lectures}, or to the exhaustive reference \cite{saff2013logarithmic}, for more details. The starting point is to realize that the energy $\FL(\bXn) = \FL(\bXn, \mm_0)$ defined in \eqref{def:FL} can be equivalently written as:
  \begin{equation*}
\FL(\bXn, \mm_0) = \hal \iint_{(x,y) \in \La \times \La, x \neq y} - \log|x-y| \dd \bXn(x) \dd\bXn(y) + \int_{\La} \VV(x) \dd \fL(x),
  \end{equation*} 
  where $\VV(x) := \int_{\La} \log|x-y| \dd \mm_0(y)$ is the logarithmic potential generated by the Lebesgue measure $\mm_0$ on $\La$ (a computation shows that $\VV$ is quadratic on $\La$) and $\fL := \bXn - \mm_0$. 
   To be consistent with our convention, if $x \notin \La$, we need to set $\VV(x) = + \infty$. We think of $\VV$ as a “reference” potential, to which we may add perturbations $\VextL$. So for the moment, we have the correspondence 
   $$
   \VextL = 0 \leftrightarrow \mm = \mm_0, \zeta = 0.
   $$
   Adding $\int_{\La} \VextL(x) \dd \fL(x)$ to $\FL(\bXn)$ as in Section \ref{sec:with_external}, the energy becomes
\begin{equation}
\label{energybecomes}
\hal \iint_{(x,y) \in \La \times \La, x \neq y} - \log|x-y| \dd \bXn(x) \dd\bXn(y) + \int_{\La} \left(\VV + \VextL\right)(x) \dd \fL(x).
\end{equation}
 \newcommand{\mueq}{\mu_{\mathrm{eq}}}
We can now use a classical result, going back essentially to \cite{frostman1935potentiel}: let $\Phi : \R^2 \to \R \cup \{\infty\}$ be a lower semi-continuous potential growing fast enough at infinity, there is a unique probability measure $\mueq$ that minimizes the functional
$$
\mu \mapsto \hal \iint_{\R^2} - \log|x-y| \dd \mu(x) \dd \mu(y) + \int_{\R^2} \Phi(x) \dd \mu(x),
$$
the measure $\mueq$ is compactly supported on $\{ \Phi < + \infty \}$, with a density proportional to $\Delta \Phi$ on its support (provided $\Phi$ is smooth enough), and there is a function $\zeta_\Phi$ which vanishes on the support of $\mueq$, is $\geq 0$ outside, such that the following identity holds for all probability measures $\mu$:
\begin{multline}
\label{mueqrewrite}
\hal \iint_{\R^2} - \log|x-y| \dd \mu(x) \dd \mu(y) + \int_{\R^2} \Phi(x) \dd \mu(x) \\ = \hal \iint_{\R^2} - \log|x-y| \dd \left(\mu - \mueq \right)(x) \dd \left(\mu-\mueq \right)(y) + \int_{\R^2} \zeta_\Phi(x) \dd \mu(x) + \mathrm{const.}
\end{multline}
Note that the left-hand side is of the form “energy of $\mu$, plus external potential” while the right-hand side is of the form “energy of $\mu - \mueq$, plus effect of $\zeta$”. With some rescaling and changing of variables, we can use the same techniques to find a measure $\mm$ (which plays the role of $\mueq$) and a function $\zeta$, both associated to $\VV + \VextL$ (which plays the role of $\Phi$) such that the following identity holds (cf. \eqref{mueqrewrite}, keeping in mind that additive constants are irrelevant when defining energies):
\begin{multline*}
\hal \iint_{(x,y) \in \La \times \La, x \neq y} - \log|x-y| \dd \bXn(x) \dd\bXn(y) + \int_{\La} \left(\VV + \VextL\right)(x) \dd \fL(x) \\
=  \FL(\bXn, \mm) + \nn \int_\La \zeta(x) \dd \bXn(x) + \mathrm{const.}
\end{multline*}
which connects the formulation “$\FL(\cdot, \mm_0)$ and $\VextL$” to another one of the type “$\FL(\cdot, \mm)$ and $\zeta$”. 
}

\cor{In this paper, there is an issue due to the difficulty in controlling the behavior of $\VextL$ near $\partial \La$, which might give rise to singular equilibrium measures. We postpone the necessary discussion (inspired by similar concerns in \cite{bauerschmidt2017local}) to Section \ref{sec:effect-perturb-eq}. An instance of the potential / equilibrium measure correspondence in this context is stated in Proposition \ref{prop:correspondence} below.}

\subsubsection*{Electric formalism for $\FL(\bXn, \mm)$}
We extend here some of the formalism from Sections \ref{sec:ElectricFields} and \ref{sec:logEnergy} in a fairly straightforward way. Define the “true electric potential/field” associated to a point configuration $\bXn$ and the background measure $\mm$ as (cf. Definition \ref{def:true}):
\begin{equation}
\label{true_mm}
\h^{\bXn, \mm} := - \log \ast (\bXn - \mm), \quad \nabla \h^{\bXn, \mm} = - \nabla \log \ast (\bXn- \mm),
\end{equation}  
their truncated versions being defined as in Definition \ref{def:trunc}. Then we have the following identity, which extends \eqref{secondformulation}:

\begin{lemma} 
\label{lem:extendSecond}
Assume that the total mass of $\mm$ is equal to $\nn$. Let $\{\eta(x), x \in \bXn \}$ be a truncation vector with $\eta(x) \leq \rr(x)$ (the nearest-neighbor distance introduced in \eqref{def:nn_distance}) for all $x \in \bXn$. We have:
\begin{equation}
\label{extendSecondFormulation}
\FL(\bXn, \mm) = \hal \left( \frac{1}{2\pi} \int_{\R^2} |\nabla \h^{\bXn, \mm}_{\veta}|^2 + \sum_{x \in \bXn} \log \eta(x)\right)  - \sum_{x \in \bXn} \int_{\DD(x,\eta(x))} \ff_{\eta(x)}(t-x) \dd \mm(t),
\end{equation}
provided the disks $\DD(x, \eta(x))$ do not intersect $\partial \La$.
\end{lemma}
\begin{proof}[Proof of Lemma \ref{lem:extendSecond}]
Since we place ourselves away from the singular part of $\mm$, the proof works exactly as when $\mm$ has a bounded density, see \cite[Lemma 2.2]{armstrong2019local}.
\end{proof}
We also extend the notation $\Ener$ from \eqref{def_Ener} by setting: 
\begin{equation}
\label{eq:EnermmOmega}
\Ener(\bXn, \mm, \Omega) := \int_{\Omega} |\nabla \h_{\veta}^{\bXn, \mm}|^2.
\end{equation}

\subsection{Good properties of sub-systems with good external potentials}
\label{sec:good_properties}
\subsubsection*{\cor{The background measure associated to a “good” external potential}}
\begin{proposition}
\label{prop:correspondence}
If $\VextL$ is a good external potential \cor{in the sense of Definition \ref{def:GoofPotential}}, there exists a probability measure $\muW$ on $\La$ and a function $\zetaW$ such that (cf. \eqref{informalcorrespondance}):
\begin{equation}
\label{identity_entre_points_de_vue}
\PnLV(\cdot) = \PnL\left(\cdot , \nn \muW, \zetaW \right).
\end{equation}
The key features of $\muW$ are that:
\begin{enumerate}
	\item It might be singular on $\partial \La$.
	\item It might have “holes” near $\partial \La$ (and $\zetaW > 0$ on these holes).
	\item It has a constant, positive density $\frac{1}{\nn}$ as soon as one looks at distance $\geq \bCc' \log T$ from $\partial \La$ (with $\bCc'$ depending on the “good external potential” constant $\bCc$). On that region we have $\zetaW \equiv 0$.
\end{enumerate}
\end{proposition}
We refer to Section \ref{sec:applications_to_generalized} for a definition and precise study of $\muW$ and to Section \ref{sec:applications_to_generalized} for a proof of \eqref{identity_entre_points_de_vue}. 

For good external potentials, it is thus equivalent to consider the “2DOCP with external background” $\PnLV$ (as in \eqref{def:PnLV}) or the “2DOCP with background measure” $\PnL(\cdot, \mm, \zeta)$ (with $\mm := \nn \muW$) and we will simply write $\PnLL$ for the corresponding Gibbs measure (and $\EnLL$ for expectations under $\PnLL$). 

\medskip

We now compare the properties of this “generalized 2DOCP” to the ones of the original Gibbs measure $\PNbeta$. The “good properties” are easier to obtain away from the boundary of $\La$ and for simplicity we will often work in “the bulk” $\Lab$ defined as:
\begin{equation*}
\Lab := \DD(\omega, T/2) \subset \DD(\omega, T) = \La.
\end{equation*}
\renewcommand{\Vv}{\mathrm{V}}

Recall that the event $\EE_\La$ was introduced in Definition \ref{def:goodSS}. \cor{In the following (Proposition \ref{prop:global_law_SUBSYS}, Proposition \ref{prop:local_law_SUBSYS}, and Section \ref{conseqSubSys}) we are always working \emph{conditionally on $\EE_\La$.} In particular, by the first item defining $\EE_\La$, we know that the number $\nn$ of points is “close” to the size $|\La| = \pi T^2$ of $\La$, in the sense that:
\begin{equation*}
\left| \nn - \pi T^2 \right| \leq T \log^2 T.
\end{equation*}
For this reason, there is no dependency in $\nn$ in the statements, as everything can be expressed in terms of~$T$. 
}

\subsubsection*{Global law}
\newcommand{\Cglob}{\hyperref[prop:global_law_SUBSYS]{\mathtt{C}_{\mathtt{Global}}}} 
\begin{proposition}[Global law for sub-systems with good external potential]
\label{prop:global_law_SUBSYS}
There exists a constant $\Cglob$ depending only on $\beta$ and the “good potential” constant $\bCc$ such that:
\begin{equation}
\label{global_law_SS_Statement}
\log \EnLL \left[  \exp\left( \frac{\beta}{2} \FL(\bX, \nn \muW) \right) \Big| \EE_\La \right] \leq \Cglob T^2 \log^5 T. 
\end{equation}
\end{proposition}
We prove Proposition \ref{prop:global_law_SUBSYS} in Section \ref{sec:global_law}. The bound \eqref{global_law_SS_Statement} should of course be compared with \eqref{eq:global_law_FN} which is valid for the full system. When considering sub-systems with good external potentials we are (only) losing some power of $\log T$, which we have not tried to optimize.

\newcommand{\Cloc}{\hyperref[prop:local_law_SUBSYS]{\mathtt{C}_{\mathtt{Local}}}} 

\subsubsection*{Local laws in the bulk}

\begin{proposition}[Local laws for sub-systems with good external potential]
\label{prop:local_law_SUBSYS}
There exists a universal constant $\Cc$ and a constant $\Cloc$ (depending only on $\beta$ and the “good potential” constant $\bCc$) such that if $T$ is large enough (depending on $\beta, \bCc$) then for all $\ell \geq \rho_\beta$, for all $x$ in $\La$, provided that the square $\sq(x, \ell)$ is included in $\Lab$ we have:
\begin{equation}
\label{eq:LocalLawSS}
\log \EnLL \left[ \exp\left(\frac{\beta}{2} \Ener\left(\bX_\nn, \nn \muW, \sq(x, \ell)\right)\right) \Big| \EE_\La \right] \leq \Cloc \beta \ell^2,
\end{equation}
where the “electric energy” is computed with respect to the background measure $\nn \muW$, \cor{see \eqref{eq:EnermmOmega},}  and also:
\begin{equation*}
\log \EnLL \left[ \exp\left(\frac{\beta}{\Cc} \Points\left(\bX_\nn, \sq(x, \ell)\right)\right) \Big| \EE_\La \right] \leq \Cloc \beta \ell^2.
\end{equation*}
\end{proposition}
This should be compared to \cor{the local laws for the full system recalled in} Proposition \ref{prop:local_laws} .

\begin{proof}[Proof of Proposition \ref{prop:local_law_SUBSYS}]
This follows from \cor{the analysis of} \cite{armstrong2019local}, but it deserves some explanation, which we provide in Section \ref{sec:proofLL}.
\end{proof}

\subsection{Sub-systems with good external potentials: consequences}
\label{conseqSubSys}
\subsubsection*{Discrepancy bounds}
Once local laws hold (in the bulk), we retrieve all statements that rely purely on energy considerations. In particular the analogous of \cite[(1.18)]{armstrong2019local} is valid, namely:
\begin{lemma}[Discrepancy bounds in sub-systems]
\label{claim:DiSS}
If $\sq(x, \ell) \subset \Lab$ then:
\begin{equation}
\label{discr_bounds_in_SS}
\log \EnLL \left[ \exp\left(\frac{\beta}{\Cc} \frac{\Di^2(\bXn, \sq(x, \ell))}{\ell^2}  \right) \Big| \EE_\La  \right] \leq \Cloc.
\end{equation}
In particular, we have an a priori Poisson-like bound on the number variance:
\begin{equation*}
\EnLL \left[ \Di^2(\bXn, \sq(x, \ell)) \Big| \EE_\La  \right] \leq \Cc \Cloc \ell^2,
\end{equation*}
together with a tail estimate: if $\ell$ is larger than some constant depending on $\beta, \bCc$ 
\begin{equation*}
\PnLL \left[ |\Di(\bXn, \sq(x, \ell))| \geq \ell \log \ell \Big| \EE_\La  \right] \leq \exp\left( - \frac{\log^2 \ell}{\Cc_\beta} \right).
\end{equation*}
\end{lemma}
\begin{proof}[Proof of Lemma \ref{claim:DiSS}]
The proof is as in \cite{armstrong2019local}, using an inequality that relates the presence of discrepancy to a certain energy cost, e.g. \cite[Lemma B.4]{armstrong2019local}.
\end{proof}
\begin{remark}
\label{rem:discr_shape}
As an inspection of the (short) proof of \cite[Lemma B.4]{armstrong2019local} quickly reveals, there is nothing specific to a square in the previous claim, and it also applies to a disk of radius $\ell$, or to any “reasonable” shape like the boxes introduced in Definition \ref{def:boxes}.
\end{remark}

\subsubsection*{Treating smooth linear statistics}
Following exactly the same proof as in \cite{MR3788208,serfaty2020gaussian} (or alternatively as in \cite{MR4063572}), one would obtain a control on linear statistics of smooth enough test functions supported in the bulk $\Lab$. We do not need it here, however it will be crucial for us to retrieve a specific property (the “smallness of the anisotropy”), but since this only serves as a tool for another result (the quantitative translation-invariance estimate presented in Section \ref{sec:approx_QI}) we postpone the corresponding discussion to an appendix (see Section \ref{sec:smallani}). The only result that we will quote directly is one about \emph{expectations} for fluctuations of linear statistics.

\begin{lemma}[Expectation of linear statistics in the bulk of subsystems]
\label{lem:bound_fluct_smooth_SS}
There exists a constant $\bCc$ depending only on~$\beta$, and on the “local laws” constant $\Cloc$ (thus on the “good potential” constant) such that the following holds.

Let $\varphi$ be a function of class $C^2$, compactly supported on a disk of radius $\ell \geq \rho_\beta$ included in $\Lab$. Then:
\begin{equation}
\label{borne_espereance}
\left| \EnLL \left[ \Fluct[\varphi] \Big| \EE_\La \right] \right| \leq \bCc |\varphi|_{\2} \ell^2. 
\end{equation}
\end{lemma}
\begin{proof}[Proof of Lemma \ref{lem:bound_fluct_smooth_SS}]
This follows from the proofs of \cite{MR3788208,serfaty2020gaussian} but is not explicitly written as such. Let $t$ be a small parameter, we know from \cite[Prop. 2.10]{MR3788208} that:
\begin{equation*}
\EnLL  \left[ e^{t \Fluct[\varphi]} \Big| \EE_\La \right] = \frac{\mathrm{K}_{\nn, \La}(\mm_s)}{\mathrm{K}_{\nn, \La}(\mm)} e^{\O(t^2)},
\end{equation*}
where $\mm_s$ is the measure $\mm_s := \mm - s \Delta \varphi$, where $s = \frac{t}{2\pi \beta}$. This is valid because $\varphi$ is assumed to be supported in $\Lab$, where $\mm$ has density $1$. By e.g. \cite[Lemma 3.6]{MR3788208} we know that we can replace $\mm_s$ by the approximate measure $\tilde{\mm}_s := (\Id + s \nabla \varphi) \# \mm$ up to quadratic terms, i.e.:
\begin{equation*}
\EnLL  \left[ e^{t \Fluct[\varphi]} \Big| \EE_\La \right] = \frac{\mathrm{K}_{\nn, \La}(\tilde{\mm}_s)}{\mathrm{K}_{\nn, \La}(\mm)} e^{\O(t^2)}.
\end{equation*}
Next, a consequence of \cite[Prop. 4.2]{serfaty2020gaussian} is that the ratio of partition functions can be written as:
\begin{equation*}
\frac{\mathrm{K}_{\nn, \La}(\tilde{\mm}_s)}{\mathrm{K}_{\nn, \La}(\mm)} = \exp\left(s \O\left(|\varphi|_\2 \EnerPoints(\bX_\nn, \supp \nabla \varphi)\right) + \O(s^2)\right),
\end{equation*}
and thus taking the limit $t\to 0$ (or equivalently $s \to 0$) and using the local laws we get \eqref{borne_espereance} after identifying the first order terms.
\end{proof}

\section{Quantitative estimate on translation-invariance}
\label{sec:approx_QI}

\newcommand{\Cfep}{\Cc}
\newcommand{\fep}{f^{(\epsilon)}}
\newcommand{\fbep}{\bar{f}^{(\epsilon)}}
\newcommand{\uep}{1/\epsilon}
\newcommand{\eep}{e^{\uep}}
\newcommand{\Wep}{\mathtt{W}^{(\epsilon)}}
\newcommand{\Wel}{\mathtt{W}^{(\epsilon, \ell)}}
\newcommand{\Phit}{\Phi_t^{\epsilon, \ell}}
\newcommand{\Phimt}{\Phi_{-t}^{\epsilon, \ell}}
\newcommand{\Psit}{\psi_t^{\epsilon, \ell}}

\cor{Roughly speaking, the goal in this section is to quantify how much translation-invariance there is within a finite 2DOCP by estimating how much the expectation of a test function $\G : \Conf \to \R$ changes if one applies a translation on the configuration, i.e. by comparing $\EN\left[\G(\bXN)\right]$ and $\EN\left[\G(\bXN + t \vu)\right]$ where $t \in \R$ and $\vu$ is some fixed vector - we take $\vu := (0,1)$ in $\R^2$ in the sequel. 
\begin{itemize}
	\item Due to the finiteness of the system, this error will likely be large if $\G$ depends on the entire configuration or if $|t|$ is too big. This is why below we will always make the assumption that $\bX \mapsto \G(\bX)$ depends only on the configuration $\bX$ within some finite box $\La \subset \LN$, and that $|t|$ is not too big - so that, in particular, the translated box $\La + t$ remains well within $\LN$.
	\item For our purposes, it is in fact enough to compare $\EN\left[\G(\bXN)\right]$ with the average effect of two opposite translations $\hal \left( \EN\left[\G(\bXN + t \vu)\right] + \EN\left[\G(\bXN - t \vu)\right] \right)$. Although seemingly innocent, this old trick is extremely useful - though the analysis remains difficult.
	\item We actually want to study translation-invariance not only for the canonical Gibbs measure $\PNbeta$ of the 2DOCP, but for generalized 2DOCP's as introduced in the previous sections.
\end{itemize}
As mentioned in the introduction (see Section \ref{sec:strategy}), the usual strategy consists in:
\begin{enumerate}
	\item Constructing “localized translations”, which act as translations on some large region and as the identity outside some larger region.
	\item Estimating the effect of such localized translations on the energy of the system.
	\item Using localized translations (item 1.) as a change of variables within the Gibbs measure and inserting the energy estimates (item 2.) in order to derive conclusions on the expectations.
\end{enumerate}
We need here to construct localized translations which are “dampened” very slowly in order to reduce their effect on the energy. In Section \ref{sec:SW} we prepare the construction by introducing a certain vector field which we call “spin wave” by analogy with classical constructions in the realm of spin systems. In Section \ref{sec:slowvartrans} we define the localized translations. In Section \ref{sec:EFFECTOF} we state a crucial bound on their effect on the energy, whose proof is very delicate and postponed to Section \ref{sec:EffetTranslaEnerPreuve}. We use this to derive actual “approximate translation-invariance” estimates in the last two sections.}

\subsection{The “spin wave” and its properties}
\label{sec:SW}
\subsubsection*{An auxiliary function.}
For $\epsilon \in (0,1)$, let $\fep : x = (x_1, x_2) \in \R^2 \to \R$ be defined as:
\begin{equation*}
\fep(x) := \begin{cases}
x_1 & \text{ if } |x| \leq 1 \\
x_1 (1 - \epsilon \log |x|) & \text{ if } 1 \leq |x| \leq e^{1/\epsilon} \\
0 & \text{ if } |x| \geq e^{1/\epsilon}.
\end{cases}
\end{equation*}
The function $\fep$ is continuous, piecewise $C^2$ and compactly supported on the disk of radius $e^{1/\epsilon}$. We have, by direct computations:
\begin{itemize}
\item On (the interior of) the unit disk, $\partial_1 \fep \equiv 1$, $\partial_2 \fep \equiv 0$ and the second (and third) derivatives of $\fep$ vanish.
\item For $1 \leq |x| \leq  e^{1/\epsilon}$, we have : 
\begin{itemize}
\item $\partial_1 \fep(x) = 1 - \epsilon \log |x| - \epsilon \frac{x_1^2}{|x|^2}$, $\partial_2 \fep(x) = - \epsilon \frac{x_1x_2}{|x|^2}$.
\item The second derivatives of $\fep$ satisfy the pointwise bound $|\fep|_{\2, \star}(x) \preceq \frac{\epsilon}{|x|}$.
\item The third derivatives of $\fep$ satisfy the pointwise bound $|\fep|_{\3, \star}(x) \preceq \frac{\epsilon}{|x|^2}$.
\end{itemize}
\end{itemize}
In particular, observe that the first partial derivatives of $\fep$ are bounded by $1$ with a jump of size $\O(\epsilon)$ along both $\partial \DD(0, 1)$ and $\partial \DD(0, \eep)$, while the second partial derivatives have a jump of size $\O(\epsilon)$ along $\partial \DD(0, 1)$ and of size $\O(\epsilon / \eep)$ along $\partial \DD(0, \eep)$.

Thus after applying a mollification to $\fep$ at scale $\frac{1}{2}$ near $\partial \DD(0,1)$ and at scale $\hal \eep$ near $\partial \DD(0, \eep)$ we may consider a function $\fbep$ which is smooth, compactly supported in $\DD(0, 2 \eep)$ and such that (for some universal constant $\Cfep$):
\begin{itemize}
	\item $\fbep(x) = x_1$ for $|x| \leq \hal$.
	\item The first derivatives of $\fbep$ are bounded by $2$ on $\R^2$.
	\item $|\fbep|_{\2, \star}(x) \leq \Cfep \epsilon$ for $|x| \leq 2$ and $|\fbep|_{\2, \star}(x) \leq \Cfep \frac{\epsilon}{|x|}$ for $2 \leq |x| \leq 2 \eep$.
	\item $|\fbep|_{\3, \star}(x) \leq \Cfep \epsilon$ for $|x| \leq 2$ and $|\fbep|_{\3, \star}(x) \leq \Cfep \frac{\epsilon}{|x|^2}$ for $2 \leq |x| \leq 2 \eep$.
\end{itemize}

\subsubsection*{Definition of the “spin wave”.}
Next, we define our “spin wave”\footnote{This terminology alludes to similar constructions used in the theory of continuous spin systems to prove so-called “Mermin-Wagner” theorems, see e.g. \cite[Sec.~9.2]{friedli2017statistical} or \cite[Chapter 3]{simon2014statistical}.} $\Wep$ as the vector field
\begin{equation*}
\Wep := \nabla^\perp \fbep = \left( - \partial_2 \fbep, \partial_1 \fbep \right).
\end{equation*} 
\begin{lemma}
\label{lem:spinwave}
The following properties of $\Wep$ are straightforward:
\begin{enumerate}
	\item $\Wep$ is smooth, compactly supported on $\DD(0, 2 e^{1/\epsilon})$ (because so is $\fbep$).
	\item $\Wep(x) = \vu = (0, 1)$ for $|x| \leq \hal$ (because then $\fbep(x_1,x_2) = x_1$).
	\item $\Div \Wep = 0$ on $\R^2$ (by definition of $\Wep$ as the perpendicular gradient of a smooth function).
	\item $|\Wep|_\0 \leq 2$ (it is bounded by the first derivative of $\fbep$), $|\Wep|_\1 \leq \Cfep \epsilon$ and more precisely:
\begin{equation}
\label{sizeDWep}
\left| \Wep \right|_{\1, \star}(x) \leq \Cfep \times \begin{cases}  \epsilon & \text{ for } |x| \leq 2 \\
 \frac{\epsilon}{|x|} & \text{ for } 2 \leq |x| \leq 2 \eep.
\end{cases}
\end{equation}
\item $|\Wep|_{\2, \star}(x) \leq \Cfep \epsilon$ for $|x| \leq 2$ and $|\Wep|_{\2, \star}(x) \leq \Cfep \frac{\epsilon}{|x|^2}$ for $2 \leq |x| \leq 2 \eep$.
\end{enumerate}
\end{lemma}
We have thus constructed a smooth, \emph{divergence-free} vector field which is constant near the origin \emph{and has an arbitrary small $\Hone$ norm} (of order $\epsilon$). The downside is that the size of the support of $\Wep$ is \emph{exponential} with respect to the parameter $\uep$.

\subsection{Slowly varying localized translations}
\label{sec:slowvartrans}
For $\ell > 0$, let the vector field $\Wel$ be defined for $x \in \R^2$ as $\Wel(x) := \Wep(x/\ell)$. Since $\Wel$ is continuous and compactly supported, it generates a global flow, \cor{namely a one-parameter group} $\{\Phit\}_{t \in (-\infty, \infty)}$ \cor{of maps $\R^2 \to \R^2$, such that $(t, x) \mapsto \Phit(x)$ solves the following ODE on $(-\infty, + \infty) \times \R^2$:
\begin{equation}
\label{eq:ODEPhit}
\partial_t \Phit = \Wel \circ \Phit.
\end{equation}
} The maps $\Phit$ have the following properties:
\begin{lemma} \label{lemma_propPhi}
For all $|t| \leq \frac{\ell}{10}$ we have:
\begin{enumerate}
	\item $\Phit$ is an area-preserving diffeomorphism of $\R^2$.
	\item We have $\Phit(x) = x + t \vu$ for $|x| \leq \ell/4$.
	\item We have $\Phit(x) = x$ for $|x| \geq 2 \ell \eep $.
\end{enumerate}
\end{lemma}
Thus for fixed $|t| \leq \frac{\ell}{10}$ the diffeomorphism $\Phit$ coincides with the translation by $t \vu$ on the disk $\DD(0, \ell/4)$ and with the identity outside $\DD(0, 2 \eep)$, we call it a \emph{localized translation} as in \cite{MR1886241}. The main difference with the construction of \cite{MR1886241} is that we have $\int |\Phit - \Id|_{\1, \star}^2 = \O(\epsilon)$ instead of $\O(1)$. 

Interestingly enough, a \emph{bounded (but not small)} $H^1$ norm for $\Phit - \Id$ (which induces a bounded, but not small energy cost, as we will show in Section \ref{sec:EffetTranslaEnerPreuve}) is enough to prove translation-invariance \emph{in the infinite-volume setting} \cor{(this is done in \cite{leble2024dlr} for the 2DOCP)}, but fails to give anything valuable in finite-volume. However, according to a remark in \cite[Sec. III.7]{simon2014statistical} \emph{“it appears that any model in which this weaker property is valid, the [possibility of finding a construction with arbitrarily small energy cost exists]”}. This remains very intriguing to us. 

\begin{proof}[Proof of Lemma \ref{lemma_propPhi}]
The fact that $\Phit$ is area-preserving follows from Liouville's theorem, since the vector field $\Wep$ (and thus $\Wel$) is divergence-free by construction. Moreover, since we ensured that $\Wep \equiv \vu$ on $\DD(0, \hal)$, the rescaled vector field $\Wel$ coincides with $\vu$ on the disk $\DD(0, \ell/2)$ hence we have $\Phit(x) = x + t \vu$ as long as $x + t \vu$ remains in $\DD(0, \ell/2)$. In particular, if $|x| \leq \ell/4$ and since $|t| \leq \frac{\ell}{10}$ by assumption, we have $\Phit(x) = x + t \vu$. On the other hand, $\Wel$ vanishes identically outside $\DD(0, 2\ell \eep)$ and thus the flow there coincides with the identity map. 
\end{proof}
We will study $\Phit$ a bit more closely in Section \ref{study_of_phit} for technical purposes.

\subsection{Effect of localized translations on the energy}
\label{sec:EFFECTOF}
\newcommand{\HL}{\mathsf{H}_\La}
We  now apply the “spin wave” / “localized translation” construction to a sub-system with good external potential. Let $\La$ (a disk of radius $T$), $\nn$, $\VextL$ be as in Section \ref{sec:CI}, with $\VextL$ a good external potential, let $\mm = \nn \muW$ be the corresponding non-uniform background and assume that the properties listed in Section \ref{sec:good_properties} hold. Let \cor{$L \geq 100$} and let us choose $\epsilon, \ell$ in such a way that:
\begin{equation}
\label{eq:ellrhoepsilon}
e^{\sqrt{\log \ell}} \leq L \leq \frac{\ell}{10}, \quad  5 \ell \eep \leq T \leq 10 \ell \eep, \quad \log \ell \leq \epsilon^{-1}, \quad  \epsilon^{-1} \leq \ell^2 \leq \epsilon^{-3}.
\end{equation}
Let $\ErrAv(t, \epsilon, \ell, \bX)$ denote the following “averaging error”:
\begin{equation}
\label{def:HL}
\FL(\bX, \mm)  = \hal \left( \FL(\Phit \bX, \mm)  + \FL(\Phimt \bX, \mm) \right) + \ErrAv(t\vu, \epsilon, \ell, \bX),
\end{equation}
where $\Phit \bX$ denotes the configuration obtained after applying $\Phit$ to all the points of $\bX$.
\begin{proposition}
\label{prop:effectPhiEnergy}
There exists a constant $\Cc$ depending on $\beta$ on the “local laws” constant $\Cloc$ (and thus on the good external potential constant) such that:
\begin{equation}
\label{eq:control_ErrAv}
\PnLL \left[ \sup_{|t| \leq \frac{\ell}{10}} |\ErrAv(t \vu, \epsilon, \ell, \bXn)| \leq \Cc t^2 \epsilon \log \epsilon \Big| \EE_\La \right] \geq 1 - \exp\left( - \ell^2 \right).
\end{equation}
\end{proposition}
We prove Proposition \ref{prop:effectPhiEnergy} in Section \ref{sec:EffetTranslaEnerPreuve} and use it next to “mollify” observables before taking expectations.

\renewcommand{\Phit}{\Phi_t}
\renewcommand{\Phimt}{\Phi_{-t}}

\subsection{Effect of localized translations on expectations}
\cor{We recall that $\G : \Conf \to \R$ is said to be $\La$-local if $\G$ depends only on the configuration in $\La$, and that an event $\EE$ is said to be $\La$-local if $\1_{\EE}$ is $\La$-local.}

\begin{proposition}
\label{prop:Quantitative_invariance}
Let $\G$ be a measurable function on $\Conf(\La)$, let $\EE \subset \Conf(\La)$ be an event. Assume that the function $\G$ is $\DD(0, \ell/10)$-local and that $\EE$ is $\La \setminus \Lab$-local. Recall that $|t| \leq \ell / 10$.

For all $\tau \in (0,1)$ and for all $\sigma >0$, we have:
\begin{equation}
\label{eq:quantitative_invariance}
\EnLL\left[ \G(\bX) \Cond \EE \right] = \EnLL\left[ \hal \left( \G(\bX + t \vu) + \G(\bX - t \vu) \right) \Cond \EE \right] + \ErrorQI(t\vu, \epsilon, \ell, \tau, \sigma, \G, \bX),
\end{equation}
with an error term $\ErrorQI$ bounded as follows:
\begin{multline}
\label{def:ErrorQI}
\left| \ErrorQI(t\vu, \epsilon, \ell, \tau, \sigma, \G, \bX) \right| \leq  2 \left( \EnLL \left[ \G^2(\bX \pm t \vu) \Big| \EE \right]  + \sigma^2 \right)^\hal \times \left(e^{\beta \tau} - 1 \right) \\
 + 2 e^{\beta \tau} \left( \EnLL\left[ \G^2(\bX \pm \vu) \Big| \EE \right]  + \sigma^2 \right)^\hal \times \left(\PnLL \left[ |\G(\bX \pm t \vu)| \geq \sigma \right]\right)^\hal 
 \\
  + \EnLL\left[ \G^2(\bX) \Big| \EE \right]  \times  \left(\PnLL \left[ \ErrAv(t \vu, \epsilon, \ell, \bX) \geq \tau \big| \EE \right]\right)^\hal
\end{multline}
\end{proposition}
We postpone the proof of Proposition \ref{prop:Quantitative_invariance} to Section \ref{sec:proof_quanti}.

There is of course nothing special about the vector $\vu$, \cor{and we may replace it by any unit vector in Proposition \ref{prop:effectPhiEnergy} and in Proposition \ref{prop:Quantitative_invariance}.}

\subsection{Application: expectation of discrepancies in sub-systems}
\newcommand{\bt}{\bar{t}}
\begin{proposition}
\label{prop:mollification}
Let $\boxs$ be a box of size $L$ as introduced in Definition \ref{def:boxes}. Assume that it is contained in the disk $\DD(0, \ell/10)$. Then we have:
\begin{equation}
\label{eq:esperance_une_fois_mollifiee}
\left| \EnLL \left[ \Di(\bXn, \boxs) \Big| \EE_\La \right]  \right| \leq \bCc L \left(\epsilon \log \epsilon \log L\right)^{1/3}.
\end{equation}
\end{proposition}
The point of \eqref{eq:esperance_une_fois_mollifiee} is that if $\epsilon \log \epsilon \log L$ is $\ll 1$ then the bound on the average discrepancy is $\ll L$, and thus \emph{much better than the crude estimate via standard deviation} using Lemma \ref{claim:DiSS}. This is crucial for us.

\begin{proof}[Proof of Proposition \ref{prop:mollification}]
Let us take $\G(\bX) := \Di(\bX, \boxs)$, which is clearly $\DD(0, \ell/10)$-local in view of our assumption on $\boxs$. Moreover $\EE_\La$ (by its Definition \ref{def:EELA}) is clearly $\La \setminus \Lab$-local. Fix the parameter $\bt$ as:
\begin{equation}
\label{eq:defbt}
\bt := \left( \epsilon \log \epsilon \log L\right)^{-1/3}.
\end{equation}
We first check that $\bt$ has the correct range (using \eqref{eq:ellrhoepsilon}):
\begin{equation*}
\bt \leq \epsilon^{-1/3} \leq \ell^{2/3} \leq \frac{\ell}{10}.
\end{equation*}
In particular, for all vector $v \in \R^2$ with $\|v\| \leq \bt$ we may apply the result of Proposition \ref{prop:Quantitative_invariance} to a localized translation in the direction $\frac{v}{\|v\|}$ instead of $\vu$.

Introduce a smooth cut-off function $\bchi$ equal to $1$ for $|x| \leq \frac{\bt}{2}$ and to $0$ for $|x| \geq \bt$, such that $|\bchi|_{\kk} \leq \Cc \bt^{-k}$ (for $\kk = 1, 2$). We impose that $\bchi$ be an even function. For all $|v| \leq \bt$ we have:
\begin{multline*}
\EnLL\left[ \Di(\bX, \boxs) \Cond \EE_\La \right] = \EnLL\left[ \hal \left( \Di(\bX + v, \boxs) + \Di(\bX - v, \boxs) \right) \Cond \EE_\La \right] \\
+ \ErrorQI(v, \epsilon, \ell, \tau, \sigma, \Di(\cdot, \boxs), \bX).
\end{multline*}
Integrating this against $\frac{\bchi}{|\bchi|_{L^1}}$ (which has mass $1$) and using the fact that $\bchi$ is even we obtain:
\begin{equation*}
\EnLL \left[ \Di(\bXn, \boxs) \Big| \EE_\La \right]  = \frac{1}{|\bchi|_{L^1}} \int \bchi(v) \left(\EnLL \left( \Di(\bX + v, \boxs) \Big| \EE_\La \right) + \ErrorQI(v) \right) \dd v.
\end{equation*}
 Observe that $\Di(\bX + v, \boxs) = \Fluct[\ind_{\boxs}(\cdot + v)](\bX)$, so that we may re-write the integral in the right-hand side as a convolution. Introducing the function $\varphi := \ind_{\boxs} \ast \frac{\bchi}{|\bchi|_{L^1}}$ we get:
\begin{equation}
\label{esp_esp_mollifiee}
\left| \EnLL \left[ \Di(\bXn, \boxs) \Big| \EE_\La \right] - \EnLL \left[ \Fluct[\varphi] \Big| \EE_\La \right] \right| \leq \sup_{|v| \leq \bt} \ErrorQI\left(v, \epsilon, \ell, \tau, \sigma, \Di(\cdot, \boxs), \bXn\right).
\end{equation}
By construction, the function $\varphi$ is now a smooth cut-off function, equal to $1$ on $\{x \in \boxs, \dist(x, \partial \boxs) \geq \bt\}$
 and to $0$ for $\{x \notin \boxs, \dist(x, \partial \boxs) \geq \bt\}$. In particular, the support of $\nabla \varphi$ has an area $\O(L \bt)$, and moreover we have $|\varphi|_{\kk} \leq \Cc \bt^{-\kk}$ for $\kk = 1, 2$ (by Young's convolution inequality). From Lemma~\ref{lem:bound_fluct_smooth_SS}, we thus know:
 \begin{equation}
 \label{esperancemollifiee}
\left| \EnLL \left[ \Fluct[\varphi] \right] \Big| \EE_\La \right| \preceq \frac{1}{\bt^2} \times L \bt = \O\left(\frac{L}{\bt}\right),
 \end{equation}
with multiplicative constants depending on $\beta, \delta$.

On the other hand, we know by Proposition \ref{prop:effectPhiEnergy} that choosing $\bCc$ large enough we can ensure:
\begin{equation*}
\sup_{|v| \leq \bt} \left(\PnLL \left( \left\lbrace \left| \ErrAv(v, \epsilon, \ell, \bX) \right\rbrace \right| \geq \bCc \epsilon \log \epsilon \bt^2 \big| \EE_\La \right]\right) \leq \exp\left(- \ell^2 / \bCc \right).
\end{equation*}
Moreover $\Di(\bX + v, \boxs) = \Di(\bX, \boxs - v)$ and we have by Lemma \ref{claim:DiSS}
\begin{equation*}
\sup_{|v| \leq \bt} \EnLL\left[ \Di^2(\bX, \boxs - v) \Big| \EE_\La \right] \leq \bCc L^2
\end{equation*}
because the translated object $\boxs - v$ remains a box within $\Lab$ and, still by Lemma \ref{claim:DiSS}, for $\bCc$ large enough we have:
\begin{equation*}
\sup_{|v| \leq \bt} \PnLL \left( \left\lbrace |\Di(\bX + v, \boxs)| \geq  L \log L  \right \rbrace \right) \leq e^{- \frac{1}{\bCc} \log^2 L}.
\end{equation*} 
(Remark that the $\sup$ is outside $\PnLL, \EnLL$ in those bounds - it would be significantly more challenging otherwise.)

Thus, choosing $\tau = \bCc \epsilon \log \epsilon \bt^2$, $\sigma = L \log L$ in the statement of Proposition \ref{prop:Quantitative_invariance}, the error term $\ErrorQI$ reduces to:
\begin{equation*}
\ErrorQI(v, \epsilon, \ell, \tau, \sigma, \Di(\bX, \boxs), \bX) \preceq \bCc L \log L \left(\epsilon \log \epsilon \bt^2 + \exp\left(- \ell^2 / \Cc \right) +  \exp\left( - \frac{1}{\Cc} \log^2 L \right) \right),
\end{equation*}
uniformly for $|v| \leq \bt$ and thus (using \eqref{eq:ellrhoepsilon} and keeping only the dominant term):
\begin{equation}
\label{ErrorQI_v}
\sup_{|v| \leq \bt} \ErrorQI \leq \bCc L \log L \epsilon \log \epsilon \bt^2.
\end{equation}
Combining \eqref{esp_esp_mollifiee} with \eqref{esperancemollifiee} and \eqref{ErrorQI_v}, and using our choice \eqref{eq:defbt} for $\bt$, we obtain \eqref{eq:esperance_une_fois_mollifiee}.
\end{proof}

\begin{remark}
One can use the same argument to study expectations of discrepancies \emph{in the full system}. Then one does not need to use $\EE_\La$ since local laws are know to hold unconditionally, and one can take $\epsilon$ as small as $\log^{-1} N$.
\end{remark}

\begin{corollary}[An application]
\label{coro:appli_exp}
Assume that $L, T$ satisfy the following relation:
\begin{equation}
\label{eq:T_L_additional}
T = 100 L \exp(10 L).
\end{equation}
Then we have, \cor{with $\La := \DD(\omega, T)$ and $\boxs$ a box of size $L$ as in Definition \ref{def:boxes},}
\begin{equation}
\label{eq:discrepance_appliquee}
\left| \EnLL \left[ \Di(\bX, \boxs) \right] \right| \leq \bCc L^{0.67}.
\end{equation}
\end{corollary}
\begin{proof}
We can take $\ell = 10 L$ and $\epsilon = \ell^{-1} = \frac{1}{10 L}$, and $T = 100 L \exp(10 L) = 10 \ell e^{1/\epsilon}$. The conditions of \eqref{eq:ellrhoepsilon} are clearly satisfied. The right-hand side of \eqref{eq:esperance_une_fois_mollifiee} is then bounded by $\bCc L^{2/3} \log L = \O(L^{0.67})$ (for $L$ large enough).
\end{proof}

\section{Conclusion: proof of Theorem \ref{theo:main2}}
\newcommand{\AAR}{\mathcal{A}_R}
\label{sec:conclusion}

Let $\delta > 0$ be fixed. Let $x, R$ be such that $\dist(\DD(x, R), \partial \LN) \geq \delta \sqrt{N}$ as assumed in \eqref{eq:assumption_distance}. Let $\epsilon_R$ be chosen as:
\begin{equation}
\label{choix_alpha}
\epsilon_R := \log^{-0.3}(R).
\end{equation}

For simplicity, we will focus on the case of an \emph{excess} of points, i.e. a positive discrepancy, the other case being treated similarly. Define the event $\AAR$:
\begin{equation*}
\AAR := \{ \Di\left(\bXN, \DD(x,R) \right) \geq \epsilon_R R  \} 
\end{equation*}
The conclusion that we want to reach (as stated in \eqref{eq:resultat}) is that for $\cor{N}, R, \bCc$ large enough (depending on $\beta$ and on the parameter $\delta$, but not on $x$):
\begin{equation}
\label{tail_estimate_quonprouve}
\PNbeta \left( \AAR \right) \leq \exp\left( - \log^{1.5} R \right).
\end{equation}
\corT{We have not tried to optimize $\epsilon_R$ or the exponent in \eqref{tail_estimate_quonprouve}, what matters for us is that $\epsilon_R \to 0$ and that $1.5 > 1$ so our probabilistic tail is better than algebraic. With the methods of the present paper, there is a hard limit on the smallness of $\epsilon_R$ - it has to be at least $\log^{-1} R$.}

\medskip

\subsubsection*{\cor{The choice of parameters and constants.}}
\cor{In the steps below, we work with $\delta$ (as in \eqref{eq:assumption_distance}) fixed, which serves as a threshold for defining the “bulk” of the system (points at distance $\geq \delta \sqrt{N}$ from the edge). We then:
\begin{itemize}
	\item Take $N, x, R$ arbitrary, with the only constraint that $\dist(\DD(x, R), \partial \LN) \geq \delta \sqrt{N}$. 
	\item The choice of $R$ will determine the value of the parameters $L, T, M$, as well as some auxiliary parameters $s, \omega$. If $R$ is larger than some constant depending on $\beta$, a certain number of conditions will be satisfied: \eqref{eq:s_and_L}, \eqref{eq:condi_T2logTML}, which we use in Steps 1. and 2.
	\item In Step 3. we use Proposition \ref{prop:VextIsOftenGood}. It gives us a constant $\bCc_1$, which depends on $\beta, \delta$, and a condition on $R$ being large enough (depending again on $\beta, \delta$), such that the exterior potential $\Vext$ is a good external potential with constant $\bCc_1$ (up to a negligible event).
	\item We then rely on the results of Sections \ref{sec:good_properties} and \ref{conseqSubSys} about “good sub-systems with good external potentials”. The constants found in the statements there depend on how “good” the external potential is, so here on our constant $\bCc_1$. We will denote those by $\bCc_2$.
\end{itemize}}

\subsubsection*{Step 1. Choosing $L$ and cornering the discrepancy}
Let $L$ be chosen as:
\begin{equation}
\label{choix_L}
L := \log^{0.99} R.
\end{equation}
In particular, \emph{for $R$ greater than some constant} (depending only on $\beta$) we have:
\begin{equation}
\Cins \frac{1}{\epsilon_R} = \Cins \log^{0.3} R \leq \log^{0.99} R = L \leq R/10,
\end{equation}
(where $\Cins$ is the constant depending only on $\beta$ introduced in Proposition \ref{prop:chase_the_goat}) so the first condition of \eqref{eq:s_and_L} is satisfied. 
\newcommand{\BBR}{\mathcal{B}_R}

Moreover let $s$ be chosen as $\frac{1}{\Cins} \min \left( \frac{L^3}{R}, L \epsilon_R \right)$, namely (for $R$ large enough): $s := \frac{1}{\Cins} \frac{L^3}{R}$. By definition, the second condition of \eqref{eq:s_and_L} is then satisfied. 

Let us compute:
\begin{equation*}
\exp\left( - \frac{s \epsilon_R R}{4} \right) = \exp\left( - \frac{\epsilon_R L^3}{4 \Cins} \right) \leq \exp\left( - \frac{\log^{2.67} R}{\Cc_\beta}  \right),
\end{equation*}
\cor{where we replace $4 \Cins$ by $\Cc_\beta$.}

\smallskip

For each $k$ with $0 \leq k \leq R^{2}$, let $\BBR(k)$ be the event:
\begin{equation*}
\BBR(k) := \left\lbrace \text{The discrepancy in the annulus $\DD_R \setminus \DD_{R - 2L + \frac{k L}{R^2}}$ is larger than $\frac{1}{4} \eRR$} \right\rbrace
\end{equation*}
\cor{At this point, if we were dealing with a \emph{default} of points, we would introduce instead:
\begin{equation*}
\BBR(k) := \left\lbrace \text{The discrepancy in the annulus $\DD_R \setminus \DD_{R - 2L + \frac{k L}{R^2}}$ is smaller than $-\frac{1}{4} \eRR$} \right\rbrace.
\end{equation*}
}

\smallskip

Combining Proposition \ref{prop:chase_the_goat} and Lemma \ref{lem:finding_discr} we obtain:
\begin{equation}
\label{eq:AAR_BBR}
\PNbeta( \AAR ) \leq \sum_{k=1}^{R^2} \PNbeta( \BBR(k) ) + \exp\left( - \frac{\log^{2.67} R}{\Cc_\beta}  \right),
\end{equation}
and we now focus on bounding $\BBR(k)$ (the index $k$ plays no particular role, and for simplicity we forget about it). 

For each $k$, as explained in Section \ref{sec:well-separated}, we decompose the annulus $\DD_R \setminus \DD_{R - 2L + \frac{k L}{R^2}}$ into boxes of size $L$ as in Definition \ref{def:boxes}, that we denote by $\left\lbrace \boxs_i, i \in \{0, \dots, \frac{R}{L} -1 \} \right\rbrace$.

\subsubsection*{Step 2. Choosing $T, M$ and a well-separated family} 
\newcommand{\CCR}{\mathcal{C}}
Let $T$ be chosen (as in \cor{Corollary \ref{coro:appli_exp}}) as:
\begin{equation}
\label{eq:choix_de_T}
T := 100 \times \log^{0.99} R \times \exp\left(10 \log^{0.99} R \right) = 100 L \exp \left( 10 L \right),
\end{equation}
and let $M$ be chosen as:
\begin{equation}
\label{choix_M}
M := T^6. 
\end{equation} 
Since $L = \log^{0.99} R$ (according to \eqref{choix_L}), the conditions of \eqref{eq:condi_T2logTML},
\cor{namely:
\begin{equation*}
T \geq 10 L, \quad 100 \leq M \leq \frac{R}{L}, \quad  T \leq \frac{ML}{100}, \quad T^2 \log T \leq ML.
\end{equation*}
}are clearly satisfied for $R$ large enough. 

\smallskip

Now, for $l \in \{0, \dots, M-1\}$, let $\CCR(l)$ be the event:
\begin{equation}
\label{def:CCR}
\CCR(l) :=\left\lbrace  \sum_{i = l \textrm{ mod } M} \Di(\bXN, \boxs_i) \geq \frac{\epsilon_R R}{4 M} \right\rbrace.
\end{equation}
By Lemma \ref{lem:MV1} we know that:
\begin{equation}
\label{B_TO_C}
\PNbeta( \BBR ) \leq \sum_{l = 0}^{M-1} \PNbeta \left( \CCR(l) \right), 
\end{equation}
and we now focus on bounding $\CCR(l)$ (again, the index $l$ plays no role in the sequel). \cor{Note that here again, in the case of a \emph{default} of points, we would define instead:
\begin{equation*}
\CCR(l) :=\left\lbrace  \sum_{i = l \textrm{ mod } M} \Di(\bXN, \boxs_i) \leq - \frac{\epsilon_R R}{4 M} \right\rbrace,
\end{equation*}
and then proceed similarly.}

\smallskip

The index $l$ being fixed, we only consider the boxes $\boxs_i$ for $i \equiv l \mod M$ and forget about the other boxes. We relabel those boxes as $\boxs_i$ for $i \in \{1, \dots, \Nn\}$ where $\Nn$ is the cardinality of that family of boxes, with:
\begin{equation}
\label{eq:NNRML}
\Nn \leq \Cc \frac{R}{ML}.
\end{equation} 

As in Section \ref{sec:well-separated} we let $\Lai$ be the disk $\DD(\omega_i, T)$, where $\omega_i$ is the “center” of the box $\boxs_i$, and we recall that $\dij$ denotes the distance between $\Lai$ and $\Laj$. Using \eqref{sumdij} and summing over $i = 1, \dots, \Nn = \O(\frac{R}{ML})$, we get: 
\begin{equation}
\label{sumdij_applied}
\max_{1 \leq i \leq \Nn} \sum_{j \neq i} \frac{1}{\dij} = \O\left(\frac{\log R}{ML}\right), \quad \sum_{1 \leq i \neq j \leq \Nn} \frac{1}{\dij} = \O\left(\frac{R \log R}{(ML)^2}\right).
\end{equation}

\subsubsection*{\cor{Step 3. Inserting }a “good event”}
Let $\Vext$ be the logarithmic potential generated by the system outside the $\La_i$'s as in \eqref{def:VextLi}. For $\bCc > 0$, let $\Eext(\bCc)$ be the event:
\begin{equation}
\label{def:event_good}
\Eext(\bCc) := \left\lbrace \text{$\Vext$ is a “good external potential” on each $\La_i$ with constant $\bCc$.} \right\rbrace
\end{equation}
(See Definition \ref{def:GoofPotential}). 

By Proposition \ref{prop:VextIsOftenGood} we know that if $\bCc_1$ is chosen large enough and if $T$ is large enough (i.e. $R$ is large enough) depending only on $\beta$ and on the parameter $\delta$ from \eqref{eq:assumption_distance}, then we have:
\begin{equation*}
\PNbeta \left( \Eext(\bCc_1) \right) \geq 1 - \Nn \exp\left( - \log^2 T / \bCc_1 \right).
\end{equation*}
Since $\Nn$ is always smaller than $R = e^{\log R}$ (see \eqref{eq:NNRML}) and since, by our choice \eqref{eq:choix_de_T} of $T$ we have $\log T \geq \log^{0.99} R$ and thus:
\begin{equation*}
\exp\left( - \log^2 T / \bCc_1 \right) \leq \exp\left( - \log^{1.98} R / \bCc_1 \right) \ll \exp(- \log R ),
\end{equation*}
we deduce that (still for $R, \bCc_1$ large enough depending on $\beta, \delta$):
\begin{equation}
\label{eq:ProbaEExt}
\PNbeta \left( \Eext(\bCc_1) \right) \geq 1 - \exp\left( - \log^2 T  / \bCc_1 \right).
\end{equation}
On the other hand, for each $i = 1, \dots, \Nn$, let $\EE_i$ be the event: “$\La_i$ is a good sub-system”, as in Definition \ref{def:goodSS}. By Lemma \ref{lemma:oftenGood}, we know that:
\begin{equation*}
\PNbeta\left( \bigcap_{i=1}^\Nn \EE_i \right) \geq 1 - \Nn \exp\left( - \log^2 T / \Cc_\beta \right), 
\end{equation*}
and by the same parameter comparison as above we get: $\PNbeta \left( \bigcap_{i=1}^\Nn \EE_i  \right) \geq 1 - \exp\left( - \log^2 T  / \Cc_\beta \right)$. 

We thus have (for $R, \bCc_1$ large enough - \cor{taking in particular $\bCc_1$ larger than the constant $\Cc_\beta$ in the previous inequality}):
\begin{equation}
\label{CCR_condition}
\PNbeta\left( \CCR \right) \leq \PNbeta \left( \CCR \cap \Eext(\bCc_1) \cap \bigcap_{i=1}^\Nn \EE_i  \right) + \exp\left( - \log^2 T  / \bCc_1 \right).
\end{equation}

\subsubsection*{\cor{Step 4.} Using approximate conditional independence}
Let $\omega$ be chosen as:
\begin{equation}
\label{def_omega}
\omega := L^{-1.33}.
\end{equation}
For $1 \leq i \leq \Nn$, let $\G_i$ be the following function on $\Conf$, which is clearly non-negative and $\Lai$-local:
\begin{equation}
\label{def:Gi}
\G_i(\bX) := e^{\omega \Di(\bX, \boxs_i)}.
\end{equation}
\cor{(In the case of a default, we would choose $\G_i(\bX) := e^{-\omega \Di(\bX, \boxs_i)}$) instead.}

\smallskip 

If $\bXN$ is in $\CCR$ then by definition (see \eqref{def:CCR}) we have:
\begin{equation*}
\prod_{i=1}^\Nn \G_i(\bXN) = e^{\omega \sum_{i=1}^\Nn \Di(\bXN, \boxs_i)} \geq \exp\left( \frac{\omega \epsilon_R R}{4 M} \right),
\end{equation*}
and thus by Markov's inequality:
\begin{equation}
\label{eq:CCR_Markov}
\PNbeta \left( \CCR \cap \Eext(\bCc_1) \cap \bigcap_{i=1}^\Nn \EE_i  \right) \leq \exp\left(-\frac{\omega \epsilon_R R}{4 M} \right) \EN \left[ \prod_{i=1}^\Nn \G_i(\bXN) \ind_{ \Eext(\bCc_1) \cap \bigcap_{i=1}^\Nn \EE_i } \right].
\end{equation}

We are now in a position to use Proposition \ref{prop:CI} with $\EE_N := \Eext(\bCc_1) \cap \bigcap_{i=1}^\Nn \EE_i$. We obtain:
\begin{multline}
\label{application_CI}
\EN \left[ \prod_{i=1}^\Nn \G_i(\bXN) \ind_{ \Eext(\bCc_1) \cap  \bigcap_{i=1}^\Nn \EE_i } \right] \leq \exp\left(2 \beta \sup_{\bXN \in \EE_N} \ErrorCI(\bXN) \right) \\ \times \sup_{\bXext \in \Eext, \{\tn_i\} \ \mathrm{adm.}} \prod_{i=1}^\Nn \E^\beta_{\tn_i, \Lai, \Vext} \left[ \G_i(\bX) \Cond \EE_i \right].
\end{multline}

Let us note that the event $\bigcap_{i=1}^\Nn \EE_i$ implies that the “admissible” number of points in each $\La_i$ is bounded by $10 T^2$ (see Definition \ref{def:goodSS}) which checks the first condition of \eqref{eq:boundnni}, and that the second condition of \eqref{eq:boundnni} is implied by the third condition\footnote{The distance between two centers $\omega_i, \omega_j$ is bounded below by $ML$ and thus $\dist(\Lai, \Laj) \geq ML - 2T \geq 98 T$.} of \eqref{eq:condi_T2logTML}. 

We may thus use Lemma \ref{lem:sizeCI} which, together with \eqref{sumdij}, implies that the conditional independence error $\ErrorCI$ between the $\Lai$'s is bounded by:
\begin{equation}
\label{bound_Error_CI}
\sup_{\bXN \in \bigcap_{i=1}^\Nn \EE_i} \left|\ErrorCI[\bXN | \La_1, \dots, \La_\Nn]\right| \leq \Cc \left( \frac{T^5 R \log R}{(ML)^2} \right),
\end{equation}
which controls the first term in the right-hand side of \eqref{application_CI}, \cor{with some universal constant $\Cc$.} 

We now focus on the second term in the right-hand side of \eqref{application_CI}.

\subsubsection*{\cor{Step 5. } Controlling expectations in each “sub-system”}
\newcommand{\EnLi}{\mathbb{E}^\beta_{\nn_i, \La_i}}
\newcommand{\PnLi}{\mathbb{P}^\beta_{\nn_i, \La_i}}
Let us fix $1 \leq i \leq \Nn$ and work in $\Lai$. We can assume that $\Vext$ is $\bCc_1$-good on $\Lai$ because of \eqref{def:event_good}.

\begin{claim} Let $\Dd_i := \Di(\cor{\bX}, \boxs_i)$. We have:
\begin{equation}
\label{eq:esperance_finale_SS}
\EnLi\left[ e^{\omega \Dd_i} \big| \EE_{\Lai} \right] \leq e^{\cor{\bCc_2} L^{-0.66}},
\end{equation}
\cor{with some constant $\bCc_2$ depending on $\delta, \beta$ and on the constant $\bCc_1$ for $\Vext$.}
\end{claim}
\begin{proof}
We start by decomposing the expectation as:
\begin{equation}
\label{EnLi_decomp}
\EnLi\left[ e^{\omega \Dd_i} \big| \EE_{\Lai} \right] = \EnLi\left[ e^{\omega \Dd_i} \ind_{|\omega \Dd_i| \leq \hal} \big| \EE_{\Lai} \right] + \EnLi\left[ e^{\omega \Dd_i} \ind_{|\omega \Dd_i| > \hal} \big| \EE_{\Lai} \right].
\end{equation}

Using a Taylor's expansion, we may write the first term in the right-hand side of \eqref{EnLi_decomp} as:
\begin{multline}
\label{EnLi_lin}
\EnLi\left[ e^{\omega \Dd_i} \ind_{\omega |\Dd_i| \leq \hal} \big| \EE_{\Lai} \right] = \EnLi\left[ \left(1 + \omega \Dd_i + \O\left(\omega^2 \Dd_i^2 \right)\right) \ind_{\omega |\Dd_i| \leq \hal} \big| \EE_{\Lai} \right] \\
\leq \PnLi \left[\ind_{\omega |\Dd_i| \leq \hal} \big| \EE_{\Lai} \right] + \omega \EnLi[\Dd_i \big| \EE_{\Lai}] + \O\left(\omega^2 \EnLi[\Dd_i^2 \big| \EE_{\Lai}] \right) + \EnLi\left[ \omega |\Dd_i| \ind_{\omega |\Dd_i| > \hal} \big| \EE_{\Lai} \right], \\
\leq 1 +  \omega \EnLi[\Dd_i \big| \EE_{\Lai}] + \O\left(\omega^2 \right) \EnLi[\Dd_i^2 \big| \EE_{\Lai}]   \\
\cor{+ \EnLi\left[ \omega |\Dd_i| \ind_{\omega |\Dd_i| > \hal} \big| \EE_{\Lai} \right] - \PnLi \left[\ind_{\omega |\Dd_i| > \hal} \big| \EE_{\Lai} \right]}
\end{multline}
\cor{(we use $\omega \Dd_i \ind_{\omega |\Dd_i| \leq \hal} \leq \omega \Dd_i  + \omega |\Dd_i| \ind_{\omega |\Dd_i| > \hal}$).}

Using Corollary \ref{coro:appli_exp} to control $\EnLi[\Dd_i \big| \EE_{\Lai}]$ and Lemma \ref{claim:DiSS} to control $\EnLi[\Dd_i^2 \big| \EE_{\Lai}]$, and inserting into \eqref{EnLi_decomp} we obtain:
\begin{multline}
\label{eq:esperance_temporaire}
\EnLi\left[ e^{\omega \Dd_i} \big| \EE_{\Lai} \right] \leq 1 + \bCc_2 \omega L^{0.67} + \bCc_2 \omega^2 L^2 \\
+ \EnLi\left[ e^{\omega \Dd_i} \ind_{|\omega \Dd_i| > \hal} \big| \EE_{\Lai} \right] + \EnLi\left[ \omega |\Dd_i| \ind_{\omega |\Dd_i| > \hal} \right] - \PnLi \left[\ind_{\omega |\Dd_i| > \hal} \big| \EE_{\Lai} \right],
\end{multline}
and it remains to control the second line of \eqref{eq:esperance_temporaire}, using Lemma \ref{claim:DiSS}. 

We write:
\begin{multline*}
\EnLi\left[ e^{\omega \Dd_i} \ind_{|\omega \Dd_i| > \hal} \big| \EE_{\Lai} \right] \leq \EnLi\left[ e^{2 \omega^2 L^2 \frac{\Dd^2_i}{L^2} } \ind_{|\omega \Dd_i| > \hal} \big| \EE_{\Lai} \right] 
\\
\leq \EnLi\left[ e^{4 \omega^2 L^2 \frac{\Dd^2_i}{L^2} } \big| \EE_{\Lai} \right]^\hal \EnLi\left[ \ind_{|\omega \Dd_i| > \hal} \big| \EE_{\Lai} \right]^{\hal} 
\leq e^{\bCc_2 \omega^2 L^2} \PnLi \left[ \left\lbrace|\omega \Dd_i| > \hal\right\rbrace \big| \EE_{\Lai}  \right],
\end{multline*}
where we have used Cauchy-Schwarz's inequality, then the fact that $\omega^2 L^2 = L^{-0.66} \ll 1$ (see \eqref{def_omega}) in order to apply Hölder's inequality, and the exponential moment \eqref{discr_bounds_in_SS}. 

Using \eqref{discr_bounds_in_SS} again we get that:
\begin{equation*}
\PnLi \left[ \left\lbrace|\omega \Dd_i| > \hal\right\rbrace \big| \EE_{\Lai}  \right] \leq \exp\left( - L^{0.66} / \bCc_2 \right),
\end{equation*}
using Markov's inequality \cor{and} the fact that $\frac{1}{\cor{L^2 \omega^2}} = L^{0.66}$ (see \eqref{def_omega}). We can thus write:
\begin{equation*}
\EnLi\left[ e^{\omega \Dd_i} \ind_{|\omega \Dd_i| > \hal} \big| \EE_{\Lai} \right] = \O\left( e^{- L^{0.66} / \bCc_2 } \right).
\end{equation*}
The two other terms in the second line of \eqref{eq:esperance_temporaire} are handled the same way. 

We obtain:
\begin{equation*}
\EnLi\left[ e^{\omega \Dd_i} \big| \EE_{\Lai} \right] \leq 1 + \bCc_2 \omega L^{0.67} + \bCc_2 \omega^2 L^2 + \O\left( e^{- L^{0.66} / \bCc_2 } \right).
\end{equation*}

Inserting\footnote{There might be a confusion between the \emph{error} term $\O\left( e^{- L^{0.66} / \bCc } \right)$ and the fact that we write $1 + \O\left(L^{-0.66} \right) \leq e^{\O(L^{-0.66})}$, but $e^{- L^{0.66}}$ and $e^{L^{-0.66}}$ are two different terms, \cor{the former being much smaller than the latter}.} the value $\omega = L^{-1.33}$ we obtain \eqref{eq:esperance_finale_SS}.
\end{proof}

\cor{In the case of a default, we use instead $e^{-\omega \Dd_i}$ as a test function, and the same computation holds.}

\subsubsection*{\cor{Step 6. } Conclusion.}
Inserting \eqref{bound_Error_CI} and \eqref{eq:esperance_finale_SS} into \eqref{application_CI} (and using \eqref{eq:NNRML}) we obtain:
\begin{equation*}
\EN \left[ \prod_{i=1}^\Nn \G_i(\bXN) \ind_{ \Eext(\bCc_1) \cap \bigcap_{i=1}^\Nn \EE_i } \right] \leq \exp\left( \cor{\Cc} \frac{T^5 R \log R}{(ML)^2}  + \cor{\Cc \bCc_2} \frac{R}{ML} L^{-0.66} \right).
\end{equation*}
Since $M = T^6$ (by \eqref{choix_M}) it is easy to check that (for $R$ large enough) the dominant term in the exponent is by far the second one because $T \gg L$. 

Returning to \eqref{eq:CCR_Markov} we thus obtain:
\begin{equation*}
\PNbeta \left( \CCR \cap \Eext(\bCc_1) \cap  \bigcap_{i=1}^\Nn \EE_i  \right) \leq \exp\left(-\frac{\omega \epsilon_R R}{4 M} \right) \exp\left( \Cc \bCc_2 \frac{R}{ML} L^{-0.66} \right),
\end{equation*}
and since $\omega = L^{-1.33}$ (by \eqref{def_omega}), $L = \log^{0.99}R$ (by \eqref{choix_L}) and $\epsilon_R = \log^{-0.3} R$ (by \eqref{choix_alpha}) the first factor dominates the second one for $R$ large enough:
\begin{equation*}
\PNbeta \left( \CCR \cap \Eext(\bCc_1) \cap \bigcap_{i=1}^\Nn \EE_i  \right) \leq \exp\left(-\frac{\omega \epsilon_R R}{8 M} \right),
\end{equation*}
and thus after inserting the value of $M$ \eqref{choix_M} defined in terms of $T$ as in \eqref{eq:choix_de_T} we obtain the following Markov-type inequality:
\begin{equation*}
\PNbeta \left( \CCR \cap \Eext(\bCc_1)  \bigcap_{i=1}^\Nn \EE_i  \right) \leq \exp\left( - \exp\left( \log R + o(\log R)  \right)  \right).
\end{equation*}
In particular, for $R$ large enough, this is smaller than $\exp\left( - R^{1/2} \right)$. 

Returning successively to \eqref{CCR_condition}, \eqref{B_TO_C}, \eqref{eq:AAR_BBR} and using the fact that $\log T \geq \log^{0.99} R$ we obtain:
\begin{equation*}
\PNbeta \left( \CCR \right) \leq \exp\left( - \log^{1.98} R / \bCc_1 \right), \quad \PNbeta \left( \BBR \right) \leq \exp\left( \Cc \log R - \log^{1.98} R / \bCc_1 \right) \leq \exp\left( - \log^{1.98} R / (2\bCc_1) \right),
\end{equation*}
and finally:
\begin{equation*}
\PNbeta \left( \AAR \right) \leq \exp\left( \Cc \log R - \log^{1.98} R / (2\bCc_1) \right) \leq \exp\left( - \log^{1.98} R / (4\bCc_1) \right)
\end{equation*}
which proves \eqref{tail_estimate_quonprouve} and thus concludes the proof of Theorem \ref{theo:main2}.

\appendix
\section{Discussion of the model and terminology}
\label{sec:discussion_model}
\subsection{\cor{On the definition of the model}}
\subsubsection*{Possible definitions of the model}
There are several slightly different ways to define the two-dimensional one-component plasma. 
\begin{itemize}
	\item Some papers work with an “infinitely extended equilibrium” Coulomb system, e.g. \cite{martin1980charge,lebowitz1983charge,jancovici1993large}. The mathematical existence of such infinite-volume limits is not yet clear for $\beta \neq 2$, see however \cite[Corollary 1.1]{armstrong2019local} for existence of infinite-volume \emph{limit points} in the weak topology, \cor{and \cite{leble2024dlr} for a “physical” description through DLR equations}.

	\item In the statistical physics literature, it is common to place $N$ particles in a “uniform neutralizing background of opposite charge” which occupies a certain domain $\Sigma_N$ with constant density $\rho_N := \frac{-N}{|\Sigma_N|}$. There is then \emph{perfect confinement} in the sense that the particles are not allowed to live outside $\Sigma_N$. The domain is not always explicitly chosen, though it often ends up being a disk, mostly by default or for the convenience of symmetry\footnote{There is also some interest for studying the 2DOCP on a sphere, which avoids having to deal with a boundary.}. Some authors state their results for different “reasonable shapes” as e.g. \cite{sari1976nu}. 

	\cor{The presence of the background can be equivalently be seen as the application of} some potential to each point charge, while they all interact with each other. Indeed, one may write:
	\begin{equation}
	\label{FNwithVback}
\FN(\bXN) = \hal \iint_{\cor{x \neq y}} - \log|x-y| \dd \bXN(x) \dd \bXN(y) + \int \Vback(x) \dd \bXN(x) + C_N, 
	\end{equation} 
	where $\Vback$ is the logarithmic potential generated by the background, namely:
	\begin{equation*}
\Vback(x) := \int_{\Sigma_N} - \log |x-y| \rho_N \dd y,
	\end{equation*}
	and $C_N$ is a constant (the self-interaction of the background with itself) that does not depend on $\bXN$ (only on $N$) and can thus be absorbed in $\KNbeta$.

	\item In the mathematical physics literature around the planar Coulomb gas (e.g. \cite{zabrodin2006large,ameur2011fluctuations,MR3353821}) the particles/eigenvalues $\XN$ are usually not confined \emph{a priori} in a certain domain of the Euclidean space, but are rather subject to a certain external “confining” potential/field/weight $\Vv$ acting as $\Vback$ in \eqref{FNwithVback} (this model is sometimes called a \emph{(two-dimensional) $\beta$-ensemble} by analogy with well-known one-dimensional models coming from random matrix theory). 

	Via a certain mean-field energy functional, the choice of $\Vv$ determines\footnote{Alternatively, one can associate to $\Vv$ a “thermal equilibrium measure” as in \cite{armstrong2019local}, which has unbounded support and plays the role of those three objects.}:
	\begin{enumerate}
		\item A compact subset $\Sigma_N$, sometimes called the droplet.
		\item An \emph{equilibrium measure} $\mu_N$ on $\Sigma_N$.
		\item An effective confining potential $\zeta_N$.
	\end{enumerate}
	After “splitting” (see \cite[Lemma 3.1]{MR3353821}), the energy takes the following form, cf. \eqref{def:FN} and \eqref{FNwithVback}:
 \begin{equation}
 \label{withzeta}
\tFN(\bXN) := \hal \iint_{x \neq y} - \log |x-y| \dd \left(\bXN - \mu_N\right)(x) \dd \left(\bXN - \mu_N\right)(y) 
+ \sum_{i=1}^N \zeta_N(x_i) +C_N
 \end{equation} 
 \cor{(we write $\tFN$ instead of $\FN$ to emphasize that this is a more general definition).}
 
The confining potential $\zeta_N$ vanishes on $\Sigma_N$ and is positive outside of it - thus penalizing particles that leave $\Sigma_N$: \cor{one expects that} few particles will fall outside $\Sigma_N$, and that if they do they will stay close by, see \cite{ameur2019localization} for a quantitative statement. \cor{However this potential is usually \emph{not} chosen to be infinite outside $\Sigma_N$, and thus there is only “imperfect” or “soft” confinement.} 

 The canonical choice is the quadratic potential $\Vv(x) = |x|^2$ for which $\Sigma_N$ is a large disk of radius comparable to $\sqrt{N}$ and $\mu_N$ is the uniform measure on $\Sigma_N$.
\end{itemize}
\cor{Another way to speak of the nature of the confinement is to use the term “hard” edge/wall for the situation where the points are forced to live in $\LN$, and “soft” edge when there is a smooth (non-infinite) confining external potential instead.}

In this paper, we choose to work in a disk (for convenience) and with a perfect confinement (as this is the usual choice in the physics literature). However Theorem \ref{theo:main} still holds if one replaces $\LN$ by a square or by any “reasonable” shape, still with a perfect confinement. Indeed, our proof makes no use of the global geometry of the system.

\subsubsection*{\cor{Validity of our results for a soft confinement}}
\cor{More importantly, Theorem \ref{theo:main} also holds for a two-dimensional $\beta$-ensemble/Coulomb gas with quadratic potential. Indeed:
	\begin{itemize}
		\item All the tools that we import in Section \ref{sec:preliminary}: global laws, local laws, control on fluctuations, Wegner's estimates... are proven in that case.
		\item The analysis of Section \ref{sec:locating} relies only on those tools and takes place entirely in the bulk.
		\item The approximate conditional independence stated in Proposition \ref{prop:CI} (the only place in Section \ref{sec:approx_CI} where the Gibbs measure actually appears) would carry over effortlessly to a “soft edge” setting, because the presence of an effective confinement term $\zeta_N$ in the energy (see \eqref{withzeta}) would be only felt by the “exterior” part of the system, which anyway gets integrated out in \eqref{eq:CI}.  
		\item Let us temporarily skip over Section \ref{sec:good_good}. Then the “subsystems” introduced in Section \ref{sec:PnLV} are models on their own and their analysis (global laws, local laws) as done in Section \ref{sec:good_properties}, Section \ref{conseqSubSys} is independent of any hard/soft edge choice for the 2DOCP. 
		\item The approximate translation-invariance stated in Section \ref{sec:approx_QI} also relies only on bulk properties of the system. In particular, the proof of Proposition \ref{prop:effectPhiEnergy}, which is the main ingredient, uses only a deterministic analysis together with the local laws in the bulk.
		\item One could then combine all the ingredients of the proof as in Section \ref{sec:conclusion}.
	\end{itemize}}

\cor{In fact, the only moment in our analysis where we actually “see” the boundary of the 2DOCP is in Section~\ref{sec:good_good}, more precisely in the proof of Proposition \ref{prop:VextIsOftenGood}. Indeed, the potential $\Vext$ (see \eqref{def:VextLi}) felt in a sub-system $\La$ takes into account the entire system in $\LN \setminus \La$ \emph{up to the boundary} (and, in the case of an imperfect confinement, it would take into account the entire system in $\R^2 \setminus \La$). Within the proof of Proposition \ref{prop:VextIsOftenGood}, we can pinpoint the role of the boundary: it occurs in the course of the proofs of Lemma~\ref{lem:nablaRiwi} and Lemma \ref{lem:2nablaRiwi} during the “Step 3. The rest of the system”.}

\cor{In that step, we introduce “an artificial cut-off $\gamma$ supported on $2 \LN$”, which in the case of perfect confinement is completely transparent because all the points live in $\LN$ anyway, and we then apply Lemma \ref{lem:apriori} (fluctuations of Lipschitz linear statistics), which is valid for both hard or soft edges.}

\cor{The only difference in the case of a soft confinement would thus be that there might exist points living outside $2 \LN$. However, “localization estimates” as in \cite[Thm. 1.12]{chafai2018concentration} or \cite[Thm. 3]{ameur2019localization} ensure that up to changing $2\LN$ by $\Cc_\beta \LN$ for some $\Cc_\beta > 0$ large enough, we can write:
\begin{equation*}
\PNbeta\left[ \{ \text{There is no particle outside $\Cc_\beta \LN$} \} \right] \geq 1 - \exp\left(- \frac{1}{\Cc_\beta} N\right),
\end{equation*}
so we can work instead with an artificial cut-off $\gamma$ supported on $\Cc_\beta \LN$, and conclude similarly.
}

\subsection{\cor{On terminology}}
\cor{We call our model, defined by \eqref{def:FN} and \eqref{def:PNbeta}, a \emph{one-component plasma} in accordance with the name given in many early important physics papers, for instance: 
\begin{itemize}
	\item “Monte Carlo study of a one-component plasma“ \cite{brush1966monte} (this is about the three-dimensional model). \textit{The model treated in this study consists of a system of identical point charges immersed in a uniform background; the continuous charge density of the background is chosen equal and opposite to the average charge density of the point charges, so that the system as a whole is electrically neutral.}
	\item “On the $\nu$-dimensional one-component plasma” \cite{sari1976nu}. \textit{This paper deals with the one-component classical plasma constituted, in a domain $\La$, by $N$ point charges (electrical charge equal to $-e$), immersed in a uniform neutralizing background of density $\rho = \frac{N}{|\La|}$.}
	\item “On the classical two-dimensional one-component Coulomb plasma” \cite{alastuey1981classical}. \textit{A one-component plasma is a system of N identical particles (...) embedded in a uniform neutralizing background of opposite charge.}
\end{itemize}
Another possible name for our model would be “jellium”, as appears e.g. in \cite{jancovici1993large}: 
\begin{quote}
\textit{We shall mainly consider the mathematically simpler one-component plasma (OCP or Jellium). This frequently used model system consists of charged particles of one sign embedded in a uniform background of the opposite charge"}.
\end{quote}
The term “jellium” seems to have been coined in \cite{herring1952discussion}, see \cite[p. 20]{lewin2022coulomb}. Note however, as mentioned in \cite[footnote 386]{lewin2022coulomb}, that: \textit{In the physics literature, the name ‘Jellium’ is often employed for electrons (which are quantum with spin), whereas the ‘one-component plasma’ is mainly used for classical particles}.}

\cor{In the mathematics literature, this model has recently been studied either under the name “one-component plasma” (e.g. \cite{bauerschmidt2017local}), “Coulomb plasma” (e.g. \cite{MR4063572}), or “Coulomb gas” (e.g. \cite{MR3353821}). It is also sometimes informally referred to as “$\beta$-Ginibre”, “2d log-gas”, “2d $\beta$-ensemble”, “2d Dyson gas”...}

\section{Auxiliary proofs of preliminary results}
\newcommand{\Cepel}{\Cc_{\epsilon \ell^2}}

\subsection{Proof of Proposition \ref{prop:local_laws}}
\label{sec:proof_locAS}
\begin{proof}[Proof of Proposition \ref{prop:local_laws}]
The first item corresponds to the first statement in \cite[Theorem~1]{armstrong2019local}, combined with their Lemma B.2. To be precise, the statement of \cite[Theorem~1]{armstrong2019local} involves the quantity $\F^{\square_R(x)}$ which is short for $\F^{\square_R(x)}(\XN, U)$ with $U = \R^2$ as defined in \cite[Eq. (2.24)]{armstrong2019local}. The potential $u$ appearing there is defined in \cite[Section 2.3]{armstrong2019local} but as used in the main statement of \cite{armstrong2019local} it coincides with the true potential (because, with their notation, “$U = \R^2$” in this case). It remains to observe that Theorem 1 in \cite{armstrong2019local} states a control on $\F^{\square_R(x)}$ which can be turned into a control on the electric energy only (as we write in Proposition \ref{prop:local_laws}), this is precisely the purpose of \cite[Lemma B.2]{armstrong2019local} and in particular their equation (B.8), which shows that one can indeed control the electric energy in terms of $\F^{\square_R(x)}$. As a last technical comment for the careful reader: note that since here “$U = \R^2$” (their notation) the various truncations $\mathsf{r}, \tilde{\mathsf{r}}, \tilde{\tilde{\mathsf{r}}}$ all coincide.

The second item of Proposition \ref{prop:local_laws} is \cite[(1.18)]{armstrong2019local}. This, or  \cite[(1.19)]{armstrong2019local}, implies  \eqref{eq:number_of_points_LL}.
\end{proof}

\subsection{Proof of Lemma \ref{lem:apriori}}
\label{sec:proof_apriori}
\begin{proof}[Proof of Lemma \ref{lem:apriori}]
For $x$ in $\bX$ let us define the truncation $\eta(x)$ as:
\begin{equation*}
\eta(x) = \begin{cases}
0 &  \text{if } x \notin \Omega \\
\rr(x) & \text{if } x \in \Omega,
\end{cases}
\end{equation*}
where $\rr$ is the “nearest-neighbor” distance as in \eqref{def:nn_distance}. Let us recall that, by definition, we always have $|\rr| \leq 1/4$. Let $\Elec$ be an electric field compatible with $\bX$ on $\supp \varphi$ and let $\Elece$ be the electric field truncated accordingly. We have, in the sense of distributions on $\supp \nabla \varphi$:
\begin{equation*}
- \dive \Elece = 2 \pi \left(\sum_{x \in \bX} \delta^{(\eta(x))}_{x} - \mm_0 \right),
\end{equation*}
where we replace the Dirac mass $\delta_x$ by its “smeared out” version $\delta^{(\eta(x))}_{x}$, which is the uniform measure of mass $1$ on the circle of center $x$ and radius $\eta(x)$. Let us write $\bX_{\veta} := \sum_{x \in \bX} \delta^{(\eta(x))}_{x}$. The measures $\bX$ and $\bX_{\veta}$ coincide outside a $1$-neighborhood of $\supp \nabla \varphi$ and the atoms there have been smeared out at a distance at most $1$, thus:
\begin{equation*}
\left| \int_{\R^2} \varphi(x)  \dd \left( \bX - \bX_{\veta} \right) (x) \right| \leq |\varphi|_{\1, \Omega} \times \Points(\bX, \Omega),
\end{equation*}
which can be localized, if $\tilde{\Omega}_1, \dots, \tilde{\Omega}_m$ cover $\supp \nabla \varphi$ we have:
\begin{equation*}
\left| \int_{\R^2} \varphi(x)  \dd \left( \bX - \bX_{\veta} \right) (x) \right| \leq \sum_{k=1}^m \sup_{x \in \Omega_i} |\nabla \varphi(x)| \times \Points(\bX, \Omega_i),
\end{equation*}
where $\Omega_i$ contains a $1$-neighborhood of $\tilde{\Omega}_i$. Now, let us re-write the fluctuations of $\varphi$ as:
\begin{multline*}
\left|\int_{\R^2} \varphi(x) \dd\left(\bX - \mm_0\right)(x)\right| \leq \left| \int_{\R^2} \varphi(x) \dd \left(  \bX_{\veta}(x) - \mm_0\right)(x) \right| + |\varphi|_{\1, \Omega} \times \Points(\bX, \Omega) \\
= \frac{1}{2\pi} \left| \int_{\R^2} \varphi(x) \dive \Elece \right| + |\varphi|_{\1, \Omega} \times \Points(\bX, \Omega).
\end{multline*}
Integrating by parts and using Cauchy-Schwarz's inequality yields:
\begin{equation*}
\left|\int_{\R^2} \varphi(x) \dive \Elece \right| = \left|\int_{\supp \nabla \varphi} \nabla \varphi \cdot \Elece\right| \leq \left(\int_{\R^2} |\nabla \varphi|^2 \right)^\hal \left(\int_{\supp \nabla \varphi} |\Elece|^2 \right)^\hal.
\end{equation*}
Finally, it remains to observe that with our choice of truncation $\eta$ as above and the definition of $\Omega$, we have:
\begin{equation*}
\int_{\supp \nabla \varphi} |\Elece|^2 \leq \int_{\Omega} |\Elec_{\vr}|^2,
\end{equation*}
where $\vr$ is the nearest-neighbor truncation, which concludes the proof of \eqref{eq:apriori_one_var}. This last step can also be localized by decomposing $\supp \nabla \varphi$ into several domains.
\end{proof}

\subsection{Proof of Proposition \ref{prop:bound_fluct_radial}}
\label{sec:proofCLT}
\begin{proof}[Proof of Proposition \ref{prop:bound_fluct_radial}]
\cor{The proof follows the same lines as the analogous results proven in \cite{MR3788208,serfaty2020gaussian}, namely:
\begin{enumerate}
	\item Rewriting the Laplace transform of the fluctuations as a ratio of partition functions of two 2DOCP's with  two different background measures - one of them being $\mm_0$.
	\item Constructing a transport that maps $\mm_0$ to the second background measure.
	\item Estimating the effect of such a transport on the energy.
	\item Conclusion.
\end{enumerate}}

\subsubsection*{\cor{Step 1. Rewriting the Laplace transform.}}
We start by rewriting the Laplace transform of the fluctuations as a ratio of partition functions, this is a standard trick that goes back (at least) to \cite{MR1487983}. In the two-dimensional context, it can be found e.g. in \cite[Proposition 2.10]{MR3788208}. Let $\mus$ be the signed measure defined by:
\begin{equation}
\label{def:mut}
\mus := \muc - \frac{s}{2\pi \beta} \Delta \varphi, 
\end{equation}
which coincides with $\muc$ outside of $\AA$ (let us recall that $\AA$ is  some annulus containing $\supp \Delta \varphi$). Since $\nabla \varphi$ is compactly supported, the total mass of $\Delta \varphi$ (seen as a measure) is $0$ so $\mus$ and $\muc$ have the same total mass on $\AA$. Moreover as soon as the parameter $s$ satisfies the condition \eqref{eq:condition_on_s} ($|s| \leq \frac{\pi \beta}{4 |\varphi|_\2}$) then $\mus$ has a non-negative density which is bounded between $\frac{1}{2}$ and $\frac{3}{2}$ on $\AA$. We introduce the following notation:
\begin{equation*}
\FN(\bXN, \mus) := \hal \iint_{x \neq y} - \log |x-y| \dd \left(\bXN - \mus \right)(x) \dd \left(\bXN - \mus \right)(y).
\end{equation*}

\begin{claim}[Laplace transform as ratio of partition functions] 
\label{claim:rewriting_Laplace_I}
The following identity holds:
\begin{equation}
\label{eq:rewriting_Laplace_I}
\EN \left[  \exp\left( s \Fluct[\varphi]  \right) \right] 
= \exp\left(\frac{s^2}{4\pi \beta} \int_{\R^2} |\nabla \varphi|^2 \right) \frac{ \int_{\LN^N} \exp\left( - \beta \FN(\bXN, \mus) \right) \dd \XN}{\int_{\LN^N} \exp\left( - \beta \FN(\bXN, \muc) \right) \dd \XN}
\end{equation}
\end{claim}
\begin{proof}[Proof of Claim \ref{claim:rewriting_Laplace_I}]
We have by definition:
\begin{equation*}
 \EN \left[  \exp\left( s \Fluct[\varphi]  \right) \right] = \frac{\int_{\LN^N} \exp\left( - \beta \left(\FN(\bXN, \muc) - \frac{s}{\beta} \Fluct[\varphi] ) \right) \right) \dd \XN}{\int_{\LN^N} \exp\left( - \beta \FN(\bXN, \muc) \right) \dd \XN}.
\end{equation*}
Since $\frac{1}{2\pi} \Delta \log = \delta_0$ (see \eqref{eq:identiteLaplaceLog}) one may write $- \frac{s}{\beta} \Fluct[\varphi]$ as: 
\begin{equation*}
- \frac{s}{\beta} \Fluct[\varphi] := - \frac{s}{\beta} \int_{\LN} \varphi(x) \dd\left(\bXN - \muc\right)(x) =  \iint - \log |x-y| \left(\frac{s}{2\pi \beta} \Delta \varphi(y) \right) \dd y \dd( \bXN - \muc)(x),
\end{equation*}
and then “complete the square”.
\end{proof}

\subsubsection*{\cor{Step 2. Finding a good transport from $\muc$ to $\mus$.}}
In order to compare the partition functions in \eqref{eq:rewriting_Laplace_I}, it is common to introduce some sort of transportation map from $\muc$ to $\mus$. The fact that the density of $\muc$ is constant and the one of $\mus$ has radial symmetry reduces the computation to a one-dimensional problem which can be solved exactly and explicitely (the general, non-radial case requires an abstract argument or resorting to an approximate transportation, which makes the analysis more involved). Proposition \ref{prop:bound_fluct_radial} is then a consequence of \cite[Prop. 4.2]{serfaty2020gaussian}. 

\begin{claim}[Radial monotone rearrangement]
\label{claim:radial_transport}
For convenience let us assume that $\varphi$ is radially symmetric around $0$, and let $\rmax$ be the outer radius of the annulus $\AA$. Let $\Fs : [0, \rmax] \to \R_+$ be the cumulative radial distribution function of the density $\mus$, namely:
 \begin{equation*}
\Fs := r \mapsto \int_{0}^r \left(1 -  \frac{s}{2\pi \beta} \Delta \varphi(\rho) \right) 2 \pi \rho \dd \rho = \pi r^2 - \frac{s}{\beta}  \int_{0}^r \Delta \varphi(\rho) \rho \dd \rho,
\end{equation*}
where we denote by $\Delta \varphi(\rho)$ the value of $\Delta \varphi(x)$ at any point $x$ with $|x| = \rho$. Let $\Phis$ be the transport map defined by:
\begin{equation}
\label{def:Phis}
\Phis(r) := \Fs^{-1}(\pi r^2) \ (r \geq 0).
\end{equation}
Finally, let $\vPhis$ be the map $\vPhis(x) := \Phis(|x|) \frac{x}{|x|}$. Then:
\begin{enumerate}
\item  $\vPhis$ is a $C^1$-automorphism of the disk $\DD(0, \rmax)$, which transports $\muc$ onto $\mus$.
\item The map $\psis := x \mapsto \vPhis(x) - x$ is supported on the annulus $\AA$ and satisfies:
\begin{equation}
\label{eq:psisvarphi}
|\psis|_{\1} \preceq s |\varphi|_{\2}.
\end{equation}
\end{enumerate}
\end{claim}
\begin{proof}[Proof of Claim \ref{claim:radial_transport}]
Since $s$ satisfies \eqref{eq:condition_on_s}, $\Fs$ is strictly increasing and continuous, thus $\Phis = \Fs^{-1}$ is well-defined (it is the so-called “monotone rearrangement”). Since $\muc, \mus$ are radially symmetric the map $\vPhis$ transports $\muc$ onto $\mus$ (as two-dimensional measures), and we have $\vPhis(x) = x$ for $x \notin \AA$. 

A computation in polar coordinates shows that $|\psis|_{\1} \leq |\Phis' - 1|_{\0}$, on the other hand the derivatives of $\Phis$ can be estimated in terms of the perturbation measure $\Delta \varphi$ by using the transportation identity \eqref{def:Phis}, which reads:
\begin{equation*}
\Phis(r)^2 - r^2 = \frac{s}{\pi \beta} \int_{0}^{\Phis(r)} \Delta \varphi( \rho ) \rho \dd \rho,
\end{equation*}
and taking derivatives. In particular one finds (after an integration by parts) that $|\Phis(r) - r| \preceq s |\varphi|_{\2} r$, and then that: $|\Phis'(r) - 1| \preceq s |\varphi|_{\2}$. 
\end{proof}

\subsubsection*{\cor{Step 3. Effect on the energy}}
Next let us extend the notation $\vPhis$ to $\vPhis(\bXN) := \sum_{i = 1}^N \delta_{\vPhis(x_i)}$.
\begin{claim}[Effect of the transportation on the energy]
\label{claim:effect_transport_energy}
We have:
\begin{equation*}
\left|\FN(\vPhis(\bXN), \mus) - \FN(\bXN, \muc)\right| \preceq s |\varphi|_{\2} \EnerPoints(\bXN, \AA).
\end{equation*}
\end{claim}
\begin{proof}[Proof of Claim \ref{claim:effect_transport_energy}]
We apply \cite[Prop 4.2]{serfaty2020gaussian} with (in the notation of that paper) $U_\ell = \mathcal{A}$, which contains the support of $\psi_s$ and thus of its derivative. Then \cite[(4.5)]{serfaty2020gaussian} states that the electric energy in $\AA$ stays bounded along the transport by $\vPhis$, and \cite[(4.6)]{serfaty2020gaussian} that in fact its derivative is bounded by $|\psis|_\1$ times the initial electric energy plus the number of points in $\AA$.
\end{proof}

\subsubsection*{\cor{Conclusion.}}
To conclude, we change variables in \eqref{eq:rewriting_Laplace_I} using $\vPhis$ and use Claim \ref{claim:effect_transport_energy} to estimate the effect on the energy. The Jacobian term that appears is again of order $s |\varphi|_\2$ times the number of points in the support of $\psis$ and can be incorporated in the previous error term.
\end{proof}

\subsection{Cornering the discrepancy: proof of Proposition \ref{prop:chase_the_goat}}
\label{sec:proof_chase_the_goat}
\begin{proof}[Proof of Proposition \ref{prop:chase_the_goat}]
\cor{We treat here the case of a positive discrepancy (an excess of points), i.e. of the event $\DisInsP$, the other case being similar.}

Let $\varphi$ be a nonnegative, non increasing, compactly supported $\CC^2$ test function with radial symmetry around $z$ such that:
\begin{enumerate}
\item $\varphi \equiv 1$ on $\DD_{R-2L}$, $\varphi \equiv 0$ outside $\DD_{R-L}$, $\varphi$ takes values in $[0,1]$.
\item $|\varphi|_{\1} \leq \frac{100}{L}, \quad |\varphi|_{\2} \leq \frac{100}{L^2}$.
\end{enumerate}
Let $\tphi : \R^2 \to \R$ be the function defined by $\tphi(x) := \varphi(|x-z|)$. Assume that:
\begin{equation}
\label{assumption_is_spread}
\bXN \in \DisInsP := \bigcap_{R - 2L \leq r \leq R-L} \left\lbrace \Di(\bXN, \DD(z, r)) \geq \hal \eRR \right\rbrace 
\end{equation}
\begin{claim}[$\tphi$ detects the discrepancy]
	\label{claim:varphi_detects} Under \eqref{assumption_is_spread} we have:
\begin{equation}
\label{eq:varphi_detects}
\Fluct[\tphi] \geq \frac{\eRR}{2}.
\end{equation}
	\end{claim}
	\begin{proof}[Proof of Claim \ref{claim:varphi_detects}]
Without loss of generality we may assume that $z = 0$. We take advantage of the radial symmetry of $\tphi$ and re-write $\Fluct[\tphi]$ as: 
\begin{equation*}
\Fluct[\tphi] = \int_{\DD_{R-L}} \tphi(x) \dd \fN(x) = \int_{0}^{R-L} \varphi(r) \left( \frac{\dd}{\dd r} \Di(\bXN, \DD_r) \right) \dd r.
\end{equation*}
Integrating by parts, we get: $\Fluct[\tphi] = - \int_{R-2L}^{R-L} \varphi'(r) \Di(\bXN, \DD_r) \dd r.$ By assumption \eqref{assumption_is_spread} we have $\Di(\bXN, \DD_r) \geq \frac{\eRR}{2}$ for all $r$ in the domain of integration, moreover by construction $- \varphi' \geq 0$ and its integral is $1$. We thus obtain \eqref{eq:varphi_detects}.
\end{proof}
We have constructed a radially symmetric test function $\tphi$ that is of class $C^2$ and detects a fraction of the discrepancy. Let us compare \eqref{eq:varphi_detects} with the control on the size of $\Fluct[\tphi]$ given by Proposition \ref{prop:bound_fluct_radial}.

\begin{claim}[Fluctuations of $\tphi$]
\label{claim:typical_fluct_varphi} With $s$ as in \eqref{eq:s_and_L}
\begin{equation}
\label{eq:typical_fluct_varphi}
\log \EN \left[ e^{s \Fluct[\tphi]} \right] \leq \Cc_\beta (s^2 + s) \frac{R}{L}.
\end{equation}
\end{claim}
\begin{proof}[Proof of Claim \ref{claim:typical_fluct_varphi}]
By construction, $\tphi$ is a $\CC^2$ test function with compact support and radial symmetry, and $|\tphi|_{\2} \leq \frac{100}{L^2}$. We assumed $L \leq \frac{R}{10}$ and $|s| \leq \frac{1}{\Cins} \frac{L^3}{R}$ (see condition \eqref{eq:s_and_L}), hence up to choosing the constant $\Cins$ large enough (depending on $\beta$) we can ensure that $s$ satisfies $|s| \leq \frac{\pi \beta}{4 |\tphi|_{\2}}$. Moreover the support of $\Delta \tphi$ is an annulus. In particular we may apply Proposition~\ref{prop:bound_fluct_radial} and control the exponential moment of $\Fluct[\tphi]$ by:
\begin{equation}
\label{usingCLT}
\log \EN \left[  \exp\left( s \Fluct[\tphi]  \right) \right] = \frac{s^2}{4 \pi \beta} \int_{\R^2} |\nabla \tphi|^2 
+ \log \EN \left[\exp \left(s |\tphi|_{\2} \O\left( \Cc_\beta \EnerPoints\left(\DD_{R-L} \setminus \DD_{R-2L}\right) \right) \right) \right].
\end{equation}
The quantity $\int_{\R^2} |\nabla \tphi|^2$ is readily bounded (up to some multiplicative constant) by $RL$ (the area of the annulus where $\Delta \tphi$ is supported) times $\frac{1}{L^2}$ (the order of magnitude of $|\varphi|_\1^2$), and thus:
\begin{equation}
\label{Hund}
\frac{s^2}{4 \pi \beta}\int_{\R^2} |\nabla \tphi|^2 = \O \left(\frac{s^2R}{L}\right).
\end{equation}
It remains to estimate the contribution of the following term:
\begin{equation*}
\Remain := \EN \left[\exp \left(s |\tphi|_{\2} \O\left( \Cc_\beta \EnerPoints\left(\DD_{R-L} \setminus \DD_{R-2L}\right) \right) \right)  \right].
\end{equation*}
By construction we know that $|\tphi|_2$ is of order $\frac{1}{L^2}$. Cover the annulus $\DD_{R-L} \setminus \DD_{R-2L}$ by a family $\{\sq_i\}_{i \in I}$ of $\#I = \O \left( \frac{R}{L} \right)$ squares of sidelength $\O\left(L\right)$ and write:
\begin{equation*}
\exp \left(s |\tphi|_{\2} \O\left( \Cc_\beta \EnerPoints\left(\DD_{R-L} \setminus \DD_{R-2L}\right) \right) \right) \leq \exp \left(\frac{\Cc_\beta s R}{L^3} \frac{1}{\# I} \sum_{i \in I} \EnerPoints(\sq_i) \right).
\end{equation*}
By convexity we get:
\begin{equation}
\label{remain2}
\Remain \leq \frac{1}{\# I} \sum_{i \in I} \EN \left[\exp\left( \frac{\Cc_\beta s R}{L^3} \EnerPoints(\sq_i)  \right) \right].
\end{equation}
Up to chosing $\Cins$ large enough we can ensure that the parameter $\frac{\Cc_\beta sR}{L^3}$ in the right-hand side of \eqref{remain2} is smaller than any fixed constant. Then for each $i \in I$ the local laws (Proposition \ref{prop:local_laws}) yield:
\begin{equation*}
\log \EN \left[ \exp\left( \frac{\Cc_\beta s R}{L^3} \left( \EnerPoints(\sq_i) \right) \right) \right] \leq \frac{\Cc_\beta s R}{L^3} \times \CLL L^2 = \O_\beta\left(\frac{sR}{L} \right).
\end{equation*}
Taking an average over $i \in I$ yields: 
\begin{equation}
\label{conclusion:remain}
\log \Remain = \O_\beta\left( \frac{s R}{L} \right),
\end{equation}
and combining it with \eqref{usingCLT}, \eqref{Hund} and \eqref{conclusion:remain}, we obtain \eqref{eq:typical_fluct_varphi} as claimed. 
\end{proof}
Next, applying Markov's inequality in exponential form to \eqref{eq:typical_fluct_varphi}, we get:
\begin{equation*}
\PNbeta\left(\Fluct[\tphi] \geq \frac{\eRR}{2}\right) \leq \exp \left( - \frac{s \eRR}{2} + (s^2 + s) \Cc_\beta\left(\frac{R}{L}\right) \right).
\end{equation*}
Since we assume $L \geq \Cins \frac{1}{\eR}$ and $s \leq \frac{1}{\Cins} L \eR$, up to choosing $\Cins$ large enough we can ensure that:
\begin{equation*}
- \frac{s\eRR}{2} + (s^2 + s) \Cc_\beta\left(\frac{R}{L}\right) \leq - \frac{s \eRR}{4}, 
\end{equation*}
which concludes the proof.
\end{proof}

\section{Study of the external potential: proof of Proposition \ref{prop:VextIsOftenGood}}

\let\oldrho\rho
\renewcommand{\rho}{T}
\renewcommand{\hrho}{\hT}

\label{sec:proofVext}
\begin{proof}[Proof of Proposition \ref{prop:VextIsOftenGood}]
Let us fix an index $i \in \{1, \dots, \Nn\}$ and study the external potential $\Vext$ (as defined in \eqref{def:Vext}) on $\Lai$. Let $\Vvi$ be the potential generated on $\Lai$ by “everything outside $\Lai$”, namely:
\begin{equation}
\label{def:Vi}
\Vvi(x) := \int_{\LN \setminus \Lai} - \log |x-y| \dd \fN(y), \ x \in \Lai.
\end{equation}
For all $x$ in $\Lai$ we have, by definition:
\begin{equation}
\label{VextVi}
\VextL(x) = \Vvi(x) - \sum_{j, j \neq i} \int_{\Laj} - \log |x-y| \dd \fN(y).
\end{equation}
For convenience, we will first study $\Vvi$ itself and then use a rough bound on the remaining terms in \eqref{VextVi}. In the rest of this section we choose an auxiliary length scale $\hrho$ as:
\begin{equation*}
\hrho := \log \rho.
\end{equation*}

\cor{The proof goes into five steps given in the following five subsections. The steps 2-3-4 each correspond to their own technical result (the combination of which eventually yields the statement of Proposition~\ref{prop:VextIsOftenGood} in step 5), which is stated at the beginning of each subsection and proven immediately thereafter.}

\subsection{Decomposition and regularization}
\label{sec:DecompoRegul}
\subsubsection*{Decomposition of $\Vv_i$}
Let us introduce a smooth cut-off function $\chi_i$ equal to $1$ on $\DD(\omega_i, \rho + \hrho)$ and to $0$ outside $\DD(\omega_i, \rho + 2 \hrho)$, with $|\chi_i|_\kk \leq \Cc \hrho^{-\kk}$ and let:
\begin{itemize}
\item $\RR'_i$ be the field generated by the background measure in the annulus $\AA := \DD(\omega_i, \rho + 2 \hrho) \setminus \DD(\omega_i, \rho)$, weighted by $\chi_i$, namely:
\begin{equation*}
\RR'_i := -\log \ast \left( - \chi_i \ind_{\Sigma_N \setminus \La_i} \mm_0 \right).
\end{equation*}
\item $\h^{\nu_i}$ be the field generated by the positive measure $\nu_i$ corresponding to the point charges in $\AA$, weighted by $\chi_i$, namely: 
\begin{equation*}
\h^{\nu_i} := - \log \ast \left( \chi_i \ind_{\Sigma_N \setminus \La_i} \bX_N \right).
\end{equation*}
\item $\RR_i$ be the field generated by the signed measure $(1-\chi_i) \fN$, namely: 
\begin{equation}
\label{def:Ri}
\RR_i := -\log \ast \left( (1 - \chi_i) \fN \right).
\end{equation}
\end{itemize}
Of course, we have $\h^{\nu_i} + \RR'_i = - \log \ast \left(\chi_i \fN \right)$ and thus the following decomposition holds (cf. \eqref{Vextassum}):
\begin{equation}
\label{decompose_VVi}
\Vv_i = \h^{\nu_i} + \RR'_i + \RR_i.
\end{equation}
\begin{remark}
\label{remark_sur_rprimei}
By Newton's theorem, the logarithmic potential $\RR'_i$ is \emph{constant} within the disk $\La_i$ thanks to its radial symmetry. So its derivative and its “interior” normal derivative $\partial_n^{-} \RR'_i$ (see \eqref{def:jumpR}) are identically $0$ there, and it will play no further role in the proof.
\end{remark}

\subsubsection*{Regularization of $\Vv_i$ near the edge.}
Since there might be charges in $\nu_i$ located very close to $\partial \La_i$, we cannot expect in general to have an upper bound on $\Vvi$ in $\La_i$ valid \emph{up to the edge}. In order to study $\Vvi$ near $\partial \La_i$, we introduce a smooth cut-off function $\chir : \R^2 \to [0,1]$ such that:
\begin{equation*}
\chir(z) = 0 \text{ if } |z| \leq \frac{1}{4}, \quad \chir(z) = 1  \text{ if } |z| \geq 1, \quad |\chir|_{\1} \leq 10,
\end{equation*} 
and we use $\chir$ to define a \emph{regularized} version of $\Vvi$ as follows:
\begin{equation}
\label{def:tVi}
\tVi (x) := \int_{\LN \setminus \Lai} - \log |x-y| \chir(x-y) \dd \fN(y), \ x \in \Lai.
\end{equation}
\begin{claim} \label{claimtVi} We have:
\begin{itemize}
    \item $\tVi(x) \leq \Vvi(x) + 100$ for all $x \in \Lai$.
    \item $\Vvi(x) = \tVi(x)$ for all $x \in \Lai$ with $\dist(x, \partial \Lai) \geq 1$.
\end{itemize}
\end{claim}
\begin{proof}[Proof of Claim \ref{claimtVi}]
For the first point, observe that adding the short-distance cut-off $\chir$ diminishes the (non-negative) influence of point particles at distance less than $1$, as well as the influence of the background in a small disk (which is negative but uniformly bounded). The second point is obvious.
\end{proof}
We introduce the (regularized) “boundary” term $B_i$:
\begin{equation}
\label{eq:defBi}
B_i : x \mapsto B_i(x) := \int_{\LN \setminus \Lai} \chi_i(y) \times  \log |x -y| \chir(x -y)\  \dd \fN(y).
\end{equation} 
We can decompose $\tVi$ as (compare with \eqref{decompose_VVi}):
\begin{equation}
\label{decompose:tVi}
\tVi = B_i + \RR_i,
\end{equation}
with $\RR_i$ as above in \eqref{def:Ri} (observe that for $x \in \La_i$ and $y$ in the support of $1 - \chi_i$ we have $\chir(x-y) = 1$ i.e. the regularization by $\chir$ is transparent and does not appear in $\RR_i$).

\subsection{Study of the boundary term}
\label{sec:nablaBi}
\begin{claim}
\label{claim:nablaBi}
We have, with $B_i$ as in \eqref{eq:defBi}:
\begin{equation}
\label{eq:nablaBi}
\PNbeta\left[ \sup_{x \in \La_i} |\nabla B_i(x)| \leq \Cc_\beta \log^2 T \ \right] \geq 1 - \exp\left(- \frac{1}{\Cc_\beta} \log^2 T \right).
\end{equation}
\end{claim}
\begin{proof}[Proof of Claim \ref{claim:nablaBi}]
For $x \in \La_i$ arbitrary and $y$ in the support of $\chi_i$ we have:
\begin{equation*}
\| \nabla \left(\log |x-y| \chir(x-y)\right) \| \preceq \frac{1}{1 + |x-y|}  
\end{equation*}
(this is possible thanks to the regularization due to $\chir$ near $\partial \La_i$). To control $|\nabla B_i(x)|$ it is thus enough to bound, for $x \in \La_i$:
\begin{equation}
\label{blablaLNLai_1}
\int \ind_{\LN \setminus \La_i}(y) \chi_i(y) \frac{1}{1 + |x-y|} \left(\dd \bXN(y) + \dd y\right).
\end{equation}
By construction $\chi_i \ind_{\Sigma_N \setminus \La_i}$ is supported in the annulus $\DD(\omega_i, \rho + 2 \hrho) \setminus \DD(\omega_i, \rho)$. Let us cover this annulus by $\O(\frac{\rho}{\hrho})$ squares of sidelength $\hrho$. In each square, by the local laws, we can ensure that there are at most $\Cc_\beta \hrho^2$ points with probability $1 - \exp\left(- \frac{1}{\Cc_\beta} \hrho^2 \right)$ and so after an union bound we can ensure that \emph{all of them} contain at most $\Cc_\beta \hrho^2$ points with probability $1 - \exp\left(- \frac{1}{\Cc_\beta} \hrho^2 \right)$ ($\hrho^2$ has been chosen in order to beat the combinatorial loss due to such union bounds). If so, then we have:
\begin{equation*}
\int \ind_{\LN \setminus \La_i}(y) \chi_i(y) \frac{1}{1 + |x-y|} \dd \bXN(y) \preceq \hrho \int_{y \in \partial \La_i} \frac{1}{1 + |x-y|} \dd y,
\end{equation*}
for which a rough bound is $\O\left(\log^{2} \rho\right)$. The continuous part of \eqref{blablaLNLai_1} is of the same order.
\end{proof}

We now focus on studying $\RR_i$. In the next sections the constants $\Cc_\beta$ will also depend on $\delta$ as in Assumption \eqref{eq:assumption_distance}.

\subsection{Estimate on the first derivative at the center}
\begin{lemma}
\label{lem:nablaRiwi}
We have, with $\RR_i$ as in \eqref{def:Ri}:
\begin{equation}
\label{eq:sizenablaV}
\PNbeta \left[ |\nabla \RR_i(\omega_i)| \leq \Cc_\beta \right] \geq 1 - \exp\left(-\frac{\log^2 T}{\Cc_\beta}\right). 
\end{equation}
\end{lemma}
\begin{proof}[Proof of Lemma \ref{lem:nablaRiwi}]
Let us bound the first partial derivative of $\RR_i$. By definition of $\RR_i$ we have:
\begin{equation*}
\partial_1 \RR_i(x) = - \int_{\LN} (1 - \chi_i(y)) \partial_1 \log|x-y| \dd \fN(y). 
\end{equation*}

\textbf{Step 1. The first outer layer.}
Let $\sigma$ be a cut-off function with radial symmetry around $\omega_i$ such that:
\begin{equation*}
\sigma(z) = 1 \text{ if } |z - \omega_i| \leq  2 \rho, \quad \sigma(z) = 0 \text{ if } |z - \omega_i| \geq  3 \rho,\quad |\sigma|_{\kk} \leq \Cc T^{-\kk} \text{ for $\kk \in \{1, 2\}$}.
\end{equation*}
We start by evaluating the contribution to $\RR_i$ coming from the support of $\sigma$.
\begin{claim}
\label{claim:OuterLayer}
Let $A := \int (1 - \chi_i(y)) \sigma(y) \partial_1 \log|\omega_i-y| \dd \fN(y)$. We have:
\begin{equation}
\label{eq:OuterLayer}
\PNbeta \left( |A| \leq \Cc_\beta \right) \geq 1 - \exp\left(- \frac{1}{\Cc_\beta} \hrho^2 \right).
\end{equation}
\end{claim}
\begin{proof}[Proof of Claim \ref{claim:OuterLayer}]
We are looking at the fluctuations of $\varphi := y \mapsto \left(1- \chi_i(y) \right) \sigma(y) \partial_1 \log|\omega_i - y|$ which is of class $\CC^1$ and compactly supported within $\DD(\omega_i, 3\rho) \setminus \DD(\omega_i, \rho + \hrho)$. Let us recall that:
\begin{itemize}
    \item  $(1- \chi_i), \sigma$ are bounded and $\partial_1 \log|\omega_i - y|$ is of order $\rho^{-1}$ for $y$ in the support of $\varphi$.
    \item $ |\chi_i|_{\1} $ is of order $\hrho^{-1}$, $|\sigma|_\1$ is of order $\rho^{-1}$ and $\left|\partial_1 \log|\omega_i - \cdot| \right|_{\1, \supp \varphi}$ is of order $\rho^{-2}.$
\end{itemize}
Moreover let us make the following observations:
\begin{itemize}
\item On the annulus $\AA_1 := \DD(\omega_i, \rho + 2 \hrho) \setminus \DD(\omega_i, \rho + \hrho)$, of area $\O(\rho \hrho)$, $\|\nabla \varphi\|$ is of order $(\rho \hrho)^{-1}$ (the dominant contribution comes from differentiating $\chi_i$).
\item On the annulus $\AA_2 := \DD(\omega_i, 3 \rho) \setminus \DD(\omega_i, \rho + 2 \hrho)$, of area $\O(\rho^2)$, $\|\nabla \varphi \|$ is of order $\rho^{-2}$ ($\chi_i$ does not play a role anymore).
\end{itemize}
Using Lemma \ref{lem:apriori} in its localized version \eqref{eq:apriori_one_var_localized} we are left to bound:
\begin{equation*}
A_1 := \left(\rho \hrho\right)^{-1} \times \left( \Ener(\bXN, \AA_1)^\hal \times \left(\rho \hrho\right)^\hal + \Points\left(\bXN, \AA_1 \right) \right)
\end{equation*}
and 
\begin{equation*}
A_2 := \rho^{-2} \times \left(  \Ener(\bXN, \AA_2)^\hal \times (\rho^2)^\hal + \Points\left(\bXN, \AA_2 \right) \right).
\end{equation*}
To control the first term $A_1$, we cover the annulus $\AA_1$ by $\O(\frac{\rho}{\hrho})$ squares of sidelength $\hrho$ and use the local laws. We obtain:
\begin{equation*}
\PNbeta \left( \EnerPoints(\bXN, \AA_1) \geq \Cc_\beta \rho \hrho \right) \leq \exp\left(- \frac{1}{\Cc_\beta} \hrho^2 \right),
\end{equation*}
where we used a union bound on $\exp\left( \O(\log \rho) \right)$ events and the fact that $\log \rho \ll \hrho^2$. In particular this ensures that:
\begin{equation*}
\PNbeta \left( |A_1| \leq \Cc_\beta \right) \geq 1 - \exp\left(- \frac{1}{\Cc_\beta} \hrho^2 \right).
\end{equation*}
On the other hand, to control the second term $A_2$ we may apply the local laws to the full disk $\DD(\omega_i, 3 \rho)$ and see that:
\begin{equation*}
\PNbeta \left( \EnerPoints\left(\bXN, \DD(\omega_i, 3\rho) \right) \geq \Cc_\beta \rho^2 \right) \leq \exp\left(- \frac{\rho^2}{\Cc_\beta} \right), 
\end{equation*}
which ensures that $\PNbeta \left( |A_2| \leq \Cc_\beta \right) \geq 1 - \exp\left(- \frac{\rho^2}{\Cc_\beta} \right)$. We deduce \eqref{eq:OuterLayer}.
\end{proof}

\textbf{Step 2. Dyadic scales up to the boundary.}
Let $K := \log_2 \left(\frac{1}{10} \dist\left( \La_i, \partial \LN \right)\right)$. By our assumption \eqref{eq:assumption_distance} we have $K \geq \log_2 \left(\frac{\delta \sqrt{N}}{10}\right) $.

For $\log_2(2T) \leq k \leq K$ we let $\tau_k$ be a smooth cut-off function such that:
\begin{itemize}
    \item $\tau_k$ has radial symmetry around $\omega_i$ and is supported on the annulus $\AA'_k := \DD(\omega_i, 2^{k+1}) \setminus \DD(\omega_i, 2^{k})$ of area $\O(2^{2k})$.
    \item We have $|\tau_k|_\1 \leq \Cc 2^{-k}$ and $|\tau_k|_\2 \leq \Cc 2^{-2k}$ (with $\Cc$ independent of $k$).
    \item We have the following “partition of unity”-type of identity:
    \begin{equation}
    \label{partition_unity_tauk}
\sigma(z) + \sum_{k = \log_2(2 \rho)}^{K} \tau_k(z) \equiv 1 \text{ for } |z-\omega_i| \leq 2^{K-1}.
    \end{equation}
\end{itemize}
For each $k$ we let $\phi_k$ be the following map:
\begin{equation*}
\phi_k : y \mapsto - \tau_k(y) \partial_1 \log|\omega_i - y|.
\end{equation*}
One can check that:
\begin{equation*}
|\phi_k|_\1 \leq \Cc 2^{-2k}, \quad |\phi_k|_\2 \leq \Cc 2^{-3k}.
\end{equation*}

\begin{claim}
\label{claim:dyadic_scales}
We have, for $k = \log_2(2 \rho), \dots, K$:
\begin{equation}
\label{eq:dyadic_scales}
\PNbeta\left[ \left| \Fluct[\phi_k] \right| \geq \Cc_\beta 2^{-k/2} \right] \leq \exp\left(- \frac{1}{\Cc_\beta} 2^{k/2}\right).
\end{equation}
\end{claim}
\begin{proof}[Proof of Claim \ref{claim:dyadic_scales}]
We use the “fine bounds” on fluctuations given by Proposition \ref{prop:bound_fluct_radial}. We choose the parameter $s = c 2^{k}$ for some constant $c$, as allowed by \eqref{eq:condition_on_s}, and observe that, in view of the bounds mentioned above:
\begin{equation*}
s^2 \int_{\AA'_k} |\nabla \phi_k|^2 \preceq s^2 2^{2k} \times \left(2^{-2k} \right)^2 = \O\left(c^2 \right).
\end{equation*}
Moreover we may use the local laws to write (note that we deliberately stopped before reaching the boundary thanks to our choice of $K$):
\begin{equation*}
\log \EN\left(e^{\Cc_\beta s |\phi_k|_\2 \left(\EnerPoints(\bXN, \AA'_k)\right)} \right) \leq \log \EN\left(e^{\Cc_\beta c 2^{-2k} \left(\EnerPoints(\bXN, \AA'_k)\right)} \right) \leq c \Cc_\beta.
\end{equation*}
In summary, we have:
\begin{equation*}
\log \EN\left[e^{c 2^{k} \Fluct[\phi_k]}\right] = \O\left(c^2 \right) + \O_\beta\left(c\right),
\end{equation*}
and Markov's inequality yields the claim.
\end{proof}
Summing up the controls on $\phi_k$ and using a union bound, we obtain:
\begin{equation}
\label{eq:sum_phik}
\PNbeta\left[ \left| \sum_{k = \log_2(2 \rho)}^K \Fluct[\phi_k] \right| \geq \Cc_\beta T^{-1/2} \right] \leq \exp\left(- \frac{1}{\Cc_\beta} T^{1/2} \right) \ll \exp\left(- \log^2 T\right).
\end{equation}
\end{proof}

\textbf{Step 3. The rest of the system.} It remains to study the contribution to $\nabla \RR_i$ coming from the part of the system  far from  $\omega_i$. Let us introduce an artificial cut-off $\gamma$ supported on $2 \LN$ - \cor{this is fine because in our model all the points must live in $\LN$.}

We have
\begin{multline*}
\int_{\LN} \left( 1 - \left(\sigma(y) + \sum_{k = \log_2(2 \rho)}^{K} \tau_k(y)\right) \right) \partial_1 \log|y-\omega_i| \dd \fN(y), 
\\
= \int \gamma(y) \left( 1 - \left(\sigma(y) + \sum_{k = \log_2(2 \rho)}^{K} \tau_k(y)\right) \right) \partial_1 \log|y-\omega_i| \dd \fN(y) 
\end{multline*}
and we may thus apply Lemma \ref{lem:apriori} to the function:
\begin{equation*}
\varphi : y \mapsto \gamma(y) \left( 1 - \left(\sigma(y) + \sum_{k = \log_2(2 \rho)}^{K} \tau_k(y)\right) \right) \partial_1 \log|y-\omega_i|,
\end{equation*}
whose derivative is of order at most $\frac{1}{\delta^2 N}$ in view of Assumption \ref{eq:assumption_distance} (this is where the dependency in $\delta$ comes in). The support of $\varphi$ has area $N$, the number of points is obviously bounded by $N$ and the energy there is $\O(N)$ with high probability by the \emph{global} law. We obtain:
\begin{equation*}
|\Fluct[\varphi]| \leq \Cc_\beta \frac{1}{\delta^2}, \text{ with probability } 1 - \exp\left(-\frac{1}{\Cc_\beta} N\right).
\end{equation*}

Keeping the dominant contributions of these three steps, we obtain \eqref{eq:sizenablaV}.
\end{proof}

\subsection{Estimate on the second derivative up to the boundary.}
\begin{lemma}
\label{lem:2nablaRiwi}
We have, with $\RR_i$ as in \eqref{def:Ri}:
\begin{equation}
\label{eq:size2nablaV}
\PNbeta \left[ \sup_{x \in \La_i} |\Dd^2 \RR_i(x)| \times \left(1 + \dist(x, \partial \La_i) \right) \leq \Cc_\beta \right] \geq 1 - \exp\left(- \frac{1}{\Cc_\beta} \hrho^2 \right). 
\end{equation}
\end{lemma}
\begin{proof}[Proof of Lemma \ref{lem:2nablaRiwi}]
We bound one of the four second partial derivatives (they can all be treated the same way) and use the same decomposition as in the proof of Lemma \ref{lem:nablaRiwi}.

\medskip

\textbf{Step 1. The first outer layer.} To control the first outer layer, we proceed like for Claim \ref{claim:nablaBi}. For $x \in \La_i$ arbitrary and $y$ in the support of $(1- \chi_i) \sigma$ we have:
\begin{equation*}
\| \Dd^2 \log|x-y| \| \preceq \frac{1}{\left(\hrho + |x-y|\right)^2}  
\end{equation*}
(the points $x,y$ being at distance at least $\hrho$ from each other). It is thus enough to bound:
\begin{equation}
\label{blablaLNLai}
\int \sigma(y) \left(1 - \chi_i(y)\right) \frac{1}{\left(\hrho + |x-y|\right)^2} \left(\dd \bXN(y) + \dd y\right).
\end{equation}
Assume as above that each square of sidelength $\hrho$ covering the annulus $\DD(\omega_i, \rho + 2 \hrho) \setminus \DD(\omega_i, \rho)$ contains at most $\Cc \hrho^2$ points (an event that occurs with probability $1 - \exp\left(- \frac{1}{\Cc} \hrho^2 \right)$ and does not depend on $x$). We may then compare \eqref{blablaLNLai} to the following one-dimensional integral:
\begin{equation*}
\hrho \int_{y \in \partial \La_i} \frac{1}{\left(\hrho + |x-y|\right)^2} \dd y,
\end{equation*}
for which a rough bound is $\Cc \frac{1}{1 + \dist(x, \partial \La_i)}$. 

\medskip

\textbf{Step 2. Dyadic scales up to the boundary.} We proceed as in the previous proof, but instead of using the fine bounds of Proposition \ref{prop:bound_fluct_radial} (which are only valid function-wise, so $x$-wise here) we use the rougher controls of Lemma \ref{lem:apriori} which allow for a uniform control. For each $k$ we let $\phi_k$ be the following map:
\begin{equation*}
\phi_{k,x} : y \mapsto - \tau_k(y) \partial_{12} \log|x - y|.
\end{equation*}
If $x$ is in $\La_i$ and $y \in \AA'_k$ for $k \geq \log_2(2\rho)$ we have $| \phi_{k,x} |_\1 \leq \Cc 2^{-3k}$. Hence using \eqref{eq:apriori_one_var} we know that we can bound $\Fluct[\phi_{k,x}]$ by:
\begin{equation*}
2^{-3k} \times \left(2^{2k} + \EnerPoints(\bXN, \AA'_k)\right).
\end{equation*}
Using the local laws, we may control $\EnerPoints(\bXN, \AA'_k)$ by $\Cc_\beta 2^{2k}$ with probability $1 - \exp\left(-\frac{1}{\Cc_\beta} 2^{2k} \right)$. We deduce that:
\begin{equation*}
\PNbeta\left[ \sup_{x \in \La_i} \left| \Fluct[\phi_{k,x}] \right| \geq \Cc_\beta 2^{-k} \right] \leq \exp\left(- \frac{1}{\Cc_\beta} 2^{2k} \right).
\end{equation*}
Summing again the contributions through an union bound, we get that the contribution to the second derivative of $\RR_i$ due to the dyadic annuli is bounded by $\Cc \rho^{-1}$ with probability $1 - \exp\left(-\frac{\rho^2}{\Cc_\beta}\right)$.

\medskip

\textbf{Step 3. The rest of the system.} We argue as in the previous proof, but thanks to the additional derivative we gain a factor $\frac{1}{\delta \sqrt{N}} \ll \dist(x, \La_i)^{-1}$.
\end{proof}

\subsection{Summary and conclusion}

\subsubsection*{Bounding the contributions of other sub-systems}
Returning to \eqref{def:Vi}, \eqref{VextVi}, a rough bound on $\|\Dd \log|x-y| \|$ for $x \in \La_i$ and $y \in \La_j$, together with the fact that each sub-system contains by assumption $\O(T^2)$ points, ensures that (using \eqref{sumdij}):
\begin{equation}
\label{eq:VviVext}
\left| \Vvi - \Vext \right|_{\1, \Lai} \leq \Cc T^2 \sum_{j, j \neq i} \frac{1}{\dij} = \O\left(T^2 \frac{\log R}{ML} \right).  
\end{equation}
(Let us note that it is fairly easy to improve \eqref{eq:VviVext} by expanding the interaction in a more precise way and controlling some fluctuations, but it would add technicalities while not making a big difference in our final statement.)

\subsubsection*{Summary}
\textbf{1. Control up to the edge.}
 Combining Lemma \ref{lem:nablaRiwi} and Lemma \ref{lem:2nablaRiwi} and integrating between $\omega_i$ and any given $x$ in $\La_i$ we deduce that $|\nabla \RR_i(x)| \leq \Cc \log \rho$ for all $x$ in $\Lai$. Combining this with the decomposition \eqref{decompose:tVi} and Claim \ref{claim:nablaBi} we deduce that
    \begin{equation*}
    \|\nabla \tVi\| \leq \Cc_\beta \log^{2} T \text{ on $\Lai$...}
    \end{equation*}
...and thus after integrating between $\omega_i$ and any given $x$ in $\La_i$ we obtain:
    \begin{equation*}
    \left| \tVi(x) - \tVi(\omega_i) \right| \leq \Cc_\beta T \log^{2} T
    \end{equation*}

\medskip

\textbf{2. Properties of the decomposition.} We have already bounded $|\RR_i|_{\1, \La}$, and we have used the control of $\nu$ at scale $\hT = \log T$ several times in the argument above.

In conclusion, with probability  $1 - \exp\left(- \frac{\hrho^2}{\Cc_\beta} \right)$ the potential $\Vv_i$ satisfies the requirements to be a “good external potential” on $\Lai$. Thus in view of \eqref{eq:VviVext}, since we enforce \eqref{eq:condi_T2logTML}, with probability $1 - \exp\left(- \frac{\hrho^2}{\Cc_\beta} \right)$ the potential $\Vext$ is a “good external potential with constant $\bCc$”, where $\bCc$ is some large enough constant depending only on $\beta$ and $\delta$.

We conclude with a union bound on $\Nn$ such events.

\section{\cor{Background measure and global laws for sub-systems: proof of Proposition \ref{prop:correspondence} and Proposition \ref{prop:global_law_SUBSYS}}}
\label{sec:global_law}
\renewcommand{\rho}{\oldrho}

\cor{Here is the plan of this section:
\begin{itemize}
	\item In Section \ref{sec:effect-perturb-eq}, we give some general results on the background/equilibrium measure associated to harmonic external potentials as the ones appearing in \eqref{def:Vext}. This is strongly inspired by \cite{bauerschmidt2017local}.
	\item In Section \ref{sec:applications_to_generalized}, we specifically use these results in the case of “good” external potentials as defined in Section \ref{sec:good_good} and we show that our generalized 2DOCP's with external potentials are “equivalent” to generalized 2DOCP's with a certain background measure, proving Proposition \ref{prop:correspondence}.
	\item In Section \ref{sec:proofglobal}, we combine our study of the background measure, the properties of good external potentials / sub-systems and explicit computations in order prove the global law for good sub-systems stated in Proposition \ref{prop:global_law_SUBSYS}.
\end{itemize}
}

\subsection{Effect of an harmonic perturbation on the equilibrium measure}
\label{sec:effect-perturb-eq}
In this section, we revisit (parts of) the analysis of \cite[Sec.3]{bauerschmidt2017local}. Broadly speaking, our goal is to understand the effect of the external potential $\Vext$ on the typical repartition of charges within the subsystem. In particular, we want to use the fact that $\Vext$ is \emph{harmonic} on $\La$ and is generated by a “nice” external configuration. We start by recalling some elements of logarithmic potential theory.

\subsubsection{Some potential theory}
\label{sec:potentialTheory}
Let $\VV : \R^2 \to \R \cup \{+ \infty\}$ be a lower semi-continuous function satisfying, for some $\epsilon > 0$, the following growth condition:
\begin{equation*}
\lim_{|z| \to \infty} \left(\VV(z) - (2 + \epsilon) \log |z|\right) = + \infty.
\end{equation*}
\cor{Assume that $\VV < + \infty$ on some open disk.} Denote by $\IV$ the following functional defined on the space $\PP$ of probability measures on $\R^2$:
\begin{equation}
\label{def:IV}
\IV : \mu \mapsto  \iint - \log |x-y| \dd \mu(x) \dd \mu(y) + \int \VV \dd \mu.
\end{equation}
The next proposition\footnote{The statements of Proposition \ref{prop:frostman} only hold “quasi everywhere” (q.e.), which means “up to a set of capacity zero”. This makes no difference for us and, for simplicity, we omit it.
} covers well-known properties of the minimization problem associated to $\IV$.
\begin{proposition}[The equilibrium measure]
\label{prop:frostman} 
There exists a unique minimizer $\muV$ of $\IV$, which we call the “equilibrium measure” associated to $\VV$. Its support, denoted by $\SV$, is compact.

The equilibrium measure $\muV$ is characterized by the following fact: there exists a constant $\cV$ (depending on $\VV$) such that the logarithmic potential generated by $\muV$, namely $\hmuV(x) := \int - \log |x-y| \dd \muV(y)$ satisfies Euler-Lagrange equations of the form:
	\begin{equation}
	\label{eq:EulerLagrange}
\begin{cases}
\hmuV + \hal \VV = \cV & \text{ on } \SV \\
\hmuV + \hal \VV \geq \cV & \text{ on } \R^2 \setminus \SV.
\end{cases}
	\end{equation}

Moreover, the equilibrium measure is connected to the solution of an \emph{obstacle problem} in the following sense: if we define $\uV : \R^2 \to \R \cup \{+ \infty\}$ by setting: 
\begin{equation}
\label{def:uV}
\uV(z) := \sup \left\lbrace v(z), v \text{ is subharmonic},\ v \leq \frac{\VV}{2} \text{ on } \R^2,\ \limsup_{|z| \to \infty} \left( v(z) - \log |z| \right) < \infty  \right\rbrace,
\end{equation}
and let the “coincidence set” $\SSV$ be $\SSV := \left\lbrace z \in \R^2, \uV(z) = \hal \VV(z) \right\rbrace$, then $\SV \subset \SSV$ (in words: the support of the equilibrium measure is contained in the coincidence set of the obstacle problem) and we have:
\begin{equation}
\label{uVhmuV}
\uV(z) = \cV - \hmuV(z) \text{ on $\R^2$}.
\end{equation}
In particular the density of $\muV$ is given equivalently by $\frac{1}{2\pi} \Delta \uV$ or $\frac{1}{4 \pi} \Delta \VV$ on $\SV$.
\end{proposition}
\begin{proof}[Proof of Proposition \ref{prop:frostman}]
The first part is a classical result by Frostman \cite{frostman1935potentiel}. The connection with the obstacle problem has been investigated in \cite{hedenmalm2013coulomb} and an exposition can be found in \cite[Sec.2]{serfaty2015coulomb}. See also \cite[Sec. 2.1, 2.2]{bauerschmidt2017local} and the references therein.
\end{proof}

\subsubsection{Harmonic perturbations}
\renewcommand{\VV}{\mathrm{V}^\mathrm{ref}}
\label{sec:HarmonicPerturbations}
We now study how the equilibrium measure reacts to certain perturbations of the external potential. Here we closely follow the exposition of \cite[Sec. 3.3]{bauerschmidt2017local} while making several changes. \cor{We first choose explicitly our “reference” situation - for us, a quadratic potential with strict confinement, giving rise to a uniform equilibrium measure - then we describe the class of perturbations that we allow, and which is strongly inspired by the discussion in \cite{bauerschmidt2017local}, althoug slightly different. The main statements are Proposition \ref{prop:supportmuW}, Proposition \ref{prop:singcomp}: there we give a precise description of the equilibrium measure associated to an external potential of the form “reference potential $+$ admissible perturbation”. The proof of those propositions is again strongly inspired by similar results in \cite{bauerschmidt2017local}, but we need to make several non-trivial modifications - this will then occupy the next sub-sub-sections.}

\medskip

\textbf{A. The reference measure \cor{and reference potential}.}
Let $\rho \geq 1$ be some fixed radius, we work on the disk $\rho \DD := \DD(0, \rho)$ (for convenience we adopt here the notation of \cite{bauerschmidt2017local}). Let us assume that the reference potential $\VV$ is given by:
\begin{equation}
\label{def:VV}
\VV(z) = \begin{cases}
\frac{|z|^2}{\rho^2} & \text{ for } |z| \leq \rho, \\
+ \infty & \text{ for } |z| > \rho.
\end{cases}
\end{equation}
It satisfies the assumptions of Section \ref{sec:potentialTheory}, and it is easy to check that the associated equilibrium measure $\muV$ is the uniform measure on the disk $\rho \DD$ - in other words, we have $\mu_V = \frac{1}{|\rho \DD|} \mm_0$ on $\rho \DD$. 

The density of $\mu_V$ is equal $\frac{1}{\pi \rho^2}$ on its support $\rho \DD$, in particular the assumption that “$\frac{1}{4 \pi} \Delta \VV \geq \alpha$ in $\rho \DD$” made in \cite[Sec. 3.3]{bauerschmidt2017local} is satisfied here but with $\alpha$ proportional to $\rho^{-2}$.

\medskip

\textbf{B. Class of perturbations.}
We work with the same class of potentials $\WW$ as in \cite{bauerschmidt2017local}, namely we take $\WW$ as:
\begin{equation}
\label{def:W}
\WW := \tau \VV + \rho^{-2} \left(\hnu +  \RR \right), \text{ where:}
\end{equation}
\begin{itemize} 
    \item $\VV$ is the reference potential of \eqref{def:VV} and $\tau$ is a real parameter. 

    For us $\tau$ will always be close to $1$, and for convenience we can assume that $|\tau-1| \leq \frac{1}{10}$.
    \item  $\hnu$ is the logarithmic potential generated by a \emph{positive}, finite distribution $\nu$ of charges located \emph{outside} of the disk $\rho \DD$.

    For convenience (and since that is enough for us), we will assume that $\nu$ is supported on an annulus $\DD(0, \rho + \rho') \setminus \DD(0, \rho)$ for some $\rho' \leq \frac{\rho}{100}$.

    In \cite[Prop. 3.3, 3.4]{bauerschmidt2017local} the controls are given in terms of the \emph{total mass} $\|\nu\|$ of $\nu$, but we want more precise estimates and to that aim we introduce the following quantity: 
    \begin{equation}
    \label{def:nuloc}
        \nuloc(\rho') := \sup_{x \in \partial \rho \DD} \nu(\DD(x, \rho')),
        \end{equation} 
    which controls the mass of $\nu$ in a \emph{local} fashion, at scale $\rho'$ along $\partial \rho \DD$.

\item $\RR$ is some harmonic function on $\rho \DD$ which is continuous up to the boundary. \\
An important role is played by the \emph{normal derivative of $\RR$ on $\partial \rho \DD$}. We introduce the notation:
\begin{equation}
\label{def:jumpR}
\partial_n^{-} \RR(z) := \lim_{\epsilon \to 0} \frac{\RR(z) - \RR(z-\epsilon \vec{n})}{\epsilon}, \quad \|\partial_n^{-} \RR\|_{\infty, \partial \rho \DD} := 
\sup_{z \in \partial \rho \DD} |\partial_n^{-} \RR(z)|,
\end{equation}
where for $z \in \partial \rho \DD$ we denote by $\vec{n}$ the outer unit normal vector (to the circle) at $z$. 
\end{itemize}

Our scaling $\rho^{-2}$ in front of the perturbation terms in \eqref{def:W} is not present in \cite[(3.7)]{bauerschmidt2017local}, but it is of course equivalent to taking the mass of $\nu$ and the potential $\RR$ to be of order $\rho^{-2}$ in their statements. This scaling will compensate the fact that for us the quantity $\alpha$ of \cite{bauerschmidt2017local} is of order $\rho^{-2}$ (see above).

\cor{Before moving on, let us give a quick heuristic explanation for the class of pertubations chosen in \eqref{def:W}. The presence of the parameter $\tau \approx 1$ in front of $\VV$ corresponds to the fact that the expected number of particles in $\rho \DD$ will be close to, but not always exactly equal to $\pi \rho^2$. The distribution $\nu$ corresponds to a system of point charges living outside $\rho \DD$, while the harmonic function $\RR$ will take into account (possibly among other things) the background measure of a system living outside $\rho \DD$ - cf. \eqref{def:Vext}.}

\medskip

\textbf{C. Effect of the perturbation.} Passing from $\VV$ to $\WW$ as in \eqref{def:W} has three effects on the equilibrium measure:
\begin{enumerate}
	\item The support loses some parts located near the circle $\partial \rho \DD$.
	\item A singular measure appears on the circle $\partial \rho \DD$.
    \item The continuous density in the new support changes from $\Delta \VV$ to $\tau \Delta \VV$.
\end{enumerate}
In other words, the \cor{corresponding equilibrium} measure $\muW$ can be written as: 
\begin{equation*}
\muW := \tau \muV + \left(\eta - \tau \muV \ind_{B} \right),
\end{equation*}
with $B$ a subset of $\rho \DD$ located near $\partial \rho \DD$ and $\eta$ a measure that is absolutely continuous with respect to the arclength measure $\dd s$ on $\partial \rho \DD$.

Of course the combination of all effects must have total mass $0$ so that the resulting equilibrium measure $\muW$ still be a probability measure. 

\medskip

The following propositions give quantitative controls on the location of $B$ (first item above) and the density of $\eta$ (second item). The techniques are heavily borrowed from \cite[Prop. 3.3 \& 3.4]{bauerschmidt2017local}, with some simplications due to our specific context, and additional precisions due to our use of the quantity $\nuloc$ as in \eqref{def:nuloc}. \cor{We recall that $\rho'$ is a certain lengthscale smaller $\frac{\rho}{100}$ such that $\nu$ is supported on an annulus $\DD(0, \rho + \rho') \setminus \DD(0, \rho)$, and that we used it in the definition of $\nuloc$.}

\begin{proposition}[The support of $\muW$ contains a large sub-disk]
\label{prop:supportmuW}
The support of $\muW$ contains the set of points $z \in \rho \DD$ at distance at least $\kappa$ from $\partial \rho \DD$, with $\kappa$ \cor{given by}: 
\begin{equation}
\label{condi:kappa}
\kappa = \Cc \max\left(20 \rho',\ \frac{\nuloc(\rho')}{\rho'},\ |1-\tau| \rho + \|\partial^-_n \RR\|_{\infty, \partial \rho \DD}\right).
\end{equation} 
\end{proposition}
 
\begin{proposition}[The singular component has a controlled density]
\label{prop:singcomp}
The Radon-Nikodym derivative of $\eta$ with respect to the \emph{arclength} measure on $\partial \rho \DD$ is bounded by:
\begin{equation}
\label{bound:eta_s}
\| \frac{\dd \eta}{\dd s} \|_\infty \leq \Cc \left( \frac{1}{\rho} \|\eta\| + \rho^{-2} \frac{\nuloc(\rho')}{\rho'} + \rho^{-2} \|\partial^-_n \RR\|_{\infty, \partial \rho \DD} + |1-\tau| \rho^{-1} \right),
\end{equation}
\cor{where $\|\eta\|$ is the mass of $\eta$.}
\end{proposition}

\cor{Observe that $\kappa$ and $\| \frac{\dd \eta}{\dd s} \|_\infty$ become smaller when the perturbation is smaller i.e. $\tau$ is close to $1$, $\nuloc$ is small, $\RR$ is small...}

The next two subsections are devoted to the proofs of Propositions \ref{prop:supportmuW} and \ref{prop:singcomp}. We follow the proofs of the corresponding results in \cite{bauerschmidt2017local} while emphasizing the required modifications.

\subsubsection{Proof of Proposition \ref{prop:supportmuW}}
\textbf{\cor{A. Importing an explicit construction.}}
\cor{We want to prove that the support of $\muW$ contains a certain disk $D := \{z \in \rho \DD, \dist(z, \partial \rho \DD) > \kappa \}$. To do so, one can rely on the connection with the “obstacle problem” recalled in Proposition \ref{prop:frostman} - in view of \eqref{uVhmuV}, which here reads:
\begin{equation*}
u_\WW(z) = c_\WW - \mathfrak{h}^{\muW}(z) \text{ on $\R^2$}, 
\end{equation*}
and applying the Laplacian on both sides, we see that in order to prove that $D$ belongs to $\supp \muW$ i.e. that $\muW$ has a positive density on $D$, it is enough to show that $\Delta u_\WW \neq 0$ on $D$. The strategy is then to compute $u_\WW$ explicitly (at least on $D$) and to show that $u_\WW = \hal \WW$ on $D$, which concludes because $\Delta \WW = \Delta \VV > 0$ on $D$.}

\cor{The computation of $u_\WW$ is based on its pointwise characterization \eqref{def:uV} as a supremum of values of a certain class of sub-harmonic functions. The key ingredient of the proof is thus the explicit construction of certain sub-harmonic functions.}   

To quote from \cite[Proof of Prop 3.3]{bauerschmidt2017local}: “let $D := \{ z \in \rho \DD, \dist(z, \partial \rho \DD) \geq \kappa \}$, we show that $D$ [is included in the support of $\muW$] by exhibiting for every $z_0 \in D$ a function $v = v_{z_0}$ that satisfies $v(z_0) = \hal \WW(z_0)$ and:
\begin{equation}
\label{goalvz}
\begin{cases}
\text{$v$ is sub-harmonic on $\R^2$}, \\
\text{$v \leq \hal \WW$ on $\R^2$}, \\
\lim_{|z| \to \infty} \left( v(z) - \log|z| \right) < \infty.
\end{cases}
\end{equation}
Thus we have $u_\WW = \hal \WW$ in $D$ (cf. \eqref{def:uV}) and since the perturbation $\WW - \VV$ is harmonic on $\rho \DD \supset D$ we have\footnote{This follows from the last statement in Proposition \ref{prop:frostman}.} (as densities) $\muW = \muV > 0$ on $D$ and thus $D$ is contained in the support of $\muW$.” 

It remains to construct such a function $v$. \cor{The following lemma from \cite{bauerschmidt2017local} provides the crucial tool for it}. For any $r > 0$, let $l_r$ be the logarithmic potential generated by the uniform probability measure spread on $\DD(0, r)$.
\begin{lemma}[An explicit construction]
\label{lem:explicitcons}
\label{lem:sigmalr}
Let $z_0, w$ be two points in $\R^2$ and let $r$ such that $r \leq \hal |z_0 - w|$. Then there exists $\tilde{z} \in \R^2$ and $k \in \R$ (both depending on $z_0$ and $w$) such that:
\begin{equation}
\label{eq:sigmalr} \hal \left(l_r(z_0 - \tilde{z}) + k \right) = - \hal \log |z_0 - w|, \quad \hal \left(l_r(z - \tilde{z}) + k\right) \leq - \hal \log |z - w| \text{ for } z \in \R^2.
\end{equation}
Moreover the point $\tilde{z}$ can be found on the line segment between $z_0$ and $w$ with:
\begin{equation}
\label{disttzz0}
|\tilde{z} - z_0| = \frac{r^2}{|z_0 - w|}.
\end{equation}
\end{lemma}
\begin{proof}[Proof of Lemma \ref{lem:explicitcons}]
This is \cite[Lemma 3.6]{bauerschmidt2017local} with three minor differences:
\begin{itemize}
	\item First of all, in their statement they assume that $r \leq \hal |z_0 - w|$ \emph{and} that $r \in (0, 1)$ (which is always the case in their setting), whereas we take $r$ arbitrary (still with the condition $r \leq \hal |z_0 - w|$). There is in fact no additional generality in our statement: the case $r \in (0,1)$ extends to the general case by scaling.
	\item Secondly, they allow for any value $\sigma \geq \hal$ of a certain parameter $\sigma$, but we will only need $\sigma = \hal$ (so “$\sigma$” does not appear here). 
	\item   Finally, the estimate \eqref{disttzz0} is not written down in \cite{bauerschmidt2017local}, however it is a straightforward consequence of the fact that $|\tilde{z} - z_0| \leq r \leq  \hal |z_0 - w|$ (which is given by their statement and assumption) and the first equality in \cite[(3.14)]{bauerschmidt2017local} (with $\sigma = \hal$).
\end{itemize}
\end{proof}

\textbf{\cor{B. The proof of Proposition \ref{prop:supportmuW} }}
Next, we follow \cite[Proof of Prop 3.3]{bauerschmidt2017local}, with a slight adaptation, to prove our Proposition \ref{prop:supportmuW}.
\begin{proof}[Proof of Proposition \ref{prop:supportmuW}]
Contrarily to \cite{bauerschmidt2017local} we will not scale everything back to the unit circle. 

We fix $z_0$ in the disk $D := \{ z \in \rho \DD, \dist(z, \partial \rho \DD) \geq \kappa \}$, with $\kappa$ satisfying the three conditions of \eqref{condi:kappa}, and we seek to construct a function $v$ satisfying \eqref{goalvz}. 

\medskip

\textbf{\cor{Step 1. Introducing some objects.}} The function $v$ will be “made” with the following quantities\footnote{In \cite{bauerschmidt2017local} there is an additional free parameter $\sigma$ but we can always take $\sigma = \hal$, thus $\sigma$ does not appear here.}:
\begin{itemize}
	\item Let $\uV$ be the solution to the obstacle problem for the reference potential $\VV$ as in \eqref{def:uV}
	\item Let $\tilde{\RR}$ be the harmonic extension of $\RR$ outside $\rho \DD$ (as in \cite[(3.17)]{bauerschmidt2017local})
	\item Let $z \mapsto G(z) : = \max(0, \log |\rho^{-1} z|)$.
	\item Let $\gamma := 2 \|\partial_n^- \RR\|_{\infty, \partial \rho \DD}$ as defined in \eqref{def:jumpR}.
\end{itemize}
There remains to define one last object, the function $\tilde{z}$, for which we differ slightly from \cite[(3.18)]{bauerschmidt2017local} as we explain next.

\medskip 

\textbf{\cor{Step 2. Definition of $\tilde{z}$, $k$ and $L$.}}
For each $w$ in the support of $\nu$, we apply Lemma \ref{lem:explicitcons} to $z_0, w$ as above and (which is new compared to \cite{bauerschmidt2017local}) we choose the parameter $r$ depending on $w$ as follows:
\begin{equation*}
r = r(w) = \frac{1}{10} |z_0 - w|,
\end{equation*}
which is valid choice as it is obviously smaller than $\hal |z_0 - w|$. We obtain a point\footnote{$\tilde{z}$ depends also on $z_0$ but we will not write down this dependency as we work for any fixed $z_0$.} $\tilde{z} = \tilde{z}(w)$ and a real number $k = k(w)$ such that \eqref{eq:sigmalr} is satisfied. 

We now make the following important observation: we know by Lemma \ref{lem:sigmalr} that $\tilde{z}$ lies on the line segment between $z$ and $w$, at distance $\frac{1}{100} |z_0 - w|$ from $z_0$. Since by assumption we have on the one hand: 
\begin{equation*}
\dist(z_0, \partial \rho \DD) \geq \kappa \geq 20 \rho',
\end{equation*}
and on the other hand $\supp \nu \subset \DD(0, \rho + \rho') \setminus \DD(0, \rho)$ it is easy to check that:
\begin{equation*}
|\tilde{z} - z_0| \leq \frac{1}{100} \left(\dist(z_0, \partial \rho \DD) + \rho'\right) \leq \frac{1}{4} \dist(z_0, \partial \rho \DD),
\end{equation*}
and thus the point $\tilde{z}$ remains in $\rho \DD$, in fact \emph{the entire disk $\DD\left(\tilde{z}, r(w)\right)$ is contained in $\rho \DD$}.

This being done for all $w \in \supp \nu$ we may define a map $L : \R^2 \to \R$ by:
\begin{equation}
\label{eq:defLz}
L(z) :=  \hal \int_{\supp \nu} \left(l_{r(w)}(z - \tilde{z}(w)) + k(w) \right) \nu(\dd w).
\end{equation}

\medskip 

\textbf{\cor{Step 3. Definition of $v$.}} 
Finally, as in \cite[(3.18)]{bauerschmidt2017local} we form the map $v = v_{z_0}$ by setting:
\begin{equation}
\label{eq:defvz0}
v : z \mapsto \tau \uV(z) + \rho^{-2} \left( L(z) + \tilde{\RR}(z) \right) + \rho^{-1} \gamma G(z),
\end{equation}
and we claim that $v$ satisfies \eqref{goalvz}.

\medskip

\textbf{\cor{Step 4. Checking sub-harmonicity.}}
The distributional Laplacian \cor{of $v$} is given by:
\begin{equation}
\label{distrib_laplacian}
\Delta v = \tau \Delta \uV + \rho^{-2} \Delta L  + \rho^{-2} (2 \partial_n^- \RR + \gamma ) \dd s,
\end{equation}
the measure $\dd s$ being the arclength measure on $\partial \rho \DD$. We used the general formula of \cite[(3.10)]{bauerschmidt2017local} \cor{to compute the Laplacian of $\tilde{\RR}$ and $G$: these functions are, by definition, harmonic outside $\partial \rho \DD$, so their Laplacian lives on $\partial \rho \DD$ and is given by the jump of their normal derivatives along $\partial \rho \DD$. For $\tilde{\RR}$ this is simply $2 \partial_n^- \RR$, and for $G$ this is $\rho^{-1}$.} 

In the interior of $\DD$ we can consider only the first two terms in the right-hand side of \eqref{distrib_laplacian}, use the last item of Proposition \ref{prop:frostman} to evaluate $\Delta \uV = \hal \frac{1}{\rho^2}$, and an explicit computation of $\Delta l_r$ - we recall that $l_r$ is the logarithmic potential generated by the uniform probability measure on the disk of radius $r$) - to get:
\begin{equation*}
\Delta v = \frac{\tau}{2 \rho^2} - \rho^{-2} \left( \int_{\supp \nu} \frac{1}{r(w)^2} \ind_{\DD(\tilde{z}(w), r(w))} \ \nu(\dd w) \right)
\end{equation*} 
(The differences with the corresponding identity in \cite{bauerschmidt2017local} is that here $\alpha$ (their notation) is of order $\rho^{-2}$, $\sigma$ (their notation) is always $\hal$, that $L$ here is scaled by $\rho^{-2}$ and most importantly that the distance $r$ is here chosen \emph{depending on $w$} as above.) 

In order to check that $\Delta v \geq 0$ (which is a requirement of \eqref{goalvz}) we thus need to guarantee that: 
\begin{equation}
\label{rw2leq12}
\int_{\supp \nu} \frac{1}{r(w)^2} \ind_{\DD(\tilde{z}(w), r(w))} \ \nu(\dd w) \leq \frac{\tau}{2}, \text{ with $\frac{\tau}{2} \geq \frac{9}{20}$ because $|\tau-1| \leq \frac{1}{10}$ by assumption.}
\end{equation}
Here we stop following the route of \cite{bauerschmidt2017local} and instead recall that by construction $r(w) = \frac{1}{10} |z_0 - w|$, thus we may write (bounding the indicator function by $1$):
\begin{equation*}
\int_{\supp \nu} \frac{1}{r(w)^2} \ind_{\DD(\tilde{z}(w), r(w))} \ \nu(\dd w) \leq 100 \int_{\supp \nu} \frac{1}{|w - z_0|^2} \nu(\dd w).
\end{equation*}
Cover the annulus $\DD(0, \rho + \rho') \setminus \DD(0, \rho)$, which by assumption contains the support of $\nu$, by $\O(\rho/\rho')$ squares of sidelength $\rho'$. The quantity $\nuloc$ allows us to bound the mass of $\nu$ on each such square. Then for $z \in \rho \DD$, if $z$ is at distance at least $20 \rho'$ from the boundary then we can compare the integral $\int \frac{1}{|z-w|^2}\  \nu(\dd w)$ to
\begin{equation*}
 \frac{\nuloc(\rho')}{\rho'} \times \int_{w \in \partial \DD} \frac{1}{|z-w|^2} \dd w,
\end{equation*}
which is $\O\left(\nuloc(\rho') \times \frac{1}{\rho' \dist(z, \partial \DD)}\right)$ and thus smaller than $\Cc \frac{\nuloc(\rho')}{\kappa \rho' }$ if $z$ is at distance at least $\kappa$ from the boundary. In particular, if $\kappa$ is larger than some constant times $\frac{\nuloc(\rho')}{\rho'}$ we do have \eqref{rw2leq12}. 

On the other hand, for the singular part of $\Delta v$ living on the boundary of $\rho \DD$, we need to ensure that $2 \partial_n^- \RR + \gamma \geq 0$, which is true by definition of $\gamma$.

\medskip 

\textbf{\cor{Step 5. Checking the obstacle property.}} We want to check that: 
\begin{equation*}
v(z_0) = \hal \WW(z_0), \quad v(z) \leq \hal \WW(z) \text{ on $\R^2$}
\end{equation*} 
(in fact it is enough to check the last inequality on the disk $\rho \DD$ as $\WW$ is infinite outside). 

This step goes on exactly like in the original proof: we already know that $\uV(z) = \hal \VV(z)$ on $\rho \DD$ (because $\uV$ solves the “reference” obstacle problem), that $\tilde{\RR}(z) = \RR(z)$ on $\rho \DD$ (by definition), and that $G$ vanishes on $\rho \DD$. Thus in view of \eqref{def:W} it suffices to have:
\begin{equation*}
L(z_0) = \hnu(z_0), \quad L(z) \leq \hnu(z) \text{ for $z \neq z_0$}, 
\end{equation*}
which is precisely guaranteed by construction of $L$ using Lemma \ref{lem:explicitcons}, see \eqref{eq:sigmalr} and \eqref{eq:defLz}.

\medskip

\textbf{\cor{Step 6. Checking the growth at infinity}.} \cor{It remains to check the last property required in \eqref{def:uV} or \eqref{goalvz}, namely the behavior of $v - \log$ at infinity.} Here we proceed again a bit differently than in \cite{bauerschmidt2017local}. 

Recall that we have chosen our function as:
\begin{equation*}
\tau \uV(z) + \rho^{-2} \left( L(z) + \tilde{\RR}(z) \right) + \rho^{-1} \gamma G(z).
\end{equation*}
As $|z| \to \infty$ we have $\uV(z) \sim \log |z|$, $L(z) \sim - \|\nu\| \log |z|$, $\tilde{\RR}(z) = \O(1)$ and $G(z) \sim \log z$. Thus the third condition of \eqref{goalvz} is satisfied if we have: 
\begin{equation*}
\|\nu\|  \geq \rho^2 (1-\tau) + \rho \gamma.
\end{equation*} 
We argue that we may always assume that this is the case. 

Indeed, if not, then we may distribute\footnote{This is inspired by a different but similar argument in the original proof of \cite{bauerschmidt2017local}.} positive charges uniformly with density $\gamma + \rho |1-\tau|$ all around the circle $\partial \rho \DD$ on an annulus of width~$1$, and redefine the perturbation $\RR$ accordingly in order for the total potential to remain the same. This operation changes $\nu$: it increases its total mass to a convenient level and increases $\nuloc(1)$ by $\gamma + \rho |1-\tau|$. In view of the “Sub-harmonicity” paragraph above, this is harmless as long as we assume that $\kappa$ larger than some constant times $\gamma + \rho |1-\tau|$ in \eqref{condi:kappa}.

On the other hand, this operation does not affect the quantity  $\|\partial_n^- \RR\|_{\infty, \partial \rho \DD}$ as defined in \eqref{def:jumpR}, because placing a radially symmetric density of charges outside $\rho \DD$ creates a \emph{constant} potential within $\rho \DD$.

Hence without loss of generality we may assume that $\|\nu\| \geq \rho^2 (1-\tau) + \rho \gamma$. 

\medskip
This concludes the proof of Proposition \ref{prop:supportmuW}.
\end{proof}

\subsubsection{Proof of Proposition \ref{prop:singcomp}}
\begin{proof}[Proof of Proposition \ref{prop:singcomp}]
\cor{The goal is now to bound the Radon-Nykodym derivative $\frac{\dd \eta}{\dd s}$, where $\eta$ is the singular component of $\muW$ and $s$ is the arclength measure on $\rho \DD$.} This time the changes are minor compared to \cite[Proof of Prop. 3.4]{bauerschmidt2017local}. 

\cor{We follow their computations until the moment when they argue that $\frac{\dd \eta}{\dd s}$ can be bounded pointwise by\footnote{\cor{In view of \eqref{def:W}, our $\rho^{-2} \RR$ is their $R$ and our $\rho^{-2} \nu$ is their $\nu$.}} (see \cite[(3.25)]{bauerschmidt2017local}):
\begin{equation}
\label{aubord}
\frac{\dd \eta}{\dd s}(z) \leq - \frac{\tau}{2} \partial_n^{-} \VV(z) + \rho^{-2} \int \frac{z-w}{|z-w|^2} \cdot \vec{n} \ \nu(\dd w) - \rho^{-2} \partial_n^{-} \RR(z) + \int \frac{z-w}{|z-w|^2} \cdot \vec{n} \left(\mu + \eta\right)(\dd w).
\end{equation}
We can then write (as they do) that: 
\begin{equation*}
\int \frac{z-w}{|z-w|^2} \cdot \vec{n}\ \mu(\dd w)  = \hal \partial_n^{-} \VV(z), \quad \int \frac{z-w}{|z-w|^2} \cdot \vec{n}\ \eta(\dd w) \leq \frac{1}{\rho} \|\eta\|,
\end{equation*} 
and we bound $\partial_n^{-} \VV$ by $\rho^{-1}$ (see \eqref{def:VV}), which in view of \eqref{aubord} leads to the pointwise bound:
\begin{equation}
\label{PreDiff} \frac{\dd \eta}{\dd s}(z) \leq \frac{1-\tau}{2} \rho^{-1} + \rho^{-2} \int \frac{z-w}{|z-w|^2} \cdot \vec{n} \  \nu(\dd w) - \rho^{-2} \partial_n^{-} \RR(z) + \frac{1}{\rho} \|\eta\|.
\end{equation}
}

\cor{The only real difference is that we bound $\int \frac{z-w}{|z-w|^2} \cdot \vec{n} \ \nu(\dd w)$ in a different way.} Here $z$ is some fixed point on the circle and $\vec{n}$ is the normal vector to the circle at $z$. We distinguish between $|w -z| \leq 10 \rho'$ and $|w -z| \geq 10 \rho'$. 
\begin{itemize}
	\item For the first contribution, we note\footnote{As in \cite[Proof of Prop. 3.4]{bauerschmidt2017local}, except that they scale everything back to $\rho = 1$.} that $\frac{z-w}{|z-w|^2} \cdot \vec{n}$ is always smaller than $\frac{1}{2\rho}$ and thus:
\begin{equation*}
\int_{|w-z| \leq 10 \rho'} \frac{z-w}{|z-w|^2} \cdot \vec{n} \  \nu(\dd w) \leq \frac{\Cc \nuloc(\rho')}{2\rho}
\end{equation*}
\item We bound the rest of the integral by $\int_{|w-z| \geq 10 \rho'} \max\left(0, \frac{z-w}{|z-w|^2} \cdot \vec{n} \right) \nu(\dd w)$, which we may compare to:
\begin{equation*}
\frac{\nuloc(\rho')}{\rho'} \int_{w \in \partial \DD \rho, |w-z| \geq 10 \rho'} \max\left(0, \frac{z-w}{|z-w|^2} \cdot \vec{n} \right) \dd w, 
\end{equation*}
which is itself smaller than:
\begin{equation*}
\frac{\nuloc(\rho')}{\rho'} \int_{w \in \partial \DD \rho} \max\left(0, \frac{z-w}{|z-w|^2} \cdot \vec{n} \right) \dd w = \frac{\nuloc(\rho')}{\rho'} \times \O(1)
\end{equation*}
by one-dimensional scale-invariance.
\end{itemize}
The dominant contribution is the second one, so in summary:
\begin{equation}
\label{differencefromBBNY}
\int \frac{z-w}{|z-w|^2} \cdot \vec{n} \ \nu(\dd w) \preceq \frac{\nuloc(\rho')}{\rho'}.
\end{equation}

\cor{Inserting \eqref{differencefromBBNY} into \eqref{PreDiff} yields \eqref{bound:eta_s}.}
\end{proof}

\subsection{Application to a generalized 2DOCP with good external potential: \cor{proof of Proposition \ref{prop:correspondence}}}
\label{sec:applications_to_generalized}
\newcommand{\tmu}{\tilde{\mu}}
\renewcommand{\gg}{\mathsf{g}}
\newcommand{\Pp}{\mathrm{P}}
\renewcommand{\VV}{\mathrm{V}^\mathrm{ref}_\La}
\renewcommand{\muV}{\mu_{\VV}}
\renewcommand{\IV}{\mathcal{I}_{\mu_{\VV}}}
We now apply the knowledge gained in Section \ref{sec:effect-perturb-eq} to justify \cor{show \eqref{identity_entre_points_de_vue}}, i.e. to pass from a 2DOCP with good external potential to a 2DOCP with well-controlled non-uniform neutralizing background (see Sections \ref{sec:good_good} and \ref{sec:PnLV}).

Let $\La$ be a disk of radius $\rho = T$ and $\nn$ be the number of points in $\La$. Let $\VextL$ be a “good external potential” on $\La$. Consider the potential $\WW$ defined as:
\begin{equation}
\label{rappelWW}
\WW := \tau \VV + \frac{2}{\nn} \VextL, \quad \tau := \frac{|\La|}{\nn}, 
\end{equation}
where $\VV$ is as in \eqref{def:VV} (with $\rho = T$) and let $\muW$ be the associated equilibrium measure, which we write (with the notation of Section \ref{sec:HarmonicPerturbations}) as:
\begin{equation}
\label{rappelmuW}
\muW = \tau \muV + \left( \eta - \tau \ind_B \muV \right).
\end{equation}
We will use repeatedly below the assumption that (see Definition \ref{def:goodSS} of a good sub-system):
\begin{equation}
\label{taumoinsun}
\tau = 1 + \O\left(\frac{\log^2T}{T}\right).
\end{equation}

\subsubsection*{A. Some properties of the new equilibrium measure}
\begin{claim}
\label{claim:Effect_KDeta}
For some constant $\Cc_\mu$ depending on the “good external potential” $\bCc$ constant of $\VextL$, the following holds:

The measure $\muW - \tau \muV$ has mass $1 - \tau = 1 - \frac{|\La|}{\nn}$, and is made of:
\begin{enumerate}
    \item A continuous (negative) part of density $- \frac{\tau}{|\La|}$ supported on a subset $B$ located at distance $\leq \kappa$ from $\partial \rho \DD$, with $\kappa$ as in \eqref{condi:kappa}. We can take: 
    \begin{equation}
    \label{choix_kappa}
    \kappa = \Cc_\mu  \times \log^2  T.
    \end{equation}
     \item A singular component $\eta$ living on $\partial \La$, with a density $\frac{\dd \eta}{\dd s}$ bounded as in \eqref{bound:eta_s}. We have:
    \begin{equation}
    \label{bound_deta_ds}
        \| \frac{\dd \eta}{\dd s} \| \leq \Cc_\mu \frac{\log^2 T}{T^2}.
    \end{equation}
 The total mass of the continuous part, of the singular part, and the total variation of $\muW - \tau \muV$ is bounded by $\Cc_\mu \frac{\log^2 T}{T}$.
 Moreover, we have:
\begin{equation}
\label{interaction_diff_measure}
\nn^2 \iint -\log|x-y| \dd \left(\eta - \tau \muV \ind_B \right)(x) \dd \left(\eta - \tau \muV \ind_B\right)(y) = \O\left(\Cc_\mu^2 T^2 \log^5(T)\right).
\end{equation}
\end{enumerate}
\end{claim}
\begin{proof}[Proof of Claim \ref{claim:Effect_KDeta}]
Inserting Definition \ref{def:GoofPotential} of a “good external potential” into the bounds \eqref{condi:kappa} and \eqref{bound:eta_s} with $\rho' = \hat{T} = \log T$, and using \eqref{taumoinsun} we obtain \eqref{choix_kappa} and \eqref{bound_deta_ds}.

The total mass of the continuous part is of order $T \times \Cc_\mu  \log^2 T \times T^{-2} = \frac{\log^2 T}{T}$ (it is contained in a region at distance $\leq \Cc_\mu \log^2 T$ from the boundary of the disk of radius $T$, and its density is of order $T^{-2}$).

The mass of the singular part matches the mass of the continuous part up to an error $1-\tau$, but \eqref{taumoinsun} holds. Hence the mass of the singular part is also bounded by $\Cc_\mu \frac{\log^2 T}{T}$, and so is the total variation of $\muW - \tau \muV$.

Scaling everything back to a disk of radius~$1$, and bounding the self-interaction using \eqref{choix_kappa}, \eqref{bound_deta_ds}, we get \eqref{interaction_diff_measure}.
\end{proof}

\subsubsection*{B. Inserting the equilibrium measure into the energy: proof of Proposition \ref{prop:correspondence}}
Let $\zetaW$ be given by:
\begin{equation}
\label{def:zetaW}
\zetaW = \h^{\muW} + \hal \WW - \IW(\muW) + \hal \int \WW \dd \muW.
\end{equation}
It is a standard fact of logarithmic potential theory that the function $\zetaW$ vanishes on the support of $\muW$ and is non-negative outside of it, and we refer to $\zetaW$ as the “effective confining potential” (see \cite[Definition 2.18]{serfaty2015coulomb}).

\begin{proof}[Proof of Proposition \ref{prop:correspondence}]
\cor{Using the notation of \eqref{def:PnLV} and \eqref{def:PnLV2}, our goal is to prove that:
\begin{equation*}
\mathbb{P}_{\nn, \La, \VextL}^\beta( \cdot ) = \mathbb{P}_{\nn, \La}^\beta\left(\cdot , \nn \muW, \zetaW \right),
\end{equation*}
and thus that \eqref{identity_entre_points_de_vue} is justified.}

\medskip

\textbf{\cor{Step 1. An algebraic identity.}} We start with the following “splitting formula”.
\begin{claim}
\label{claim:algebraic_identity1}
Let $\mu_\nn := \frac{1}{\nn} \sum_{i=1}^\nn \delta_{x_i} = \frac{1}{\nn} \bX$ be the empirical measure associated to a $\nn$-tuple in $\La$. We have:
\begin{multline*}
\FL(\bX) + \int_\La \VextL(x) \dd \bX(x) =  \hal \nn^2 \left(\iint_{x \neq y} -\log|x-y| \dd \mu_\nn(x) \dd \mu_\nn(y)  + \int \left(\frac{|\La|}{\nn} \VV(x) + \frac{2}{\nn} \VextL(x) \right) \dd \mu_\nn(x)  \right) 
 \\
- \hal \left( \left(2 \nn |\La| - |\La|^2\right) \IV(\muV) -  |\La| \left(\nn - |\La| \right) \int \VV \dd \muV \right).
\end{multline*}
\end{claim}
\begin{proof}
The proof is elementary but it does require some care because of the possible non-neutrality and the various scalings involved. We start with expanding the definition \eqref{def:FL} of $\FL(\bX)$ as:
\begin{equation*}
2 \FL(\bX) = \iint_{x \neq y} -\log|x-y| \dd \bX(x) \dd \bX(y) + \iint_{\La \times \La} -\log|x-y| \dd x \dd y - 2 \iint_{\La \times \La} -\log|x-y| \dd \bX(x) \dd y.
\end{equation*}
Introducing the empirical measure $\mu_\nn$ and the reference probability measure $\muV$ (uniform on $\La$) we get:
\begin{multline}
\label{FNusingMun}
2 \FL(\bX) = \nn^2 \iint_{x \neq y} -\log|x-y| \dd \mu_\nn(x) \dd \mu_\nn(y) + |\La|^2 \iint_{\La \times \La} -\log|x-y| \dd \muV(x) \dd \muV(y) \\
 - 2 \nn |\La| \iint_{\La \times \La} -\log|x-y| \dd \mu_\nn(x) \dd \muV (y).
\end{multline}
Let us recall the following standard identity valid on $\La$ (see \eqref{eq:EulerLagrange} or e.g. \cite[Thm. 2.1]{serfaty2015coulomb}):
\begin{equation}
\label{hmuVV}
\int_{\La} -\log|x-y| \dd \muV(y) = - \frac{1}{2} \VV(x) + \left( \IV(\muV)- \hal \int \VV \dd \muV \right),
\end{equation}
we obtain after some computations: 
\begin{multline}
\label{eq:FL_as_munn}
2 \FL(\bX) = \nn^2 \left( \iint_{x \neq y} -\log|x-y| \dd \mu_\nn(x) \dd \mu_\nn(y) + \frac{|\La|}{\nn} \int \VV \dd \mu_\nn  \right) \\
- \left(2 \nn |\La| - |\La|^2\right) \IV(\muV)+ |\La| \left(\nn - |\La| \right) \int \VV \dd \muV.
\end{multline}
Dividing by $2$ and inserting the contribution of $\VextL$, we obtain the claim.
\end{proof}

\medskip

\textbf{\cor{Step 2. Another use of the splitting formula.}}
\begin{claim}
\label{claim:introduce}
We have:
\begin{equation*}
\FL(\bX) + \int_{\La} \Vext \dd \bX = \FL(\bX, \nn \muW) + \int \zetaW \dd \bX + \mathrm{Const.}(\nn, \La, \VextL)
\end{equation*}
\end{claim}
\begin{proof}[Proof of Claim \ref{claim:introduce}]
On the one hand, we know by Claim \ref{claim:algebraic_identity1} that:
\begin{equation*}
\FL(\bX) + \int_{\La} \Vext(x) \dd \bX(x) = \hal \nn^2 \left( \iint_{x \neq y} -\log|x-y| \dd \mu_\nn(x) \dd \mu_\nn(y) + \int \WW \dd \mu_\nn  \right) + \mathrm{Const.}(\nn, \La).
\end{equation*}
On the other hand, using the “splitting formula” of Sandier-Serfaty (see \cite[Lemma 2.1]{MR3353821}) we have:
\begin{multline}
\label{use_splitting}
\hal \nn^2 \left( \iint_{x \neq y} -\log|x-y| \dd \mu_\nn(x) \dd \mu_\nn(y) + \int \WW \dd \mu_\nn  \right) \\
= \hal \nn^2 \IW(\muW) + \FL(\bX, \nn \muW) + \nn \int \zetaW(x) \dd \bX(x).
\end{multline}
Discarding quantities that do not depend on the configuration $\bX$, we obtain the claim.
\end{proof}

\textbf{Conclusion.} We may thus equivalently use $\FL(\bX, \nn \muW) + \int \zetaW \dd \bX$ instead of $\FL(\bX) + \int_{\La} \Vext  \dd \bX$ in the Boltzmann's factor, up to some constant \cor{(with respect to $\bX$)} that gets absorbed in the partition function.
\end{proof}

\subsection{\cor{Global law for good sub-systems}: proof of Proposition \ref{prop:global_law_SUBSYS}}
\label{sec:proofglobal}

\begin{proof}[Proof of Proposition \ref{prop:global_law_SUBSYS}]
By definition of $\PnLL$ we have:
\begin{multline}
\label{calculGL}
\log \EnLL \left[  \exp\left( \frac{\beta}{2} \FL(\bX, \nn \muW) \right) \Big| \EE_\La \right] 
= \log \int_{\La^\nn} \ind_{\EE_\La}(\bX) e^{- \frac{\beta}{2} \left( \FL(\bX, \nn \muW) + 2 \int \zetaW(x) \dd \bX(x) \right)  } \dd \X_\nn \\
- \log \int_{\La^\nn} \ind_{\EE_\La}(\bX) e^{- \beta \left( \FL(\bX, \nn \muW) + \int \zetaW(x) \dd \bX(x) \right) } \dd \X_\nn.
\end{multline}

\cor{We study each term in the right-hand side of \eqref{calculGL} separately. First we will prove a lower bound on the energy appearing in the exponent, which can be turned into an \emph{upper bound} on the first term. Next, we give a lower bound on the “partition function” (the second integral in the right-hand side) which yields an \emph{upper bound} on the second term.}

\subsubsection*{1. Configuration-wise lower bound on the energy.}
Since $\zetaW$ is non-negative, we have: 
\begin{equation*}
\FL(\bX, \nn \muW) + 2 \int \zetaW \dd \bX \geq \FL(\bX, \nn \muW).
\end{equation*}
Because of the possible singularity of $\muW$ along the boundary, we cannot directly use results like \cite[Lemma B.2.]{armstrong2019local} to bound $\FL$ from below. Instead, for each $i = 1, \dots, \nn$ we let $\eta(x_i) := \frac{1}{4} \min\left(1, \dist(x_i, \partial \La) \right)$ and use this as a “truncation vector”. By construction, the disks $\DD(x_i, \eta(x_i))$ do not intersect the support of the singular part of $\muW$, we may thus use the monotonicity property of \cite[Lemma B.1]{armstrong2019local} (or Onsager's lemma) and write that:
\begin{equation*}
\FL(\bX, \nn \muW) \geq \frac{1}{2} \left( \frac{1}{2\pi} \int_{\R^2} |\nabla \h^{\bX, \nn \muW}_{\veta}|^2 + \sum_{i=1}^\nn \log \eta(x_i) \right) - \O(\nn).
\end{equation*}
The integral is obviously non-negative, and it remains to find a lower bound for the negative contribution coming from $\sum_{i=1}^\nn \log \eta(x_i)$. This is where we use our restriction to the event $\EE_\La$. Since $\bX$ is assumed to satisfy the conditions of Definition \ref{def:EELA} we know that:
\begin{enumerate}
    \item There is no index $i$ such that $\eta(x_i) \leq e^{-\log^2 T}$.
    \item There are at most $T \log T$ indices $i$ such that $\eta(x_i) \leq \frac{1}{4}$.
\end{enumerate}
We thus have the rough lower bound: $\sum_{i=1}^\nn \log \eta(x_i) \geq - \nn \log 4 - T \log T \times \log^2 T$. Hence for $\bX \in \EE_\La$, we have the configuration-wise lower bound: $\FL(\bX, \nn \muW) \geq - \O\left(T^2\right)$, which we can integrate in order to obtain (with an implicit constant depending only on $\beta$):
\begin{equation}
\label{eq:GlobalLawA}
\log \int_{\La^\nn} \ind_{\EE_\La}(\bX) e^{- \frac{\beta}{2} \left( \FL(\bX, \nn \muW) + 2 \nn \int \zetaW \dd \bX \right)  } \dd \X_\nn \leq \log |\EE_\La| + \O\left(T^2\right),
\end{equation}
where $|\EE_\La|$ denotes the volume of the event $\EE_\La$ under the Lebesgue measure $\dd \X_\nn$ on $\La^\nn$.

\subsubsection*{2. Lower bound on the partition function.}
To find an upper bound on the second term in the right-hand side of \eqref{calculGL}, we rely on a “Jensen's trick” inspired by \cite{garcia2019large} and write:
\begin{multline}
\label{with_restriction}
- \log \int_{\La^\nn} \ind_{\EE_\La}(\bX) e^{ - \beta \left( \FL(\bX, \nn \muW) + \nn \int \zetaW(x) \dd \bX(x) \right)  } \dd \X_\nn 
\\
\leq - \log |\EE_\La| + \beta \int_{\La^\nn} \ind_{\EE_\La}(\bX) \left(\FL(\bX, \nn \muW) + \nn \int \zetaW(x) \dd \bX(x) \right) \frac{\dd \X_\nn}{|\EE_\La|}.
\end{multline}

Let us start by doing computations without the indicator $\ind_{\EE_\La}(\bX)$.
\begin{claim}
\label{claim:sansEELA}
We have:
\begin{equation}
\label{eq:without_the_restriction}
\int_{\La^\nn} \left(\FL(\bX, \nn \muW) + \nn \int \zetaW(x) \dd \bX(x) \right) \frac{\dd  \X_\nn}{|\La^\nn|} \leq \O(T^2 \log^5 T).
\end{equation}
\end{claim}
\begin{proof}[Proof of Claim \ref{claim:sansEELA}]
Let us recall that the integrand can be written (see \eqref{use_splitting}) as:
\begin{multline}
\label{rewriteintegrand}
\FL(\bX, \nn \muW) + \nn \int \zetaW(x) \dd \bX(x) \\
= \hal \nn^2 \left( \iint_{x \neq y} -\log|x-y| \dd \mu_\nn(x) \dd \mu_\nn(y) + \int \WW \dd \mu_\nn  \right) - \hal \nn^2 \IW(\muW). 
\end{multline}
By elementary computations, we obtain:
\begin{equation*}
\hal \nn^2 \int_{\La^\nn} \left( \iint_{x \neq y} -\log|x-y| \dd \mu_\nn(x) \dd \mu_\nn(y) \right) \frac{\dd \X_\nn}{|\La|^\nn}  = \hal \nn (\nn-1) \iint -\log|x-y| \dd \muV(x) \dd \muV(y)
\end{equation*}
and similarly:
\begin{equation*}
\hal \nn^2 \int_{\La^\nn} \left( \int \WW \dd \mu_\nn  \right)  \frac{\dd \X_\nn}{|\La|^\nn}  = \hal \nn^2 \int \WW(x) \dd \muV(x),
\end{equation*}
and thus after some simplifications we get:
\begin{multline*}
\int_{\La^\nn} \left(\FL(\bX, \nn \muW) + \nn \int \zetaW(x) \dd \bX(x) \right) \frac{\dd \X_\nn}{|\La|^\nn} 
\\
= \hal \nn(\nn-1) \iint -\log|x-y| \dd \muV(x) \dd \muV(y) + \hal n^2 \int \WW \dd \muV - \hal n^2 \IW(\muW). 
\end{multline*}
We now insert the expressions \eqref{rappelWW}, \eqref{rappelmuW} for $\WW$ and $\muW$, the definition \eqref{def:IV} of $\IW(\muW)$, expand and use the identity \eqref{hmuVV}. We obtain:
\begin{multline}
\label{after_simplification_lower_bound}
\hal \nn(\nn-1) \iint -\log|x-y| \dd \muV(x) \dd \muV(y) + \hal n^2 \int \WW \dd \muV - \hal n^2 \IW(\muW) = \\
\nn^2 \iint -\log|x-y| \dd \left(\eta - \tau \muV \ind_B \right)(x) \dd \left(\eta - \tau \muV \ind_B\right)(y)
- \hal \nn \iint -\log|x-y| \dd \muV(x) \dd \muV(y) \\
+ \nn \int \VextL(x) \left( \dd \muV - \dd \muW \right)(x).
\end{multline}
Using \eqref{interaction_diff_measure} we control the first term in the right-hand side of \eqref{after_simplification_lower_bound} by $\O(T^2 \log^5(T))$. On the other hand, by a direct estimate, we have $\nn \iint -\log|x-y| \dd \muV(x) \dd \muV(y) = \O(T^2 \log T)$. It remains to bound $\nn \int \VextL \left( \dd \muV - \dd \muW \right)$ from above. Since $\muW, \muV$ have the same mass, it is equivalent to bound:
\begin{equation*}
\nn \int \left(\VextL(x) - \VextL(\omega) \right) \left( \dd \muV - \dd \muW \right)(x),
\end{equation*}
where $\omega$ is the center of $\La$. Let us decompose the integral into two parts:

\medskip

\textbf{Away from the boundary.} On $\{ \dist(z, \partial \La) \geq \kappa \}$ we know that $\muW$ coincides with $\tau \muV$ and thus (since we take $\kappa \geq 1$):
\begin{equation*}
\left| \nn \int_{\{ \dist(z, \partial \La) \geq \kappa \}} \left( \VextL(z) - \VextL(\omega) \right) \left( \dd \muW - \dd \muV \right) \right| \leq \nn \times |1-\tau| \times \sup_{\{\dist(z, \partial \rho \DD) \geq 1\}} |\VextL(z) - \VextL(\omega)|.
\end{equation*}
We know that $|1-\tau| = \O\left(\frac{\log^2 T}{T}\right)$ and by assumption we have $|\VextL(z) - \VextL(\omega)| = \O(T \log^3 T)$ on $\{\dist(z, \partial \rho \DD) \geq 1\}$. Thus the contribution “away from the boundary” is bounded by $\O\left(T^2 \log^5 T\right)$.

\medskip

\textbf{Near the boundary.} On $\{\dist(z, \partial \La) \leq \kappa \}$ we control each contribution separately. On the one hand, we have, using the mean value formula for $\VextL$ (which is harmonic on $\La$):
\begin{equation*}
\nn \int_{\{ \dist(z, \partial \La) \geq \kappa \}} \left( \VextL(z) - \VextL(\omega) \right) \dd \muV(z) = 0.
\end{equation*}
On the other hand, $\muW$ being a non-negative measure, we may write: 
\begin{equation*}
- \nn \int_{\{ \dist(z, \partial \La) \geq \kappa \}} \left(\VextL(z) - \VextL(\omega) \right) \dd \muW(z) \leq  - \nn \int_{\{ \dist(z, \partial \La) \geq \kappa \}} \left(\wVextL(z) - \wVextL(\omega) \right) \dd \muW(z),  
\end{equation*}
where we used the auxiliary function $\wVextL$ from Definition \ref{def:GoofPotential}. The mass of $\muW$ near the boundary is bounded by $\O\left(\frac{\kappa}{T} + T \|\frac{\dd \eta}{\dd s}\|_{\infty}  \right)$. Using the results of Claim \ref{claim:Effect_KDeta} and the definition of a “good external potential” we bound the contribution “near the boundary” by $\O\left(T^2 \log^5 T\right)$ also.

This concludes the proof of the claim. 
\end{proof}

\cor{Note that the integral in the left-hand side of \eqref{eq:without_the_restriction} is not exactly the same as the second integral in the right-hand side of \eqref{calculGL}.} It \cor{thus remains to check} that restricting our integrals to $\EE_\La$ as in \eqref{with_restriction} instead of \eqref{eq:without_the_restriction} has no consequence on the estimate, which is not totally obvious. \cor{Using \eqref{rewriteintegrand}, and in view of \eqref{with_restriction}, we see that it amounts to comparing
\begin{equation}
\label{tocomparing}
\int_{\La^\nn} \ind_{\EE_\La}(\bX)  \left( \hal \nn^2 \left( \iint_{x \neq y} -\log|x-y| \dd \mu_\nn(x) \dd \mu_\nn(y) + \int \WW \dd \mu_\nn  \right) - \hal \nn^2 \IW(\muW) \right) \frac{\dd \X_\nn}{|\EE_\La|}
\end{equation}
namely the case “with restriction”, to our previous expression, namely
\begin{equation*}
\int_{\La^\nn}  \left( \hal \nn^2 \left( \iint_{x \neq y} -\log|x-y| \dd \mu_\nn(x) \dd \mu_\nn(y) + \int \WW \dd \mu_\nn  \right) - \hal \nn^2 \IW(\muW) \right)  \frac{\dd  \X_\nn}{|\La^\nn|},
\end{equation*}
which we have just bounded.
}

In the proof below, we will make a crucial use of the following observation: the “cluster bound” for independent points is not worse than the one mentioned in \eqref{eq:cluster_bound}, and we thus have:
\begin{equation}
\label{EE_LavsLa}
\frac{|\EE_\La|}{|\La^\nn|} \geq 1 - \exp\left(- \log^2 T \right).
\end{equation}

\cor{In the following claim, we focus on the second term in the integrand of \eqref{tocomparing}.}
\begin{claim}[\cor{The restriction does not matter too much}] 
\label{lem:virer_le_conditionnement}
We have:
\begin{equation}
\label{eq:sans_condi_lavie}
\hal \nn^2 \int_{\La^\nn} \ind_{\EE_\La}(\bX) \left( \int \WW \dd \mu_\nn  \right)  \frac{\dd \X_\nn}{|\EE_\La|} \leq \hal \nn^2 \int_{\La^\nn} \left( \int \WW \dd \mu_\nn  \right)  \frac{\dd \X_\nn}{|\La|^\nn} + \exp\left(- \log^2 T \right).
\end{equation}
\end{claim}
\begin{proof}[Proof of Lemma \ref{lem:virer_le_conditionnement}]
Let $\Cc$ be some large constant, and let
\begin{equation}
\label{def:hatWW}
\hat{\WW} := \WW - \VextL(\omega) + \Cc T^2.
\end{equation}
Substracting the same constant to both sides, we may write that:
\begin{multline}
\label{ladifference}
\hal \nn^2 \int_{\La^\nn} \left( \int \WW \dd \mu_\nn  \right)  \frac{\dd \X_\nn}{|\La|^\nn}  - \hal \nn^2 \int_{\La^\nn} \ind_{\EE_\La}(\bX) \left( \int \WW \dd \mu_\nn  \right)  \frac{\dd \X_\nn}{|\EE_\La|} 
\\
= \hal \nn^2 \int_{\La^\nn} \left( \int \hat{\WW} \dd \mu_\nn  \right)  \frac{\dd \X_\nn}{|\La|^\nn}  - \hal \nn^2 \int_{\La^\nn} \ind_{\EE_\La}(\bX) \left( \int \hat{\WW} \dd \mu_\nn  \right)  \frac{\dd \X_\nn}{|\EE_\La|}.  
\end{multline}
On the other hand, in view of the definition \eqref{rappelWW} of $\WW$ and the properties of the “good external potential” $\VextL$ as listed in Definition \ref{def:GoofPotential}, we know that by choosing the constant $\Cc$ large enough we can guarantee that $\hat{\WW} \geq 0$, in which case it is clear that:
\begin{equation}
\label{avechatWW}
\hal \nn^2 \int_{\La^\nn} \ind_{\EE_\La}(\bX) \left( \int \hat{\WW} \dd \mu_\nn  \right)  \frac{\dd \X_\nn}{|\EE_\La|} \leq \hal \nn^2 \int_{\La^\nn} \left( \int \hat{\WW} \dd \mu_\nn  \right)  \frac{\dd \X_\nn}{|\La|^\nn} \times \left( \frac{|\La^\nn|}{|\EE_\La|} \right).
\end{equation}
We thus see that:
\begin{equation*}
\hal \nn^2 \int_{\La^\nn} \left( \int \WW \dd \mu_\nn  \right)  \frac{\dd \X_\nn}{|\La|^\nn}  - \hal \nn^2 \int_{\La^\nn} \ind_{\EE_\La}(\bX) \left( \int \WW \dd \mu_\nn  \right)  \frac{\dd \X_\nn}{|\EE_\La|}  \leq \hal \nn^2 \int_{\La^\nn} \left( \int \hat{\WW} \dd \mu_\nn  \right)  \frac{\dd \X_\nn}{|\La|^\nn}  \left(\frac{|\La^\nn|}{|\EE_\La|} - 1 \right).
\end{equation*}
On the other hand, the same computation as in the proof of Claim \ref{claim:sansEELA} yields:
\begin{equation*}
\hal \nn^2 \int_{\La^\nn} \left( \int \hat{\WW} \dd \mu_\nn  \right)  \frac{\dd \X_\nn}{|\La|^\nn} = \hal \nn^2 \int \hat{\WW}(x) \dd \muV(x),
\end{equation*}
which we can evaluate explicitely using our choice \eqref{def:hatWW} for $\hat{\WW}$ and the expression \eqref{rappelWW} for $\WW$. We obtain some polynomial in $T$, which gets absorbed by the sub-algebraic tail of \eqref{EE_LavsLa}. This concludes the proof of \eqref{eq:sans_condi_lavie}.
\end{proof}
We could proceed similarly for the first term \cor{in the integrand of \eqref{tocomparing}} (which is easier because $\WW$ does not play any role) \cor{and for the third term (which is a constant)}.

This concludes the proof of \eqref{global_law_SS_Statement} and thus of Proposition \ref{prop:global_law_SUBSYS}.
\end{proof}

\section{Local laws for sub-systems: proof of Proposition \ref{prop:local_law_SUBSYS}}
\label{sec:proofLL}
\newcommand{\ccR}{\mathcal{R}}

\subsubsection*{\cor{A summary of the local laws of \cite{armstrong2019local}.}}
\cor{By “local laws”, we mean \emph{good controls on the electric energy down to (large enough) microscopic scales} - controls on the number of points are then obtained as a byproduct. By “good control”, we mean ideally an estimate of the form (for some constant $\Cc_\beta$ depending only on $\beta$):
\begin{equation}
\label{LLAlleg}
\log \mathbb{E}\left[ \exp\left( \frac{\beta}{\Cc} \Ener\left(\bX, \sq(x, \ccR) \right) \right) \right] \leq \Cc_\beta \ccR^2
\end{equation}
where $\ccR$ is the lengthscale (for convenience, we switch here to the notation of \cite{armstrong2019local} for scales) and the expectation is taken under the Gibbs measure of the system.}

\cor{Such controls are easier to obtain when $\ccR$ is large, but for many purposes it is crucial to get them for $\ccR$ as small as some constant\footnote{\cor{Getting local laws \emph{down to the microscopic scale} is one of the big improvements of \cite{armstrong2019local} over the “mesoscopic” ones proven in \cite{MR3719060,bauerschmidt2017local}.}} depending only on $\beta$, but not on the size of the system - this is the constant $\rho_\beta$ in the statement. Finally, for technical reasons, local laws happen to be more difficult to prove if the region $\sq(x, \ccR)$ under consideration is close to the boundary of the system, and one usually places a restriction to be “in the bulk” in some sense (our notion of bulk here is, for convenience, much more restrictive than the one in \cite{armstrong2019local}).}

\medskip

\cor{In \cite{armstrong2019local}, the authors consider (one-component) Coulomb gases in dimension $d \geq 2$, which includes the 2DOCP with soft edge (see Section \ref{sec:discussion_model} for a dicussion), and they can also treat the case of a 2DOCP with hard edge and Neumann boundary conditions, which we will not discuss here. The background measure in their case does not have to be uniform.} 

\cor{The proof strategy in \cite{armstrong2019local}, as in \cite{MR3719060,bauerschmidt2017local}, is a bootstrap in scales, starting from the macroscopic one ($\ccR = $ size of the system) and going down to the microscopic one ($\ccR \geq \rho_\beta$, of order $1$).}

\subsubsection*{\cor{Description of the bootstrap.}}
\cor{In a nutshell, the aforementioned “bootstrap in scales” proceeds as follows:
\begin{enumerate}
	\item Start from $\ccR$ of order $T$ (our notation for the total size of the system), use a “global law” of the form:
	\begin{equation*}
\log \mathbb{E}\left[ \exp\left( \frac{\beta}{\Cc} \FL(\bX) \right) \right] \leq \Cc_\beta T^2,
	\end{equation*}
	(see \eqref{eq:global_law_FN} in the case of the 2DOCP) and convert it into an estimate on the global  “electric” energy via identities like \eqref{extendSecondFormulation}.
\item Show that if the local laws \eqref{LLAlleg} hold at scale $\ccR$ with a constant $\Cc_\beta$, then they also hold at scale $\hal \ccR$ \emph{with the same constant $\Cc_\beta$}, as long as $\ccR$ satisfies certain conditions which boil down to $\ccR$ being larger than a certain lengthscale $\rho_\beta$.
\end{enumerate}
This clearly gives the desired result. The difficult point is the second one, it relies on the \emph{screening construction} of Serfaty et al., whose role we now briefly present.
}

\subsubsection*{\cor{The role of screening.}}
\cor{For the purpose of the exposition, let us write the Gibbs density as $\frac{1}{\int e^{-\beta \F(\bX)} \dd \X} e^{-\beta \F(\bX)}$, where $\F$ could be the energy $\FN$ of the 2DOCP as in \eqref{def:FN}, or the interaction energy of a $2d$ Coulomb gas with soft confinement as in \cite{armstrong2019local}, or the energy $\FL$ of our generalized 2DOCP's as in \eqref{def:FLnuW}...}
\newcommand{\Fint}{\F^{\Int}}
\newcommand{\Fext}{\F^{\Ext}}
\cor{Fix some “interior” region $\Int := \sq(x, \ccR)$ on which one wants to prove local laws, i.e. to control the exponential moment:
\begin{equation*}
\E\left[ e^{\frac{\beta}{2} \Ener(\bX, \sq(x, \ccR))} \right] = \frac{1}{\int e^{-\beta \F(\bX)} \dd \X} \int e^{-\beta \F(\bX) + \frac{\beta}{2} \Ener(\bX, \sq(x, \ccR))} \dd \X \overset{?}{\leq} e^{\Cc_\beta \ccR^2}.
\end{equation*}
The main issue here is that $\F(\bX)$ (appearing in both integrands) is a global quantity of order $T^2$ whereas $\Ener(\bX, \sq(x, \ccR))$ is a local quantity (which we hope to show is of order $\ccR^2$) and it is thus hard to have them “play together”.}

\medskip

\cor{Now, assume for a moment that the energy $\F$ splits into something of the form:
\begin{equation}
\label{eq:Ffintfext}
\F(\bX) = \Fint(\bXint) + \Fext(\bXext),
\end{equation}
with $\Fint$ the energy of the configuration in $\Int$ and $\Fext$ the energy in $\Ext := \overline{\Int}$. Then it would be easy to check that we can integrate out the $\Ext$-configuration and obtain:
\begin{equation}
\label{localization}
\mathbb{E}\left[ \exp\left( \frac{\beta}{2} \Fint \right) \right] = \mathbb{E}_{\Int} \left[ \exp\left( \frac{\beta}{2} \Fint \right) \right], 
\end{equation}
where $\mathbb{E}_{\Int}$ is the “local” Gibbs density $\frac{1}{\int e^{-\beta \Fint(\bXint)} \dd \Xint} e^{-\beta \Fint(\bXint)}$. The situation would become much nicer because all the terms now live at the same scale $\ccR$, and we could use the same strategy as when proving the “global” law in order to derive the desired estimate.
}

\medskip

\cor{In fact, in view of the identities \eqref{secondformulation} or \eqref{extendSecondFormulation}, we do have (cf. \eqref{eq:Ffintfext}):
\begin{multline}
\label{presquesplit}
\F(\bX) = \frac{1}{4\pi} \left( \Ener(\bX, \Int) + \Ener(\bX, \Ext) \right) \\
+ \sum_{x \in \bXint} \left(\frac{1}{4\pi} \log \eta(x) - \int_{\DD(x,\eta(x))} \hspace{-0.3cm} \ff_{\eta(x)}(t-x) \dd t\right) + \sum_{x \in \bXext} \left(\frac{1}{4\pi} \log \eta(x) - \int_{\DD(x,\eta(x))} \hspace{-0.3cm} \ff_{\eta(x)}(t-x) \dd t\right),  
\end{multline}
and it thus looks like we can split $\F$ according to $\Int$ and $\Ext$. However $\Ener(\bX, \Int)$ is not a priori equal to $\Ener(\bXint, \Int)$ - the electric field living on $\Int$ depends \emph{on the entire configuration}, not only on $\bXint$. This subtle lack of locality is the central issue here.}

\medskip

\cor{The “screening construction” initiated in \cite{MR3353821} helps us making sense of \eqref{eq:Ffintfext} and \eqref{localization} - in an approximate way. In short, it starts with the data of $\bXext$, and proceeds to placing an arbitrary “nice” configuration $\bXint$ inside of $\Int$ while ensuring that (cf. \eqref{eq:Ffintfext}):
\begin{equation}
\label{PreScr}
\Fint(\bXint) + \Fext(\bXext) \leq \F(\bXint + \bXext) \leq \Fint(\bXint) + \Fext(\bXext) + \ErrorScr(\ccR), 
\end{equation}
with $\ErrorScr(\ccR)$ the screening error at scale $\ccR$. We will not define $\Fint$, $\Fext$ here, but these are “localized” versions of $\frac{1}{4\pi} \Ener(\bX, \Int)$, $\frac{1}{4\pi} \Ener(\bX, \Ext)$ (cf. \eqref{presquesplit}).}

\medskip

\cor{Obtaining \eqref{PreScr} allows us to get something of the form (cf. \eqref{localization}):
\begin{equation}
\label{PreScrCons}
\mathbb{E}\left[ \exp\left( \frac{\beta}{2} \Fint \right) \right] \leq  \mathbb{E}_{\Int} \left[ \exp\left( \frac{\beta}{2} \Fint \right) \right] \times \exp\left(\beta \times \text{A bound on } \ErrorScr(\ccR) \right),
\end{equation}
the interested reader might look e.g. at \cite[(4.30)]{armstrong2019local} for a precise formulation.}

\medskip

\cor{In fact, the screening error term in \eqref{PreScr} is not bounded uniformly with respect to $\bXext$, but it can be bounded \emph{if we know that the energy of $\bXext$ in a slightly larger square $\sq(x, 2 \ccR)$ obeys a certain bound} (this follows from the screening construction itself, which is both quite involved and quite explicit). \\ The bootstrap in scales can now take place, by saying at each step:
\begin{enumerate}
	\item We know that the local laws hold at scale $2 \ccR$ i.e. we control the typical energy of $\bXext$ in $\sq(x, 2 \ccR)$.
	\item Hence we can control the typical screening error $\ErrorScr(\ccR)$ at scale $\ccR$ appearing in \eqref{PreScr}.
	\item We then run the screening construction and get \eqref{PreScrCons}. We deduce the local laws at scale $\ccR$.
\end{enumerate}
}

\subsubsection*{\cor{In which way(s) is our situation different.}}
Compared to \cite{armstrong2019local}, there are two main differences:
\begin{enumerate}
	\item We are in presence of a background measure $\nn \muW$ \cor{which is non-uniform in a strange way: it is  uniform as soon as one is far enough from the boundary, but also has holes near the boundary and a singular component on the boundary, as described in Proposition \ref{prop:correspondence}.}
	\item The \emph{global law} that we were able to derive is as in Proposition \ref{prop:global_law_SUBSYS}, with a global energy estimate that is slightly larger than the total volume. \cor{We thus start the bootstrap in scales with an information of the form:
\begin{equation}
\label{instead}
\text{Energy at scale $T$} \leq \Cc_\beta T^2 \log^5 T \text{ instead of } \leq \Cc_\beta T^2,
\end{equation}
(in exponential moments) and yet we would like to obtain, for $\ccR \leq \hal T$: 
\begin{equation}
\label{andyet}
\text{Energy at scale $\ccR$} \leq \Cc_\beta \ccR^2.
\end{equation}}
\end{enumerate}
Our goal is thus twofold:
\begin{enumerate}
	\item \cor{Check that the method of \cite{armstrong2019local} still applies with our background measure.}
	\item Show that the bootstrap in scale not only \emph{propagates} \cor{good energy estimates} to smaller scales, but in fact \emph{improves} (if needed) the estimates at each step, which allows to get rid of the logarithmic correction \cor{that we start with in \eqref{instead}.}
\end{enumerate}

The first point is, in fact, not an issue. Indeed, although the methods of \cite{armstrong2019local} are \emph{not} suited to situations where the background measure has singularities or holes (in fact this is the reason why local laws are not proven near the edge), they are local by design and they will apply readily as soon as we look at distances $\geq \Cc \log T$ from the boundary, because then the background measure $\nn \muW$ \cor{has a constant density} equal to $1$ (see Section \ref{sec:good_properties}). Since we only care about the bulk of $\La$ \cor{in our statement}, this is fine. \cor{It remains to understand why starting with \eqref{instead} still leads to \eqref{andyet}.}

\subsubsection*{\cor{Why the bootstrap still works.}}
\cor{The main screening error $\ErrorScr$ in \eqref{PreScr} is the first term in \cite[(4.7)]{armstrong2019local} which reads “$\frac{\ell}{\tell} S$”, where $\ell, \tell$ are parameters and $S$ depends on the configuration $\bXext$. More precisely:
\begin{itemize}
	\item The quantity $\ell$ must be such that $\ell^3 \geq \frac{S}{\tell}$ (\cite[(4.4)]{armstrong2019local}).
	\item The quantity $\tell$ can be chosen\footnote{The interested reader can find a list of conditions on $\ell, \tell$ at the end of the proof of \cite[Prop. 4.5]{armstrong2019local}. The choice of $\tell$ mentioned in \cite[(4.27)]{armstrong2019local} correspond to the \emph{smallest} possible choice, but increasing $\tell$ up to $\mathcal{R}$ still gives a valid choice.} of the same order at the scale $\ccR$.
	\item The quantity $S$ is bounded by the energy at scale $2 \ccR$.
\end{itemize}
so in fact the screening error can be controlled by: 
\begin{equation}
\label{theErrorScr}
\ErrorScr(\ccR) \preceq \left( \frac{S}{\ccR} \right)^{4/3} \preceq \left( \frac{\text{Energy at scale $2 \ccR$}}{\ccR} \right)^{4/3}. 
\end{equation}
In particular, if we know (by a previous step of the bootstrap) that the energy at scale $2 \ccR$ is typically smaller than $\Cc_\beta \times (2 \ccR)^2$, we get an error of order $\ccR^{4/3} \ll \mathcal{R}^2$ which is indeed negligible\footnote{This can in fact be improved further by chosing $\ell, \tell$ more cleverly, see the proof of \cite[Prop. 2.5]{armstrong2019local}.}}.

\medskip

\cor{For our purposes, the key observation is that even if we start from a poorer estimate\footnote{Having a poorer estimate is equivalent to having the parameter $\mathcal{C}$ in \cite[(4.25)]{armstrong2019local} depend on $\mathcal{R}$. The conditions written at the end of the proof of \cite[Proposition 4.5]{armstrong2019local} can still be satisfied as long as $\mathcal{C}$ is much smaller than $\mathcal{R}$, which corresponds to an initial energy estimate in $o(\mathcal{R}^3)$. \\ If $\mathcal{C}$ is smaller than $\mathcal{R}^{3/2}$ then the error term will become smaller than $\mathcal{R}^2$ \emph{in one step}, otherwise one would need to apply the bootstrap in scales for some time (i.e. go down in scales) before reaching the desired local laws. Our situation corresponds to $\mathcal{C} = \log^5 \mathcal{R}$, cf. \eqref{global_law_SS_Statement}.} on the energy at scale $2 \ccR$, for example one with some big logarithmic corrections of the form (cf. our global law \eqref{global_law_SS_Statement}):
\begin{equation*}
\text{Energy at scale $2 \ccR$} \leq \left(2\ccR \right)^2 \log^{100} (\ccR)
\end{equation*}
(say, in exponential moments) then the right-hand side of \eqref{theErrorScr} \emph{remains much smaller than $\mathcal{R}^2$}.} 

\medskip

\cor{Returning to \eqref{PreScrCons}, we see that even if we start from \eqref{instead} then after \emph{a single step} of the bootstrap, i.e. when working at the scale $\ccR = \hal T$, we have something of the form:
\begin{equation*}
\mathbb{E}\left[ \exp\left( \frac{\beta}{2} \Fint \right) \right] \leq  \mathbb{E}_{\Int} \left[ \exp\left( \frac{\beta}{2} \Fint \right) \right] \times \exp\left(\beta \times o(T^2) \right),
\end{equation*}
and thus the “quality” of our local law depends entirely on how good of an estimate we can get for $\mathbb{E}_{\Int} \left[ \exp\left( \frac{\beta}{2} \Fint \right) \right]$. But we are now at scale $\hal T$ so “in the bulk”, where the background measure has density $1$, and the standard techniques apply.
}

\medskip

\cor{ In fact, a careful inspection of the proof (which we will not need here) reveals that one could start with a “global law” as poor as $\O(T^{5/2})$, or even $o(T^3)$ instead of \eqref{global_law_SS_Statement} and still recover good local laws, possibly after a performing a certain number of steps of the bootstrap.
}

\section{Effect of localized translations: Proof of Proposition \ref{prop:effectPhiEnergy}}
\label{sec:EffetTranslaEnerPreuve}
In all this section, $\epsilon$ and $\ell$ are fixed and we assume that $|t| \leq \frac{\ell}{10}$. To lighten notation, we will drop some dependencies with respect to $\epsilon$ and $\ell$. \cor{The “universal” constants $\Cc$ in this section depend of course on our construction of the flow, but not on the other parameters ($t, \ell, x$ etc.).}

\subsubsection*{\cor{Some reminders}}
\cor{In Section \ref{sec:slowvartrans}, we have defined a flow  $\{\Phit\}_t$ for $|t| \leq \frac{\ell}{10}$, such that (see Lemma \ref{lemma_propPhi}):
\begin{enumerate}
	\item $\Phit$ is an area-preserving diffeomorphism of $\R^2$.
	\item $\Phit(x) = x + t \vu$ for $|x| \leq \ell/4$.
	\item  $\Phit(x) = x$ for $|x| \geq 2 \ell \eep $.
\end{enumerate}
The defining property of $\Phit$ is that it solves the ODE \eqref{eq:ODEPhit} $\partial_t \Phit = \Wel \circ \Phit$.
}

\subsection{Some additional properties of localized translations}
\label{study_of_phit}
\renewcommand{\Psit}{\psi_t}
\renewcommand{\Phit}{\Phi_t}
\newcommand{\rt}{\gamma_t}
\newcommand{\rmt}{\gamma_{-t}}
We decompose $\Phit$ in two ways, either as:
\begin{equation}
\label{redefPsit}
\Phit(x) = x + \Psit(x),
\end{equation}
which defines a vector field $\Psit$, or alternatively as:
\begin{equation}
\label{defrt}
\Phit(x) = x + t \Wel(x) + \rt(x),
\end{equation}
which defines a vector field $\rt$.

\begin{proposition}
\label{prop:taut_prop}
\begin{enumerate}
    \item We have $|\Phit - \Id|_\0 = |\Psit|_\0 \leq 2 |t| \leq \frac{\ell}{5}$.
    \item For $|x| \leq \ell/4$, we have $|\Psit|_{\1, \star}(x) = 0$, $|\Psit|_{\2, \star}(x) = 0$.
    \item For $\ell/4 \leq |x| \leq  2 \ell \eep$, we have:
\begin{equation}
\label{Psit1X}
    |\Psit|_{\kk, \star}(x) \leq \frac{\Cc t \epsilon}{|x|^\kk} \text{ for } \kk = 1, 2.
\end{equation}
\item For $|x| \leq \ell/4$, we have $\rt \equiv 0$.
\item For $\ell/4 \leq |x| \leq  2 \ell \eep$, we have:  
\begin{equation}
\label{Psit2X}
  |\rt|_{\kk, \star}(x) \leq \frac{\Cc t^2 \epsilon}{|x|^{\kk + 1}}, \text{ for } \kk = 0, 1, 2, 3.
\end{equation}
\end{enumerate}
\end{proposition}

\begin{proof}[Proof of Proposition \ref{prop:taut_prop}]
Let us start by some general observations.

\subsubsection*{\cor{Step 1.} Preliminary claims.}
Since the vector field $\Wel$ is bounded by $2$ (see Lemma \ref{lem:spinwave}) the distance $|\Phit(x) - x|$ is always smaller than $2|t|$ and thus $|\Psit|_\0 \leq 2 |t|$. This proves the first item.  On the other hand, as stated in Lemma \ref{lemma_propPhi}, the flow $\Phit$ acts as a translation on $\DD(0, \ell/4)$, which means that $\Psit$ is a constant and thus the derivatives of $\Psit$ vanish identically there. This proves the second item, and also implies that the second-order correction $\rt$ vanishes identically on $\DD(0, \ell/4)$.
\medskip
In view of the bounds on $\Wep$ given in Lemma \ref{lem:spinwave} and the definition of $\Wel := \Wep(\cdot / \ell)$, after scaling we obtain for $\kk = 1, 2$:
\begin{equation*}
|\Wel|_{\kk, \star}(x) \leq \Cc \epsilon 
\begin{cases} \frac{1}{|x|^\kk} & \text{ for } |x| \geq 2 \ell \\
\frac{1}{|\ell|^\kk} & \text{ for } |x| \leq 2 \ell. 
\end{cases}
\end{equation*}
Thus we may always write that $|\Wel|_{\kk, \star}(x) \leq \Cc \epsilon \frac{1}{|x|^\kk}$. 

\begin{claim} 
\label{claim:WAtPhitx}
If $|x| \geq \frac{\ell}{4}$ then:
\begin{equation}
\label{theboundWel}
|\Wel|_{\kk, \star}(\Phit(x)) \leq \Cc \epsilon \frac{1}{|x|^\kk}
\end{equation}
\end{claim}
\begin{proof}
Since $|t| \leq \frac{\ell}{10}$ and $|\Phit(x) - x| \leq 2|t|$, we can ensure that if $|x| \geq \frac{\ell}{4}$, then $|\Phit(x)| \geq \frac{1}{5} |x|$, so in view of the previous estimates on $\Wel$, we obtain the claim.
\end{proof}

Let us end this paragraph with a simple general fact which will be useful to prove the remaining bounds:
\begin{equation}
\label{ddfnorm}
\frac{\dd}{\dd t} \|f(t)\| \leq \| \frac{\dd}{\dd t} f(t)\|.
\end{equation}

\subsubsection*{\cor{Step 2.} Initial controls on $\Dd^\kk \Phit$.}
\begin{claim}
\label{claim:control_DdkPhit}
We have, for $|x| \geq \frac{\ell}{4}$:
\begin{equation}
\label{eq:DDkPhit}
\|\Dd \Phit(x) \| \leq 1 + \Cc \frac{\epsilon t}{|x|}, \quad \|\Dd^\kk \Phit(x) \| \leq \Cc \frac{\epsilon t}{|x|^\kk} \text{ for $\kk = 2, 3$}
\end{equation}
\end{claim}
\newcommand{\ddt}{\frac{\dd}{\dd t}}

\begin{proof}[Proof of Claim \ref{claim:control_DdkPhit}]
At $t = 0$ we have $\Phit = \Id$ and thus $\Dd \Phit \equiv \Id$. Then we can compute:
\begin{equation*}
\ddt \| \Dd \Phit(x) \| \leq \| \ddt \Dd \Phit(x) \| = \| \Dd \left(\Wel \circ \Phit(x) \right) \| \leq \| \Dd \Wel \circ \Phit(x) \| \times \| \Dd \Phit(x) \|. 
\end{equation*}
Using \eqref{theboundWel} and \cor{Grönwall's inequality}, we see that:
\begin{equation*}
 \| \Dd \Phit(x) \|  \leq e^{\frac{\Cc \epsilon t}{|x|}} = 1 + \Cc' \frac{\epsilon t}{|x|}.
\end{equation*}
Arguing similarly, we have: $\Dd^2 \Phit = 0$ at $t=0$, and:
\begin{equation*}
\ddt \| \Dd^2 \Phit(x) \| \leq \| \ddt \Dd^2 \Phit(x) \| = \| \Dd^2 \left(\Wel \circ \Phit(x) \right) \|.
\end{equation*}
Applying Leibniz's rule, using the previous bound under the rough form $\|\Dd \Phit\| \leq \Cc$ and inserting \eqref{theboundWel} again, we obtain:
\begin{equation*}
\ddt \| \Dd^2 \Phit(x) \| \leq \frac{\Cc \epsilon}{|x|^2} + \frac{\Cc \epsilon}{|x|} \times   \| \Dd^2 \Phit(x) \|.
\end{equation*}
Solving again the differential inequality with initial condition $0$, we obtain:
\begin{equation*}
\| \Dd^2 \Phit(x) \| \leq \Cc \frac{\epsilon t}{|x|^2}
\end{equation*}
We proceed the same way for the third derivative.
\end{proof}

\subsubsection*{\cor{Step 3. Proof of \eqref{Psit1X} (properties of $\Psit$).}}
Returning to the definition \eqref{redefPsit} of $\Psit$, we see that $\Dd \Psit = 0$ at $t = 0$ and we may then compute:
\begin{equation*}
\ddt \| \Dd \Psit(x) \| \leq \| \ddt \Dd \Psit(x) \| = \| \Dd \left(\Wel \circ \Phit(x) \right) \| \leq | \Dd \Wel \circ \Phit(x) \| \times \| \Dd \Phit(x) \|. 
\end{equation*}
By the same computations as above, we find $\ddt \| \Dd \Psit(x) \| \leq \Cc \frac{\epsilon}{|x|}$ which implies the bound \eqref{Psit1X} for $\kk = 1$.

To study $|\Psit|_\2$, observe that $\Dd^2 \Phit = \Dd^2 \Psit$ and use \eqref{eq:DDkPhit}.


\subsubsection*{\cor{Step 4. Proof of \eqref{Psit2X} (properties of $\rt$)}.}
If we want to control, say $|\rt|_\0$, we compute:
\begin{equation*}
\left| \ddt \rt \right| = \left| \Wel \circ \Phit(x) - \Wel(x) \right| \leq \Cc \frac{\epsilon |\Psit|_\0}{|x|} \leq \Cc \frac{\epsilon t}{|x|},
\end{equation*}
where we used the bound on $|\Psit|_\0$ stated as the first item of Proposition \ref{prop:taut_prop} together with Claim \ref{claim:WAtPhitx} in order to bound the Lipschitz constant of $\Wel$ between $x$ and $\Phit(x)$. Integrating on $t$ yields $|\rt|_\0 \preceq \epsilon t^2 |x|^{-1}$ as claimed.

Higher derivatives are controlled the same way, using Leibniz's rule together with Claims~\ref{claim:WAtPhitx} and~\ref{claim:control_DdkPhit}.
\end{proof}

\subsubsection*{Upper bounds on the derivatives of $\Psit$.} 
\newcommand{\PsiU}{\Psi_{1}}
\newcommand{\PsiD}{\Psi_{2}}
Let us introduce two functions $\PsiU, \PsiD$:
\begin{equation}
\label{def:PsiUD}
\PsiU : x \mapsto \frac{\Cc t \epsilon}{(\ell + |x|)}, \quad \PsiD : x \mapsto \frac{\Cc t \epsilon}{(\ell + |x|)^2}
\end{equation}
Choosing the constant $\Cc$ suitably we have, in view of Proposition \ref{prop:taut_prop}, the pointwise bounds:
\begin{equation*}
|\Psit|_{\1, \star}(x) \leq \PsiU(x), \quad |\Psit|_{\2, \star}(x) \leq \PsiD(x), \text{ for all } x \in \R^2,
\end{equation*}
thus when looking for \emph{upper bounds} we may replace occurrences of $|\Psit|_{\1, \star}, |\Psit|_{\2, \star}$ by $\PsiU, \PsiD$. The upside of working with $\PsiU, \PsiD$ is that they enjoy the following properties deduced from elementary calculus:
\begin{itemize}
\item (Slow variation at scale $\ell$.) For all $x$ we have:
\begin{equation}
\label{eq:Psi_slowvar_l}
\sup_{y, |x-y| \leq \ell} \frac{|\PsiU(y)|}{|\PsiU(x)|} \leq 2, \quad \sup_{y, |x-y| \leq \ell} \frac{|\PsiD(y)|}{|\PsiD(x)|} \leq 2.
\end{equation}
\item (Slow variation.) For all $x$ we have:
\begin{equation}
\label{eq:Psi_slowvar_}
\sup_{y, |x-y| \leq \hal |x|} \frac{|\PsiU(y)|}{|\PsiU(x)|} \leq 2, \quad \sup_{y, |x-y| \leq \hal |x|} \frac{|\PsiD(y)|}{|\PsiD(x)|} \leq 2.
\end{equation}
\end{itemize}
Both \eqref{eq:Psi_slowvar_l} and \eqref{eq:Psi_slowvar_} will be convenient to simplify computations later on. We also record the following simple facts:
\begin{claim}  
\label{claim:simple_facts}
If $\epsilon$ is chosen smaller than some universal constant, then:
\begin{itemize}
\item $|\Psit|_{\1} \leq \frac{1}{5}$ (globally).
\item For all $x,x'$ we have:
\begin{equation*}
\hal |x-x'| \leq |\Phit(x) - \Phit(x')| \leq 2 |x-x'|. 
\end{equation*}
\item For all $x$ we have:
\begin{equation}
\label{PsitPhitPsiU}
|\Psit|_{\1, \star}(\Phit(x)) \leq 2 \PsiU(x), \quad |\Psit|_{\2, \star}(\Phit(x)) \leq 2 \PsiD(x).
\end{equation}
    \end{itemize}
\end{claim}
\begin{proof}[Proof of Claim \ref{claim:simple_facts}]
The first item follows directly from Proposition \ref{prop:taut_prop} (see in particular \eqref{Psit1X}). It implies the second item straightforwardly. To prove \eqref{PsitPhitPsiU} let us recall that $|\Phit(x) - x| \leq 2|t| \leq \frac{\ell}{5}$ (by the first item of Proposition \ref{prop:taut_prop} and the assumption $|t| \leq \frac{\ell}{10}$) and that $|\Psit|_{\1, \star}(\Phit(x)) \leq \PsiU(\Phit(x))$ by construction (similarly for $|\Psit|_{\2, \star}, \PsiD$). We thus have:
\begin{equation*}
|\Psit|_{\1, \star}(\Phit(x)) \leq \PsiU(\Phit(x)) \leq \sup_{|y-x| \leq \frac{\ell}{10}} \PsiU(y),
\end{equation*}
which is smaller than $2 \PsiU(x)$ according to \eqref{eq:Psi_slowvar_l}, as desired.
\end{proof}

\subsection{The “well-spread” event}
\label{sec:well_spread}
\begin{definition}[The $\WS$ event]
\label{defi:WS}
Let $\Omega$ be some subset of $\R^2$, let $\ell \geq 1$ be a length-scale, let $K \geq 10$. We define the event $\WS(\Omega, \ell, K)$ as follows:
\begin{equation}
\label{def:WS}
\WS(\Omega, \ell, K) := \bigcap_{x \in (\ell \Z)^2 \cap \Omega} \left\lbrace \Points(\bX, \sq(x, \ell)) \leq K \ell^2 \right\rbrace \cap \left\lbrace \Ener(\bX, \sq(x, \ell)) \leq K \ell^2 \right\rbrace.
\end{equation}
\end{definition}
Saying that $\bX \in \WS(\Omega, \ell, K)$ essentially means that if we cover $\Omega$ by squares of side-length $\ell$, then the number of points and the electric energy in each square are of order $\ell^2$, which is what we expect in view of the local laws. It has the following consequences:

\begin{enumerate}
    \item If $\bX \in \WS(\Omega, \ell, K)$ then for each $x \in (\ell \Z)^2 \cap \Omega$ the quantity $\Ener_s(\bX, \sq(x, \ell))$ (as defined in \eqref{def_Ener}) is bounded by $K \ell^2 (1 + \log s)$ for $s \in (0,1)$.
    \item Let us say that a function $f$ varies slowly at scale $\ell$ when:
\begin{equation}
\label{slowlyvar} 
\sup_{|x' - x| \leq 2 \ell} |f(x')| \leq 10 |f(x)|. 
\end{equation}
Then if $f$ satisfies \eqref{slowlyvar} and is supported in $\Omega$, assuming that $\bX \in \WS(\Omega, \ell, K)$ allows us to make two computational simplifications:
\begin{enumerate}
    \item An energy density upper bound, namely (for $s \in (0,1)$):
    \begin{equation}
    \label{eq:EnerDensUB}
    \left| \int f(x) |\nabla h_{s \vr}|^2 \dd x \right| \preceq K \|f\|_{L^1} (1 + |\log s|). 
    \end{equation}
    \item Sum-integral comparisons:
    \begin{equation}
    \label{eq:SumIntegral}
     \left| \sum_{x \in \bX} f(x) \dd x \right| \preceq K \|f\|_{L^1}. 
    \end{equation}
\end{enumerate}
(To prove \eqref{eq:EnerDensUB} and \eqref{eq:SumIntegral}, cover the support of $f$ by squares of sidelength $\ell$ and use \eqref{slowlyvar}).
\end{enumerate}

\begin{lemma}[The Well-Spread event is frequent]
\label{lem:WS_is_frequent}
For $K$ larger than some constant (depending on $\beta$ and the “good external potential” constant $\bCc$) and if $\epsilon \ell^2 \geq 1$ (which is guaranteed by \eqref{eq:ellrhoepsilon}), we have:
\begin{equation}
\label{eq:PnLLgoodgood}
\PnLL \left( \WS(\Lab, \ell, K) \big| \EE_\La   \right) \geq 1 - \exp\left( - \ell^2  \right)
\end{equation}
\end{lemma}
\begin{proof}[Proof of Lemma \ref{lem:WS_is_frequent}]
Covering $\Lab$ by $\O(T^2 \ell^{-2})$ squares of sidelength $\ell$ and using the local laws in $\Lab$ on each one, together with a union bound, we get:
\begin{equation}
\label{eq:WSfrequent}
\PnLL \left( \WS(\Lab, \ell, K)  \big| \EE_\La \right) \geq 1 - \Cc \frac{T^2}{\ell^2} \exp\left(- \frac{K \ell^2}{\Cc_\beta} \right).
\end{equation}
Using the relation between $T, \ell$ and $\epsilon$ as in \eqref{eq:ellrhoepsilon}, we obtain:
\begin{equation*}
\PnLL \left( \WS(\Lab, \ell, K) \big| \EE_\La   \right) \geq 1 - \Cc \exp\left( \frac{2}{\epsilon} - \frac{K \ell^2}{\Cc_\beta} \right),
\end{equation*}
and thus if we choose $K$ larger than some constant (depending on $\beta$ and $\bCc$) and impose that $\epsilon \ell^2 \geq 1$ then \eqref{eq:PnLLgoodgood} holds.
\end{proof}

\subsection{The measure-preserving case}
Let $|t| \leq \frac{\ell}{10}$ and let $\Phi = \Phi_t^{\epsilon, \ell}$ be the localized translation constructed in Section \ref{sec:slowvartrans}. Let $\psi = \psi_t^{\epsilon, \ell} = \Phi - \Id$ as studied in Proposition \ref{prop:taut_prop}. In this section, we carefully inspect the proof of \cite[Prop 4.2]{serfaty2020gaussian} in order to obtain the following:
\begin{proposition}[Energy comparison along a measure-preserving map]
\label{prop:meas_pres}
Let $\bX$ be a point configuration with $\nn$ points in $\La$. Assume that $\bX$ belongs to $\WS(\Lab, \ell, K)$ for a certain $K > 1$. Then:
\begin{equation}
\label{bli_measure_pres}
\FL(\PhiX \bX) = \FL(\bX) + \AUN \left[ \bX, \psi \right] + \O\left(K^2 t^2 \epsilon \log \epsilon  \right).
\end{equation}
\end{proposition}
\cor{The quantity $\AUN$, called the anisotropy, is given by:
\begin{equation}
\label{predefAUN}
\AUN \left[ \bX, \psi \right] := \frac{1}{4\pi} \int \left\langle \nhetaX, 2 \Dd \psi \nhetaX \right\rangle + \sum_{1 \leq i \leq \nn} \dashint_{|u| = \eta_i} \nabla \hciX(x_i + u) \cdot \left(\psi(x_i + u) - \psi(x_i)\right) \dd u,
\end{equation}
where the truncation vector $\veta$ is chosen as $\veta : = s \vr \text{ with } s = \epsilon^{2}$ (see \eqref{eq:vetasvr} below), $\vr$ being the vector of nearest-neighbor distances, and $\nhetaX$ is the truncated electric field generated by $\bX$ as in Definition \ref{def:trunc}. We recall all these notions in the course of the proof, and we refer to Section \ref{sec:smallani} for more details on the anisotropy. The presence of this “first-order” term is not new, in fact the most important feature in \eqref{bli_measure_pres} is the size of the error term $\O\left(K^2 t^2 \epsilon \log \epsilon  \right)$, which cannot be deduced from existing results.}

\begin{remark}[Comparison with existing analyses]
\label{rem:meas_pres}
Compared to the analysis done in \cite{serfaty2020gaussian} (on which we rely heavily) there are two important modifications:
\begin{enumerate}
    \item We estimate the energy cost of transporting by $\Phi$ through the \emph{local} density of electric energy  and points instead of using the global one (denoted by $\cXi(t)$ in \cite{serfaty2020gaussian}). \cor{This is why we use the $\WS$ event.}
    \item The quantity $|\psi|_{L^\infty} |\psi|_{C^2}$ which appears in the control on the second derivative of the energy (along a transport) in \cite[Prop. 4.2]{serfaty2020gaussian} does not appear in our computations. 
\end{enumerate}
Both items are crucial for us:
\begin{itemize}
 	\item The localization allows us to bound the error in terms of $\int_\La |\psi|^2_{1,\star}(x) \dd x$ \cor{(recall the notation $|\psi|_{1,\star}(x)$ for the norm of the derivative of $\psi$ at $x$, see Section \ref{sec:notation})} instead of $|\psi|^2_1 \times |\La|$. 

 	In the case of our localized translation, the former is $\O(\epsilon)$ while the latter is gigantic.
 	\item Even after localizing, the contribution of $|\psi|_{L^\infty} |\psi|_{C^2}$, \cor{namely (with our notation) $\int |\psi|_\0 |\psi|_{\2, \star}(x) \dd x$} would be of order $1$ but not small\footnote{It is in fact impossible to make it small by choosing a different “localized translation”, because in dimension $2$ the (homogeneous) Sobolev space $\dot{W}^{2,1}$ is embedded in $L^\infty$. This is, fortunately, not the case for $\dot{W}^{1,2}$.} so it was necessary to get rid of it.
 \end{itemize} 

Obtaining these two refinements requires significant adaptations. On the other hand, the “measure-preserving” character of $\Phi$ will bring several small simplifications: the background measure is not affected by the transport so all distinctions between $\mu$ \cor{and the push-forward} $\nu := \Phi \# \mu$ \cor{of $\mu$ by $\Phi$} (using the notation of \cite{serfaty2020gaussian}) are void. We will in particular repeatedly use the fact that (for various functions $f$):
\begin{equation*}
\int f(x) \left(\dd (\PhiX \bX)(x) - \dd x\right) = \int f(\Phi(x)) \left( \dd \bX(x) - \dd x \right).
\end{equation*}
\end{remark}

\newcommand{\VLa}{\mathrm{V}_\La}
Since we are working on a disk with a constant background, the logarithmic potential generated by said background is explicitly computable and given by:
\begin{equation}
\label{def:VLa}
\VLa(x) := \int_{\La} - \log |x-y| \dd y = \frac{|x|^2}{4}.
\end{equation}
We use the explicit expression \eqref{def:VLa} for simplicity a couple times below, although one could work with a more general shape and proceed to a more careful analysis instead.

\begin{proof}[Proof of Proposition \ref{prop:meas_pres}] \renewcommand{\tell}{\ell}

We follow the steps of \cite[Appendix A]{serfaty2020gaussian} while making several important changes. We will only use a couple technical results as “black boxes” and copy or adapt all the main arguments and computations. \cor{Before diving into the proof, we give an overview: the main goal is to study the difference $\FL(\PhiX \bX) -\FL(\bX)$, where $\Phi$ is a measure-preserving diffeomorphism. In order to do this with a good level of precision, the method introduced in \cite{MR3788208} and refined in \cite{serfaty2020gaussian} consists in:
\begin{enumerate}
	\item Expressing $\FL$ in terms of electric fields as in \eqref{extendSecondFormulation} (see also \eqref{basic_comparison} below).
	\item Using $\Phi$ to transport \emph{the} electric field associated to $\bX$ onto \emph{an} electric field compatible with $\PhiX \bX$ in an explicit way.
	\item Writing $\Phi = \Id + \psi$ and performing a Taylor's expansion of the electric energy.
\end{enumerate}
In that process, two difficulties arise:
\begin{enumerate}
	\item There is never a unique electric field compatible (in the sense of Definition \ref{def:compatible}) with a given configuration $\bX$, see the discussion in Section \ref{sec:ElectricFields}. The expression of $\FL(\bX)$ in terms of electric energy requires to use the “true” electric field $\nabla \h^{\bX}$ as in \eqref{true_mm}. However, transporting $\nabla \h^{\bX}$ by $\Phi$ will usually not produce the true electric field associated to $\PhiX \bX$, but only \emph{some} compatible electric field. One needs to take care of the difference.
	\item At a technical level, $\FL(\bX)$ is not directly expressed in terms of $\nabla \h^{\bX}$ but rather with its \emph{truncated} version $\nhetaX$ in the sense of Definition \ref{def:trunc}, which is an electric field compatible with the “spread-out” configuration $\sum_{x \in \bX} \delta_{x}^{(\eta(x))}$ (see \eqref{DeltaHHeta}), where $\delta_{x}^{(\eta)}$ is the probability measure on the circle of center $x$ and radius $\eta$. When applying a transport, a Dirac mass will be sent to a Dirac mass, but those “spread-out” charges might not be exactly mapped onto other ones.\\
	This might seem innocent, but controlling (to the desired level of precision) the error coming from this difference ends up taking a big part of the proof, as in \cite{serfaty2020gaussian}. One reason is that there are “many” points (thus many spread-out masses), and that because of the long-range interaction they “all interact with each other”.
\end{enumerate}
The proof goes in several steps:
\begin{enumerate}
	\item In Section \ref{sec:Transporting} we define the aforementioned transport of electric fields, and in Section \ref{sec:applying_transport} we apply it to the “electric formulation” of the energy. This yields in \eqref{MainRemA}, \eqref{MainRemB} a first decomposition of the difference $\FL(\PhiX \bX) -\FL(\bX)$ into the sum of two terms $\Main$ and $\Rem$.
	\item The $\Main$ term is already well-understood, and the so-called “anisotropy” $\AUN \left[ \bX, \psi \right]$ comes out of it (see \eqref{bli_measure_pres}, \eqref{predefAUN}). The existing estimates of \cite{serfaty2020gaussian} are easy to localize, and one can check that in contrast to those, no term of the form $|\psi|_\0 |\psi|_{\2, \star}$ appears - this is due to our assumption on $\Phi$ being measure-preserving. This is done in Section \ref{sec:MainTransport}.
	\item The $\Rem$ term, on the other hand, which comes from the spread-out masses, turns out to be trickier. We split it into three parts (Section \ref{sec:RemPrelim}) and analyze them separately in Sections \ref{sec:Rem1}, \ref{sec:Rem2}, \ref{sec:Rem3}.
	\item Finally, in Section \ref{sec:ConclusionMeasPres} we combine all the previous steps to prove Proposition \ref{prop:meas_pres}.
\end{enumerate}
}

\subsubsection{Transporting vector fields}
\label{sec:Transporting}
\begin{definition}[Transport of vector fields]
If $v$ is a vector field on $\R^2$ we define  $\PhiS v$, the “transport of $v$ by $\Phi$”, as:
\begin{equation}
\label{def:PhiS}
\PhiS v := \left(\Dd \Phi \circ \Phi^{-1} \right)^{T} v \circ \Phi^{-1}.
\end{equation}
\end{definition}
The point is that when $\dive v$ is a measure we have:
\begin{equation}
\label{transport_PhiS}
\dive \PhiS v = \PhiS \left( \dive v \right),
\end{equation}
where on the left-hand side there is a \emph{transport of vector fields} while the right-hand side is a \emph{push-forward of measures}. The identity \eqref{transport_PhiS} is a special case of \cite[Lemma A.3]{serfaty2020gaussian}. 

On top of it, we make the following simple observation:
\begin{claim}
\label{claim:transport_gradient}
If $v = \nabla h$ is a gradient, then:
\begin{equation*}
\PhiS \nabla h = \left( \Dd \Phi \circ \Phi^{-1} \right)^T \nabla\left(h \circ \Phi^{-1}\right) \left(\Dd \Phi \circ \Phi^{-1} \right).
\end{equation*}
\end{claim}
Thus is $v$ is a gradient and $\Phi$ is close to the identity map, then $\PhiS v$ is “almost” a gradient. The proof of Claim \ref{claim:transport_gradient} is straightforward using the definition \eqref{def:PhiS} and some calculus.

\subsubsection{Setting up the energy comparison}
\label{sec:applying_transport}
We want to compare $\FL(\bY)$ to $\FL(\bX)$. We start by recalling known expressions for both quantities. 

\textbf{\cor{Step 1. Expressing the energy in terms of electric fields.}}
Let $\hetaX, \hetaY$ be the true electric potentials generated by $\bX, \bY$ in $\La$ in the sense of Definition \ref{def:true}. As a truncation vector, let us choose\footnote{Such a choice of a very small truncation parameter appears in \cite{MR3788208,serfaty2020gaussian}. It might seem “unphysical” but it is very convenient to get rid of annoying error terms, while only costing $\log s$ in view of \eqref{eq:EnersEner}.} 
 \begin{equation}
 \label{eq:vetasvr}
\veta : = s \vr \text{ with } s = \epsilon^{2}
 \end{equation}
(in particular $s \leq \frac{1}{10}$), the distances $\vr$ being computed with respect to the configuration $\bX$. We recall that, by Claim \ref{claim:simple_facts}:
\begin{equation}
\label{yiyjxixj}
\hal |x_i - x_j| \leq |y_j - y_i| \leq 2 |x_i - x_j|,
\end{equation} 
and thus if we compute the nearest-neighbor distances $\vr$ with respect to $\bY$ instead of $\bX$ we still have $\veta \leq 2 s \vr \leq \frac{1}{5} \vr$. 

From Lemma \ref{lem:extendSecond} we know that:
\begin{equation}
\label{basic_comparison}
\FL(\bX) - \FL(\bY) = \frac{1}{4 \pi} \left(\int |\nhetaX|^2 -  \int |\nhetaY|^2\right),
\end{equation}
where $\h^{\bX} = \h^{\bX, \mm}$, $\h^{\bY} := \h^{\bY, \mm}$ (the background will be $\mm = \nn \muW$ everywhere and  we omit it) and $\veta$ is as in \eqref{eq:vetasvr}. 

\medskip

\textbf{\cor{Step 2. Introducing the transported fields.}}
We introduce two additional vector fields, using the notation of \eqref{def:PhiS} for the first one:
\begin{equation}
\label{new_vf}
\Elec_{\veta} := \PhiS \nhetaX, \quad \nhh := \nabla (- \log) \ast \left(\sum_{i=1}^\nn \PhiS \delta_{x_i}^{(\eta_i)} - \mm \right),
\end{equation}
(we recall that $\delta_{x_i}^{(\eta_i)}$ denotes the uniform measure of mass $1$ spread on the circle of center $x_i$ and radius~$\eta_i$). To summarize we have, besides $\nhetaX$ which is the true electric field generated by $\bX$, 
\begin{eqnarray*}
    & \nhetaY  =  \nabla (-\log) \star \left(\sum_{i=1}^\nn \delta_{\Phi(x_i)}^{(\eta_i)} - \mm  \right) \\
    & \nhh   = \nabla (- \log) \star \left(\sum_{i=1}^\nn \PhiS \delta_{x_i}^{(\eta_i)} - \mm \right) \\
    & \Elec_{\veta}  = \PhiS \nabla (-\log) \star \left(\sum_{i=1}^\nn \delta_{x_i}^{(\eta_i)} - \mm \right).
\end{eqnarray*}
There are subtle differences between these three vector fields:
\begin{itemize}
    \item $\nhetaY$ is a gradient, it is the true electric field generated by $\bY$, and in its divergence the charges are spread along a circle with centers $y_i = \Phi(x_i)$ ($1 \leq i \leq \nn$).
    \item $\nhh$ is also a gradient, but the charges are spread along deformed circles $\PhiS \delta_{x_i}^{(\eta_i)}$ which are (approximately) ellipses of center $y_i$.
    \item $\Elec_{\veta}$ is not a gradient in general, but it is obtained by transporting $\nhetaX$ (which is a gradient and for which the charges are spread along circles around the original points $x_i$'s) according to \eqref{def:Phis}.
\end{itemize}
By the identity \eqref{transport_PhiS}, $\nhh$ and $\Elec_{\veta}$ have the same divergence, and as in \cite[(A.23)]{serfaty2020gaussian} we get the decomposition:
\begin{equation*}
\int |\Elec_{\veta}|^2 = \int |\nhh|^2 + \int |\Elec_{\veta} - \nhh|^2.
\end{equation*}

\medskip

\textbf{\cor{Step 3. Bounding the difference $\Elec_{\veta} - \nhh$.}}
We control the second term in the right-hand side immediately. Using Claim \ref{claim:transport_gradient} we see that:
\begin{equation*}
\Elec_{\veta} = \left( \Dd \Phi \circ \Phi^{-1} \right)^T \nabla \left( \h_{\veta}^{\bX} \circ \Phi^{-1} \right) \left(\Dd \Phi \circ \Phi^{-1} \right), 
\end{equation*}
and thus since $\Phi = \Id + \psi$ with $|\psi|_{\1}$ smaller than $\frac{1}{5}$ we have the pointwise bound:
\begin{equation*}
| \Elec_{\veta} - \nabla \left( \h_{\veta}^{\bX} \circ \Phi^{-1} \right) | \leq \Cc \| \Dd \psi \circ \Phi^{-1} \| \times |\nhetaX \circ \Phi^{-1}|,
\end{equation*}
which implies (after integrating the previous inequality and changing variables by $\Phi$):
\begin{equation*}
\int |\Elec_{\veta} - \nabla \left( \h_{\veta}^{\bX} \circ \Phi^{-1} \right|^2 \leq \Cc \int |\psi|^2_{1, \star} |\nhetaX|^2.
\end{equation*}
This provides an upper bound on the $L^2$-distance between $\Elec_{\veta}$ and the space of gradients, and thus since $\nhh$ is its projection onto that space we get:
\begin{equation}
\label{eq:distance_projection}
\int |\Elec_{\veta} - \nhh|^2 \leq \Cc \int_\La |\psi|^2_{1, \star}(x) |\nhetaX|^2(x) \dd x.
\end{equation}
This is the first moment where we will use the notation and simple facts of Section \ref{sec:slowvartrans} combined with our $\WS$ assumption. First, replacing $|\psi|_{1, \star}$ by $\PsiU$ (as in \eqref{def:PsiUD}) provides an upper bound. Next, since $\PsiU$ has slow variations at scale $\ell$ (see \eqref{eq:Psi_slowvar_l}) and we are working on $\WS(\Lambda, \ell, K)$, we may apply the energy density upper bound \eqref{eq:EnerDensUB} (we will apply a similar chain of argument repeatedly in the rest of the proof). Here in conclusion, we have:
\begin{equation}
\label{first_example_psiPsi}
\int_\La |\psi|^2_{1, \star}(x) |\nhetaX|^2(x) \dd x \leq \int_\La \PsiU^2(x) |\nhetaX|^2(x) \dd x \preceq K \left( \int_\La \PsiU^2(x)  \dd x\right) (1 + |\log s|), 
\end{equation}
which finally implies, after a direct estimate of the $L^2$ norm of $\PsiU$ (see \eqref{def:PsiUD}), that:
\begin{equation}
\label{eq:distance_projection_2}
\int |\Elec_{\veta} - \nhh|^2 \leq \Cc K t^2 \epsilon (1 + |\log s|) = \O\left( K t^2 \epsilon \log \epsilon \right),
\end{equation}
because we have chosen $s = \epsilon^2$ in \eqref{eq:vetasvr}. 

\medskip

\textbf{\cor{Step 4. A first decomposition of the difference of energies}}
Going back to \cite[(A.25)]{serfaty2020gaussian}, and inserting \eqref{eq:distance_projection_2} we write\footnote{Let us observe that the additional term $\mathrm{Err}$ appearing in \cite[(A.23)]{serfaty2020gaussian} is $0$ in our case because $\Phi$ is measure-preserving and thus, with the notation of \cite{serfaty2020gaussian}, $\nu = \mu$.}:
\begin{equation}
\label{MainRemA}
\FL(\bY) - \FL(\bX) = \Main + \Rem + \O\left( K t^2 \epsilon \log \epsilon \right)
\end{equation}
where $\Main, \Rem$ are given by:
\begin{equation}
\label{MainRemB}
\Main := \frac{1}{4 \pi} \left( \int |\Elece|^2 - \int |\nhetaX|^2 \right), \quad \Rem := \frac{1}{4\pi} \left( \int |\nhetaY|^2 - \int |\nhh|^2 \right).
\end{equation}

\subsubsection{The Main term}
\label{sec:MainTransport}
For the term $\Main$, a direct expansion of $\Elece = \PhiS \nhetaX$ using the definition \eqref{def:PhiS} gives:
\begin{equation}
\label{analysis_main}
\Main = \frac{1}{4\pi} \int \left\langle \nhetaX, 2 \Dd \psi \nhetaX  \right\rangle + \int \O(|\psi|^2_{1, \star}) |\nhetaX|^2.
\end{equation}
This is consistent with \cite[(A.31)]{serfaty2020gaussian}, the improvement being that we have no contribution of the form $|\psi|_{L^{\infty}} |\psi|_{C^2}$ in the second order term thanks to the fact that $\Phi$ is measure-preserving. Let us also note that although $\Phi$ is measure-preserving, we may not have $\dive \psi = 0$, however it is true that $\dive \psi = \O\left(|\psi|_{1,\star}^2\right)$ pointwise and thus the $- \dive \psi$ term appearing in \cite[(A.32)]{serfaty2020gaussian} can be absorbed in our second order correction as in \eqref{analysis_main}. Arguing as in \eqref{first_example_psiPsi} we may re-write the error term in \eqref{analysis_main} and get:
\begin{equation}
\label{Main_2}
\Main = \frac{1}{4 \pi} \int \left\langle \nhetaX, 2 \Dd \psi \nhetaX  \right\rangle + \O\left( K t^2 \epsilon \log \epsilon \right).
\end{equation}
We thus see how the $H^1$ norm of $\psi$ appears as a second-order contribution to the energy change, and since by construction we have $\int |\psi|_{\1, \star}^2 = \O(\epsilon)$ we may indeed hope (if $\Main$ is indeed the “main” term) to have a \emph{small} energy cost.

\subsubsection{The Rem term}
\label{sec:RemPrelim}
The $\Rem$ term is due to the difference between the electric fields $\nhetaY$ (for which charges are spread along circles) and $\nhh$ (for which charges are spread along approximate ellipses). The fact that we can choose \emph{small} truncations via the parameter $s$ (we which recall to have chosen as $s = \epsilon^2$ in \eqref{eq:vetasvr}, we could even have used an higher power of $\epsilon$) will turn out to be crucial in order to control those errors. Let us keep the notation of \cite{serfaty2020gaussian} and use:
\begin{itemize}
\item  $\dyi$ to denote the charge spread uniformly on the circle of center $y_i = \Phi(x_i)$ and radius $\eta_i$,
\item $\dhi$ to denote the push-forward by $\Phi$ of the measure $\delta_{x_i}^{(\eta_i)}$. 
\end{itemize}
We also write $v_i$ for the function 
\begin{equation}
\label{def:vi}
v_i := - \log \star \left( \dhi - \dyi \right).
\end{equation}
As in \cite[(A. 41)]{serfaty2020gaussian} we decompose $\Rem$ as $\Rem_1 + \Rem_2 + \Rem_3$ and analyse each term separately.

\subsubsection{The \texorpdfstring{$\Rem_1$}{Rem1} term.}
\label{sec:Rem1}
 $\Rem_1$ is defined as:
\begin{equation*}
\Rem_1 = - \hal \sum_{i = 1}^\nn \int v_i \left( \dhi + \dyi \right).
\end{equation*}
We write as in \cite[(A.42)]{serfaty2020gaussian} (with $\vec{n}$ the unit normal vector to the circle)
\begin{equation*}
\Rem_1 = \hal \sum_{i = 1}^\nn \frac{1}{\eta_i} \dashint_{\partial \DD(y_i, \eta_i)} \left( \psi(y) - \psi(y_i) \right) \cdot \vec{n} + \O\left( \sum_{i=1}^\nn |\psi|^2_{\1, \loc}(y_i) \right).
\end{equation*}
Let us use $\PsiU$ as an upper bound to the derivative of $\psi$, we have:
\begin{equation*}
\sum_{i=1}^\nn |\psi|^2_{\1, \loc}(y_i) \preceq \sum_{i=1}^\nn \left(\PsiU \right)^2(x_i),
\end{equation*}
Moreover, a Taylor's expansion yields:
\begin{equation*}
 \psi(y) - \psi(y_i) = \Dd \psi(y_i) (y -y_i) + \O(|\psi|_{\2, \loc}(y_i) \eta_i^2),
\end{equation*}
but since $\Phi$ is measure-preserving we know that $\dive \psi(y_i) = \O(|\psi|^2_{\1, \star}(y_i))$, thus:
\begin{equation*}
\hal \sum_{i = 1}^\nn \frac{1}{\eta_i} \dashint_{\partial \DD(y_i, \eta_i)} \left( \psi(y) - \psi(y_i) \right) \cdot \vec{n} = \O(s \PsiD(x_i) + \PsiU^2(x_i)).
\end{equation*}
Since $\PsiU, \PsiD$ have slow variations and since we are working under the $\WS$ assumption we may compare sums to integrals as in \eqref{eq:SumIntegral}, hence:
\begin{equation*}
\sum_{i=1}^\nn \left(\PsiU \right)^2(x_i) + s \PsiD(x_i) \leq K \int \PsiU^2(x) \dd x + K t^2 s  = \O(K t^2 \epsilon).
\end{equation*}
We thus obtain:
\begin{equation}
\label{eq:Rem_1}
\Rem_1 = \O(K t^2 \epsilon).
\end{equation}

\renewcommand{\rho}{T}

\subsubsection{The $\Rem_2$ term.} 
\label{sec:Rem2}
$\Rem_2$ is defined as (see \eqref{def:vi})
\begin{equation*}
\Rem_2 := \hal \sum_{1 \leq i \neq j \leq \nn} \int v_i \left(\dhj - \dyj \right).
\end{equation*}
\cor{Our goal in this section will be to prove:
\begin{equation*}
\Rem_2 =  \O\left(s^3 K^2 \ell^2 \epsilon \log \epsilon t^2 \right) = \O\left(K^2 \epsilon \log \epsilon t^2\right),
\end{equation*}
}

We start by writing as in \cite[(A.45)]{serfaty2020gaussian}:
\begin{equation}
\label{preRem2}
\left| \sum_{1 \leq i \neq j \leq \nn} \int v_i \left(\dhj - \dyj \right)\right| \preceq \sum_{1 \leq i \neq j \leq \nn} \eta_i^2 \eta_j \left(  \frac{|\psi|^2_{\1, \loc}(x_i)}{|y_i - y_j|^3} + \frac{|\psi|_{\1, \loc}(x_i) |\psi|_{\2, \loc}(x_i)}{|y_i - y_j|^2} \right),
\end{equation}
and then proceed a bit differently. 

First, we use $\PsiU, \PsiD$ as upper bounds to the derivatives of $\psi$, use \eqref{yiyjxixj} and \eqref{eq:Psi_slowvar_} in order to replace the $y_i, y_j$'s in the right-hand side of \eqref{preRem2} by the corresponding $x_i, x_j$'s up to some multiplicative constant. Next, we decompose (as in \cite{serfaty2020gaussian}) the sum between contributions coming from “close” and “far away” pairs of points. 

\smallskip

\textbf{A. Distances smaller than $10 \ell$.}
Since we always have $|x_i - x_j|^3 \geq \rr_i^2 \rr_j$ and since we take the truncation $\eta_i = s \rr_i$ we may write for each fixed $i$
\begin{multline*}
\sum_{1 \leq j \leq \nn, j \neq i, |x_j - x_i| \leq 10 \ell} \eta_i^2 \eta_j \left(  \frac{\PsiU^2(x_i)}{|x_i - x_j|^3} + \frac{\PsiU(x_i) \PsiD(x_i)}{|x_i - x_j|^2} \right) \\
\leq s^3 \left( \PsiU^2(x_i) + \PsiU(x_i) \PsiD(x_i) \right) \times \#\left\lbrace j, |x_j-x_i| \leq 10 \ell \right\rbrace.
\end{multline*}
Since we condition on $\WS(\Lab, \ell, K)$ we may bound $\#\left\lbrace j, |x_j-x_i| \leq 10 \ell \right\rbrace$ by $\Cc K \ell^2$, and thus:
\begin{equation*}
s^3 \left( \PsiU^2(x_i) + \PsiU(x_i) \PsiD(x_i) \right) \times \#\left\lbrace j, |x_j-x_i| \leq 10 \ell \right\rbrace \preceq s^3 K \ell^2 \left( \PsiU^2(x_i) + \PsiU(x_i) \PsiD(x_i) \right).
\end{equation*}
Next we compare the sum (over $i$) of the previous quantity to an integral using \eqref{eq:SumIntegral} and get:
\begin{equation}
\label{Rem2close}
\sum_{1 \leq j \leq \nn, j \neq i, |x_j - x_i| \leq 10 \ell} \eta_i^2 \eta_j \left(  \frac{\PsiU^2(x_i)}{|x_i - x_j|^3} + \frac{\PsiU(x_i) \PsiD(x_i)}{|x_i - x_j|^2} \right) \leq \Cc s^3 K^2 \ell^2 \int \left( \PsiU^2(x) + \PsiU(x) \PsiD(x) \right) \dd x,
\end{equation}
and the right-hand side can be evaluated using \eqref{def:PsiUD}
\begin{equation}
\label{Rem2close_application}
s^3 K^2 \ell^2 \int \left( \PsiU^2(x) + \PsiU(x) \PsiD(x) \right) \dd x =\O\left( s^3 K^2 \ell^2  \epsilon t^2 \right).
\end{equation}

\begin{remark}
We do lose some information when replacing all distances by the smallest one over a large zone of size $\ell$, but for technical reasons it seemed hard to do much better, and it works well enough for us.
\end{remark}

\smallskip

\textbf{B. Distances larger than $10 \ell$.}
The function: $z \mapsto \frac{1}{|z|^3}$ “varies slowly at scale $\ell$” (in the sense of \eqref{slowlyvar}) on $\{|z| \geq 10 \ell\}$ thus using \eqref{eq:SumIntegral} for each fixed $i$ we can compare 
\begin{equation*}
\sum_{1 \leq j \leq \nn, j \neq i, |x_j - x_i| \geq 10 \ell} \frac{1}{|x_i - x_j|^3}
\end{equation*} 
to $K$ times the corresponding integral, namely: 
\begin{equation*}
\int_{\{|x - x_i| \geq 10 \ell\}} \frac{1}{|x_i - x|^3} \dd x = \O(1).
\end{equation*}
Similarly the sum $\sum_{1 \leq j \leq \nn, |x_j - x_i| \geq 10 \ell}\frac{1}{|x_i - x_j|^2}$ can be compared to:
\begin{equation*}
K \times \int_{\La \cap \{|x - x_i| \geq 10 \ell\}} \frac{1}{|x_i - x|^2} \dd x = K \times \O(\log T).
\end{equation*}
Using the obvious bound $\eta_i^2 \eta_j \leq s^3$ we thus obtain:
\begin{multline*}
A := \sum_{1 \leq j \leq \nn, |x_j - x_i| \geq 10 \ell} \eta_i^2 \eta_j \left(  \frac{|\psi|^2_{\1, \loc}(x_i)}{|x_i - x_j|^3} + \frac{|\psi|_{\1, \loc}(x_i) |\psi|_{\2, X}(x_i)}{|x_i - x_j|^2} \right) \\
\preceq s^3 K \sum_{1 \leq i \leq \nn} |\psi|^2_{\1, \loc}(x_i) + |\psi|_{\1, \loc}(x_i) |\psi|_{\2, \loc}(x_i) \log T. 
\end{multline*}
Using again $\PsiU, \PsiD$ instead and comparing again the sum to an integral, we obtain:
\begin{equation}
\label{rem2_far}
A \preceq s^3 K^2 \left( \int \PsiU^2(x) + \log T \PsiU(x) \PsiD(x)  \dd x\right),  
\end{equation}
which can be evaluated using \eqref{def:PsiUD} (and \eqref{eq:ellrhoepsilon}):
\begin{equation}
\label{eq:Rem2far_application}
A \preceq s^3 K^2 \left( \int \PsiU^2(x) + \log \rho \PsiU(x) \PsiD(x)  \dd x\right) \preceq  s^3 K^2 \left(\epsilon t^2 + \epsilon^2 t^2 \log T \right) = \O\left(s^3 K^2 \epsilon \log \epsilon t^2 \right)
\end{equation}

\smallskip

\textbf{C. Conclusion for $\Rem_2$.} Combining \eqref{Rem2close_application} and \eqref{eq:Rem2far_application} and discarding negligible terms we get, as desired:
\begin{equation}
\label{Rem2}
\Rem_2 =  \O\left(s^3 K^2 \ell^2 \epsilon \log \epsilon t^2 \right) = \O\left(K^2 \epsilon \log \epsilon t^2\right),
\end{equation}
where we have used \eqref{eq:ellrhoepsilon} to simplify the expression.

\subsubsection{The \texorpdfstring{$\Rem_3$}{Rem3} term.} 
\label{sec:Rem3}
$\Rem_3$ is defined as ($\mm$ is outside of the sum over $j$.):
\begin{equation*}
\Rem_3 := - \sum_{1 \leq i \leq \nn} \int_\La v_i \left( \sum_{1 \leq j \leq \nn, j \neq i} \dhj - \mm \right).
\end{equation*}
\cor{Our goal in this section is to prove the following estimate:
\begin{equation*}
\Rem_3 = \sum_{1 \leq i \leq \nn} \dashint_{|u| = \eta_i} \nabla \hciX(x_i + u) \cdot \left(\psi(x_i + u) - \psi(x_i)\right) \dd u +\O\left(K^2 \epsilon \log \epsilon t^2 \right). 
\end{equation*}
This is the most cumbersome part of the proof, we split it into the following five steps (plus the conclusion).
}

\textbf{\cor{Step 1: Setting up the computation.}}
Writing $\dhj - \mm = \dhj - \dyj + \dyj - \mm$ and recalling that $v_i := - \log \star \left( \dhi - \dyi \right)$, one can express $\Rem_3$ (as in \cite[Substep (5.3)]{serfaty2020gaussian}) as:
\begin{equation*}
\Rem_3 = \sum_{1 \leq i \leq \nn} \int \hci \left(\dyi - \dhi \right).
\end{equation*} 
We recall that the truncated field $\h_{\veta}$ coincides with $\tilde{\h}_i$ inside the $i^{th}$ spread-out charge.

The analysis of \cite[Step 4.]{serfaty2020gaussian} shows that, for each $i$ we have, as in \cite[Substep 5.3]{serfaty2020gaussian}:
\begin{multline}
\label{Rem3v1}
\int \hci \left(\dyi - \dhi \right) = \dashint_{|u| = \eta_i} \nabla \hci(y_i + u) \cdot \left(\psi(x_i + u) - \psi(x_i)\right) \dd u \\ + \O\left(|\hci|_{C^2(B(y_i, 2\eta_i))} \eta_i^2 |\psi|^2_{\1, \loc}(x_i) \right).
\end{multline}
Here and below we temporarily borrow the notation of \cite{serfaty2020gaussian} for local controls on derivatives of $\hci$, namely:
\begin{equation*}
|\hci|_{C^\kk(B(x,r))} = L^{\infty} \text{ norm of the $\kk$-th derivative of $\hci$ over the ball/disk of center $x$ and radius $r$}.
\end{equation*}
We cannot use \eqref{Rem3v1} as such because of the various dependencies in $\Y$, which we need to analyse.
\medskip

\textbf{\cor{Step 2: First preliminary claim: Variation of the first derivatives.}}
\cor{We study how the first derivative behaves under the transport.}
\begin{claim}
\label{claim:d1hci}[Variation of the first derivatives]
For all $a \in \{1, 2\}$, $1 \leq i \leq \nn$ and $|u| \leq \eta_i$ we have:
\begin{equation*}
\left| \partial_a \hci(y_i + u) - \partial_a \hciX(x_i + u) \right| \preceq  K \left( \frac{1}{\rr_i} \PsiU(x_i) \tell^2 +  |\psi|_\0 \log \rho \right)
\end{equation*}
\end{claim} 
\begin{proof}[Proof of Claim \ref{claim:d1hci}]
We have by definition:
\begin{multline*}
\partial_a \hci(y_i + u) - \partial_a \hciX(x_i + u) = \sum_{1 \leq j \leq \nn, j \neq i} \partial_a (-\log) (y_j - (y_i +u)) - \partial_a(-\log)(x_j - (x_i +u))  \\
- \partial_a \VLa(y_i + u) + \partial_a \VLa(x_i+u),
\end{multline*}
where $\VLa$ is the logarithmic potential generated by the background measure on $\La$, whose expression is given in \eqref{def:VLa}. Since $y_i = x_i + \psi(x_i)$ the difference $\partial_a \VLa(y_i + u) - \partial_a \VLa(x_i+u)$ is easily bounded by $|\psi|_\0$. We now focus on the contribution coming from the point particles. We have:
\begin{equation*}
\left| \partial_a \hci(y_i + u) - \partial_a \hciX(x_i + u) \right| \leq \sum_{1 \leq j \leq \nn, j \neq i} \frac{|\psi(x_i) - \psi(x_j)|}{|x_i - x_j|^2}. 
\end{equation*}
Let us split the sum into two parts corresponding to “close” and “far away” pairs of points.

\smallskip

\textbf{Distances smaller than $10 \ell$.}
On the one hand we have:
\begin{equation*}
\sum_{j, |x_j - x_i| \leq 10 \ell, j \neq i} \frac{|\psi(x_i) - \psi(x_j)|}{|x_i - x_j|^2} \preceq \frac{1}{\rr_i} \PsiU(x_i) \times  \#\left\lbrace j, |x_j-x_i| \leq 10 \ell \right\rbrace \preceq K \frac{1}{\rr_i} \PsiU(x_i) \ell^2.
\end{equation*}
To be precise, here we have used a mean value argument to argue that:
\begin{equation*}
\frac{|\psi(x_i) - \psi(x_j)|}{|x_i - x_j|} \leq \sup_{x \in [x_i, x_j]} |\psi|_{\1, \star}(x),  
\end{equation*}
then we wrote $|\psi|_{\1, \star}(x) \leq \PsiU(x)$ and finally we used \eqref{eq:Psi_slowvar_l}. Then we applied the local control on the number of points implied by $\WS(\Lab, \ell, K)$.

\smallskip

\textbf{Distances larger than $10 \ell$.}
On the other hand:
\begin{equation*}
\sum_{|x_j - x_i| \geq 10 \ell} \frac{|\psi(x_i) - \psi(x_j)|}{|x_i - x_j|^2} \preceq K |\psi|_\0 \int_{|x - x_i| \geq 10 \ell, x \in \La} \frac{1}{|x-x_i|^2} \dd x \preceq K |\psi|_\0 \log \rho.
\end{equation*}
This time we simply bounded $|\psi(x_i) - \psi(x_j)|$ by $2 |\psi|_\0$ and used the fact that $z \mapsto \frac{1}{|z|^2}$ “varies slowly at scale $\ell$” (in the sense of \eqref{slowlyvar}) on $\{|z| \geq 10 \ell\}$ in order to compare the sum to an integral as in \eqref{eq:SumIntegral}.

Combining those estimates proves the claim.
\end{proof}

\textbf{\cor{Step 3: Second preliminary claim: Variation of the second derivatives.}}
\cor{Similarly, we study how the \emph{second} derivative behaves under the transport.}
\begin{claim}[Variation of the second derivatives]
\label{claim:d2hci}
For all $a, b \in \{1, 2\}$, $1 \leq i \leq \nn$ and $|u| \leq \eta_i$ we have:
\begin{equation*}
\left| \partial_{ab}\hci(y_i + u) - \partial_{ab} \hciX(x_i + u) \right| \preceq K \left(\frac{1}{\rr_i^2} \tell^2 \PsiU(x_i) + \PsiU(x_i) \log T  + |\psi|_{\0} \frac{1}{\tell + |x_i|}\right).
\end{equation*}
\end{claim} 
\begin{proof}[Proof of Claim \ref{claim:d2hci}]
The proof is similar to Claim \ref{claim:d1hci}. We have by definition:
\begin{multline*}
\partial_{ab} \hci(y_i + u) - \partial_{ab} \hci(x_i + u) = \sum_{1 \leq j \leq \nn, j \neq i} \partial_{ab} (-\log) (y_j - (y_i +u)) - \partial_{ab}(-\log)(x_j - (x_i +u))\\
- \partial_{ab} \VLa(y_i + u) + \partial_{ab} \VLa(x_i+u),
\end{multline*}
however the second derivatives of $\VLa$ are \emph{constant} (see \eqref{def:VLa}) so we can discard those terms.
Writing $y = \Phi(x) = x + \psi(x)$, we get:
\begin{equation*}
\left| \partial_{ab} \hci(y_i + u) - \partial_{ab} \hciX(x_i + u) \right| \leq \sum_{1 \leq j \leq \nn, j \neq i} \frac{|\psi(x_i) - \psi(x_j)|}{|x_i - x_j|^3}. 
\end{equation*}
We now split the sum into three parts: $|x_j - x_i| \leq \tell$, $|x_j - x_i| \leq \hal |x_i|$ and $|x_j - x_i| \geq \left(\tell \cup \hal |x_i| \right)$. 

\textbf{Distances smaller than $10 \ell$.}
Arguing as in the proof of Claim \ref{claim:d1hci} we get:
\begin{equation*}
\sum_{1 \leq j \leq \nn, j \neq i} \frac{|\psi(x_i) - \psi(x_j)|}{|x_i - x_j|^3} \preceq  K \frac{1}{\rr_i^2} \tell^2 \PsiU(x_i).
\end{equation*}

\smallskip

\textbf{Distances between $10 \ell$ and $\hal |x_i|$.}
For $10 \ell \leq |x_j - x_i| \leq \hal |x_i|$ we write by a mean value argument:
\begin{equation*}
\frac{|\psi(x_i) - \psi(x_j)|}{|x_i - x_j|^3} \leq \sup_{x \in [x_i, x_j]} |\psi|_{\1, \star}(x) \times \frac{1}{|x_j - x_i|^2}, 
\end{equation*}
then we may again replace $|\psi|_{\1, \star}(x)$ by $\PsiU(x)$ and use property \eqref{eq:Psi_slowvar_} to bound it by $\PsiU(x_i)$ up to some multiplicative constant. Next, comparing a sum to an integral, we have:
\begin{equation*}
\sum_{j, |x_j - x_i| \geq 10 \ell} \frac{1}{|x_i - x_j|^2} \preceq K \log \rho.
\end{equation*}
In conclusion, we get for fixed $i$:
\begin{equation*}
\sum_{1 \leq j \leq \nn, j \neq i, 10 \ell \leq |x_j - x_i| \leq \hal |x_i|} \frac{|\psi(x_i) - \psi(x_j)|}{|x_i - x_j|^3} \leq K \PsiU(x_i) \log \rho.
\end{equation*}

\textbf{Large distances.}
To estimate the remaining contribution, we write:
\begin{equation*}
\sum_{1 \leq j \leq \nn, |x_j - x_i| \geq \max(10 \ell,  \hal |x_i|)} \frac{|\psi(x_i) - \psi(x_j)|}{|x_i - x_j|^3} \preceq K |\psi|_\0 \int_{|x - x_i|  \geq \max(10 \ell,  \hal |x_i|)} \frac{1}{|x-x_i|^3} \dd x \preceq K |\psi|_\0 \frac{1}{|x_i| + \tell},
\end{equation*}
where we compared the sum to an integral using \eqref{eq:SumIntegral}.

Combining all three estimates proves the claim.
\end{proof}

\begin{remark}
In Claims \ref{claim:d1hci} and \ref{claim:d2hci} we have studied the contributions coming from the point particles and the background separately. Taking cancellations between those terms into account would yield more accurate estimates, but we do not need them here.
\end{remark}
\medskip 

\textbf{Step 4: Studying the first-order term.} 
Recall that we are still trying to express $\Rem_3$ purely in terms of the original points $(x_1, \dots, x_\nn)$.

Going back to \eqref{Rem3v1}, we write the first-order term as:
\begin{multline*}
\dashint_{|u| = \eta_i} \nabla \hci(y_i + u) \cdot \left(\psi(x_i + u) - \psi(x_i)\right) \dd u = \dashint_{|u| = \eta_i} \nabla \hciX(x_i + u) \cdot \left(\psi(x_i + u) - \psi(x_i)\right) \dd u \\ 
+ \dashint_{|u| = \eta_i} \left( \nabla \hci(y_i + u) - \nabla \hciX(x_i + u) \right) \cdot \left(\psi(x_i + u) - \psi(x_i)\right) \dd u. 
\end{multline*}
We keep the first term in the right-hand side as such and we focus on the second one, which we decompose as:
\begin{multline}
\label{Rem3vA}
\dashint_{|u| = \eta_i} \left( \nabla \hci(y_i + u) - \nabla \hciX(x_i + u) \right) \cdot \left(\psi(x_i + u) - \psi(x_i)\right) \dd u \\ 
= 
\dashint_{|u| = \eta_i} \left( \nabla \hci(y_i) - \nabla \hciX(x_i) \right) \cdot \left(\psi(x_i + u) - \psi(x_i)\right) \dd u \\
+ 
\dashint_{|u| = \eta_i} \left( \left( \nabla \hci(y_i + u) - \nabla \hciX(x_i + u) \right) - \left( \nabla \hci(y_i) - \nabla \hciX(x_i) \right) \right) \cdot \left(\psi(x_i + u) - \psi(x_i)\right) \dd u.
\end{multline}

Let us also write that $\psi(x_i + u) - \psi(x_i) = \Dd \psi(x_i) \times u + \O\left(|\psi|_{\2, \loc}(x_i) \eta_i^2 \right)$, and  observe that according to Claim \ref{claim:d2hci} we have (for $|u| = \eta_i$)
\begin{multline*}
\left|\left( \nabla \hci(y_i + u) - \nabla \hciX(x_i + u) \right) - \left( \nabla \hci(y_i) - \nabla \hciX(x_i) \right)  \right| \\
\preceq \eta_i \times K \left(\frac{1}{\rr_i^2} \tell^2 \PsiU(x_i) + \PsiU(x_i) \log T  + |\psi|_{\0} \frac{1}{\tell + |x_i|}\right).
\end{multline*}

We may thus re-write the right-hand side of \eqref{Rem3vA} as:
\begin{multline*}
\dashint_{|u| = \eta_i} \left( \nabla \hci(y_i) - \nabla \hciX(x_i) \right) \cdot \Dd \psi(x_i) u \dd u + \dashint_{|u| = \eta_i} \left( \nabla \hci(y_i) - \nabla \hciX(x_i) \right) \times \O\left(|\psi|_{\2, \loc}(x_i) \eta_i^2 \right) \dd u \\
+ 
\O\left(\eta_i \times K \left(\frac{1}{\rr_i^2} \tell^2 \PsiU(x_i) + \PsiU(x_i) \log T  + |\psi|_{\0} \frac{1}{\tell + |x_i|}\right) \right) \times |\psi|_{\1, \loc}(x_i) \times \eta_i. 
\end{multline*}
The first term vanishes by symmetry, and we can bound the second term further using Claim \ref{claim:d1hci}. In conclusion, we obtain (using that $\eta_i \leq s \rr_i \leq s$):
\begin{multline}
\label{rem3FirstOrder}
\dashint_{|u| = \eta_i} \nabla \hci(y_i + u) \cdot \left(\psi(x_i + u) - \psi(x_i)\right) \dd u = \dashint_{|u| = \eta_i} \nabla \hciX(x_i + u) \cdot \left(\psi(x_i + u) - \psi(x_i)\right) \dd u \\ 
+ K \left( s^2 \PsiU(x_i) \PsiD(x_i) \ell^2 +  s^2 |\psi|_\0 \PsiD(x_i) \log T + s^2 \PsiU^2(x_i) \left(\tell^2 + \log T\right) + s^2 |\psi|_{\0} \PsiU(x_i) \frac{1}{\tell + |x_i|}\right).
\end{multline}

Summing the error term in \eqref{rem3FirstOrder} over $i$ yields:
\begin{multline}
\label{sum_error_First_order}
\sum_{i=1}^\nn K \left( s^2 \PsiU(x_i) \PsiD(x_i) \ell^2 +  s^2 |\psi|_\0 \PsiD(x_i) \log T + s^2 \PsiU^2(x_i) \left(\tell^2 + \log T\right) + s^2 |\psi|_{\0} \PsiU(x_i) \frac{1}{\tell + |x_i|}\right) 
\\ \leq K^2 \left( s^2 \ell^2  \int \PsiU(x) \PsiD(x) \dd x +  s^2 |\psi|_\0 \log T \int \PsiD(x)  \dd x + s^2 \left(\tell^2 + \log T\right) \int \PsiU^2(x) \dd x  \right. \\ \left. + s^2 |\psi|_{\0} \int \PsiU(x) \frac{1}{\ell + |x|} \dd x\right) 
\\
\leq K^2 \left(s^2 \ell^2  \epsilon^2 t^2  +  s^2 t^2 \log T + s^2 \left(\tell^2 + \log T\right) \epsilon t^2 + s^2 t^2 \right) = \O\left(K^2 \epsilon t^2\right),
\end{multline}
where we have used \eqref{eq:ellrhoepsilon} to simplify the expression involving $T, \ell, \epsilon$ and $s$ (we recall that $s = \epsilon^2$).
\medskip 

\textbf{Step 5: Re-writing the error term.} The error term in \eqref{Rem3v1} involves $|\hci|_{C^2(B(y_i, 2\eta_i))}$, which is expressed in terms of the \emph{transported} points and thus remains an issue for us. Using Claim \ref{claim:d2hci} however, we see that:
\begin{equation}
\label{hciB}
|\hci|_{C^2(B(y_i, 2\eta_i))} \leq |\hciX|_{C^2(B(x_i, 2\eta_i))} + K \O\left(\frac{1}{\rr_i^2} \ell^2 |\psi|_{\1, \loc}(x_i) + |\psi|_{\1, \loc}(x_i) \log \rho   + |\psi|_{\0} \frac{1}{\ell + |x_i|}\right).
\end{equation}
Now, the analysis of \cite[Lemma A.2]{serfaty2020gaussian} gives:
\begin{equation}
\label{eq:hciXCk}
|\hciX|_{C^2(B(x_i, 2\eta_i))} \preceq \frac{1}{\rr_i^2} \left(1 + \int_{\DD(x_i, \rr_i)} |\nabla \hciX|^2 \right).
\end{equation}
\begin{remark}
The proof of \eqref{eq:hciXCk} uses the fact that $\hciX$ is almost harmonic on the disk $\DD(x_i, \rr_i)$, and is simplified by the fact that the background measure ($\mu$ in the notation of \cite{serfaty2020gaussian}) is here constant (beware: when reading \cite[(A.5), (A.6)]{serfaty2020gaussian}, our background measure corresponds to $N \mu$ and not $\mu$ - compare \cite[(3.1)]{serfaty2020gaussian} with our Definition \ref{def:true}. Also, the nearest-neighbor distances in \cite{serfaty2020gaussian} are of order $N^{-1/2}$ where ours are of order one, this is due to a different choice of scaling, but \cite[Lemma A.2]{serfaty2020gaussian} is valid in any case).
\end{remark}

Combining \eqref{hciB} and \eqref{eq:hciXCk} the error term appearing in \eqref{Rem3v1} can be re-written as:
\begin{multline*}
|\hci|_{C^2(B(y_i, 2\eta_i))} \eta_i^2 |\psi|^2_{\1, \loc}(x_i) \preceq \PsiU^2(x_i) \frac{\eta_i^2}{\rr_i^2} \left(1 + \int_{\DD(x_i, \rr_i)} |\nabla \hciX|^2 \right) \\
+ K \PsiU^2(x_i) \eta_i^2 \left( \frac{1}{\rr_i^2} \tell^2 \PsiU(x_i) + \PsiU(x_i) \log \rho  + |\psi|_{\0} \frac{1}{\ell + |x_i|}  \right),
\end{multline*}

 Summing over $1 \leq i \leq \nn$, using the $\WS$ assumption, comparing to an integral and using again that $\eta_i \leq s \rr_i$ we obtain:
\begin{multline}
\label{RemErrorNew}
\sum_{1 \leq i \leq \nn} |\hci|_{C^2(B(y_i, 2\eta_i))} \eta_i^2 |\psi|^2_{\1, \loc}(x_i) 
\\ \preceq K (1 + |\log s|)  s^2 \int \PsiU^2(x) \dd x + K^2 s^2 \left( \left(\tell^2 + \log \rho \right) \int_\La \PsiU^3(x)  \dd x + |\psi|_\0 \int_\La \PsiU^2(x) \frac{1}{\tell + |x|} \dd x \right)
\end{multline}
Estimating everything explicitly and keeping only the dominant term we may thus control the error term in \eqref{Rem3v1} by:
\begin{equation}
\label{RemErrorNew_application}
\sum_{1 \leq i \leq \nn} |\hci|_{C^2(B(y_i, 2\eta_i))} \eta_i^2 |\psi|^2_{\1, \loc}(x_i)  = \O\left(K^2 (1 + \log |s|)  s^2 \ell^2 \epsilon t^2 \right) = \O\left(K^2 \epsilon \log \epsilon t^2\right),
\end{equation}
where we have used that $s = \epsilon^2$ and \eqref{eq:ellrhoepsilon} again.

\medskip 

\textbf{Step 6: Conclusion for $\Rem_3$.}
In conclusion we obtained that:
\begin{equation}
\label{eq:Rem3Conclusion}
\Rem_3 = \sum_{1 \leq i \leq \nn} \dashint_{|u| = \eta_i} \nabla \hciX(x_i + u) \cdot \left(\psi(x_i + u) - \psi(x_i)\right) \dd u +\O\left(K^2 \epsilon \log \epsilon t^2 \right). 
\end{equation}

\subsubsection{Concluding the proof of Proposition \ref{prop:meas_pres}}
\label{sec:ConclusionMeasPres}
Combining all the previous steps, we find that:
\begin{equation*}
\FL(\PhiX \bX) = \FL(\bX) + \AUN(\bX, \psi) + \O\left(K^2 \epsilon \log \epsilon t^2 \right),
\end{equation*}
where the “anisotropy” term $\AUN$ is defined as the sum of the linear (in $\psi$) terms obtained in $\Main$ and $\Rem_3$, namely:
\begin{equation}
\label{def:Aniv2}
\AUN(\bX, \psi) := \frac{1}{4\pi} \int \left\langle \nhetaX, 2 \Dd \psi \nhetaX \right\rangle + \sum_{1 \leq i \leq \nn} \dashint_{|u| = \eta_i} \nabla \hciX(x_i + u) \cdot \left(\psi(x_i + u) - \psi(x_i)\right) \dd u.
\end{equation}
There was a linear term appearing in $\Rem_1$ but it was found to be negligible, and comparing \eqref{def:Aniv2} to \cite[(4.8)]{serfaty2020gaussian}, the reader might observe that there is another term missing (the last term in \cite[(4.8)]{serfaty2020gaussian}), in fact for us this term is $\O\left(K \epsilon t^2 \right)$ and can thus be incorporated in the error term. This is due to the fact that our $\Phi$ is measure-preserving.

This concludes the proof of Proposition \ref{prop:meas_pres}.
\end{proof}

\newcommand{\tpsii}{\tilde{\psi}^{(i)}}
\subsection{Smallness of the anisotropy}
\label{sec:smallani}
Applying the result of Proposition \ref{prop:meas_pres} to $\Phit$ and $\Phimt$ we obtain that if $\bX$ is in $\WS(\Lab, \ell, K)$ we have:
\begin{equation}
\label{avant_decontrolerAni}
\hal \left(\FL(\Phit \bX, \mm) + \FL(\Phimt \bX, \mm) \right) - \FL(\bX, \mm) = \hal \AUN[\bX, \Psit + \Psimt] + \O\left(K^2 \epsilon \log \epsilon t^2 \right).
\end{equation}
Let us decompose $\Psit, \Psimt$ as in \eqref{redefPsit}, \eqref{defrt}. The terms of first order in $t$ cancel each other, and it remains to bound $\AUN[\bX, \rt + \rmt]$.  At this point, using the “rough” bounds of \cite[Prop. 4.2]{armstrong2019local} on the anisotropy~$\AUN$, even in a localized way, would yield a bounded, \emph{but not small}, error term - which would make the whole approach pointless. We thus need to rely on a finer understanding of anisotropy terms as put forward in \cite{MR3788208,serfaty2020gaussian} (see also \cite{MR4063572} for similar concerns about their “angle term”).

For simplicity, we will focus on the context of interest for us, namely when the background measure $\mm$ is given by $\nn \muW$ as above - in particular it is constant on $\Lab$.

\subsubsection*{Anisotropy.} 
Let $\psi$ be a continuous vector field supported on $\Lab$ and let $U$ be a neighborhood of $\supp \psi$. We define $\AUN(\bX, \mm, \psi)$ as:
\begin{equation}
\label{def:AUN}
\AUN(\bX, \mm, \psi) := \iint_{(x,y) \in \La \times \La, x \neq y} \psi(x) \cdot \nabla \log|x-y| \dd \left( \bX_n - \mm \right)(x) \dd \left( \bX - \mm \right)(y).
\end{equation}
(It has an alternative expression using the electric field, which is the one that we used above in \eqref{def:Aniv2}, see \cite[(4.14)]{serfaty2020gaussian}). The expression \eqref{def:AUN} makes sense because $\mm$ has a continuous density near $\supp \psi$ and is such that $\iint -\log|x-y| \dd \mm(x) \dd \mm(y)$ is finite. Thanks to a clever integration by parts, one can control $\AUN$ as follows (see \cite[(4.10)]{serfaty2020gaussian}):
\begin{equation}
\label{eq:control_AUN}
\AUN(\bX, \mm, \psi) \leq \Cc |\psi|_{\1} \times \left( \EnerPoints(\bX, \supp \psi) \right).
\end{equation} 
If $\psi$ lives on a disk of radius $\ell$, if $|\psi|_{\1} = \O(\ell^{-2})$ and if local laws hold then we can expect the anisotropy to typically be $\O(1)$. Let us now explain how $\AUN$ shows up in the computations.

\subsubsection*{Energy comparison.}
\newcommand{\Cpsi}{\Cc_\psi}
Assume that $\psi$ is a vector field of class $C^3$, supported on a disk of radius $\ell$ within $\Lab$, and such that:
\begin{equation}
\label{eq:assumpsi}
|\psi|_{\kk} \leq \Cpsi \ell^{-\kk - 1}, \text{ for } \kk = 1, 2, 3.
\end{equation}
Let $\tau$ be a real parameter such that (for some universal $\Cc$ large enough):
\begin{equation}
\label{condi_sur_tau}
|\tau| |\psi|_\1 \leq \frac{1}{\Cc}, \quad |\tau| |\psi|_{\2} \leq \log \ell \frac{1}{\Cc}
\end{equation}
For all $\tau$ such that \eqref{condi_sur_tau} is satisfied, let $\Phi_\tau := \Id + \tau \psi$ and $\mm_\tau := \Phi_\tau \# \mm$.
\begin{lemma}
\label{lem:ComparisonTransport}
We have the following expansion for all $\bX$
\begin{equation}
\label{eq:FsousPhit}
\FL(\Phi_\tau \cdot \bX, \mm_\tau) = \FL(\bX_\nn, \mm) + \tau \AUN(\bX, \mm, \psi) + \tau^2 \ErrorEnergy(\bX, \psi),
\end{equation}
where $\ErrorEnergy(\bX, \psi)$ is controlled by:
\begin{equation}
\label{eq:ErrorEnergy}
\ErrorEnergy(\bX, \psi) \leq \Cpsi^2 \frac{\log \ell}{\ell^4} \EnerPoints(\bX, \supp \psi),
\end{equation}
the energy being computed with $\mm$ as neutralizing background.
\end{lemma}
\begin{proof}[Proof of Lemma \ref{lem:ComparisonTransport}]
This follows from the second-order expansion of the energy as found in \cite[Lemma 4.1, Prop 4.2]{serfaty2020gaussian}. There is some care required in order to check that \cite[(4.11)]{serfaty2020gaussian} does indeed yield the claimed second-order correction, but this is made simpler by our assumption \eqref{eq:assumpsi} and the fact that $\mm$ is constant on $\Lab$.
\end{proof}
Since $\supp \psi$ has volume $\O(\ell^2)$, in view of the local laws we expect $\ErrorEnergy$ to be $\O(\ell^{-2} \log \ell)$ (with a constant depending on $\Cpsi$). The anisotropy is thus the first-order contribution to the energy change induced by a transport which is a small perturbation of the identity map.

\paragraph{$\AUN$ versus $\Ani$.}
What is called the “anisotropy” in \cite{MR3788208} and \cite{serfaty2020gaussian} is not exactly the same term, in fact \cite{MR3788208} refers to $\Ani$ and \cite{serfaty2020gaussian} to $\AUN$, where the latter is as defined above and the former corresponds to:
\begin{equation}
\label{def:Ani}
\Ani(\bX, \mm, \psi) : = \AUN(\bX, \mm, \psi) - \frac{1}{4} \int \Div \psi(x) \dd \bX(x).
\end{equation}
So in fact “$\AUN$” contains a possibly non-vanishing contribution $\frac{1}{4} \int \Div \psi(x) \dd \bX(x) \approx \frac{1}{4} \int \Div \psi(x) \dd \mm(x)$ which we need to substract in order to obtain “$\Ani$” which is the term that will eventually be found to be negligible. The term appearing in the energy expansion \eqref{eq:FsousPhit} is $\AUN$. 

Both $\AUN$ and $\Ani$ satisfy the same control \eqref{eq:control_AUN}. In most relevant cases, $\psi$ happens to be such that $|\psi|_{\1} \times |\supp \psi| = \O(1)$. Since (by local laws) the electric energy is typically proportional to the volume, we deduce that (for such “usual” $\psi$'s) \emph{the anisotropy $\Ani$ is (at most) of order $1$}. Analytically speaking, i.e. as far as deterministic, function-wise bounds are concerned, it is very challenging to do better. However, a key result underlying some of the recent progress in the study of 2DOCP's is that $\Ani$ is, so to speak, \emph{often smaller than it seems}. This is a \emph{probabilistic} statement found in \cite[Corollary 4.4]{MR3788208}, \cite[Lemma 7.2.]{serfaty2020gaussian}, see also (with a different formalism) \cite[Section 8.5]{MR4063572}. Let us sketch the proof of this fact.

\subsubsection*{Smallness of the anisotropy via “Serfaty's trick”.}
Let $\KnLL(\mm_\tau)$ be the partition function associated to the background $\mm_\tau$ (we keep the same “effective confinement” $\zetaW$ for all $\tau$) namely:
\begin{equation*}
\KnLL(\mm_\tau) := \int_{\La^\nn} \exp\left( - \beta \left(\FL(\bXn, \mm_\tau)  +  \nn \sum_{i=1}^\nn \zetaW(x_i) \right) \right) \dd \X_\nn.
\end{equation*}
The key point is that there are two ways to evaluate the ratio $\frac{\KnLL(\mm_\tau)}{\KnLL(\mm)}$:
\begin{enumerate}
    \item By the transportation approach of \cite{MR3788208,serfaty2020gaussian}, involving a change of variables $(x_1, \dots, x_N) = \left(\Phi_\tau(x'_1), \dots, \Phi_\tau(x'_N)\right)$ in the very definition of $\KnLL(\mm_\tau)$ and an analysis of its effect on the energy. As seen in Lemma \ref{lem:ComparisonTransport}, the anisotropy of $\psi$ appears there as one of the contributions.
    \item By using “free energy expansions”, i.e. explicit expressions of (the logarithm of) the two partition functions up to some error term that has to be negligible. This was done in \cite{MR3788208} with a non-quantitative error term originating in the analysis of \cite{MR3735628}, in \cite{MR4063572} using their own expansion, and much improved in \cite{serfaty2020gaussian} using the error bounds of \cite{armstrong2019local}. 
\end{enumerate}
This gives two expressions for the same quantity, and since the anisotropy appears only (in exponential moments) in the first one, then it must be confounded with some error terms of the second one. This is fruitful for $\tau$ large (but not too large), and thus also for smaller values of $\tau$ by Hölder's inequality.

This “trick”, which yields a form of “smallness of the anisotropy”, is used in \cite{MR3788208,MR4063572,serfaty2020gaussian} as a tool to prove central limit theorems for fluctuations of smooth linear statistics. Unfortunately, it is hard to pinpoint a clear general statement in the literature, so we recall the proof in the next paragraphs. Recall that we take $\mm = \nn \muW$ as background measure, which has constant density $1$ in the bulk.

\newcommand{\muWz}{\muW^{(0)}}
\newcommand{\muWs}{\muW^{(s)}}
\newcommand{\tmuWs}{\tilde{\mu}^{(s)}_\WW}
\newcommand{\Cphi}{\Cc_\varphi}

\let\oldell\ell
\renewcommand{\ell}{\hat{\oldell}}

\paragraph{1. Comparison along a transport.}
Here for technical reasons we need to work with partition functions restricted to $\EE_\La$, and we write:
\begin{equation*}
\KnLL(\mm_\tau \big| \EE_\La) := \int_{\La^\nn} \ind_{\EE_\La}(\bX_\nn) \exp\left( - \beta \left(\FL(\bX_\nn, \mm_\tau)  +  \nn \sum_{i=1}^\nn \zetaW(x_i) \right) \right) \dd \X_\nn.
\end{equation*}
We have the following “comparison of partition functions”:
\begin{claim}
\label{claim:compa_transport}
\begin{multline}
\label{eq:partition_function_compa1}
\log \frac{\KnLL(\mm_\tau \big| \EE_\La)}{\KnLL(\mm \big| \EE_\La)} = \left(\frac{\beta}{4} - 1 \right) \left( \int \log \mm_\tau \dd \mm_\tau  - \int \log \mm \dd \mm  \right) \\ + \log \EnLL \left[\exp\left(\tau \Ani[\bX, \mm, \psi] + \left(\tau |\psi|_{\2} + \tau^2 |\psi|^2_\1 \right) \EnerPoints(\supp \psi) + \tau^2 \ErrorEnergy\right)  \big| \EE_\La \right].
\end{multline}
\end{claim}
\begin{proof}[Proof of Claim \ref{claim:compa_transport}]
We follow the same steps as in \cite[Prop 4.3]{MR3788208}. First we write $\KnLL(\mm_\tau \big| \EE_\La)$ as:
\begin{equation*}
\KnLL(\mm_\tau \big| \EE_\La) = \int_{\La^\nn} \ind_{\EE_\La}(\bY_\nn) \exp\left( - \beta \left(\FL(\bY_\nn, \mm_\tau) + \nn \int \zetaW(x) \dd \bY_\nn(x) \right) \right) \dd \Y_\nn
\end{equation*}
and perform the change of variables $\bY_\nn = \Phi_\tau \cdot \bX_\nn$. By construction we have $\Phi_\tau = \Id$ outside $\Lab$ and in particular the term $\nn \int \zetaW(x) \dd \bY_\nn(x)$ (which only detects points outside the support of $\muW$) is not affected by this, nor is the indicator $\EE_\La$ which only cares about points very close to $\partial \La$. We obtain:
\begin{multline}
\label{apres_transport}
\KnLL(\mm_\tau \big| \EE_\La) \\
= \int_{\EE_\La}  \exp\left( - \beta \left(\FL\left(\Phi_\tau \cdot \bX_\nn, \Phi_\tau \# \mm \right)  + \nn \int \zetaW(x) \dd \bX_\nn(x) \right) + \int \log \det \Dd \Phi_\tau(x) \dd \bX_\nn(x) \right) \dd \X_\nn,
\end{multline}
where the last term in the integrand is the Jacobian of the tranformation. Using \eqref{eq:FsousPhit} we may compare the energy before and after the transport. 
\begin{equation*}
\FL\left(\Phi_\tau \cdot \bX_\nn, \Phi_\tau \# \mm \right) = \FL\left(\bX_\nn, \mm \right) + \tau \AUN[\bX, \mm, \psi] + \tau^2 \ErrorEnergy,
\end{equation*}
which we can re-write (using \eqref{def:Ani}) as:
\begin{equation*}
\FL\left(\Phi_\tau \cdot \bX_\nn, \Phi_\tau \# \mm \right) = \FL\left(\bX_\nn, \mm \right) + \frac{\tau}{4} \int \Div \psi(x) \dd \bX_\nn(x) + \tau \Ani[\bX, \mm, \psi] + \tau^2 \ErrorEnergy.
\end{equation*}
The terms $\int \log \det \Dd \Phi_\tau(x) \dd \bX_\nn(x)$ and $\tau \int \Div \psi(x) \dd \bX_\nn(x)$ are both equal to a deterministic quantity up to small error terms. On the one hand, we have, \cor{with the notation of \eqref{def:Fluct} for fluctuations of linear statistics}:
\begin{equation*}
\int \log \det \Dd \Phi_\tau(x) \dd \bX_\nn(x) = \int \log \det \Dd \Phi_\tau(x) \dd x + \Fluct[\log \det \Dd \Phi_\tau]
\end{equation*}
and on the other hand: $\tau \int \Div \psi(x) \dd \bX_\nn(x) = \int \log \det \Dd \Phi_\tau(x) \dd \bX_\nn(x) + \tau^2 \sum_{i=1}^\nn |\psi(x_i)|^2_{\1}$. 
The quantity $\int \log \det \Dd \Phi_s(x) \dd x$ coincides (see. \cite[(4.11)--(4.13)]{MR3788208}) with:
\begin{equation}
\label{difference_entropies}
\int \log \det \Dd \Phi_s(x) \dd x = \int \log \mm \dd \mm  - \int \log \mm_\tau \dd \mm_\tau ,
\end{equation}
hence we obtain:
\begin{multline*}
\FL\left(\Phi_\tau \cdot \bX_\nn, \Phi_\tau \# \mm \right) = \FL\left(\bX_\nn, \mm \right) + \frac{1}{4} \left(\int \log \mm \dd \mm  - \int \log \mm_\tau \dd \mm_\tau\right) \\
+ \Fluct[\log \det \Dd \Phi_\tau] + \tau^2 \sum_{i=1}^\nn |\psi(x_i)|^2_{\1} + \tau^2 \ErrorEnergy.
\end{multline*}
Using Lemma \ref{lem:apriori} we can control: $\Fluct[\log \det \Dd \Phi_\tau] \preceq \tau |\psi|_{\2} \EnerPoints(\supp \psi)$ and on the other hand we have: $\tau^2 \sum_{i=1}^\nn |\psi(x_i)|^2_{\1} \leq \tau^2 |\psi|^2_{\1} \EnerPoints(\supp \psi)$, thus we can write:
\begin{multline*}
\FL\left(\Phi_\tau \cdot \bX_\nn, \Phi_\tau \# \mm \right) = \FL\left(\bX_\nn, \mm \right) + \frac{1}{4} \left(\int \log \mm \dd \mm  - \int \log \mm_\tau \dd \mm_\tau \right) \\
+ \left(\tau |\psi|_{\2} + \tau^2 |\psi|^2_\1 \right) \EnerPoints(\supp \psi) + \tau^2 \ErrorEnergy,
\end{multline*}
and inserting this in \eqref{apres_transport} yields \eqref{eq:partition_function_compa1}.
\end{proof}

\begin{remark}
\label{rem:difference_entropies}
Each term in the right-hand side of \eqref{difference_entropies} might be infinite when taken separately (because e.g. $\mm$ may have a singularity on $\partial \La$ and hence infinite entropy) but the difference makes sense as the two measures coincide outside $\Lab$.
\end{remark}

\newcommand{\tmuW}{\tilde{\mu}_\WW}
\newcommand{\KNeu}{\mathrm{K}^{\beta, \mathrm{Neu}}}
\newcommand{\tmm}{\tilde{\mm}}
\subsubsection*{Free energy comparisons}
As we have seen above, one can compare two partition functions using a “transportation” approach. On the other hand, we have the following.
\begin{claim}[Free energy comparison, the “direct” approach]
\label{claim:comparison_two}
Assume that the support of $\tmm - \mm$ is contained in a square $\Omega$ of sidelength $\ell$ included in $\Lab$.
\begin{equation}
\label{eq:compa_two_free}
\left| \log \frac{\KnL(\tmm \big| \EE_\La)}{\KnL(\mm \big| \EE_\La)} \right| = \left(\frac{\beta}{4} - 1 \right) \left( \int \log \tmm \dd \tmm - \int \log \mm \dd \mm  \right) + \O\left( \ell \log \ell \right)
\end{equation}
\end{claim}
\begin{proof}[Proof of Claim \ref{claim:comparison_two}]
This is essentially the result of \cite[Proposition 6.4]{serfaty2020gaussian}, except that our reference measure does not necessarily have a $C^1$ density near the edge of $\La$. This is in fact not a problem as long as we are doing comparisons inside $\Lab$ (i.e. as long as the other measure coincides with $\mm$ outside $\Lab$), but it requires an explanation. 

The proof of \cite[Proposition 6.4]{serfaty2020gaussian} relies on two ingredients:
\begin{enumerate}
       \item \cite[Proposition 6.3]{serfaty2020gaussian} (free energy expansion for general density in a rectangle). This we can import directly as it has nothing to do with our specific setup.
    \item  \cite[Proposition 3.6]{serfaty2020gaussian} (almost additivity of the free energy). It says that one can decompose $\KnL(\mm)$ into two parts: inside/outside $\Omega$, with a small error. This is proven by two inequalities: one is easy and corresponds to the sub-additivity of Neumann energies whereas the other one uses the screening procedure and the local laws on (a neighborhood of) $\Omega$ in order to control the screening error terms (see \cite[Prop. 3.6]{serfaty2020gaussian}). Since the screening procedure takes place in a neighborhood of $\Omega$, the possible singularities near $\partial \La$ are irrelevant. The only adaptation needed is to replace the local laws used in \cite{serfaty2020gaussian} by ours (which is the reason why we “condition” on the event $\EE_\La$).
\end{enumerate}
\end{proof}

\newcommand{\sstar}{s_{\star}}
\subsubsection*{Conclusion 1: smallness of $\Ani$} We may now apply “Serfaty's trick”. Comparing the statements of Claim \ref{claim:compa_transport} and Claim \ref{claim:comparison_two} we see that necessarily:
\begin{multline*}
\log \EnLL \left[\exp\left(\tau \Ani[\psi, \bX, \mm] + \left(\tau |\psi|_{\2} + \tau^2 |\psi|^2_\1 \right) \EnerPoints(\supp \psi) + \tau^2 \ErrorEnergy\right) \big| \EE_\La \right] \\ = \O\left( \ell \log \ell \right).
\end{multline*}
Using the local laws and our assumptions \eqref{eq:assumpsi} on $\psi$ we may control the exponential moments of the error terms:
\begin{multline*}
\log \EnLL \left[\exp\left( \left(\tau |\psi|_{\2} + \tau^2 |\psi|^2_\1 \right) \EnerPoints(\supp \psi) + \tau^2 \ErrorEnergy\right) \big| \EE_\La \right] \\ = \O\left( \tau \Cpsi \ell^{-1} + \tau^2 \Cpsi^2 \ell^{-2} + \Cpsi^2 \tau^2 \ell^{-2} \log \ell  \right).
\end{multline*}
This is valid for all $\tau$ smaller than $\frac{\ell^2}{\Cc \Cpsi \log \ell}$ so that \eqref{condi_sur_tau} are satisfied. Taking $\tau = \ell^{3/2}$, we obtain the following statement:

\begin{lemma}[“The anisotropy is small”]
\label{lemma:Ani_is_Small}
If $\psi$ is a $C^3$ vector field compactly supported on a disk of radius $\ell$ within $\Lab$, and satisfying \eqref{eq:assumpsi}, then we have, for all $\ell$ large enough (depending on the constant $\Cpsi$)
\begin{equation}
\label{eq:momentExpAni}
\EnLL \left[ \exp\left(\ell^{3/2} \Ani(\psi, \bX_\nn, \mm) \right) \big| \EE_\La \right] \leq e^{ \O(\ell \log \ell) +  \O\left(\Cpsi^2 \ell \log \ell\right)},
\end{equation}
with an implicit constant depending on $\beta$ and the “local laws” constant $\Cloc$. In particular, we have the following tail estimate on the distribution of $\Ani(\psi)$: for $\ell$ large enough,
\begin{equation}
\label{eq:tail_estimate_ani_psi}
\PnLL\left[  |\Ani(\psi, \bX_\nn, \mm) | \geq \frac{\log^2 \ell}{\ell^{\hal}} \big| \EE_\La \right] \leq \exp\left( - \ell \log^2 \ell \right).
\end{equation}
\end{lemma}
\begin{remark}
The same analysis holds for the full system, with no need for a conditioning on $\EE_\La$ and as soon as \eqref{condi:LL} is satisfied, so for a broader notion of “bulk”.
\end{remark}

\subsubsection*{Conclusion 2: proof of Proposition \ref{prop:effectPhiEnergy}}
\renewcommand{\ell}{\oldell}
We may now conclude the proof of Proposition \ref{prop:effectPhiEnergy}.

\begin{proof}[Proof of Proposition \ref{prop:effectPhiEnergy}]
Let us return to \eqref{avant_decontrolerAni} and use the fact that by definition (see \eqref{defrt}):
\begin{equation*}
\psi_t = t \Wel(x) + \rt(x), \quad \psi_{-t} = -t \Wel(x) + \rmt(x), 
\end{equation*}
we obtain (for $\bX$ in $\WS(\Lab, \ell, K)$ and thus with probability $1 - \exp(-\ell^2)$ up to choosing $K$ large enough as in Lemma \ref{lem:WS_is_frequent}):
\begin{equation}
\label{precontrolwithWS}
\hal \left(\FL(\Phit \bX, \mm) + \FL(\Phimt \bX, \mm) \right) - \FL(\bX, \mm) = \hal \Ani[\bX, \rt + \rmt] + \O(K^2 \epsilon \log \epsilon t^2).
\end{equation}
Let us introduce dyadic length scales $\ell_i := 2^{i}$ for $0 \leq i \leq \log(T/2)$ and associated cut-off functions $\chi_i$. We decompose $\rt$ as: $\rt = \sum_i \chi_i \rt$. Using Proposition \ref{prop:taut_prop} and in particular \eqref{Psit2X} we see that the vector field $\tpsii := \frac{1}{\epsilon t^2} \chi_i \rt$ satisfies $|\tpsii|_\kk \preceq  \ell_i^{-\kk-1}$. Using Lemma \ref{lemma:Ani_is_Small}, we know that for each $i$ large enough:
\begin{equation*}
\PnLL\left[  |\Ani(\tpsii) | \geq \frac{\log^2 \ell_i}{\ell_i^{\hal}} \big| \EE_\La  \right] \leq \exp\left( - \ell_i \log^2 \ell_i \right).
\end{equation*}
Since this bound is only interesting (probabilistically speaking) for large enough $i$, we use it for $i \geq \log^2 \ell$, in which case we have:
\begin{equation*}
\sum_{i = \log^2 \ell}^{\log(T/2)} \Ani(\tpsii) \leq \sum_{i = \log^2 \ell}^{\log(T/2)} \frac{\log^2 \ell_i}{\ell_i^{\hal}} = \O(\epsilon t^2)
\end{equation*}
with probability $\geq 1 - \sum_{i = \log^2 \ell}^{\log(T/2)} \exp\left( - \ell_i \log^2 \ell_i \right) \geq 1 - \exp\left( - \log^2 \ell \right)$. The first contributions (for $0 \leq i \leq \log^2 \ell$) are controlled using the “rougher” control \eqref{eq:control_AUN}, we get:
\begin{equation*}
\sum_{i=0}^{\log^2 \ell} \Ani(\tpsii) = \O \left( \log \ell^2 \right).
\end{equation*}
By \eqref{eq:ellrhoepsilon} we know that $\log \ell^2$ is comparable to $\log \epsilon$. In conclusion, we obtain:
\begin{equation}
\label{eq:controleanifinal}
\PnLL\left( |\Ani[\bX, \rt + \rmt]| \leq \Cc \epsilon \log \epsilon t^2  \big| \EE_\La  \right) \geq 1 - \exp\left( - \ell^2 \right), 
\end{equation}
with a constant $\Cc$ depending on $\beta$ and the constant $\Cloc$.

Combined with \eqref{precontrolwithWS} which is valid under an event of comparable probability (see Lemma \ref{lem:WS_is_frequent}), we conclude the proof of Proposition \ref{prop:effectPhiEnergy}.
\end{proof}

\subsection{Effect on expectations: proof of Proposition \ref{prop:Quantitative_invariance}}
\renewcommand{\EnLV}{\EnLL}
\renewcommand{\PnLV}{\PnLL}
\label{sec:proof_quanti}
\begin{proof}[Proof of Proposition \ref{prop:Quantitative_invariance}]
For $\tau \in (0, 1)$, let us introduce the event $\Fta$ as:
\begin{equation}
\label{def:Fta}
\Fta := \left\lbrace \ErrAv \leq \tau \right\rbrace
\end{equation}
(recall the definition \eqref{def:HL} of $\ErrAv$). We split the proof into several steps.

\subsubsection*{The case of non-negative functions}
In this paragraph, we make the additional assumption that $\G$ \emph{is non-negative}. We can obviously decompose $\EnLV[ \G(\bX) \big| \EE]$ as:
\begin{equation}
\label{indFtaFtabar}
\EnLV[ \G(\bX) \big| \EE] = \EnLV\left[ \G(\bX) \ind_{\Fta} \big| \EE\right] + \EnLV\left[ \G(\bX) \ind_{\overline{\Fta}} \big| \EE\right].
\end{equation}

\begin{claim}
\label{claim:indFta}
\begin{equation}
\label{bound_indFta}
\EnLV[ \G(\bX) \ind_{\Fta}(\bX) \Big| \EE ]  \leq \EnLV\left[ \hal \left( \G(\Phit \cdot \bX) + \G(\Phimt \cdot \bX) \right) \Big| \EE \right]  \times e^{\beta \tau}.
\end{equation}
\end{claim}
\begin{proof}[Proof of Claim \ref{claim:indFta}]
According to the definition of $\EnLV$ in Section \ref{sec:PnLV}, the expectation $\EnLV[ \G(\bX) \big| \EE ]$ admits the following expression:
\begin{equation*}
\frac{\int_{\La^\nn} \ind_{\EE}(\bX) \ind_{\Fta}(\bX) \G(\bX) e^{- \beta \left( \FL(\bX, \mm) + \nn \sum_{i=1}^\nn \zeta(x_i) \right) } \dd \X}{\int_{\La^\nn} \ind_{\EE}(\bX) e^{- \beta \left( \FL(\bX, \mm) + \nn \sum_{i=1}^\nn \zeta(x_i) \right) } \dd \X}.
\end{equation*}
Since the function $\zeta$ vanishes identically on $\Lab$ and since our localized translations act as the identity outside $\Lab$, the term $\sum_{i=1}^\nn \zeta(x_i)$ will not be affected by the operations below. For simplicity, we omit it altogether. Let us focus on the integral appearing in the numerator. By definition of $\ErrAv$ and of $\Fta$ as in \eqref{def:HL}, \eqref{def:Fta}, we may write:
\begin{equation*}
 \ind_{\Fta}(\bX)  e^{-\beta \FL(\bX, \mm)} =  \ind_{\Fta}(\bX)  e^{-\beta \hal \left(\FL(\Phit \cdot \bX, \mm) + \FL(\Phimt \cdot \bX, \mm) \right) + \beta \ErrAv} \leq  e^{\beta \tau} e^{-\beta \hal \left(\FL(\Phit \cdot \bX, \mm) + \FL(\Phimt \cdot \bX, \mm) \right)},
\end{equation*}
where we haved bounded the indicator function by $1$ in the right-hand side. Using the convexity of $x \mapsto \exp(-\beta x)$ and the fact that $\G$ is assumed to be non-negative, we deduce that:
\begin{multline*}
\int_{\La^\nn} \ind_{\EE}(\bX) \ind_{\Fta}(\bX) \G(\bX) e^{-\beta \FL(\bX, \mm)} \dd \X
\\
\leq \hal \left[ \int_{\La^\nn} \ind_{\EE}(\bX) \G(\bX) e^{-\beta \FL(\Phit \cdot \bX, \mm)} \dd \X + \int_{\La^\nn} \ind_{\EE}(\bX) \G(\bX) e^{-\beta \FL(\Phit \cdot \bX, \mm)} \dd \X \right] \times e^{\beta \tau}.
\end{multline*}
By construction, the change of variable $(x_1, \dots, x_\nn) \mapsto (\Phit(x_1), \dots, \Phit(x_\nn))$, which maps $\bX$ to $\Phit \bX$, has a Jacobian equal to $1$. We thus have, looking at the first integral on the right-hand side:
\begin{equation*}
\int_{\La^\nn} \ind_{\EE}(\bX) \G(\bX) e^{-\beta \FL(\Phit \cdot \bX, \mm)} \dd \X = \int_{\La^\nn} \ind_{\EE}(\Phimt \bX) \G(\Phimt \bX) e^{-\beta \FL(\bX, \mm)} \dd \X.
\end{equation*}
Moreover, since $\Phimt \equiv \Id - t \vu$ on the disk $\DD(0, \ell/4)$ (see Lemma \ref{lemma_propPhi}), since $\G$ is assumed to be $\DD(0, \ell/10)$-local, and since $|t|$ is taken smaller than $\ell/10$, we have $\Phimt\left(\DD(0, \ell/10) \right) \subset \DD(0, \ell/4)$ and thus:
\begin{equation*}
\G(\Phimt(\bX)) = \G(\bX - t \vu).
\end{equation*}
On the other hand since by construction $\Phit$ coincides with the identity map outside $\Lab$ and since the event $\EE$ is assumed to be $\La \setminus \Lab$-local, we have:
\begin{equation}
\label{indEE_ne_le_sent_pas}
\ind_{\EE}(\Phimt(\bX)) = \ind_{\EE}(\bX).
\end{equation}
Hence we can ensure that:
\begin{equation*}
\int_{\La^\nn} \ind_{\EE}(\Phimt \bX) \G(\Phimt \bX) e^{-\beta \FL(\bX, \mm)} \dd \X = \int_{\La^\nn} \ind_{\EE}(\bX) \G(\bX - t \vu) e^{-\beta \FL(\bX, \mm)} \dd \X,
\end{equation*}
and similarly for the other term (reversing the roles of $-t$ and $t$). In conclusion, we obtain:
\begin{equation*}
\int_{\La^\nn} \ind_{\EE}(\bX) \ind_{\Fta}(\bX) \G(\bX) e^{-\beta \FL(\bX, \mm)} \dd \X
\leq\left[ \int_{\La^\nn} \ind_{\EE}(\bX)  \hal \left(\G(\bX + t \vu) + \G(\bX - t \vu) \right)  e^{-\beta \FL(\bX,\mm)} \dd \X \right] \times e^{\beta \tau}.
\end{equation*}
Dividing back by the partition function, we obtain \eqref{bound_indFta}. On the other hand, we have by Cauchy-Schwarz's inequality:
\begin{equation*}
\EnLV[ \G(\bX) \ind_{\overline{\Fta}} (\bX) \Big| \EE] \leq \left(\EnLV\left[ \G^2(\bX) \Big| \EE \right]\right)^\hal  \times \left( \PnLV \left( \left\lbrace \ErrAv \geq \tau \right\rbrace \Big| \EE \right) \right)^\hal.
\end{equation*}
In summary, we have obtained \emph{under the extra assumption that $\G$ is non-negative}, and for all $\tau \in (0,1)$:
\begin{multline}
\label{eq:Ggeq0}
\EnLV[ \G(\bX) \Big| \EE] \leq \EnLV\left[ \hal \left( \G(\Phit \bX) + \G(\Phimt \bX) \right) \Big| \EE \right]  \times e^{\beta \tau} \\
 + \left(\EnLV\left[ \G^2(\bX) \Big| \EE \right]\right)^\hal \times \left( \PnLV \left( \left\lbrace \ErrAv \geq \tau \right\rbrace \Big| \EE \right) \right)^\hal.
\end{multline}
\end{proof}

\subsubsection{The general case.}
We no longer assume that $\G \geq 0$. For $\sigma > 0$, let us introduce the event $\Gsi := \left\lbrace \G + \sigma \geq 0 \right\rbrace$. We can write:
\begin{equation}
\label{general_case_av}
\EnLV[ \G(\bX) + \sigma  \big| \EE] = \EnLV[ \G(\bX) \big| \EE] + \sigma  
= \EnLV[ \left(\G(\bX) + \sigma\right) \ind_{\Gsi}(\bX) \big| \EE] + \EnLV[ \left(\G(\bX) + \sigma\right) \ind_{\overline{\Gsi}}(\bX) \big| \EE]. 
\end{equation}
Since $\bX \mapsto \left(\G(\bX) + \sigma\right) \ind_{\Gsi}(\bX)$ is non-negative by construction and has the same local character as $\G$, we may apply the conclusions of the previous paragraph and write:
\begin{multline}
\EnLV[ \left(\G(\bX) + \sigma\right) \ind_{\Gsi}(\bX) \big| \EE]  \\
\leq \EnLV\left[ \hal \left( \G(\bX + t\vu) + \sigma \right) \ind_{\Gsi}(\bX + t \vu) + \hal \left( \G(\bX - t\vu) + \sigma \right)  \ind_{\Gsi}(\bX - t\vu) \big| \EE \right] \times e^{\beta \tau} \\ 
 +\EnLV\left[ \left(\G(\bX) + \sigma\right)^2 \Big| \EE \right] \times \left( \PnLV \left( \left\lbrace \ErrAv \geq \tau \right\rbrace \Big| \EE \right) \right)^\hal,
\end{multline}
where in the last expectation term we have bounded an indicator function by $1$. We would like to get rid of the two remaining indicator functions in the right-hand side.

Using again Cauchy-Schwarz's inequality (and the definition of $\Gsi$) we see that:
\begin{equation*}
\left| \EnLV[ \left(\G(\bX + t \vu) + \sigma \right) \ind_{\overline{\Gsi}}(\bX + t \vu) \big| \EE] \right| \leq 2 \left( \EnLV\left[ \G^2(\bX + t\vu) \Big| \EE \right]  + \sigma^2 \right)^\hal \times \left(\PnLV \left[ |\G(\bX + t \vu)| \geq \sigma \right]\right)^\hal,
\end{equation*}
and similarly after replacing $t$ by $-t$. Returning to \eqref{general_case_av}, we thus obtain:
\begin{multline*}
\EnLV[ \G(\bX) \big| \EE] + \sigma \leq \left( \EnLV\left[  \hal \left( \G(\bX + t\vu) + \G(\bX - t\vu) \right)  \big| \EE \right] + \sigma \right) \times  e^{\beta \tau}
\\
 + 2 e^{\beta \tau} \left( \EnLV\left[ \G^2(\bX \pm \vu) \Big| \EE \right]  + \sigma^2 \right)^\hal \times \left(\PnLV \left[ |\G(\bX \pm t \vu)| \geq \sigma \right]\right)^\hal 
 \\
  + \EnLV\left[ \G^2(\bX) \Big| \EE \right]  \times  \left(\PnLV \left[ \ErrAv \geq \tau \big| \EE \right]\right)^\hal.
\end{multline*}
We may substract $\sigma$ on both sides, and observe that we have:
\begin{multline*}
\left( \EnLV\left[  \hal \left( \G(\bX + t\vu) + \G(\bX - t\vu) \right)  \big| \EE \right] + \sigma \right) \times  e^{\beta \tau} - \sigma =  \EnLV\left[  \hal \left( \G(\bX + t\vu) + \G(\bX - t\vu) \right)  \big| \EE \right] \\
+ \EnLV\left[  \hal \left( \G(\bX + t\vu) + \G(\bX - t\vu) \right)  + \sigma \big| \EE \right] \left(e^{\beta \tau} - 1 \right).
\end{multline*}
Using Cauchy-Schwarz's inequality one more time to bound the second line, we obtain:
\begin{multline*}
\EnLV[ \G(\bX) \big| \EE] \leq \EnLV\left[  \hal \left( \G(\bX + t\vu) + \G(\bX - t\vu) \right)  \big| \EE \right] + 2 \left( \EnLV\left[ \G^2(\bX \pm t \vu) \Big| \EE \right]  + \sigma^2 \right)^\hal \times \left(e^{\beta \tau} - 1 \right) \\
 + 2 e^{\beta \tau} \left( \EnLV\left[ \G^2(\bX \pm \vu) \Big| \EE \right]  + \sigma^2 \right)^\hal \times \left(\PnLV \left[ |\G(\bX \pm t \vu)| \geq \sigma \right]\right)^\hal 
 \\
  + \EnLV\left[ \G^2(\bX) \Big| \EE \right]  \times  \left(\PnLV \left[ \ErrAv \geq \tau \big| \EE \right]\right)^\hal.
\end{multline*}
Replacing $\G$ by $-\G$ (which is possible as there is no more a sign constraint on $\G$), we obtain the converse inequality, which yields \eqref{eq:quantitative_invariance}.
\end{proof}

\begin{remark}
\label{rem:nofreelunch}
The identity \eqref{indEE_ne_le_sent_pas} requires that $\EE$ be $\La \setminus \Lab$-local. One could try to “pass down” as much information as possible but if $\EE$ is “too rich” then it risks to be perturbed by our localized translation.
\end{remark}

\printbibliography

\end{document}